%% file: main.tex
\documentclass[captions=tableabove]{myDiss} 

\input{defs}

\title{Channel Coding based on Skew Polynomials\\[3pt] and Multivariate Polynomials}
\author{Hedongliang Liu}
\date{}

\head{Prof.\ Dr.-Ing.~Wolfgang Kellerer}
\firstexaminer{Prof.\ Dr.-Ing.~Antonia Wachter-Zeh}
\secondexaminer{Prof.\ Felice Manganiello, Ph.D.}
\datein{14.02.2024} 
\dateout{25.06.2024}

\hypersetup{
  pdftitle={Channel Coding based on Skew Polynomials and Multivariate Polynomials},
  pdfauthor={Hedongliang Liu},
  pdfkeywords={PhD, Dissertation, coding theory, skew polynomials, multivariate polynomials}
}

\logoLetterPlacementtrue

\begin{document}
\frontmatter

\chapter*{Abstract}
\input{abstract}

\chapter*{Acknowledgments}
\input{acknowledgments}

\setcounter{tocdepth}{2}
\tableofcontents
\pagebreak
\chapter*{Nomenclature}
\input{nomenclature}
\section*{Abbreviations}
\input{abbreviations}

\listoffigures
\listoftables

\mainmatter%


\chapter{Motivation and Overview}
\label{chap:overview}
\input{overview}

\chapter{Introduction to Codes based on Polynomials}
\label{chap:intro_polys}
\input{chap_1.tex}

\chapter{Dual-Containing Polycyclic Codes over Rings based on Skew Polynomials}
\label{chap:mod_ring}
\input{chap_mod_skew.tex}

\chapter{Support-Constrained Evaluation Codes based on Skew Polynomials}
\label{chap:eva_skew}
\input{chap_eva_skew.tex}
\chapter{Locally Recoverable Evaluation Codes based on Multivariate Polynomials}
\label{chap:eva_multivar}
\input{chap_mul_poly.tex}

\chapter{Joint Decoding of Interleaved Evaluation Codes}
\label{chap:dec_eva}
\input{chap_dec_eva.tex}

\hphantom{\cite{holzbaur2019decoding,bartz2022rank,maringer2022analysis,huang2022list,porwal2022interleaved,ott2022covering,ott2023geometrical}}
\include{conclusion}

\appendix
\renewcommand*{\chapterformat}{}
\renewcommand*{\chapterheadendvskip}{%
  {%
    \setlength{\parskip}{0pt}%
    \noindent\rule{\linewidth}{1pt}\par%
		}%
		\ORIGchapterheadendvskip%
              }%

\renewcommand*{\chaptermarkformat}{}
\chapter{Appendix}
\input{appendix.tex}
\backmatter

\sloppy 

\chapter{Bibliography}
\printbibliography[notkeyword={Hedongliang},heading=subbibliography,title={References}]
\addcontentsline{toc}{section}{References}
\newpage
\printbibliography[keyword={Hedongliang},keyword={indiss},heading=subbibliography,title={Publications Containing Parts of this Dissertation}]
\addcontentsline{toc}{section}{Publications Containing Parts of this Dissertation}
\newpage
\printbibliography[keyword={Hedongliang},notkeyword={indiss},heading=subbibliography,title={Other Joint Publications}]

\end{document}

%% file: defs.tex
\usepackage{lipsum}

\usepackage[T1]{fontenc}
\usepackage[utf8]{inputenc}

    {\begin{center}%
        \bfseries{Abstract} \end{center}
      \begin{changemargin}{1cm}{1cm} \footnotesize}
    {\end{changemargin}}

\definecolor{TUMBlau}{RGB}{0,101,189} 
\definecolor{TUMBlauDunkel}{RGB}{0,82,147} 
\definecolor{TUMBlauHell}{RGB}{152,198,234} 
\definecolor{TUMBlauMittel}{RGB}{100,160,200} 

\definecolor{TUMElfenbein}{RGB}{218,215,203} 
\definecolor{TUMGruen}{RGB}{162,173,0} 
\definecolor{TUMOrange}{RGB}{227,114,34} 
\definecolor{TUMGrau}{gray}{0.6} 

\definecolor{TUMGruenDunkel}{RGB}{0,124,48} 
\definecolor{TUMRot}{RGB}{196,7,27} 

\definecolor{TUMBlue}{RGB}{0,101,189} 
\definecolor{TUMBlueDark}{RGB}{0,82,147} 
\definecolor{TUMBlueLight}{RGB}{152,198,234} 
\definecolor{TUMBlueMedium}{RGB}{100,160,200} 

\definecolor{TUMIvory}{RGB}{218,215,203} 
\definecolor{TUMGreen}{RGB}{162,173,0} 
\definecolor{TUMGray}{gray}{0.6} 
\definecolor{TUMGrayDark}{gray}{0.3} 

\definecolor{TUMGreenDark}{RGB}{0,124,48} 
\definecolor{TUMRed}{RGB}{196,7,27} 

\definecolor{TUMPink}{RGB}{181,92,165}
\definecolor{TUMPinkDark}{RGB}{155,70,141}
\definecolor{TUMPink1}{RGB}{198, 128, 187}
\definecolor{TUMPink2}{RGB}{214, 164, 206}
\definecolor{TUMPink3}{RGB}{230, 199, 225}
\definecolor{TUMPink4}{RGB}{246, 234, 244}
\definecolor{plotColor1}{RGB}{0,101,189} 
\definecolor{plotColor2}{RGB}{0,124,48} 
\definecolor{plotColor3}{RGB}{196,7,27} 
\definecolor{plotColor4}{RGB}{227,114,34} 
\definecolor{plotColor5}{RGB}{0,82,147} 
\definecolor{plotColor6}{RGB}{162,173,0} 
\definecolor{plotColor7}{gray}{0.3} 

\definecolor{plotColorIA1}{RGB}{0,124,48} 
\definecolor{plotColorIA2}{RGB}{0,82,147} 
\definecolor{plotColorIA3}{RGB}{100,160,200} 
\definecolor{plotColorIA4}{RGB}{152,198,234} 
\definecolor{plotColorIA5}{gray}{0.3} 
\definecolor{plotColorIA6}{RGB}{227,114,34} 
\definecolor{plotColorIA7}{RGB}{196,7,27} 

\definecolor{darkgreen}{rgb}{0,0.7,0}

\newif\ifshowComment
 \showCommentfalse

\newif\ifhighlightChanges
\highlightChangesfalse

\usepackage{amsmath,amsfonts}
\numberwithin{equation}{chapter}
\usepackage{amssymb,mathtools,amsbsy}
\usepackage{bbm}
\usepackage{epsfig}
\usepackage{mathdots,mathabx}
\usepackage{booktabs}
\usepackage{tabularx}
\usepackage{longtable} 
\usepackage{diagbox}
\usepackage{array}
\usepackage{multirow}
\usepackage{arydshln}
\newcolumntype{C}{>{$}c<{$}} 

\usepackage{hyperref}

\usepackage[noend]{algpseudocode}
\usepackage[linesnumbered,ruled,vlined,titlenumbered,resetcount,algochapter]{algorithm2e}
\usepackage{algorithmicx}
\SetKwBlock{Repeat}{repeat}{}

\usepackage{pgfplots}
\pgfplotsset{compat=newest}
\usepgfplotslibrary{fillbetween}
\usepackage{tikz}
\usetikzlibrary{matrix,arrows,arrows.meta,positioning,calc,shapes,tikzmark}

\pgfplotscreateplotcyclelist{sims_list}{%
{violet,solid, thick, mark=o, mark size=1pt},
{teal,solid, thick, mark=star,mark size=1pt},
{cyan,solid, thick, mark=otimes,,mark size=1pt},
{brown,solid, thick, mark=triangle,mark size=1pt},
{olive,solid, thick, mark=o,mark size=1pt},
{orange, solid, thick, mark=diamond,mark size=1pt},
{lime,solid, thick, mark=square,mark size=1pt},
}
\pgfplotscreateplotcyclelist{bounds_list}{
{TUMBlue,dotted, thick, mark=*, mark size=2pt},
{TUMGreenDark,dotted, thick, mark=triangle*,mark size=3pt},
{TUMRed,dotted, thick, mark=square*,mark size=2pt},
{brown,dotted, thick, mark=star,mark size=1pt},
{olive,dotted, thick, mark=o,mark size=1pt},
{orange, dotted, thick, mark=diamond,mark size=1pt},
{lime,dotted, thick, mark=square,mark size=1pt},
}

\tikzset{
  >=stealth',
  mybox_block/.style={rectangle,rounded corners,draw=black, thick,text width=1em,minimum height=2em,minimum width=4.75em,text centered},
}
\usepackage{hf-tikz}

\pgfkeys{tikz/mymatrixenv/.style={decoration={brace},every left delimiter/.style={xshift=5pt, yshift=-1pt},every right delimiter/.style={xshift=-5pt, yshift=-1pt}}}

\pgfkeys{tikz/mymatrix/.style={matrix of math nodes,nodes in empty cells,left delimiter={(},right delimiter={)},inner sep=1pt,outer sep=3pt,column sep=4pt,row sep=4pt,nodes={minimum width=9pt,minimum height=6pt,anchor=center,inner sep=0pt,outer sep=0pt}}}

\pgfkeys{tikz/mymatrixbrace/.style={decorate,thick}}

\newcommand*\mymatrixbraceleft[4]{
    \draw[mymatrixbrace, decoration={raise=4pt}] (#1.west|-#1-#3-1.south west) -- node[left=4pt] {#4} (#1.west|-#1-#2-1.north west);
}
\newcommand*\mymatrixbraceright[4]{
    \draw[mymatrixbrace,decoration={raise=4pt}] (#1.east|-#1-#2-1.north east) -- node[right=7pt] {#4} (#1.east|-#1-#3-1.south east);
}
\newcommand*\mymatrixbracetop[4]{
    \draw[mymatrixbrace] (#1.north-|#1-1-#2.north west) -- node[above=2pt] {#4} (#1.north-|#1-1-#3.north east);
}
\newcommand*\mymatrixbracebottom[4]{
    \draw[mymatrixbrace, decoration={mirror, raise=1pt}] (#1.south-|#1-1-#2.south west) -- node[below=2pt] {#4} (#1.south-|#1-1-#3.north east);
}
\tikzset{style green/.style={
    set fill color=TUMGreenDark!80!lime!20,fill opacity=0.5,
    set border color=TUMGreenDark!60!lime!40,draw opacity=1.0,
  },
  style cyan/.style={
    set fill color=cyan!90!blue!60, draw opacity=0.4,
    set border color=blue!70!cyan!30,fill opacity=0.1,
  },
  style orange/.style={
    set fill color=TUMOrange!40,fill opacity=0.3,
    set border color=TUMOrange!90,  draw opacity=0.8,
  },
  style brown/.style={
    set fill color=brown!70!orange!40, draw opacity=0.4,
    set border color=brown, fill opacity=0.3,
  },
  style purple/.style={
    set fill color=violet!90!pink!20, draw opacity=0.5,
    set border color=violet, fill opacity=0.3,
  },
  kwad/.style={
    above left offset={-0.07,0.23},
    below right offset={0.07,-0.23},
    #1
  },
  pion/.style={
    above left offset={-0.07,0.2},
    below right offset={0.07,-0.32},
    #1
  },
  poz/.style={
    above left offset={-0.03,0.18},
    below right offset={0.03,-0.3},
    #1
  },set fill color/.code={\pgfkeysalso{fill=#1}},
  set border color/.style={draw=#1}
}

\usepackage[backend=bibtex,
style=alphabetic,
minalphanames=3,
maxalphanames=4,
firstinits=true,
sorting=nyvt,
maxbibnames=99
]{biblatex}

\usepackage{caption}
\usepackage{subcaption}
\usepackage{enumitem}

\usepackage{stackengine} 

\usepackage{footnote}
\makeatletter
\newcommand{\footnoteref}[1]{\protected@xdef\@thefnmark{\ref{#1}}\@footnotemark}
\makeatother

\usepackage{nicefrac} 

\usepackage{amsthm}
\usepackage{thm-restate}

\usepackage[capitalise]{cleveref}

\newtheorem{theorem}{Theorem}[chapter]

\newtheorem{definition}{Definition}[chapter]
\newtheorem{remark}{Remark}[chapter]
\newtheorem{lemma}{Lemma}[chapter]
\newtheorem{corollary}{Corollary}[chapter]
\newtheorem{construction}{Construction}[chapter]
\newtheorem{example}{Example}[chapter]

\newtheorem{proposition}{Proposition}[chapter]

\newtheorem{claim}{Claim}

\newcommand{\cA}{\mathcal{A}}
\newcommand{\cB}{\mathcal{B}}
\newcommand{\cC}{\mathcal{C}}
\newcommand{\cD}{\mathcal{D}}
\newcommand{\cE}{\mathcal{E}}
\newcommand{\cF}{\mathcal{F}}
\newcommand{\cG}{\mathcal{G}}

\newcommand{\cI}{\mathcal{I}}
\newcommand{\cJ}{\mathcal{J}}
\newcommand{\cK}{\mathcal{K}}
\newcommand{\cL}{\mathcal{L}}
\newcommand{\cM}{\mathcal{M}}
\newcommand{\cN}{\mathcal{N}}

\newcommand{\cP}{\mathcal{P}}

\newcommand{\cR}{\mathcal{R}}
\newcommand{\cS}{\mathcal{S}}
\newcommand{\cT}{\mathcal{T}}
\newcommand{\cU}{\mathcal{U}}
\newcommand{\cV}{\mathcal{V}}

\newcommand{\cX}{\mathcal{X}}

\newcommand{\cZ}{\mathcal{Z}}

\newcommand{\bA}{\boldsymbol{A}}
\newcommand{\bB}{\boldsymbol{B}}
\newcommand{\bC}{\boldsymbol{C}}
\newcommand{\bD}{\boldsymbol{D}}
\newcommand{\bE}{\boldsymbol{E}}
\newcommand{\bF}{\boldsymbol{F}}
\newcommand{\bG}{\boldsymbol{G}}
\newcommand{\bH}{\boldsymbol{H}}
\newcommand{\bI}{\boldsymbol{I}}

\newcommand{\bM}{\boldsymbol{M}}
\newcommand{\bN}{\boldsymbol{N}}

\newcommand{\bR}{\boldsymbol{R}}
\newcommand{\bS}{\boldsymbol{S}}
\newcommand{\bT}{\boldsymbol{T}}

\newcommand{\bV}{\boldsymbol{V}}

\newcommand{\bX}{\boldsymbol{X}}
\newcommand{\bY}{\boldsymbol{Y}}

\newcommand{\ba}{\boldsymbol{a}}
\newcommand{\bb}{\boldsymbol{b}}
\newcommand{\bc}{\boldsymbol{c}}
\newcommand{\bd}{\boldsymbol{d}}
\newcommand{\be}{\boldsymbol{e}}

\newcommand{\bg}{\boldsymbol{g}}

\newcommand{\bi}{\boldsymbol{i}}

\newcommand{\bm}{\boldsymbol{m}}
\newcommand{\bn}{\boldsymbol{n}}

\newcommand{\bp}{\boldsymbol{p}}

\newcommand{\br}{\boldsymbol{r}}
\newcommand{\bs}{\boldsymbol{s}}

\newcommand{\bu}{\boldsymbol{u}}
\newcommand{\bv}{\boldsymbol{v}}
\newcommand{\bw}{\boldsymbol{w}}
\newcommand{\bx}{\boldsymbol{x}}
\newcommand{\by}{\boldsymbol{y}}

\newcommand{\0}{\boldsymbol{0}}
\newcommand{\1}{\boldsymbol{1}}

\newcommand{\balpha}{\boldsymbol{\alpha}}
\newcommand{\bbeta}{\boldsymbol{\beta}}

\newcommand{\bbF}{\mathbb{F}}

\newcommand{\bbM}{\mathbb{M}}
\newcommand{\bbN}{\mathbb{N}}

\newcommand{\bbR}{\mathbb{R}}

\newcommand{\bbZ}{\mathbb{Z}}

\newcommand{\sfC}{\mathsf{C}}

\newcommand{\sfE}{\mathsf{E}}

\newcommand{\sfH}{\mathsf{H}}

\newcommand{\sfL}{\mathsf{L}}
\newcommand{\sfM}{\mathsf{M}}

\newcommand{\sfQ}{\mathsf{Q}}
\newcommand{\sfR}{\mathsf{R}}
\newcommand{\sfS}{\mathsf{S}}

\newcommand{\sfc}{\mathsf{c}}

\newcommand{\sfs}{\mathsf{s}}

\renewcommand{\ring}{\cA}
\newcommand{\subring}{\cB}

\newcommand{\field}{F}
\newcommand{\ideal}{I}
\newcommand{\GrobBases}{G}
\newcommand{\module}{\cM}
\newcommand{\polyRing}{\cR}
\newcommand{\IntRing}{\bbZ}

\newcommand{\F}{\mathbb{F}}
\newcommand{\Fq}{\mathbb{F}_q}
\newcommand{\Fqm}{\mathbb{F}_{q^m}}

\DeclareMathOperator{\diag}{diag}
\DeclareMathOperator{\supp}{supp}

\DeclareMathOperator{\rk}{rank}
\DeclareMathOperator{\rank}{rank} 

\DeclareMathOperator{\wt}{wt}

\DeclareMathOperator{\gcrd}{gcrd}
\DeclareMathOperator{\lclm}{lclm}
\DeclareMathOperator{\lcm}{lcm}

\DeclareMathOperator{\Id}{id}
\DeclareMathOperator{\lt}{lt} 
\DeclareMathOperator{\lc}{lc} 
\DeclareMathOperator{\lm}{lm} 
\DeclareMathOperator{\lex}{lex}
\DeclareMathOperator{\grlex}{grlex}
\DeclareMathOperator{\grevlex}{grevlex}
\DeclareMathOperator{\mdeg}{mdeg} 
\DeclareMathOperator{\rem}{rem} 
\DeclareMathOperator{\quo}{quo} 

\newcommand{\defeq}{:=}
\newcommand{\defeqrev}{=:}

\newcommand{\card}[1]{\ensuremath{\left|{#1}\right|}}
\newcommand{\ceil}[1]{\ensuremath{\left\lceil #1 \right\rceil}}
\newcommand{\floor}[1]{\ensuremath{\left\lfloor #1 \right\rfloor}}

\renewcommand{\leq}{\leqslant}
\renewcommand{\geq}{\geqslant}

\newcommand{\RS}{\sfR\sfS}
\newcommand{\RM}{\sfR\sfM}
\newcommand{\QC}{\sfQ\sfC}
\DeclareMathOperator{\dist}{\ensuremath{\mathrm{d}}}

\DeclareMathOperator{\wtH}{\wt_{\sfH}} 
\DeclareMathOperator{\dH}{\dist_{\sfH}}

\DeclareMathOperator{\wtR}{\wt_{\sfR}}
\DeclareMathOperator{\dR}{\dist_{\sfR}}
\newcommand{\extbasis}[1]{\mathrm{ext}_{#1}}

\newcommand{\wtSR}[1]{\ensuremath{\wt_{\sfS\sfR,#1}}}
\newcommand{\dSR}[1]{\ensuremath{\dist_{\sfS\sfR,#1}}}

\newcommand{\myspan}[1]{\left\langle #1 \right\rangle}

\newcommand{\Endom}{\theta}
\newcommand{\Frobaut}{\sigma}
\newcommand{\Deriv}{\delta}

\newcommand{\SkewVar}{X}
\newcommand{\RingSkewPolys}{\ring[\SkewVar;\Endom,\Deriv]}
\newcommand{\FrobSkewPolys}{\Fqm[\SkewVar;\Frobaut]}

\newcommand{\GLRS}{\bG^{(\sfL\sfR\sfS)}}
\newcommand{\GLRSi}[1]{\bG^{(\sfL\sfR\sfS)}_{#1}}
\newcommand{\alphas}[1]{\alpha_1,\dots,\alpha_{#1}}
\newcommand{\colMulVec}{\bb}

\newcommand{\betalt}{\beta_{l,t}}
\newcommand{\betal}[1]{\beta_{l,#1}}

\newcommand{\Mod}{\ensuremath{\ \mathrm{mod}\ }}
\newcommand{\modns}{\ensuremath{\mathrm{mod}}} 
\newcommand{\modstar}{\ensuremath{\mathrm{mod}^*}}
\newcommand{\modstarq}{\ (\mathrm{mod}^*\ q)}

\DeclareMathOperator{\gap}{gap_2}
\newcommand{\eps}{\varepsilon}
\newcommand{\quadbinom}[2]{\genfrac{[}{]}{0pt}{}{#1}{#2}}
\DeclarePairedDelimiter\set{\{}{\}}
\DeclarePairedDelimiter\parenv{\lparen}{\rparen}

\newcommand{\PindSet}{\Omega}
\newcommand{\locSet}{\cL}
\newcommand{\rootSet}{\cZ}
\newcommand{\zeroSet}{Z}

\newcommand{\FrobautPolyt}[2]{\Frobaut^{#2}(#1)}

\newcommand{\rowpoly}{f}

\newcommand{\priEle}{\gamma}

\newcommand{\indMap}{\varphi}
\newcommand{\indBlock}{l}

\newcommand{\evapt}[1]{\hat{#1}}

\newcommand{\numRows}{s}
\newcommand{\subRows}{\nu}
\newcommand{\numMulPolyVar}{n}
\newcommand{\mulVarRng}{R_\numMulPolyVar}
\newcommand{\fZt}{\mathsf{f}}
\newcommand{\tfzt}{\tau}
\newcommand{\Skewnk}{\cS_{n,k}}
\newcommand{\stepone}{\emph{(Step 1) }}
\newcommand{\steptwo}{\emph{(Step 2) }}
\newcommand{\matProne}{\emph{(I)}} 
\newcommand{\matPrtwo}{\emph{(II)}}

\newcommand{\FrobPolysn}{\mulVarRng[\SkewVar;\Frobaut]}

\newcommand{\up}{u}
\newcommand{\upc}{u}
\newcommand{\vp}{v}
\newcommand{\vpc}{v}

\newcommand{\mincut}{w}

\newcommand{\intOrd}{s}
\newcommand{\GRS}{\mathsf{GRS}}
\newcommand{\GRSp}{\GRS_{\balpha,\bv}^d}
\newcommand{\GRSallp}{\mathbb{G}_{\balpha}^d}
\newcommand{\ALTallp}{\mathbb{A}_{\balpha}^d}

\newcommand{\code}{\cC}
\newcommand{\kopt}{k_q^{\mathsf{opt}}}

\newcommand{\dGoppa}{d_{\mathsf{Goppa}}}
\newcommand{\kGoppa}{k_{\mathsf{Goppa}}}

\newcommand{\EB}[3]{\mathbb{E}_{#1}^{(#2,#3)}}
\newcommand{\Ebbad}[1]{\mathbb{E}^{#1}_{\mathsf{bad}}}
\newcommand{\vLambda}{\boldsymbol{\Lambda}}
\newcommand{\supbrac}[2]{\ensuremath{#1^{(#2)}}}

\newcommand{\Pmisc}{P_{\mathsf{misc}}}
\newcommand{\Psuc}{P_{\mathsf{suc}}}
\newcommand{\Pfail}{P_{\mathsf{fail}}}

\newcommand{\tmax}{\ensuremath{t_{\max}}}
\newcommand{\tInt}{\ensuremath{t_{\mathsf{Int}}}}
\newcommand{\labelRS}{\mathsf{L.RS}}
\newcommand{\labelMain}{\mathsf{L.A}}
\newcommand{\labelSingleton}{\mathsf{L.A1}}
\newcommand{\labelLz}{\mathsf{L.A2}}
\newcommand{\labelLarge}{\mathsf{L.T}}
\newcommand{\labelMisc}{\mathsf{M}}
\newcommand{\labelLower}{\mathsf{U}}
\newcommand{\labelSim}{\mathsf{SIM}}



%% file: abstract.tex

The exponential growth of data generated nowadays has created a high demand for novel solutions to increase efficiency in communication networks and the reliability of large-scale storage systems.
Error-correcting codes with related properties have been studied intensively in recent years.
Error correction is also essential for the development of quantum computers that can run useful algorithms with negligible miscalculation rate.


This dissertation considers new constructions and decoding approaches for error-correcting codes based on non-conventional polynomials, with the objective of providing new coding solutions to the applications mentioned above.

With skew polynomials, we construct codes that are dual-containing,
which is a desired property of quantum error-correcting codes.
By considering evaluation codes based on skew polynomials, a condition on the existence of optimal support-constrained codes is derived and an application of such codes in the distributed multi-source networks is proposed. For a class of multicast networks, the advantage of vector network coding compared to scalar network coding is investigated.

Multivariate polynomials have been attracting increasing interest in constructing codes with repair capabilities by accessing only a small amount of available symbols, which is required to build failure-resistant distributed storage systems. A new class of bivariate evaluation codes and their local recovery capability are studied. Interestingly, the well-known Reed-Solomon codes are used in a class of locally recoverable codes with \emph{availability} (multiple disjoint recovery sets) via subspace design.

Aside from new constructions, decoding approaches are considered in order to
increase the error correction capability in the case where the code is fixed.
In particular, new lower and upper bounds on the success probability of joint decoding interleaved \emph{alternant} codes by a syndrome-based decoder are derived, where alternant codes are an important class of algebraic codes containing Goppa codes, BCH codes and Reed-Muller codes as sub-classes.


%% file: acknowledgments.tex
This dissertation is based on the research I conducted during my doctoral studies at the Institute for Communication Engineering (ICE) at TUM, within the Coding and Cryptography group led by Antonia Wachter-Zeh. I am deeply grateful to have had the opportunity to be part of such an inspiring and motivating research environment, and I wish to express my heartfelt thanks to everyone who contributed to this incredible journey.

First and foremost, I would like to express my deepest gratitude to my advisor, Antonia, for her exceptional supervision throughout my doctoral training. Antonia has provided me with the freedom to explore a wide range of research topics, fostering my scientific curiosity. Her passion for initiating new collaborations and her insightful guidance have been instrumental in my growth, both academically and personally. I am especially thankful for her continued support from my Master’s studies through to my PhD, her encouragement during challenging moments, and her exemplary role as a researcher, teacher, and leader.

I would also like to extend my sincere appreciation to Moshe Schwartz and Hengjia Wei for their invaluable guidance and collaboration throughout our collaborative projects. Working with them has been an enriching experience, and their knowledge and insights have been vital to my research. I am particular grateful for the warm hospitality from and the fruitful discussions with Moshe, Hengjia and Han Cai during my visit to Ben-Gurion University of Negev, as well as Felix Ulmer, Pierre Loidreau and Delphine Boucher during my stay in Université de Rennes 1.

I want to thank Sven Puchinger for his mentorship and all his insightful inputs through our collaborations in the early stage of my PhD.
I would also like to thank Felice Manganiello for serving as my second examiner, and Wolfgang Kellerer for chairing the examination committee.

Over the past four years, I have had the privilege of collaborating with many talented researchers, and these joint efforts have been both joyful and intellectually rewarding. In particular, I would like to thank my co-authors: Alessandro Neri, Alexander Zeh, Anmoal Porwal, Antonia Wachter-Zeh, Chih-Chang Huang, Cornelia Ott, Frank Kschischang, Felix Ulmer, Georg Maringer, Hannes Bartz, Hengjia Wei, Hugo Sauerbier-Couvée, Ilya Vorobyev, Johan Rosenkilde, Julian Renner, Lukas Holzbaur, Marvin Xhemrishi, Moshe Schwartz, Nikita Polianskii, Rawad Bitar, Sabine Pircher, Sven Puchinger, Tim Janz, Thomas Jerkovits, Violetta Weger, and Volodya Sidorenko. Collaborating with all of you has expanded my knowledge in many meaningful ways.

I am also grateful to Violetta Weger and Stefan Ritterhoff for their help in proofreading parts of this dissertation.

To my office mates, Lorenz Welter, Marvin Xhemrishi, Sven Puchinger, Sabine Pircher, and Stefan Ritterhoff, thank you for the camaraderie and the positive atmosphere we shared. I would also like to acknowledge all members of the institute for creating such a welcoming and vibrate environment. Many thanks to the professors and organizers of events like JWCC, doctoral seminars, and the Instituteausflug -- these events have enriched my experience and fostered a spirit of collaboration. Outside of work, I will always cherish the memories of after-work gatherings, board games and weekend brunches, which added balance to my PhD life.

To my friends Hongdou, Shelton, Johnathan, Tobit, Keyue, Yanqin, and many others in Munich, thank you for pulling me out of the office and home-office from time to time, and for the unforgettable trips we shared. Han and Lisa, your constant support and generous hospitality during my visits to France and Austria have meant so much to me. Chen, Shiying, and Xuanxuan, thank you for our long-distance calls, deep conversations, and the warm welcomes whenever I returned to China—your friendship has been a source of strength throughout this journey.

Lastly, I am deeply thankful to my family, especially my mom, for her unwavering love, support, and care, and my dad, for always being my role model. Their dedication to my education and constant encouragement have been the foundation for everything I have accomplished.
\\[1em]

\noindent
\emph{Hedongliang Liu}\\
Munich, September 2024

%% file: nomenclature.tex
\section*{Sets, Rings and Fields}%
\begin{longtable}{p{4cm}|p{\linewidth-4cm}}
  $[a,b]$ & Set of integers $\{i \ | \ a\leq i\leq b\}$\\
  $[b]$ & Set of integers $\{i \ | \ 1\leq i\leq b\}$\\

  $\ring$ or $\polyRing$ & Ring \\
  $\IntRing_p$ & Integer ring of order $p$\\
  $\field$ & Field (may be infinite) \\
  $\F$ (or $\Fq$) & Finite field (with $q$ elements)\\
  $\Fqm$ & Extension field with $q^m$ elements\\

\end{longtable}

\section*{ Vectors, Matrices and Vector Spaces}%
\begin{longtable}{p{4cm}|p{\linewidth-4cm}}
  $\ba$ & Vector\\
  $\supp(\ba)$ & Set of indices of nonzero positions of $\ba$\\
  $\ba|_{\cI}$ & Vector $\ba$ restricted to the positions indexed by $\cI$\\
  $\ba\star \bb$ & Entry-wise multiplication of vectors $\ba$ and $\bb$\\
  $\diag(\ba)$ & Diagonal matrix with the entries of $\ba$ on its diagonal\\
  $\bA$ & Matrix\\
  $\bA^\top$ & Transpose of the matrix $\bA$\\
  $\supp(\bA)$ & Set of indices of nonzero columns of $\bA$\\
  $\bA|_{\cI}$ & Matrix $\bA$ restricted to the columns indexed by $\cI$\\
  $\ba_i$ & The $i$-th row of $\bA$\\ 
  $\myspan{\bA}$ & Row span of $\bA$\\
  $\cG_q(n,k)$ & Grassmannian of dimension $k$ of $\Fq^n$ (i.e., the set of all $k$-dimensional subspaces of $\Fq^n$)\\
\end{longtable}

\section*{Codes}%
\begin{longtable}{p{4cm}|p{\linewidth-4cm}}
  $\cC$, $\cC^\perp$  & Code (a set of vectors/matrices), Dual code\\
  $[n,k]_q$ & A $\Fq$-linear code of length $n$ and dimension $k$\\
  $[n,k,d]_q$ & A $\Fq$-linear code of length $n$, dimension $k$ and minimum distance $d$\\ 
  $\dH$ & Minimum Hamming distance\\
  $\dR$ & Minimum rank distance\\
  $\dSR{\bn_\ell}$ & Minimum sum-rank distance w.r.t.~an ordered partition $\bn_\ell$ of $n$\\
  $\bG$ & Generator matrix of a code\\
  $\bH$ & Parity-check matrix of a code\\

\end{longtable}

\section*{Polynomials}
\begin{longtable}{p{4cm}|p{\linewidth-4cm}}
  $\ring[x]$ or $\F[x]$ & Univariate polynomial ring in $x$ over a ring $\ring$ or a finite field $\F$\\
  $\polyRing_n$ & Multivariate polynomial ring in $n$ variables\\
  $I=\myspan{f_1,\dots, f_s}$
  & Ideal generated by $f_1,\dots, f_s$\\
  $\Endom$ & Endomorphism \\
  $\Deriv$ & $\theta$-derivation\\
  $\RingSkewPolys$ & Skew polynomial ring in variable $\SkewVar$ over a ring $\ring$\\
  $\gcrd$ & Greatest common right divisor\\
  $\lclm$ & Least common left multiple\\
\end{longtable}


%% file: abbreviations.tex
\begin{longtable}{p{4cm}|p{\linewidth-4cm}}
  AAD & Almost affinely disjoint \\
  AS & Almost sparse \\
  BCH & Bose--{Chaudhuri}--Hocquenghem\\
  GRS & Generalized Reed--Solomon\\
  i.i.d. & Independently and identically distributed\\ 
  LRS & Linearized Reed-Solomon\\
  MDS & Maximum distance separable\\
  MRD & Maximum rank distance\\
  MSRD & Maximum sum-rank distance\\
  QEC & Quantum error-correcting\\
  QLRS & Quadratic lifted Reed-Solomon codes\\
  RM & Reed--Muller\\
  RS & Reed--Solomon\\
\end{longtable}


%% file: overview.tex
Channel coding
originates from the seminal work by Shannon \cite{Shannon}, which laid the mathematical foundation of \emph{reliable communication in the presence of noise}.
The channel coding theorem by Shannon shows that reliable communication is achievable as long as the code rate is below the capacity of the (noisy) channel.
However, the proof of this result is non-constructive and focuses on the asymptotic behavior of error-correcting codes in a probabilistic setting.
On the contrary, the work by Hamming \cite{Hamming-1950} around the similar time extracted the combinatorial basis for the theory of error-correcting codes.
Their works are deeply intertwined and perfectly complementary.
Due to the clear difference between the probabilistic, asymptotic viewpoint of Shannon and the combinatorial, constructive perspective of Hamming, prosperous studies following their footprints are growing into two main respective areas\footnote{We refer the interested reader to \cite{slepian1974key} and \cite{berlekamp1974key} for the influential papers in the development of the respective areas.}: \emph{information theory} and \emph{coding theory}.
The former focuses on characterizing the capacity (i.e., asymptotically achievable code rate) for various channel models with different statistical behaviors, while the cornerstone of the latter is the finite behavior of codes for scenarios where error correction is needed.

Many well-known and widely used codes are based on polynomials.
For instance, Reed-Muller codes, used in deep-space communication~\cite{massey1992deep}, wireless communications~\cite{arikan2008performance,mondelli2014polar} and probabilistic checkable proofs in computational complexity theory~\cite{arora1997improved,sudan1999pseudorandom}, can be described as low-degree multivariate polynomials.
BCH codes, used in satellite communication~\cite{cheung1988phobos} and solid-state drives~\cite{micheloni2013bch}, are suitable for implementations on small and low-power hardwares because of their underlying polynomial structure.
Nowadays, error correction is not only used for communications, but also in abundant scenarios ranging from digital data storage in the daily life to the frontier research on quantum computing.

As the exponential growth of data generated and exchanged nowadays, large-scale distributed storage systems are needed to store vast amounts of data.
The main goal of such systems is to guarantee the integrity of the stored data, i.e., to protect the data from loss even if some storage disks are defective. Instead of simple replication, several distributed storage systems have utilized error-correcting codes to provide reliable services,
e.g., Facebook's f4 storage system~\cite{muralidhar2014f4}, Baidu's Atlas Cloud Storage~\cite{lai2015atlas}, Hadoop~\cite{hadoop2023hdfs} and Backblaze Vaults cloud storage~\cite{beach2019backblaze} use Reed-Solomon codes.

  In quantum computing, a \emph{qubit} is the basic unit of quantum information that can carry richer states beyond just $0$ and $1$.
  The challenge is that the qubits are so sensitive that even stray light or slight temperature change can cause errors~\cite{chuang1995quantum}. The state-of-the-art quantum processors typically have error rates around $10^{-3}$ per interaction between \emph{physical} qubits~\cite{foxen2020demonstrating,wu2021strong}, which is far beyond the error rate required to run useful algorithms.
  Quantum error-correcting (QEC) codes are proposed to suppress error rates of calculation by constructing \emph{logical} qubits, where each logical qubit is composed of multiple physical qubits (i.e., by adding redundancy to reduce the error rate per logic operation).
  \emph{Surface codes}~\cite{kitaev2003fault} have been thoroughly studied for QEC architectures and have been demonstrated in small examples by teams at IBM~\cite{chen2022calibrated,sundaresan2023demonstrating} and Google Quantum AI~\cite{google2023suppressing}.
  However, surface codes have the drawback that they require too many physical qubits, possibly 200 million qubits for problems of interest, which makes them impractical due to the cost and complexity.
  Recently, low-density-parity-check (LDPC) codes were proposed as a promising candidate for QEC codes as they feature a more than ten-fold reduction in the number of physical qubits compared to surface codes under similar error rate level~\cite{bravyi2023high}. Advancements in finding codes with better code rate and fast decoding algorithms suitable for quantum circuits are still highly demanded.

This dissertation intends to provide new constructions from non-conventional polynomials and decoding approaches for error-correcting codes with the desired properties in the aforementioned applications.
The structure of this dissertation is as follows.

\textbf{\cref{chap:intro_polys}} provides the basics of the polynomials and the metrics used in the remaining chapters.
We first give a brief introduction of multivariate polynomials and present a powerful tool, Gr\"obner basis, for solving polynomial equation systems.
We then introduce skew polynomials and their properties.
Finally, we cover the Hamming, the rank and the sum-rank metrics. 

In \textbf{\cref{chap:mod_ring}} we construct \emph{dual-containing} codes over rings based on skew polynomials. We first define \emph{$(\Endom,\Deriv)$-polycyclic codes} (in short, $(\Endom,\Deriv)$-codes)
and derive a parity-check matrix of this class of codes within the framework of skew polynomials.
Based on the properties of skew polynomials and dual-containing codes, we develop an algorithm to compute all Euclidean-/Hermitian-dual-containing $(\Endom,\Deriv)$-codes constructed from skew polynomials
and apply this algorithm to several rings $\ring$ of order 4.
Moreover, we give an algorithm to test whether the dual code is also a $(\Endom,\Deriv)$-code and apply it to the resulting dual-containing codes found by the previous algorithm.

\textbf{\cref{chap:eva_skew}} is devoted to a class of evaluation codes of skew polynomials, \emph{linearized Reed-Solomon (LRS)} codes, and network coding. LRS codes are \emph{maximum sum-rank distance (MSRD)} codes.
Motivated by the practical and theoretical interest in support-constrained codes,
we derive a necessary and sufficient condition on the existence of an MSRD code fulfilling certain support constraints and give an upper bound on the field size to construct such a code.
With the help of the condition, we develop a scheme to design \emph{distributed LRS codes}
for multi-source networks.
The second focus of this chapter is the advantage of vector network coding versus scalar network coding for a family of multicast networks, \emph{generalized combination networks}.
The task of this multicast network coding problem is to find the coding coefficients of the relay nodes at the middle layer of the network, so that all the receivers that connect to a fixed number of middle layer nodes can decode all the messages.
The \emph{solution} of such a network is the set of coding coefficients at each relay nodes.
We investigate the advantage
by bounding the gap between the minimum required alphabet size of the scalar solutions and the vector solutions.

\textbf{\cref{chap:eva_multivar}} deals with codes with local properties constructed from multivariate polynomials.
We first propose a class of bivariate evaluation codes, so called \emph{quadratic lifted Reed-Solomon (QLRS)} codes, where the codeword symbols whose coordinates lie on a quadratic curve form a local recovery set, so that any missing symbol in the set can be recovered within the set.
We study the dimension and minimum Hamming distance of the QLRS codes and compare them to other multivariate evaluation codes, \emph{linearized Reed-Solomon codes}, in terms of the performance in local recovery.
As the second part of this chapter, we investigate an \emph{almost affinely disjoint} (AAD) family of subspaces which is motivated by \emph{batch} codes, a class of locally recoverable codes with availability.
The subspaces in an AAD family form a partial spread where any affine transformation of any subspace in the family intersects with only a few other subspaces in the family.
We give a construction for the AAD family using the best-known evaluation codes -- Reed-Solomon codes.
Aside from the explicit construction, we also provide upper and lower bounds on the cardinality of this family.

\textbf{\cref{chap:dec_eva}} concerns joint decoding of \emph{interleaved} evaluation codes. A codeword of an interleaved code can be seen as $\intOrd$ parallel codewords from linear codes. When $\intOrd$ additive errors have a common support with restricted size, we can decode beyond half the minimum Hamming distance of the code with high probability.
\emph{Alternant codes} are subfield subcodes of Reed-Solomon. They contain Goppa codes and BCH codes as sub-classes.
We apply the
Schmidt-Sidorenko-Bossert joint decoding algorithm, which is known for decoding interleaved Reed-Solomon codes, to interleaved alternant codes, and derive a necessary and sufficient condition such that this algorithm succeeds.
Based on this condition, we derive lower and upper bounds on the success probability of decoding interleaved alternant codes by the Schmidt-Sidorenko-Bossert decoder.
Moreover, we briefly summarize the results on joint decoding of interleaved \emph{generalized Goppa codes} and on improvements in decoding radius by utilizing list decoding for interleaved alternant codes.



%% file: chap_1.tex
\noindent

Polynomials were firstly considered for error control by David E.~Muller \cite{muller1954application} for simplifying switching circuits with multiple outputs via \emph{polynomial representations} in \emph{Boolean algebra} and by Irving S.~Reed \cite{Reed1954AClass}, who exhibited the ability of Muller's polynomial codes to correct multiple errors
and proposed the first efficient decoding algorithm.
Reed-Muller codes can be described as evaluations of low-degree multivariate polynomials.
Their works not only constructed one of the oldest classes of codes
that have been extensively studied, but also brought some preliminary indications about a large body of knowledge about finite algebraic structures (rings, fields, vector spaces) that could be used for error correction.
Since then, various well-known code constructions based on polynomials have appeared.
For instance, Reed-Solomon codes \cite{reed1960polynomial} can be seen as a set of low-degree univariate polynomials and 
BCH codes \cite{Hocquenghem_1959,Bose_RayChaudhuri_1960} can be seen as principle ideals in a quotient polynomial ring.

This chapter gives an introduction to codes constructed from polynomials.
\cref{sec:notations} provides basic notations used in this thesis.
\cref{sec:mul-var-polys,sec:skew-polys} contain the basics on multivariate polynomials and skew polynomials that concern most of this thesis.
\cref{sec:block-codes} provides two methods of constructing linear block codes from polynomials.
In \cref{sec:metrics}, we present three metrics for measuring the error-correction capability of a code.

\section{Basic Notations}
\label{sec:notations}

Denote by $[a,b]$ the set of integers $\{a, a+1,\dots, b-1, b\}$, and $[b]:=[1,b]$. 
Let $\bbN$ be the set of nonnegative integers.
For any set $\cA$, denote by $\cA^*\defeq\cA\setminus\set*{0}$ the set of all the nonzero elements in $\cA$.
A ring $\cA$ is \emph{unitary} if there exists $1\in \ring^*$, such that $1\cdot a=a\cdot 1=a, \forall a\in \ring$. A ring $\cA$ is \emph{commutative} if $ab=ba, \forall a,b \in\ring$.
Denote by $\F$ a finite field, by $\F_q$ a finite field of size $q$, and by $\Fqm$ the extension field over $\Fq$ of extension degree $m$.
The integer ring of size $q$ is denoted by $\mathbb{Z}_q$.
Note that for a prime $p$, $\bbZ_p=\F_p$.
Denote by $[n,k]_{q}$ a linear block code of length $n$ and dimension $k$ over an alphabet of size $q$.
If the minimum distance $d$ of the code is also of importance, we denote the code by $[n,k,d]_q$.

Given two vectors $\ba=(a_1,\dots,a_n),\bb=(b_1,\dots, b_n)$, we denote the entry-wise multiplication of $\ba$ and $\bb$ by $\ba\star\bb\defeq (a_1b_1, a_2b_2, \dots, a_nb_n)$.
For a vector $\ba$ of length $n$, 
we denote by $\supp(\ba)$ the set of indices of the nonzero entries of $\ba$ and by $\diag(\ba)$ the $n\times n$ diagonal matrix with the entries of $\ba$ on its diagonal.
Given a set $\cE\subseteq[n]$, we denote by $\ba|_{\cE}$ the restriction of $\ba$ to the entries indexed by the set $\cE$.
For an $m\times n$ matrix $\bE$, we denote by $\supp(\bE)$ the set of indices of the nonzero columns of $\bE$ and by $\be_{i}$ the $i$-the row of $\bE$.
Given a set $\cE\subseteq[n]$, we denote by $\bE|_{\cE}$ the restriction of $\bE$ to the columns indexed by $\cE$.

Denote by $\cG_q(n,k)$ the \emph{Grassmannian} of dimension $k$, which is a set of all $k$-dimensional subspaces of $\bbF_q^n$. The cardinality of $\cG_q(n,k)$ is the well-known \emph{$q$-binomial coefficient}:
\begin{equation*}
  |\cG_q(n,k)|=\quadbinom{n}{k}_q \defeq \prod\limits_{i=0}^{k-1} \frac{q^n-q^i}{q^k-q^i}=\prod\limits_{i=0}^{k-1} \frac{q^{n-i}-1}{q^{k-i}-1}\ .
\end{equation*}

Given a ring $\ring$, we denote by $\ring[x]$ the univariate commutative polynomial ring in variable $x$ with coefficients from $\ring$.
The \emph{degree} of a nonzero polynomial
$f=\sum_{i\in\bbN}a_ix^i\in\ring[x]$
is $\deg(f)\defeq \max \set*{i\in\bbN\ |\ a_{i}\neq 0}$.
We use the convention that the degree of a zero polynomial is defined as $-\infty$.
The leading coefficient of $f$ is denoted by $\lc(f)$.
A polynomial $f\in\ring[x]$ is \emph{monic} if $\lc(f)=1$.
For $g,f\in\ring[x]$, we denote by $g\divides f$ if $g$ divides $f$, by $\lcm(g,f)$ the least common multiplier of $g$ and $f$, and by $\gcd(g,f)$ the greatest common divisor of $g$ and $f$.

Throughout the thesis, the indices
start from $1$.
For convenience, the indices of the coefficients of polynomials start from $0$. Hence, for vectors associated with polynomials, the indices of their entries are corresponding to the polynomials, i.e., the coefficient vector $(a_0,a_1,\dots,a_d)$ corresponds to the polynomial $f=\sum_{i=0}^da_ix^i$.

\section{Multivariate Polynomials}
\label{sec:mul-var-polys}
In this thesis, we use multivariate polynomials in a wide range. Hence, we introduce the basics of multivariate polynomials and a powerful tool -- \emph{Gr\"obner bases} -- in solving polynomial equations in this section.

Given a field $\field$ (may be infinite), we denote by $\polyRing_n=\field[x_1,\dots,x_n]$ the commutative polynomial ring in $n$ variables over $\field$.
We associate a vector $\bd = (d_1,\dots, d_n)\in\bbN^n$ to the exponents of a monomial by
\begin{align*}
  \bx^{\bd} = x_1^{d_1}x_2^{d_2}\cdots x_n^{d_n} \in\polyRing_n\ .
\end{align*}

\begin{definition}[Monomial order]
  A \emph{monomial order} in $\polyRing_n=\field[x_1,\dots, x_n]$ is a relation $\prec$ on $\bbN^n$ such that
  \begin{itemize}
  \item for all $\ba,\bb\in\bbN^n$, either $\ba=\bb$ or $\ba\prec\bb$ or $\bb\prec\ba$,
  \item for all $\ba,\bb,\bc \in\bbN^n$, $\ba\prec\bb \implies \ba+\bc \prec \bb+\bc$,
  \item for all $\ba\in\bbN^n$, $\0\prec\ba$ or $\0=\ba$.
  \end{itemize}
\end{definition}
The following orders are monomial orders \cite[Theorem 21.6]{von2013modern}:
\begin{itemize}
\item \emph{Lexicographic order:}\label{item:lexicographic-order}
  \begin{align*}
\ba\prec_{\lex} \bb \iff \textrm{ the leftmost nonzero entry in } \ba-\bb \textrm{ is negative.}
  \end{align*}
\item \emph{Graded lexicographic order:}
  \begin{align*}
    \ba \prec_{\grlex} \bb \iff \sum_{i=1}^n a_i < \sum_{i=1}^nb_i \textrm{ or } \parenv*{ \sum_{i=1}^n a_i = \sum_{i=1}^nb_i \textrm{ and } \ba \prec_{\lex} \bb}\ .
  \end{align*}
\item \emph{Graded reverse lexicographic order:}
  \begin{align*}
    \ba \prec_{\grevlex} \bb \iff &\sum_{i=1}^n a_i < \sum_{i=1}^nb_i \textrm{ or } \big( \sum_{i=1}^n a_i = \sum_{i=1}^nb_i \textrm{ and the rightmost}\\
    &\textrm{nonzero entry in }\ba-\bb\in\bbZ^n \textrm{ is positive}\big).
  \end{align*}
\end{itemize}

Since there are multiple variables in a multivariate polynomial, the properties of a polynomial (such as ``degree'') are different from those of a univariate polynomial. We give the formal definitions of these properties in the following.
\begin{definition}
  Let $\field$ be a field, $\polyRing_n=\field[x_1,\dots,x_n]$ be a polynomial ring, $f=\sum_{\bd\in\bbN^n}c_{\bd}\bx^{\bd} \in\polyRing_n$ be a nonzero polynomial with all $c_{\bd}\in\field$, $\bd\in\bbN^n$,
  and $\prec$ be a monomial order.
  \begin{itemize}
  \item Each $c_{\bd}\bx^{\bd}$ with $c_{\bd}\neq 0$ is a \emph{term} of $f$.
  \item The \emph{(total) degree} of $f$ is $\deg(f) \defeq \max_{\bd\in\bbN^n}\set*{\sum_{i=1}^nd_i \ |\ c_{\bd}\neq 0}\in\bbN$.
  \item The \emph{multidegree} of $f$ is $\mdeg(f) \defeq \max_{\prec}\set*{\ \bd\ |\ c_{\bd}\neq 0}\in\bbN^n$, where $\max_{\prec}$ is the maximum with respect to $\prec$.
  \item The \emph{$x_i$-degree} of $f$ is $\deg_{x_i} (f) \defeq \max_{\bd\in\bbN^n}\set*{\ d_i\ |\ c_{\bd}\neq 0}\in\bbN$.
  \item The \emph{leading coefficient} of $f$ is $\lc(f)\defeq c_{\mdeg(f)}\in\field^*$.
  \item The \emph{leading monomial} of $f$ is $\lm(f) \defeq \bx^{\mdeg(f)}\in\polyRing$.
  \item The \emph{leading term} of $f$ is $\lt(f) \defeq \lc(f)\cdot \lm(f)\in\polyRing$.
  \end{itemize}
\end{definition}
Addition, subtraction, multiplication and division between two polynomials in $\polyRing_n$ follow naturally from univariate polynomials.
We present in \cref{algo:Multivariate-division} a division algorithm with multiple divisors, which is used in Buchberger's algorithm (\cref{algo:Buchberger}) presented later to compute a Gr\"obner basis.
\begin{algorithm}[htb!]
  \caption{Multivariate division algorithm (cf.~\cite[Algorithm 21.11]{von2013modern})}
  \label{algo:Multivariate-division}
  \KwIn{Nonzero polynomials $f, f_1,\dots,f_s\in\polyRing_n=\field[x_1,\dots,x_n]$ where $\field$ is a field; a monomial order $\prec$ on $\polyRing_n$.}
  \KwOut{Quotients $q_1,\dots,q_s\in\polyRing_n$ and remainder $r\in\polyRing_n$ such that $f= q_1f_1+\dots+q_sf_s+r$ and $\forall i\in[s]$, $\lt(f_i)\not\divides r$.}
  $r\gets 0$, $p\gets f$, $q_i\gets 0, \forall i\in[s]$\;
  \While{$p\neq 0$}{
    \If{for some $i\in[s]$, $\lt(f_i) \divides\lt(p)$ \label{step:fi-divides-f}}{
      $q_i\gets q_i+\tfrac{\lt(p)}{\lt(f_i)}$, $p\gets p - \tfrac{\lt(p)}{\lt(f_i)}f_i$
    }
    \Else{
      $r\gets r+\lt(p)$, $p\gets p-\lt(p)$
    }
  }
  \Return{$q_1,\dots, q_s, r$}
\end{algorithm}

The output of this kind of division may not be unique, since there may be more than one $i\in[s]$ such that $\lt(f_i)$ divides $\lt(p)$ at \cref{step:fi-divides-f}.
In \cref{algo:Multivariate-division}, if we always chooses the smallest possible $i$ at \cref{step:fi-divides-f}, then the quotients $q_1,\dots, q_s$ and the remainder $r$, denoted by
\begin{align*}
  (q_1,\dots,q_s)&=f\ \quo\ (f_1,\dots,f_s)\ ,\\
  r&=f\ \rem\ (f_1,\dots,f_s)\ ,
\end{align*}
are uniquely determined.

\begin{theorem}[Combinatorial Nullstellensatz~{\cite[Theorem 1.2]{alon1999combinatorial}}]
  \label{lem:nullstellensatz}
  Let $\field$ be an arbitrary field and $f$ be a nonzero polynomial in $\field[x_1,\dots,x_n]$ of total degree $\deg(f)=\sum_{i=1}^n t_i$, where $t_i\in\bbN, \ \forall i$.
  Then, if $\cX_1,\dots, \cX_\numMulPolyVar$ are subsets of $\field$ with $|\cX_i|> t_i$, then there are $\evapt{x}_1\in \cX_1,\dots, \evapt{x}_n\in \cX_n$ so that
  \begin{align*}
    f(\evapt{x}_1,\dots,\evapt{x}_n)\neq 0\ .
  \end{align*}
\end{theorem}

\subsection{Ideals and Variety}
\label{sec:ideals-variety}
Let $f_1,\dots, f_s$ be polynomials in $\polyRing_n$. The polynomials generate
an \emph{ideal} in $\polyRing_n$
\begin{align*}
  \ideal = \myspan{f_1,\dots,f_s} \defeq \set*{\left. \sum_{i=1}^s p_if_i\ \right|\ p_i\in\polyRing_n}\ .
\end{align*}
An ideal is \emph{principle} if it is generated by a single element of the ring. For example, $\ideal=\myspan{f_1}\subseteq \polyRing_n$ is a principle ideal.

The \emph{variety} of $\ideal$ is
\begin{align*}
  V(\ideal)\defeq \set*{\bu\in\field^n\ |\ f(\bu)=0,\ \forall f\in \ideal}=\set*{\bu\in\field^n\ |\ f_1(\bu)=\cdots =f_s(\bu)=0}\ .
\end{align*}
We also write $V(f_1,\dots, f_s)$ instead of $V(\myspan{f_1,\dots,f_s})$ for short.
It can be readily seen that the variety $V(\ideal)$ is the set of all solutions to the \emph{system of polynomial equations} $\set*{f_1=0,\dots,f_s=0}$.

Interesting questions about $\ideal$ and $V(\ideal)$ that also concern solving the system of polynomial equations include:
\begin{itemize}
\item How ``big'' is $V(\ideal)$? Is $V(\ideal)\neq \varnothing$?
\item Ideal membership problem: given $f\in\polyRing_n$, is $f\in \ideal$?
\end{itemize}

The famous Hilbert's \emph{Nullstellensatz} \cite{hilbert1893vollen} says the following: if $\field$ is algebraically closed, then, for $f\in\ideal=\myspan{f_1,\dots, f_s}$,
there is an integer $e\in\bbN$ such that $f^e\in\ideal$.
This implies that for any ideal $\ideal$ over an algebraically closed field, the variety $V(\ideal)=\varnothing$ if and only if $1\in \ideal$.

For $n=1$, $\polyRing_1=\field[x]$, the ideal membership problem is easy to check. Let $g = \gcd(f_1,\dots, f_s)$ be the greatest common divisor of $f_1,\dots, f_s$. Then, the ideal $\ideal=\myspan{f_1,\dots, f_s} = \myspan{g}$ \cite[Proposition 1.3.8]{adams2022introduction}. Hence, for any $f\in\polyRing_n$, $f\in\ideal$ if and only if $g \divides f$.
For $n\geq 2$ and $s=1$, the ideal membership problem can be solved by \cref{algo:Multivariate-division}: $f\in\myspan{f_1}$ if and only if $f\ \rem\ f_1=0$. However, this method fails in general for $s\geq 2$.
The next subsection introduces a special type of bases of an ideal where the ideal membership can be easily (only conceptually, not computationally) determined (see \cref{thm:IMbyGB}). These special bases are the analogue to the greatest common divisor for multivariate polynomials.

\subsection{Gr\"obner Bases}
\label{sec:grobBases}
A \emph{Gr\"obner basis} of an ideal $\ideal$ is a special ``basis'' for $\ideal$, in which the questions about $V(\ideal)$ mentioned in \cref{sec:ideals-variety} are easy to answer.
Heisuke Hironaka introduced in \cite{hironaka1964resolution} a special type of basis for polynomial ideals, called ``standard basis''. Bruno Buchberger invented them independently in his dissertation~\cite{buchberger1965algorithm} and named them as Gr\"obner bases after his advisor Wolfgang Gr\"obner.

\begin{definition}[Gr\"obner basis]
  Let $\prec$ be a monomial order  and $\ideal\subseteq \polyRing_n$ be an ideal. A finite set $\GrobBases\subseteq  \ideal$ is a \emph{Gr\"obner basis} for $\ideal$ with respect to $\prec$ if $\myspan{\lt(\GrobBases)}=\myspan{\lt(\ideal)}$, where
  $\lt(\GrobBases)\defeq\set*{\left. \lt(g)\ \right|\ g\in\GrobBases}$ and $\lt(\ideal)\defeq\set*{\left. \lt(g)\ \right|\ g\in\ideal}$.
\end{definition}
For a polynomial ring $\polyRing_n=\field[x_1,\dots,x_n]$, every ideal $\ideal\subseteq \polyRing$ has a Gr\"obner basis \cite[Corollary 21.26]{von2013modern}.
The following theorem shows that given a Gr\"obner basis of an ideal, we can solve the ideal membership problem.
\begin{theorem}[{\cite[Theorem 21.28]{von2013modern}}]
  \label{thm:IMbyGB}
  Let $\GrobBases$ be a Gr\"obner basis for the ideal $\ideal\subseteq \polyRing_n$ with respect to a monomial order $\prec$, and $f\in\polyRing_n$. Then
  \begin{align*}
    f\in \ideal \iff f\ \rem\ \GrobBases = 0\ .
  \end{align*}
\end{theorem}

To present the Buchberger's algorithm that computes a Gr\"obner basis of an ideal, we need the following definition of an $S$-polynomial.
\begin{definition}[$S$-polynomial]
  Let $g,h\in\polyRing_n$, $\ba=\mdeg(g)$, $\bb=\mdeg(h)$, and $\bc= (\max\set*{a_1,b_1},\dots,\max\set*{a_n,b_n})$. The \emph{$S$-polynomial} of $g$ and $h$ is
  \begin{align*}
    S(g,h) = \frac{\bx^{\bc}}{\lt(g)}g - \frac{\bx^{\bc}}{\lt(h)}h\ \in\polyRing_n\ .
  \end{align*}
\end{definition}
The following theorem shows the importance of the $S$-polynomials for computing a Gr\"obner basis.
\begin{theorem}[{\cite[Theorem 21.31]{von2013modern}}]
  A finite set $\GrobBases=\set*{g_1,\dots, g_s}\subseteq\polyRing_n$ is a Gr\"obner basis of the ideal $\myspan{\GrobBases}$ if and only if
  \begin{align*}
    S(g_i,g_j)\ \rem\ (g_1,\dots, g_s) = 0 \textrm{ for } 1\leq i<j\leq s\ .
  \end{align*}
\end{theorem}
We now present a simplified version of Buchberger's algorithm \cite{buchberger1965algorithm} in \cref{algo:Buchberger}.
\begin{algorithm}
  \caption{Buchberger's algorithm for Gr\"obner basis computation (cf.~\cite[Algorithm 21.33]{von2013modern})}
  \label{algo:Buchberger}
  \KwIn{Nonzero polynomials $f_1,\dots,f_s\in\polyRing_n$, and a monomial order $\prec$.}
  \KwOut{A Gr\"obner basis $\GrobBases\subseteq \polyRing_n$ for the ideal $\ideal = \myspan{f_1,\dots, f_s}$ with respect to $\prec$.}
  $\GrobBases\gets \set*{f_1,\dots,f_s}$\;
  \Repeat{
    $\cS\gets \varnothing$\;
    order the elements in $\GrobBases$ as $g_1,\dots, g_t$ according to $\prec$\;
    \ForEach{$i\in[t-1], j\in[i+1,t]$}{
      $r\gets S(g_i,g_j)\ \rem\ (g_1,\dots, g_t)$\tcc*[r]{Apply \cref{algo:Multivariate-division}}
      \If{$r\neq 0$}{
        $\cS\gets \cS\cup \set*{r}$
      }
    }
    \If{$\cS=\varnothing$}{\Return{$\GrobBases$}}
    \Else{$\GrobBases\gets \GrobBases\cup \cS$}
  }
\end{algorithm}

The extended Euclidean algorithm for computing the greatest common divisor (gcd) of univariate polynomials in $\field[x]$ is a special case of Buchberger's algorithm.
A proof of the correctness of \cref{algo:Buchberger} can be found in \cite[Theorem 21.34]{von2013modern}.
In general, the Gr\"obner basis computed by Buchberger's algorithm is neither unique nor of minimal size. However, one can further process the polynomials in $\GrobBases$ to obtain a unique \emph{reduced Gr\"obner basis}, which is defined as follows. Such a unique basis exists for every ideal \cite[Theorem 21.38]{von2013modern}.
\begin{definition}
  A subset $\GrobBases\subseteq \polyRing_n$ is a \emph{minimal} Gr\"obner basis of $I=\myspan{G}$ if it is a Gr\"obner basis for $I$ and for all $g\in\GrobBases$
  \begin{enumerate}
  \item $\lc(g) = 1$ ,
  \item $\lt(g) \not\in \myspan{\lt(\GrobBases\setminus \set*{g})}$ .
  \end{enumerate}
  An element $g$ of a Gr\"obner basis $\GrobBases$ is \emph{reduced with respect to $\GrobBases$} if no monomial of $g$ is in $\myspan{\lt(\GrobBases\setminus\set*{g})}$.
  A minimal Gr\"obner basis $\GrobBases$ of an ideal $I$ is \emph{reduced} if all its elements are reduced with respect to $\GrobBases$.
\end{definition}

To understand the complexity of Buchberger's algorithm, we need to know the maximal total degree of a polynomial occurring during the computation of a Gr\"obner bases and the number of polynomials in the Gr\"obner basis.
The choice of the monomial ordering is critical to these values.
Buchberger investigated in \cite{buchberger1983note} the maximal (total) degree and the number of polynomials occurring in a Gr\"obner basis $\GrobBases$ of the ideal $\ideal=\myspan{f_1,\dots, f_s}$ for a finite set of bivariate polynomials $\cF=\set*{f_1,\dots, f_s}\subseteq \field[x_1,x_2]$.
In the case of $\prec_{\grlex}$, the maximum degree of the polynomials in $G$ is
$2\cdot\max_{f\in\cF}\deg(f)-1$.
In the case of $\prec_{\lex}$, the maximum degree of the polynomials in $G$ is
$\max_{f\in\cF}\deg(f)^2$.
For any valid monomial ordering, the number of polynomials in $\GrobBases$ is $|\GrobBases|=\min_{f\in\cF} \deg(\lt(f))+1$.
From a practical point of view, in the \href{https://magma.maths.usyd.edu.au/magma/handbook/text/1259}{computation of a Gr\"obner basis in Magma} \cite{Magma} for instance, $\prec_{\grevlex}$ is recommended for faster computation while $\prec_{\lex}$ is hard for computation, though it usually presents the most information about the ideal.

The worst-case cost of Buchberger's algorithm is still unknown today.
K\"uhnle and
Mayr \cite{kuhnle1996exponential} presented an algorithm for computing the unique reduced Gr\"obner basis for a given ideal, which requires exponential space. This gives a lower bound on the worst-case cost of Buchberger's algorithm and concludes that finding a reduced Gr\"obner basis is an $\cE\cX\cP\cS\cP\cA\cC\cE$-complete problem (see \cite[Section 25.8]{von2013modern} for the classification of computation complexities).

In this thesis, we use Gr\"obner bases for solving a system of equations. Let $f_1,\dots, f_s$ be polynomials in $\polyRing_n$. The set of solutions to the set of equations $\set*{f_1=0, \dots, f_s=0}$ is the variety $V(f_1,\dots, f_s)$. Let $\GrobBases$ be a Gr\"obner basis of the ideal $\ideal=\myspan{f_1,\dots, f_s}$. It can be shown that the variety $V(\GrobBases) = V(\ideal)$.
In \cref{sec:algo-search-via-GB}, we use the implementation of Gr\"obner bases in \href{http://magma.maths.usyd.edu.au/magma/}{Magma} \cite{Magma} to solve a system of polynomial equations.
In \cref{sec:distributed-LRS}, we use the facilities for multivariate polynomials in \href{https://doc.sagemath.org/html/en/reference/index.html}{SageMath} \cite{sagemath2022} to solve a systems of linear equations.
\section{Skew Polynomials}
\label{sec:skew-polys}
\emph{Skew polynomials} over division rings\footnote{A ring is a division ring if all the nonzero elements have a multiplicative inverse.} are \emph{non-commutative} polynomials that were introduced and studied by {\O}ystein Ore in \cite{ore1933theory}.
The theory of skew polynomials is quite rich and widely investigated in the literature. For instance, the division and factorization properties were studied in \cite{ore1933theory, giesbrecht1998factoring, baumbaugh2016results}.
Evaluation of skew polynomials and sets of roots were first considered by Tsit-Yuen Lam in \cite{lam1986general} and studied in great detail thereafter by Lam and Andr{\'e} Leroy \cite{lam1988algebraic, lam1988vandermonde, leroy1995pseudo, LamLer2004, lam2008wedderburn,leroy2012noncommutative}.
Faster algorithms for skew polynomials have been proposed for factorization and counting the number of factorizations \cite{caruso2012some}, interpolation \cite{liu2014kotter, bartz2022fastKoetter} and multiplication \cite{caruso2017fast,puchinger2018fast}. Properties of multivariate skew polynomial have been studied in \cite{martinez2019evaluation}.

The general definition of skew polynomial rings $\RingSkewPolys$ involves an \emph{endomorphism} $\Endom$ of the base ring $\ring$
and a \emph{derivation} associated with the endomorphism.
We start with the basics for the general definition in \cref{sec:generalSkewRing}, where $\ring$ is a general ring. This definition is mainly used in \cref{chap:mod_ring}.
In \cref{sec:Frob-zero-deri} we restrict our focus to a simpler class of skew polynomial ring $\FrobSkewPolys$ with the \emph{Frobenius automorphism} $\Frobaut(a)=a^q,\forall a\in\Fqm$ and the zero derivation $\Deriv=0$, on which \cref{chap:eva_skew} is based.

\subsection{The General Definition: $\RingSkewPolys$}
\label{sec:generalSkewRing}
Consider a ring $\ring$ with addition $+$ and multiplication $\cdot$ (we may omit the $\cdot$ between two elements for simplicity).
We consider the most general definition, where $\ring$ is not necessarily a division ring.
\begin{definition}[Endomorphism and derivation]
  \label{def:endo-deri}
  An \emph{endomorphism} of a ring $\ring$ is a map $\Endom: \ring\to \subring\subseteq\ring$ such that, for all $a,b\in\ring$,
  \begin{itemize}
  \item $\Endom(a+b) = \Endom(a)+\Endom(b)$,
  \item $\Endom(ab) = \Endom(a)\Endom(b)$.
  \end{itemize}
  A map $\Endom$ is an \emph{automorphism} if $\subring=\ring$.

  A \emph{$\Endom$-derivation} of $\ring$ is a map $\Deriv:\ring\to \subring\subseteq\ring$ such that, for all $a,b\in\ring$
  \begin{itemize}
  \item $\Deriv(a+b)  =\Deriv(a)+\Deriv(b)$,
  \item $\Deriv(ab) = \Deriv(a)b+\Endom(a)\Deriv(b)$.
  \end{itemize}
  A $\Endom$-derivation $\Deriv$ is an \emph{inner} $\Endom$-derivation if there exists $\beta\in \ring$ such that $\Deriv(a)= \beta a - \Endom(a)\beta$ for all $a\in \ring$.
\end{definition}
It follows from the definition that for any endomorphsim $\Endom$ and any $\Endom$-derivation $\Deriv$ of $\ring$, it holds that
\begin{itemize}
\item $\Endom(0) = 0,\ \Endom(1) = 1$,
\item $\Deriv(0) = 0, \ \Deriv(1) = 0$.
\end{itemize}
We denote by $\Id$ the \emph{identity automorphism} $\Endom(a)=a, \forall a\in\ring$.
For ease of notation, we also use the exponential notation $\Endom(a)=a{^\Endom}$ and $\Deriv(a)=a^{\Deriv}$.
\begin{lemma}
  \label{lem:inv-ele-endo}
  If $a\in\ring^*$ is invertible, then $a$ is not a \emph{zero-divisor}\footnote{An elements $b\in\ring^*$ is a zero divisor if $\exists a\in\ring^*$ such that $ab=ba=0$.} and $\Endom(a)$ invertible.
\end{lemma}
\begin{proof}
  We first show that any $a\in\ring$ cannot be both invertible and a zero-divisor. Suppose an invertible $a\in\ring$ is a zero-divisor. Let $0\neq b \in\ring$ with $ab=0$. Then $b= (a^{-1} a)b = a^{-1}(ab) = 0$, which is a contradiction.
  According to \cref{def:endo-deri}, $\Endom(a\cdot a^{-1}) = \Endom(a)\cdot\Endom(a^{-1})$. Together with $\Endom(a\cdot a^{-1}) = \Endom(1) = 1$, it can be seen that $\Endom(a)$ is invertible and its inverse is $\Endom(a^{-1})$.
\end{proof}
\begin{definition}[Skew polynomial rings]
  \label{def:skew-poly-ring}
  A \emph{skew polynomial ring} $\RingSkewPolys$ is a set of polynomials
  \begin{align*}
   \RingSkewPolys\defeq \set*{\left. \sum_{i=0}^na_i\SkewVar^i\ \right|\ a_i \in \ring, n\in \bbN}
  \end{align*}
  with addition $+$ as for usual polynomials, and multiplication $\cdot$ following the basic rule
  \begin{align}
    \SkewVar a=\Endom(a)\SkewVar +\Deriv(a), \ \forall\ a \in \ring\ . \label{eq:basic-mul-rule}
  \end{align}
  The multiplication extends to all elements in $\RingSkewPolys$ by associativity and distributivity.
  The \emph{degree} of a nonzero skew polynomial $f = \sum_{i\in\bbN} f_i \SkewVar^{i} \in\RingSkewPolys$ is $\deg f\defeq\max\set*{\ i \ |\ f_i\neq 0}$.
  By convention, the degree of the zero polynomial is $\deg(0)=-\infty$.
\end{definition}
\subsubsection{Multiplication}
Given an endomorphism $\Endom$ of a ring $\ring$, the powers of $\Endom$ are $a^{\Endom^{i+1}}=\Endom^{i+1}(a)\defeq \Endom(\Endom^{i}(a))$, for all $i\in\bbN, a\in\ring$.
Given a $\Endom$-derivation $\Deriv$ of $\ring$ and $a\in\ring$, we denote $a^{\Endom\Deriv}=\Deriv(\Endom(a))$.

For any $h=\sum_{i=0}^dh_i\SkewVar^i$ and $g=\sum_{i=0}^eg_i\SkewVar^i$ in $\RingSkewPolys$, the product of the two polynomials is
\begin{align*}
  h\cdot g = h_dg_e^{\Endom^d}\SkewVar^{d+e} + \parenv*{h_{d-1}g_e^{\Endom^{d-1}}+h_d\parenv*{g_e^{\Endom^{d-1}\Deriv}+g_e^{\Endom^{d-2}\Deriv\Endom}
  +\cdots + g_e^{\Deriv\Endom^{d-1}}}}\SkewVar^{d+e-1} +\cdots \ ,
\end{align*}
where $g_e^{\Endom^{d-2}\Deriv\Endom}=\Endom(\Deriv(g_e^{\Endom^{d-2}}))$ and $g_e^{\Deriv\Endom^{d-1}}=\Endom^{d-1}(\Deriv(g_e))$ according to the notation introduced above.
It can be seen that the explicit expression of the product is quite messy when a nonzero derivation is involved.
For skew polynomials with a zero $\Endom$-derivation, e.g., $h=\sum_{i\in\bbN}h_i\SkewVar^i$ and $g=\sum_{i\in\bbN}g_i\SkewVar^i$ in $\ring[\SkewVar; \Endom]$, the expression of the product is simply
\begin{align}
  \label{eq:prod-skew-polys}
  h\cdot g = \sum_{i\in\bbN}\sum_{j\in\bbN} h_i\Endom^i(g_j) \SkewVar^{i+j} =\sum_{s\in\bbN}\parenv*{\sum_{\substack{i\in\bbN\\i\leq s}}h_i\Endom^{i}(g_{s-i})}\SkewVar^s\ .
\end{align}
\begin{theorem}
  \label{thm:deg-prod-skew}
  For any $h,g\in\RingSkewPolys$, if the leading coefficient $\lc(g)$ is invertible, then $\deg(hg)=\deg(h)+\deg(g)$.
\end{theorem}
\begin{proof}
  Let $d= \deg(h)$ and $e= \deg(g)$.
  It can be seen from the multiplication rule in \eqref{eq:basic-mul-rule} that commuting the variable $\SkewVar$ and the coefficients does not increase the degree in $\SkewVar$. Therefore, the coefficient of $\SkewVar^{d+e+i}=0, \forall i>0$.
  It follows from \cref{lem:inv-ele-endo} that the coefficient of the monomial $\SkewVar^{d+e}$ is $\lc(h)\cdot \Endom^{d}(\lc(g))\neq 0$. 
\end{proof}

\subsubsection{Division}
\label{sec:division}

If the base ring $\ring$ of $\polyRing=\RingSkewPolys$ is a division ring, then for $f,g\in\polyRing$, one can always perform a right division on $f$ by $g$, i.e., find the quotient polynomial $q\in\polyRing$ and the remainder polynomial $r\in\polyRing$ such that $f=q\cdot g + r$ with $\deg (r)<\deg(g)$~\cite{ore1933theory}.
The right division algorithm is presented in \cref{algo:EAskew}, which is the analogue to the well-known Euclidean algorithm for skew polynomials.

For a non-division ring $\ring$, the following theorem shows that the right division can be done by \cref{algo:EAskew} as well, as long as the leading coefficient of the divisor is invertible.

\begin{theorem}
  \label{thm:right-division-over-ring}
  For any $f,g\in\polyRing$, if the leading coefficient $\lc(g)$ is invertible, then \cref{algo:EAskew} outputs a unique pair of quotient and remainder $q,r\in\polyRing$.
\end{theorem}
\begin{proof}
  By \cref{lem:inv-ele-endo}, it can seen that, with an invertible $\lc(g)$, $q_i$ at \cref{eq:right-division-skew-start} in \cref{algo:EAskew} is nonzero and \cref{eq:right-division-skew-end} can always remove the leading term in $r$. Hence, $\deg(r)$ decreases in every loop and the algorithm terminates.
  We show that the outputs $q,r$ are unique by contradiction. Suppose that $f=hg+r=\tilde{h}g+\tilde{r}$, which implies $(h-\tilde{h})g=r-\tilde{r}$. If $h-\tilde{h}$ is nonzero, then $\deg((h-\tilde{h})g)\geq\deg(g)$. However, by the termination condition of the algorithm, $\deg(r-\tilde{r})<\deg(g)$.
\end{proof}

\begin{algorithm}[htb!]
  \caption{Right division (Euclidean) algorithm for skew polynomials}
  \label{algo:EAskew}
  \KwIn{$f,g\in\RingSkewPolys$ where $\lc(g)$ is invertible}
  \KwOut{A unique pair $q,r\in\RingSkewPolys$ such that $f=q\cdot g + r$ with $\deg (r)<\deg(g)$.}
  $q\gets 0, r\gets f$\;
  \While{$\deg(r)\geq \deg(g)$}{
    $q_i \gets \lc(r)\Endom^{n-m}(g_m^{-1})\SkewVar^{n-m}$\;
    \label{eq:right-division-skew-start}
    $q \gets q+ q_i$\;
    $r \gets  r-q_i\cdot g$\;
    \label{eq:right-division-skew-end}
  }
  \Return{$q,r$}
\end{algorithm}

Left division between any $f,g\in\polyRing$ can be performed only if $\Endom$ is an automorphism of $\ring$. For left division, the repeated procedure from \cref{eq:right-division-skew-start} to \cref{eq:right-division-skew-end} in \cref{algo:EAskew} becomes
\begin{align*}
  q_i \gets& \Endom^{-m}(g_m^{-1}\lc(r))\SkewVar^{n-m}\\
  q \gets& q+ q_i\\
  r \gets & f-g \cdot q_i
\end{align*}
and the outputs $q,r\in\polyRing$ are such that $f = g\cdot q +r$. The inverse map of $\Endom$ is required in the computation of $q_i$. Hence, $\Endom$ needs to be an automorphism in order to perform the left division.

We say that $g$ right (resp., left) divides $f$ or $f$ is right (left) divisible by $g$, if the remainder of the right (left) division on $f$ by $g$ is $0$.
We denote by $g\divides_r f$ (resp., $g\divides_l f$) if $f$ is right (left) divisible by $g$.



\begin{definition}[gcrd and lclm]
  \label{def:gcrd-lclm}
  For any $f_1,f_2\in\polyRing$, the \emph{greatest common right divisor} (gcrd) of $f_1, f_2$, denoted by $\gcrd(f_1,f_2)$, is a monic polynomial $g\in\polyRing$ of the largest degree such that $g$ right divides both $f_1$ and $f_2$.

  The \emph{least common left multiplier} (lclm) of $f_1, f_2\in\polyRing$ is a monic polynomial $m\in\polyRing$ of the lowest degree which is right divisible by both $f_1$ and $f_2$, i.e.,
  \begin{align*}
    m= g_1\cdot f_1 = g_2 \cdot f_2 \quad \textrm{for some }g_1,g_2\in\polyRing\ .
  \end{align*}
\end{definition}

It has been shown in \cite[Section 2]{ore1933theory} that for any $f_1,f_2\in\polyRing$ where $\ring$ is a division ring, there is a unique $g=\gcrd(f_1,f_2)$ and it can be computed from the output of \cref{algo:EEAskew}: $\gcrd(f_1,f_2) = r_\ell$ (up to a scalar multiple).
If $\gcrd(f_1,f_2)=1$, then $f_1$ and $f_2$ are \emph{relatively prime}.

Similarly, it has been shown in \cite[Section 2]{baumbaugh2016results} that for any $f_1,f_2\in\polyRing$ where $\ring$ is a division ring, there is a unique $m=\lclm(f_1,f_2)$ and it can be computed up to a scalar multiple from the output of \cref{algo:EEAskew}: $\lclm(f_1,f_2) = s_{\ell+1}f_1 = -t_{\ell+1}f_2$.

The degree of the lclm of $f_1,f_2$ and the degree of the gcrd can be related via
\begin{align*}
  \deg \lclm(f_1,f_2) = \deg f_1 +\deg f_2 -\deg \gcrd(f_1,f_2)\ .
\end{align*}

\begin{algorithm}[htb!]
  \caption{Extended Euclidean algorithm for skew polynomials.}
  \label{algo:EEAskew}
  \KwIn{$f_1,f_2\in\polyRing$ where $\ring$ is a division ring.}
  \KwOut{$\ell\in\bbN$, $r_i,s_i,t_i\in\polyRing, i\in[\ell+1]$, and $q_i\in\polyRing,i\in[\ell]$}
  $r_1\gets f_1, s_1 \gets 1, t_1\gets 0$\;
  $r_2\gets f_2, s_2 \gets 0, t_2\gets 1$\;
  $i\gets 2$
  \While{$r_i\neq 0$}{
    $q_i, r_{i+1}\gets$ right divide $r_i$ by $r_{i-1}$
    \tcc*[r]{Apply \cref{algo:EAskew}}
    $s_{i+1}\gets s_{i-1}-q_is_i$\;
    $t_{i+1}\gets t_{i-1}-q_it_i$\;
    $i\gets i+1$
  }
  $\ell\gets i-1$\;
  \Return{$r_i,s_i,t_i, \forall i\in[\ell+1]$, and $q_i, \forall i\in[\ell]$}
\end{algorithm}


\subsubsection{Evaluation}
For a commutative polynomial $f\in\ring[x]$, the process of evaluating $f$ at $a\in\ring$, denoted by $f(a)$, is simply ``plugging in'' the value $a$ in place of $x$ in $f$ and carry out the proper operation in $\ring$. The result by this simple method coincides with the result by the \emph{remainder evaluation}, where $f(a)\in\ring$ is the remainder of dividing $f$ by $(x-a)$. In other words, the evaluation $f(a)$ is such that $f = g\cdot(x-a)+f(a)$ for some $g\in\ring[x]$.
However, for skew polynomials, simple plugging-in does not always give the same result as the remainder evaluation,
as shown in \cref{eg:plugin-neq-remainder}.
Nevertheless, we show in \cref{thm:skewEvaNform} that another way of ``plugging-in'' is equivalent to the remainder evaluation.
\begin{definition}[Remainder evaluation of skew polynomials]
\label{def:eva_division}
Let $\ring$ be a ring and $\polyRing=\ring[\SkewVar;\Endom, \Deriv]$.
For any $f\in\polyRing$ and $a\in\ring$, the \emph{evaluation} of $f$ at $a$, denoted by $f(a)$, is the remainder from the right division on $f$ by $\SkewVar-a$. In other words, we can write
\begin{align*}
    f(a) = f - g\cdot (\SkewVar-a) \quad \textrm{ for some } g\in\polyRing\ .
\end{align*}
\end{definition}
Since $\SkewVar-a$ is monic, it follows from \cref{thm:right-division-over-ring} that the remainder evaluation $f(a)$ is unique.
\begin{example}
  \label{eg:plugin-neq-remainder}
  Let $\F_4[\SkewVar; \Frobaut,\Deriv]$ be a skew polynomial ring with the Frobenius automorphism $\Frobaut(a)=a^2, \forall a \in\F_4$ and an inner $\Frobaut$-derivation $\Deriv(a) = \Frobaut(a)-a, \forall a \in\F_4$ (recall from \cref{def:endo-deri} that $\beta =1$).
  Let $\alpha$ be a primitive element of $\F_4$. To evaluate $f=\SkewVar^3+\SkewVar+1$ at $\alpha\in\F_4$, by the ``plugging-in'' method, we get
  \begin{align*}
    f(\alpha) = \alpha^3+\alpha+1 = \alpha\ .
  \end{align*}
  By the remainder evaluation from \cref{def:eva_division}, we have
  \begin{align*}
    f(\alpha) = \alpha+1 = f - (\SkewVar^2+\alpha\SkewVar)(\SkewVar-\alpha)\ .
  \end{align*}
\end{example}
The following theorem shows that the remainder evaluation for skew polynomials has an equivalent form so that one can perform the evaluation by addition and multiplication, without applying any division algorithm (e.g., \cref{algo:EAskew}). The form has been proven in \cite[Lemma 2.4]{lam1988vandermonde} for division rings and applies naturally to general rings.
\begin{theorem}
  \label{thm:skewEvaNform}
  Let $\ring$ be a ring, $\Endom$ be an endomorphism of $\ring$ and $\Deriv$ be a $\Endom$-derivation. For any $a\in\ring$, define recursively the \emph{$i$-th truncated norm} of $a$ as
  \begin{align*}
    N_0(a) &\defeq 1\ ,\\
    N_{i+1}(a) &\defeq \Endom(N_i(a))\cdot a + \Deriv(N_i(a))\ ,\ \forall i\in\bbN\ .
  \end{align*}
  Then, for any $f=\sum_{i\in\bbN}f_i\SkewVar^i\in\polyRing$, the evaluation of $f$ at $a\in\ring$ is
  \begin{align*}
    f(a) = \sum_{i\in\bbN}f_iN_i(a)\ .
  \end{align*}
\end{theorem}
\begin{proof}
  We first show that for any $k\in\bbN$, $(\SkewVar-a)\divides_r(\SkewVar^k-N_k(a))$ by induction.
  This is trivial for $k=0$, since $\SkewVar^0-N_0(a) = 1-1=0$.
  Assume it is true for some $k\geq 0$. Then
  \begin{align}
    \SkewVar^{k+1}-N_{k+1}(a)&=\SkewVar^{k+1}-\Endom\parenv*{N_k(a)}\cdot a - \Deriv\parenv*{N_k(a)}\nonumber\\
                             &=\SkewVar^{k+1}-\Endom\parenv*{N_k(a)}\cdot a + \Endom(N_k(a))\cdot \SkewVar - \Endom(N_k(a))\cdot \SkewVar - \Deriv(N_k(a))\nonumber\\
                             &=\SkewVar^{k+1}+\Endom\parenv*{N_k(a)}(\SkewVar- a) - \parenv*{\Endom(N_k(a))\cdot \SkewVar + \Deriv(N_k(a))}\nonumber\\
                             &= \Endom(N_k(a))(\SkewVar-a) + \SkewVar^{k+1}-\SkewVar\cdot N_k(a) \label{eq:commuteX}\\
                             &= \Endom(N_k(a))(\SkewVar-a) + \SkewVar(\SkewVar^{k}-\cdot N_k(a))\ ,\nonumber
  \end{align}
  where \eqref{eq:commuteX} follows from the multiplication rule in \eqref{eq:basic-mul-rule} for commuting the variable and the coefficient.
  With the induction hypothesis, $(\SkewVar-a)\divides_r(\SkewVar^k-N_k(a))$, we conclude that both terms on the last line are right divisible by $\SkewVar-a$.

  Then we can see that
  \begin{align*}
    f - f(a) &= \sum_{i\in\bbN}f_i\SkewVar^i - \sum_{i\in\bbN}f_iN_i(a)\\
    &= \sum_{i\in\bbN}f_i\cdot\parenv*{\SkewVar^i-  N_i(a)}\ .
  \end{align*}
  Applying the result for $\SkewVar^{k+1}-N_{k+1}(a)$ above to each term $\parenv*{\SkewVar^i-  N_i(a)}$ in the sum, we conclude that $(\SkewVar-a)\divides_r\parenv*{f-f(a)}$. Therefore, we can write $f= q\cdot (\SkewVar-a) + f(a)$ for some $q\in\polyRing$. Since $\deg(f(a))<1=\deg(\SkewVar-a)$, it is indeed the remainder of the right division on $f$ by $\SkewVar-a$.
\end{proof}

So far, we only considered polynomials in the \emph{left form}, $f=\sum_{i\in\bbN}f_i\SkewVar^i$. If a polynomial is given in the \emph{right form}, i.e., $h=\sum_{i\in\bbN}\SkewVar^ih_i$, then for some $a\in\ring$, the evaluation $h(a)\neq\sum_{i\in\bbN}N_i(a)h_i$. Instead, we must first convert $h$ to the left form $\sum_{i\in\bbN}g_i\SkewVar^i$ by the multiplication rule \eqref{eq:basic-mul-rule} and then compute $h(a) = \sum_{i\in\bbN}g_iN_i(a)$.

With the definition of $N_i(a)$ in \cref{thm:skewEvaNform}, we define the \emph{$(\Endom,\Deriv)$-Vandermonde matrix} on a set of elements $\Omega=\set*{a_1,\dots, a_n}\subseteq\ring$.
\begin{definition}
[$(\Endom,\Deriv)$-Vandermonde matrix]
\label{def:skewVandermonde}
Let $N_i(\cdot)$ be the $i$-th truncated norm as defined in \cref{thm:skewEvaNform}.
Given a set $\Omega=\set*{a_1,\dots, a_{n}} \subseteq \ring$,
the $(\Endom,\Deriv)$-Vandermonde matrix of $\Omega$ is given by
\begin{align*}
    \bV^{\Endom,\Deriv}(\Omega) \defeq
    \begin{pmatrix}
      N_0(a_1) & N_0(a_2) & \dots & N_0(a_n)\\
      N_1(a_1) & N_1(a_2) & \dots & N_1(a_n)\\
      \vdots & \vdots & \ddots &\vdots \\
      N_{n-1}(a_1) & N_{n-1}(a_2) & \dots & N_{n-1}(a_n)
    \end{pmatrix}\ .
\end{align*}
\end{definition}
Similar to the evaluation of commutative univariate polynomials, the evaluation of a polynomial $f=\sum_{i=0}^kf_i\SkewVar^i\in\polyRing, k\in\bbN$ at all the elements in $\Omega=\set*{a_1,\dots,a_n}$ can be written as
\begin{align*}
  \parenv*{f(a_1), f(a_2), \dots, f(a_n)} &= {(f_0,\dots, f_k)}\cdot
  \underbrace{
    \begin{pmatrix}
      N_0(a_1) & N_0(a_2) & \dots & N_0(a_n)\\
      \vdots & \vdots & \ddots &\vdots \\
      N_{k}(a_1) & N_{k}(a_2) & \dots & N_{k}(a_n)
    \end{pmatrix}}_{\defeqrev\bV_k^{\Endom,\Deriv}(\Omega)}\ .
\end{align*}

Next, we seek for a formula for evaluating a product of two polynomials at a point $a\in\ring$. For this purpose, we first define the notion of \emph{$(\Endom,\Deriv)$-conjugacy}.
\begin{definition}[$(\Endom,\Deriv)$-conjugacy]
  \label{def:generalConjugate}
For any two elements $a\in\ring,c\in\ring^*$, define
\begin{align*}
  a^c\defeq \Endom(c)ac^{-1}+\Deriv(c)c^{-1}\ ,
\end{align*}
where the term $\Deriv(c)c^{-1}$ is also known as the ``logarithmic derivative'' of $c\in\ring^*$.
Two elements $a,b\in\ring$ are said to be \emph{$(\Endom,\Deriv)$-conjugate} if there exists an element $c\in\ring^*$ such that $a^c=b$.
\end{definition}
It is easy to check that $(\Endom,\Deriv)$-conjugacy fulfills the following properties of equivalence relations:
\begin{itemize}
\item Reflexivity: $a$ is $(\Endom,\Deriv)$-conjugate to $a$, for any $a\in\ring$.
\item Symmetry: If $a$ is $(\Endom,\Deriv)$-conjugate to $b$, then $b$ is $(\Endom,\Deriv)$-conjugate to $a$, for any $a,b\in\ring$.
\item Transitivity: If $a$ is $(\Endom,\Deriv)$-conjugate to $b$ and $b$ is $(\Endom,\Deriv)$-conjugate to $c$, then $a$ is $(\Endom,\Deriv)$-conjugate to $c$, for any $a,b,c\in\ring$.
\end{itemize}
Since $(\Endom,\Deriv)$-conjugacy is an equivalence relation, we can speak of $(\Endom,\Deriv)$-conjugacy classes. For instance, the $(\Endom,\Deriv)$-conjugacy class of $0$, $\set*{0^c=\Deriv(c)c^{-1}\ |\ c\in\ring^*}$, consists of the logarithmic derivatives of all the nonzero elements in $\ring$.
\begin{definition}[$(\Endom,\Deriv)$-conjugacy classes]
For any $a \in\ring$, the \emph{$(\Endom,\Deriv)$-conjugacy class} of $a$ is
  \begin{align*}
    C_{\Endom,\Deriv}(a)\defeq \set*{ a^c\ |\ c\in \ring^*}\ .
  \end{align*}
\end{definition}

Using the notion of $(\Endom,\Deriv)$-conjugacy, the next result provides us with a useful formula for evaluating a product $f\cdot g$ at $a\in\ring$. 

\begin{theorem}[{\cite[Theorem 2.7]{lam1988vandermonde}}]
  Let $f,g\in\polyRing$ and $a\in\ring$. Then, $(f\cdot g)(a) \neq f(a)\cdot g(a)$ in general. Instead,
  \begin{align*}
    (f\cdot g)(a) =
    \begin{cases}
      0 & \text{ if }g(a)=0\\
      f\parenv*{a^{g(a)}}\cdot g(a) & \text{ else }
    \end{cases}\ ,
  \end{align*}
  where $a^{g(a)}\defeq\Endom(g(a))\cdot a\cdot g(a)^{-1}+\Deriv(g(a))\cdot g(a)^{-1}$.
\end{theorem}

\subsubsection{Ideals}
\label{sec:skew-ideal}
Since the skew polynomial rings $\polyRing$ are non-commutative, we need to differentiate between left and right ideals of $\polyRing$.

A \emph{left ideal} $\myspan{f}_l \subseteq \polyRing$ generated by $f\in\polyRing$ is a set of skew polynomials
\begin{align*}
  \myspan{f}_l \defeq \set*{\left. gf\ \right|\ g\in\polyRing}\ .
\end{align*}
Similarly, a \emph{right ideal} $\myspan{f}_r\subseteq \polyRing$ generated by $f\in\polyRing$ is a set of skew polynomials
\begin{align*}
  \myspan{f}_r \defeq \set*{\left. fg\ \right|\ g\in\polyRing}\ .
\end{align*}

\begin{lemma}
  \label{lem:g-right-div-f}
  For $g,f\in\polyRing$, $\myspan{g}_l\supseteq\myspan{f}_l$ if and only if $g\divides_r f$.
\end{lemma}
\begin{proof}
  We first show the sufficiency. Suppose $g\divides_r f$, then $\forall p \in\myspan{f}_l$, $p=u\cdot f$ for some $u\in\polyRing$ and $p=u\cdot (q\cdot g)$ for $q$ being the quotient of the right division on $f$ by $g$. Hence, $\myspan{g}_l\supseteq\myspan{f}_l$.
  For the necessity, assume
  $g\not\divides_r f$.
  Then there is a nonzero $a\in\polyRing$ of $\deg(a)<\deg(g)$ such that $f= q\cdot g+a$. It can be seen that for any $u\in\polyRing$ of $\deg(u)=0$, $u\cdot f\not\in\myspan{g}_l$.
\end{proof}

\subsection{With Frobenius Automorphism and Zero Derivation}
\label{sec:Frob-zero-deri}
In this subsection, we present the properties of skew polynomials over an extension field $\Fqm$, with the Frobenius automorphism and zero derivation.

The \emph{Frobenius automorphism} of an extension field $\Fqm$ is the map
\begin{align*}
  \Frobaut : \Fqm &\to \Fqm\ \\
  a &\mapsto a^q\ .
\end{align*}
Let $\FrobSkewPolys$ be a skew polynomial ring over $\Fqm$ with Frobenius automorphism $\Frobaut$ and zero derivation $\Deriv =0$.
We follow the notation used in \cref{sec:generalSkewRing} but omit the $\Deriv$ for simplicity.

For two skew polynomials $h=\sum_{i\in\bbN} h_i\SkewVar^{i}$ and $g=\sum_{j\in\bbN} g_j\SkewVar^{j}$ in $\FrobSkewPolys$, their product $h\cdot g$ can be computed according to \eqref{eq:prod-skew-polys}. Since every nonzero element in $\Fqm$ is invertible, it follows from \cref{thm:deg-prod-skew} that $\deg(hg)=\deg(h)+\deg(g)$.

For any $\alpha\in\Fqm$, its $i$-th truncated norm $N_i(\alpha)$ (defined in \cref{thm:skewEvaNform}) becomes
\begin{align*}
  N_i(\alpha)=\prod_{j=0}^{i-1} \Frobaut^j(\alpha) =\alpha^{\sum_{j=0}^{i-1}q^j}=\alpha^{(q^i-1)/(q-1)}\ .
\end{align*}
The evaluation of any $f=\sum_{i\in\bbN} f_i \SkewVar^{i}\in \FrobSkewPolys$ at $\alpha\in\Fqm$ becomes
\begin{align}
  \label{eq:Frob-remainder-evaluation}
  f(\alpha) = \sum_{i\in\bbN}f_iN_{i}(\alpha) =  \sum_{i\in\bbN} f_i \alpha^{(q^{i}-1)/(q-1)}\ .
\end{align}
\subsubsection{$\sigma$-Conjugacy Classes}
For any $a\in\Fqm$, its \emph{$\Frobaut$-conjugacy} w.r.t.~$c\in\Fqm^*$ becomes
\begin{align}
  \label{eq:Frob-conj}
  a^c\defeq \Frobaut(c)a c^{-1} = ac^{q-1}\ ,
\end{align}
and its $\Frobaut$-conjugacy class becomes
\begin{align}
  \label{eq:Frob-conj-classes}
  C_{\Frobaut}(a)\defeq \set*{\left. ac^{q-1}\ \right|\ c\in \Fqm^*}\ .
\end{align}

We say that two elements $a, b\in\Fqm$ are \emph{$\Frobaut$-distinct} if $b\not\in C_{\Frobaut}(a)$.
The following theorem shows that the $\Frobaut$-conjugacy classes of $\Fqm$ form a partition of $\Fqm$.
\begin{theorem}[Structure of $\Frobaut$-conjugacy classes {\cite[Theorem 2.12]{FnTsurvey-Umberto}}]
  \label{thm:conjugacyClasses}
  Let $\priEle\in\Fqm^*$ be a primitive element of $\Fqm$. For the Frobenius automorphism $\Frobaut$, the $q-1$ elements $1,\priEle,\priEle^2,\dots, \priEle^{q-2}$ are pair-wise $\Frobaut$-distinct.
  There are $q$ disjoint conjugacy classes in $\Fqm$ and 
  $$\Fqm=C_{\Frobaut}(0)\cup C_{\Frobaut}(\priEle^0)\cup\cdots \cup C_{\Frobaut}(\priEle^{q-2})\ ,$$
  where $C_{\Frobaut}(0)=\set*{0}$ and $\left|C_{\Frobaut}(\priEle^i)\right|=\frac{q^m-1}{q-1}, i\in[0,q-2]$. 
\end{theorem}

\subsubsection{Polynomial Independence}
In the following, we discuss the set of roots of skew polynomials, i.e., given $f\in\FrobSkewPolys$, we determine $\set*{\alpha\in\Fqm\ |\ f(\alpha)=0}$.
Different from commutative polynomials $g\in\F[x]$, where the number of roots of $g$ is at most $\deg(g)$, the following example shows that this does not hold for skew polynomials in general.
\begin{example}
  Consider $q=2, m=2$ and the skew polynomial ring $\F_4[\SkewVar;\Frobaut]$ with the Frobenius automorphism and zero derivation. Let $\alpha$ be a primitive element of $\F_4$. For the polynomial $f=\SkewVar^2+1\in\F_4[\SkewVar;\Frobaut]$, it can be verified that the set of root of $f$ is
  \begin{align*}
    \set*{1,\alpha,\alpha+1}\ .
  \end{align*}
\end{example}
In the commutative case, we count \emph{distinct} roots. In the case of skew polynomials, we need a different notion of \emph{distinctness}, namely, \emph{polynomial independence}.
For this, we need the following definition.
\begin{definition}[Minimal polynomial]\label{def:minPoly}
  Given a nonempty set $\Omega\subseteq \Fqm$, we say $f_{\Omega}$ is a \emph{minimal polynomial} of $\Omega$ if it is
  a monic polynomial $f_{\Omega}\in\FrobSkewPolys$ of minimal degree such that $f_{\Omega}(\alpha)=0$ for all $\alpha\in\PindSet$.
\end{definition}
The following theorem shows that for any nonempty $\Omega\subseteq\Fqm$, its minimal polynomial is unique.
It also implies that $\FrobSkewPolys$ is a \emph{principle left ideal ring}\footnote{A principle left ideal ring is a ring where every left ideal can be generated by only one element in the ideal.}.
\begin{theorem}[{\cite[Theorem 2.5]{FnTsurvey-Umberto}}]
  \label{thm:uniqueMinimalSkewPoly}
  Given any nonempty set $\Omega\subseteq\Fqm$, there exists a unique monic skew polynomial $f\in\FrobSkewPolys$ such that, for any $g$ in the left ideal $\myspan{f}_l$,
  \begin{align*}
      g(\alpha)=0,\ \forall \alpha \in \Omega\ .
  \end{align*}
\end{theorem}
The minimal polynomial can be constructed by an iterative Newton interpolation approach as follows (cf.~\cite[Proposition 2.6]{FnTsurvey-Umberto}).
First, set
\begin{align*}
  g_1 =\SkewVar-\alpha_1 .
\end{align*}
Then for $i=2,3,\dots, n$, perform
\begin{align}\label{eq:newtonInterpolation}
  g_i=
  \begin{cases}
    g_{i-1}\quad &\text{if } g_{i-1}(\alpha_i)=0,\\
    \left(\SkewVar-\alpha^{g_{i-1}(\alpha_i)}\right)\cdot g_{i-1}\quad &\text{otherwise,}
  \end{cases}
\end{align}
where $\alpha^{g_{i-1}(\alpha_i)}= \Frobaut(g_{i-1}(\alpha_i))\alpha_i g_{i-1}(\alpha_i)^{-1}$ is the $\Frobaut$-conjugate of $\alpha$ w.r.t.~$g_{i-1}(\alpha_i)$ as in \eqref{eq:Frob-conj}.
Upon termination, we have $g_n=f_{\Omega}$.

It can be shown that the minimal polynomial of a set $\Omega$ can be also constructed by computing
\begin{align}\label{eq:lclm}
  f_\Omega = \lclm_{\alpha\in\Omega} \{\SkewVar-\alpha\}\ ,
\end{align}
where $\lclm$ is the least common left multiple as defined in \cref{def:gcrd-lclm}.

The following theorem summarizes some properties of the minimal polynomial of a set $\Omega\subseteq\Fqm$.
\begin{theorem}[{\cite[Theorem 7]{liu2017matroidal}}]
  \label{thm:prop-minimal-skew-poly}
  Let $\Omega_1,\Omega_2\subseteq \Fqm$ and $f_{\Omega_1},f_{\Omega_2}$ be their minimal polynomials. Then
  \begin{itemize}
  \item $f_{\Omega_1\cup\Omega_2}=\lclm(f_{\Omega_1},f_{\Omega_2})$.
  \item $f_{\Omega_1\cap\Omega_2}=\gcrd(f_{\Omega_1},f_{\Omega_2})$.
  \item $\deg(f_{\Omega_1\cup\Omega_2})=\deg(f_{\Omega_1})+\deg(f_{\Omega_2})-\deg(f_{\Omega_1\cap\Omega_2})$.
  \end{itemize}
\end{theorem}

With the notion of minimal polynomials, we can now introduce the ``distinctness'' for the roots of skew polynomials.
\begin{definition}[P-independent set] 
  \label{def:PindSet}
  A set $\Omega\subseteq \Fqm$ is P-independent in $\FrobSkewPolys$ if the degree of its minimal polynomial is equal to $|\Omega|$, i.e., $\deg (f_{\Omega}) = |\Omega|$.
\end{definition}

Recall from \cref{def:skewVandermonde} that a $\Frobaut$-Vandermonde matrix can be constructed for a set $\Omega\subseteq\Fqm$.
It has been shown in \cite[Theorem 8]{lam1986general} that
\begin{align}
  \label{eq:PindRankVand}
  \Omega \textrm{ is P-independent } \iff |\PindSet|=\rank(\bV^{\Frobaut}(\Omega))\ .
\end{align}

We can derive the following result from the definition of P-independent sets.
\begin{lemma}\label{lem:noMoreZero}
Given a P-independent set $\Omega$, for any subset $\rootSet\subset \Omega$, let $f_{\rootSet}(x)\in\FrobSkewPolys$ be the minimal polynomial of $\rootSet$. Then, for any element $\alpha\in\Omega\setminus\rootSet$, $f_{\rootSet}(\alpha)\neq 0$.
\end{lemma}
\begin{proof}
Assume $f_{\rootSet}(\alpha)=0$, then the minimal polynomial $f_{\rootSet\cup\{\alpha\}}=f_{\rootSet}$ and $\deg(f_{\rootSet\cup\{\alpha\}}) = |\rootSet|< |\rootSet\cup\{\alpha\}|$, which contradicts to that $\rootSet\cup\{\alpha\}\subseteq \Omega$ is P-independent.
\end{proof}

In the early work by Lam, the property of unions and subsets of P-independent sets has been derived, as summarized in the following theorem.
\begin{theorem}[{\cite[Theorem 22--23]{lam1986general}}]
  \label{thm:lam-prop-Pind-sets}
For any two sets $\Omega_1,\Omega_2\subseteq\Fqm$, let $\Omega = \Omega_1\cup \Omega_2$.
\begin{itemize}
\item If $\Omega_1$ and $\Omega_2$ are P-independent and no element in $\Omega_1$ is $\Frobaut$-conjugate to any element in $\Omega_2$, then $\Omega$ is P-independent.
\item If $\Omega$ is P-independent, then $\Omega_1$ and $\Omega_2$ are P-independent. In other words, any subset of a P-independent set is P-independent.
\end{itemize}
\end{theorem}

The relation between P-independence and linear independence and the structure of roots of the minimal polynomial $f_{\Omega}$ of some special sets $\Omega\subseteq\Fqm$ are concluded in the following.

\begin{theorem}[{\cite[Lemma 1, Theorem 10]{liu2017matroidal}}]
    Let $\Omega=\set*{\alphas{n}}\subseteq C_{\Frobaut}(\priEle^l)$ for some $l\in[m]$, and $a_1,\dots,a_n\in\Fqm$ be such that $\alpha_i=\priEle^l a_i^{q-1}$ for all $i\in[n]$. 
    Then, $\Omega$ is P-independent if and only if $a_1,\dots,a_n$ are linearly independent over $\Fq$.
    Moreover, let $f_{\Omega}$ be the minimal polynomial of $\Omega$ and $\overline{\Omega}\defeq\set*{\alpha\in\Fqm\ |\ f_{\Omega}(\alpha) = 0}$ be the set of roots of $f_{\Omega}$. Then
    \begin{align*}
      \overline{\Omega}=\set*{\left. \priEle^l a^{q-1} \ \right| \ a \in\myspan{a_1,\dots,a_n}_q}\subseteq C_{\Frobaut}(\priEle^l)\ ,
    \end{align*}
    where $\myspan{a_1,\dots,a_n}_q$ denotes the $\Fq$-subspace of $\Fqm$ generated by $a_1,\dots,a_n$.
\end{theorem}

\subsection{Applications of Skew Polynomials in Coding Theory}
Codes constructed from skew polynomials were first introduced and studied by Boucher and Ulmer \cite{boucher2007skew,boucher2009modules, boucher2011note, boucher2014self}.
Various classes of codes based on skew polynomials were constructed thereafter. For instance, codes over rings \cite{abualrub2010skew}, evaluation codes \cite{boucher2014linear,liu2015construction}, optimal codes in skew and sum-rank metric \cite{martinez2018skew} (well-known as linearized Reed-Solomon codes), skew Reed-Muller codes \cite{geiselmann2019skew}, new optimal codes in the rank metric \cite{sheekey2020new}, skew convolutional codes \cite{sidorenko2020skew}, skew-cyclic codes \cite{gluesing2021introduction}, sum-rank BCH and cyclic-skew-cyclic codes \cite{martinez2021sum}, quantum codes \cite{aydin2016optimal, dinh2021class}, and twisted linearized Reed-Solomon codes \cite{neri2022twisted}.
Skew polynomials have been considered to construct cryptographic schemes \cite{coulter2001giesbrecht, boucher2010key}, Shamir's secret sharing scheme \cite{zhang2010secret}, DNA codes \cite{gursoy2017reversible} and maximal-recoverable locally repairable codes (MR-LRCs) \cite{gopi2022improved}.

\section{Linear Block Codes Constructed from Polynomials}
\label{sec:block-codes}
Let $\polyRing$ be a polynomial ring, a polynomial code can be simply defined as a set of polynomials fulfilling some conditions,
\begin{align*}
  \cC \defeq \{\left. f\in \polyRing\ \right|\ f \text{ fulfills certain conditions }\}\ .
\end{align*}
For instance, a Reed-Solomon (RS) code of dimension $k$ is a set of univariate polynomials of degree less than $k$ and a Reed-Muller (RM) code of dimension $\binom{m+d}{m}$ is a set of $m$-variate polynomials of total degree at most $d$. We denote these two set of polynomials respectively by
\begin{equation}
  \label{eq:RS-RM-poly-def}
  \begin{split}
    \RS_q(k) &\defeq \{\left. f\in\Fq[x]\ \right|\ \deg(f)<k\}\ ,\\
    \RM_q(d,m) &\defeq \{\left. f\in\F_q[x_1,\dots, x_m]\ \right|\ \deg(f)\leq d\}\ .
  \end{split}
\end{equation}

In most of the coding theory literature, as shown in the transmission model in \cref{fig:transmission-model}, the elements in a code, called \emph{codewords}, are considered to be a vector $\bc$ of length $n$ over a certain alphabet $\cA$ (e.g., $\Fq$ for the $\RS_q(k)$ code) and are mapped from an \emph{information} vector $\bu$ of length $k$.
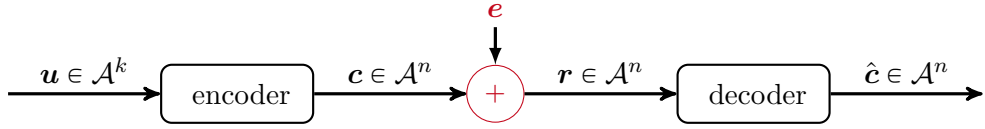
\begin{figure}[h]
  \centering
\begin{tikzpicture}
  \node (r1)[] {};
  \node (r2)[mybox_block,minimum width=2cm,right= 2cm of r1] {$\hspace{-2.4ex}$encoder};
  \node (r3)[circle,draw=TUMRed,right= 2cm of r2] {\textcolor{TUMRed}{$+$}};
  \node (r4)[mybox_block,minimum width=2cm,right= 2cm of r3] {$\hspace{-2.4ex}$decoder};
  \node (r5)[right = 2cm of r4] {};
  \node (c2)[above=0.5cm of r3] {\textcolor{TUMRed}{$\be$}};
  \path[-]
  (r1) edge[->,very thick] node[above] {$\bu\in\cA^k$} (r2)
  (r2) edge[->,very thick] node[above] {$\bc\in\cA^n$} (r3)
  (r3) edge[->,very thick] node[above] {$\br\in\cA^n$} (r4)
  (r4) edge[->,very thick] node[above] {$\hat{\bc}\in\cA^n$} (r5)
  (c2) edge[-latex, very thick] (r3);
\end{tikzpicture}
  \caption{Transmission model via a linear $[n,k]$ block code with an alphabet $\cA$.}
  \label{fig:transmission-model}
\end{figure}

Depending on how the symbols in $\bu$ and $\bc$ are related to a polynomial $f$ in a code $\cC$, we characterize polynomial codes into two categories: \emph{evaluation codes} and \emph{polycyclic codes}.

Throughout the thesis, when the code is seen as a set of polynomials as in \eqref{eq:RS-RM-poly-def}, we denote within the parenthesis ``$(\ )$'' the parameters regarding the polynomials, e.g., $k$ in $\RS_q(k)$ is the degree restriction of the polynomials. When the code is seen as a linear block code with fixed length $n$ and dimension $k$, we denote within the rectangular brackets ``$[\ ]$'' the parameters, e.g., $n$ and $k$ in $\RS_q[n,k]$ are the length and the dimension.
\subsection{Evaluation Codes}
\label{sec:evaluation-codes}
Evaluation codes are a class of codes obtained by evaluating polynomials at some evaluation points. The number of evaluation points determines the length of the codes and the degree restriction on the polynomials determines the dimension of the codes.

The RS codes $\RS_q(k)$ can be also defined as a block code of length $n$ via evaluating the polynomials at $n$ distinct points.
An $[n,k]_q$ RS code is defined as
\begin{align*}
  \RS_q[n,k]\defeq \set*{\left. \parenv*{f(\alpha_1),\dots, f(\alpha_n)}\ \right|\ f\in \F_q[x],\ \deg (f)<k}\ ,
\end{align*}
where $\alpha_1,\dots,\alpha_n$ are distinct elements in $\Fq$ and called \emph{code locators} of the RS codes.

Similarly, an RM code $\RM_q(d,m)$ is often seen as a set of $\Fq$-vectors of length $n=q^m$ via evaluating the polynomials at all the points in $\F_q^m$. Hence, $\RM_q(d,m)$ can be defined as an $[n,k]_q$ block code with $n=q^m$ and $k=\binom{m+d}{m}$,
\begin{align*}
  \RM_q[n,k] \defeq \{\parenv*{f(\bv)}_{\bv\in\F_q^m}\ | \ f\in\F_q[x_1,\dots, x_m],\ \deg(f)\leq d\}\ .
\end{align*}

For evaluation codes, with the alphabet $\cA$, each information vector $\bu\in\cA^k$ is associated to the coefficients of a polynomial of degree $k$ in $\cA[x]$, i.e.,
\begin{align*}
  (u_0, u_1,\dots,u_{k-1}) \mapsto f = u_0 + u_1x +\dots, u_{k-1} x^{k-1}\ .
\end{align*}
Evaluation codes are an important method to construct block codes from polynomials. One of the reasons is that the minimum Hamming distance can be derived by studying the roots of the underlying polynomials. The most studied maximum distance achievable codes in the Hamming, rank and sum-rank metric are all evaluation codes (see \cref{sec:metrics}).
The \emph{interpolation-based} decoders for evaluation codes are essentially developed from interpolating a polynomial from some of its evaluations.

\subsection{Polycyclic Codes}
\label{sec:polycyclic-codes}

Let $\ring$ be a commutative ring and $\polyRing= \ring[x]$ be a polynomial ring.
Denote by $\myspan{f}\subseteq\polyRing$ a principle ideal of $\polyRing$ generated by $f\in\polyRing$.
The \emph{quotient ring} (or \emph{factor ring}) of $\polyRing$ modulo the ideal $\myspan{f}$ is defined as
\begin{align*}
  \polyRing/\myspan{f}\defeq \set*{ u \Mod f\ |\ u\in\polyRing}\ .
\end{align*}

Informally speaking, a \emph{polycyclic code}, if seen as a set of polynomial, is an ideal in the quotient ring $\polyRing/\myspan{f}$. If seen as a set of vectors, it is the vector representation of the ideal.
\begin{definition}[Polycyclic codes]
  \label{def:polycyclic}
  Let $f\in\polyRing$ with $\deg(f)=n$, $g\in\polyRing$ with $\deg(g)=n-k$, and $g\divides f$.
  A \emph{polycyclic code} w.r.t.~$g, f$ is defined as the principle ideal $ \myspan{g}/\myspan{f}$ in $\polyRing/\myspan{f}$, i.e.,
  \begin{align*}
    \cC(g,f)\defeq \set*{u\cdot g \Mod f \ |\ u\in\polyRing}\ .
  \end{align*}
  A linear $[n,k]$ code is a \emph{polycyclic code} if
  \begin{align*}
    \cC[n,k] \defeq\set*{\left. \bc=(c_0, c_1,\dots, c_{n-1}) \ \right|\  \sum_{i=0}^{n-1}c_ix^i\in \cC(g,f)}\ . 
  \end{align*}
  The polynomial $g$ is a \emph{generator polynomial} of $\cC$.
\end{definition}


An information vector $\bu=(u_0,u_1,\dots,u_{k-1})$ corresponds uniquely to a codeword $\bc\in\cC[n,k]$ such that
\begin{align*}
  \parenv*{u_0+u_1 x +\cdots + u_{k-1} x^{k-1}}\cdot g = c_0+c_1x+\cdots+c_{n-1}x^{n-1}\ .
\end{align*}

\begin{remark}
  \label{remark:special-polycyclic}
  The name ``polycyclic'' is a generalization of the well-known cyclic codes.
  Polycyclic codes have the following special cases that have been defined and studied in the literature.
  \begin{itemize}
  \item If $f=x^n-1$, then $\cC$ is a \emph{cyclic code} such that
    \begin{align*}
      (c_{n-1}, c_0,\dots,c_{n-2})\in\cC\quad \textrm{ for all }(c_0,c_1,\dots, c_{n-1})\in\cC\ .
    \end{align*}
  \item If $f=x^n+1$, then $\cC$ is a \emph{negacyclic code} such that
    \begin{align*}
      (-c_{n-1}, c_0,\dots,c_{n-2})\in\cC \quad \textrm{ for all }(c_0,c_1,\dots, c_{n-1})\in\cC\ .
    \end{align*}
  \item If $f=x^n-a$ for some $a\in\ring^*$, then $\cC$ is a \emph{constacyclic code} such that
    \begin{align*}
      (ac_{n-1}, c_0,\dots,c_{n-2})\in\cC \quad \textrm{ for all }(c_0,c_1,\dots, c_{n-1})\in\cC\ .
    \end{align*}
  \end{itemize}
\end{remark}
Cyclic codes are an important class of codes in both theory and practice. Theoretically, there is rich mathematical theory in their properties, while practically, they are efficient to be encoded and decoded.
The study of cyclic codes over finite fields started from the reports by
Prange \cite{prange1957cyclic, prange1985some}. BCH codes are a class of cyclic codes with special generator polynomials which result in good distance property.
Cyclic codes over finite rings are extensively studied since the work by
Shankar \cite{shankar1979bch}, where the Chinese Remainder Theorem was used to investigate the BCH codes over integer rings.
In \cref{chap:mod_ring}, we introduce a class of polycyclic codes where $\polyRing$ is a skew polynomial ring. 
\section{Metrics}
\label{sec:metrics}
It can be seen from the last section that a \emph{code} is simply a set of vectors or polynomials.
In order to design error-correcting codes, we need to measure the distinction between the elements in the set.
A distance measure $\dist_{\sfM}(\cdot,\cdot)$ on a set $\cA$ is called a \emph{metric} if it fulfills the following conditions for all $a,b,c\in\cA$:
\begin{itemize}
\item Positive definiteness: $\dist_{\sfM}(a,b)\geq 0$ and the equality holds if and only if $a=b$.
\item Symmetry: $\dist_{\sfM}(a,b) = \dist_{\sfM}(b,a)$.
\item Triangle inequality: $\dist_{\sfM}(a,b)+\dist_{\sfM}(b,c)\geq \dist_{\sfM}(a,c)$.
\end{itemize}

In this thesis, several metrics are used to study the error-correction capability of the codes constructed from polynomials.

\subsection{Hamming Metric}
The Hamming metric is the most used distance measure for error-correcting codes. It can be used for any linear block codes.
\begin{definition}[Hamming metric]
  \label{def:Hamming-metric}
  The \emph{Hamming weight} on $\Fq^n$ is defined as
  \begin{align*}
    \wtH(\cdot)\ :\ \Fq^n&\to \bbN\\
    \ba & \mapsto |\set*{i\in[n]\ |\ a_i\neq 0}|\ .
  \end{align*}
  The \emph{Hamming distance} between two vectors is defined as
  \begin{align*}
    \dH(\cdot,\cdot) \ :\ \Fq^n\times \Fq^n &\to \bbN\\
    \ba,\bb &\mapsto |\set*{i\in[n]\ |\ a_i-b_i\neq 0}| = \wtH(\ba-\bb)\ .
  \end{align*}
  For a code $\cC\subseteq \Fq^n$, its \emph{minimum Hamming distance} is
  \begin{align*}
    \dH(\cC) &\defeq \min_{\substack{\bc_1,\bc_2\in\cC\\\bc_1\neq\bc_2}} \dH(\bc_1,\bc_2)\\
    &= \min_{\0\neq \bc\in\cC} \wtH(\bc)\ (\textrm{If }\cC \textrm{ is linear}).
  \end{align*}
\end{definition}

\begin{theorem}[Singleton bound~\cite{singleton1964maximum}]
  \label{thm:singleton-Hamming}
  For any block code $\cC\subseteq\Fq^n$ (linear or non-linear) with minimum Hamming distance $\dH(\cC)=d$,
  \begin{align*}
    |\cC|\leq q^{n-\dH(\cC)+1}\ .
  \end{align*}
  If $\cC$ is linear, then its dimension $k$ fulfills
  \begin{align*}
    k\leq n-d+1\ .
  \end{align*}
\end{theorem}
A code whose length, cardinality and minimum Hamming distance fulfill the Singleton bound is called \emph{maximum distance separable} (MDS).
(Generalized) Reed-Solomon (GRS) codes are the most studied class of MDS codes. They require a field $\Fq$ of size at least the code length $n$. There are not many other non-trivial (i.e., not $[n,1]$ or $[n,n]$) linear codes with this property and all of them are relatively short compared to their field size. In fact the famous MDS conjecture (\cite{segre1955ovals}\cite[Conjecture 6.13]{Ball2020}) asserts that within the range $4\leq k\leq q-2$, a $k$-dimensional MDS codes of length $n$ satisfies $n\leq q+1$.
\begin{theorem}[Sphere-packing~\cite{Hamming-1950} and Gilbert-Varshamov bounds~\cite{gilbert1952comparison,varshamov1957estimate}]
  Let $A^{\sfH}_q(n,d)$ be the maximum cardinality of a code 
  $\cC\subseteq\Fq^n$ with minimum Hamming distance $\dH(\cC)=d$. Then,
  \begin{align*}
    \frac{q^n}{\sum_{i=0}^{d-1}\binom{n}{i}(q-1)^i}\leq
    A_q^{\sfH}(n,d)
    \leq \frac{q^n}{\sum_{i=0}^{t}\binom{n}{i}(q-1)^i}\ ,
  \end{align*}
  where $t=\floor{\tfrac{d-1}{2}}$. The upper bound is the sphere-packing bound and the lower bound is the Gilbert-Varshamov (GV) bound.
\end{theorem}
A code whose parameters fulfill the sphere-packing bound is called a \emph{perfect} code. The GV bound is often referred as the \emph{random coding} bound. Namely, constructing a code by randomly taking $\tfrac{q^n}{\sum_{i=0}^{d-1}\binom{n}{i}(q-1)^i}$ distinct vectors from $\Fq^n$, with high probability, the code has minimal Hamming distance at least $d$.
\subsection{Rank Metric}
  Fix a basis $\bbeta=(\beta_1,\beta_2,\dots,\beta_m)$ of $\Fqm$ over $\Fq$. We define a mapping from $\Fqm^n$ to $\Fq^{m\times n}$ by
  \begin{align}
    \extbasis{\bbeta} : \Fqm^n&\to\Fq^{m\times n} \nonumber\\
    \bc =(c_1,c_2,\dots, c_n) &\mapsto \bC=
                                \begin{pmatrix}
                                  c_{1,1} & c_{1,2} & \dots & c_{1,n}\\
                                  \vdots & \vdots & \ddots & \vdots \\
                                  c_{m,1} & c_{m,2} & \dots & c_{m,n}
                                \end{pmatrix}\ ,
                                                              \label{eq:ext-map}
  \end{align}
  where $\bC$ is unique such that $c_j=\sum_{i=1}^{m}c_{i,j}\beta_i$, for all $j=1,\dots, n$.
  The $\Fq$-rank of $\bc$ is defined as $\rank_q(\bc)\defeq \rank(\bC)$.

\begin{definition}[Rank metric]
  The \emph{rank weight} on $\Fqm^n$ is defined as
  \begin{align*}
    \wtR(\cdot)\ :\ \Fqm^n&\to \bbN\\
    \ba & \mapsto \rank_q(\ba)\ .
  \end{align*}
  The \emph{Rank distance} between two vectors is defined as
  \begin{align*}
    \dR(\cdot,\cdot) \ :\ \Fqm^n\times \Fqm^n &\to \bbN\\
    \ba,\bb &\mapsto \rank_q(\ba-\bb) .
  \end{align*}
  For a code $\cC\subseteq \Fqm^n$, its \emph{minimum rank distance} is
  \begin{align*}
    \dR(\cC)& \defeq \min_{\substack{\bc_1,\bc_2\in\cC\\\bc_1\neq\bc_2}} \dR(\bc_1,\bc_2)\\
    &= \min_{\0\neq \bc\in\cC} \wtR(\bc)\ (\textrm{If }\cC \textrm{ is linear}).
  \end{align*}
\end{definition}

\begin{theorem}[Singleton bound in rank metric~{\cite[Theorem 5.4]{delsarte1978bilinear}}]
For a code $\cC\subseteq\Fqm^n$ with minimum rank distance $\dR(\cC)=d$,
\begin{align*}
  |\cC|\leq q^{\min\set*{n(m-d+1), m(n-d+1)}} = q^{\max\set*{n,m}(\min\set*{n,m}-d+1)}\ .
\end{align*}
  If $\cC$ is $\Fqm$-linear, then its dimension $k$ over $\Fqm$ fulfills
  \begin{align*}
    k\leq n-d+1\ .
  \end{align*}
\end{theorem}
A code whose length, cardinality and minimum rank distance fulfilling the Singleton bound is a \emph{maximum rank distance} (MRD) code.
\emph{Gabidulin codes} \cite{delsarte1978bilinear, gabidulin1985theory, roth1991maximum} are the most well-known MRD codes. They are a class of evaluation codes based on \emph{linearlized polynomials} (a special class of skew polynomials).

\begin{theorem}[Sphere-packing and Gilbert-Varshamov bounds in rank metric~{\cite{gadouleau2006genp1}}]
  Let $A^{\sfR}_{q^m}(n,d)$ be the maximum cardinality of a code 
  $\cC\subseteq\Fqm^n$ with minimum rank distance $\dR(\cC)=d$. Then,
  \begin{align*}
    \frac{q^{mn}}{\card{\cB_{\sfR}^{(d-1)}}}\leq
    A^{\sfR}_{q^m}(n,d)
    \leq \frac{q^{mn}}{\card{\cB_{\sfR}^{(t)}}}\ ,
  \end{align*}
  where $t=\floor{\tfrac{d-1}{2}}$ and $\card{\cB_{\sfR}^{(\tau)}}$ is a set (often called \emph{ball}) of all the vectors of rank distance at most $\tau$ to a fixed vector $\bb\in\Fqm^n$ (e.g., $\bb=\0$), i.e.,
  \begin{align*}
    \cB_{\sfR}^{(\tau)} &\defeq \set*{\ba\in\Fqm^n\ |\ \wtR(\ba)\leq \tau}\ \textrm{and}\\
    \card{\cB_{\sfR}^{(\tau)}} &=\sum_{i=0}^{\tau}\quadbinom{m}{i}_q\prod_{j=0}^{i-1}(q^n-q^j) \textrm{ with }
    \quadbinom{m}{i}_q = \prod\limits_{j=0}^{i-1} \frac{q^m-q^j}{q^i-q^j}\ .
  \end{align*}
\end{theorem}
\subsection{Sum-Rank Metric}
The sum-rank metric was first considered in coding for MIMO (multiple-input multiple-output) block-fading channels \cite{el2003design, lu2005unified} and the design of PSK-AM (phase-shift keying with amplitude modulation) constellations \cite{lu2006constructions}. It was then explicitly introduced in multi-shot network coding \cite{nobrega2009multishot}.
An explicit construction of optimal space-time codes in terms of \emph{rate-diversity trade-off} from sum-rank metric codes over a finite field was first given in \cite{shehadeh2021space}. Additionally, sum-rank metric codes have been considered in applications such as network streaming \cite{mahmood2016convolutional}, distributed storage systems \cite{martinez2019universal,martinez2020locally,cai2021construction} and post-quantum secure code-based cryptosystems \cite{dalconzo2022codeequivalence,hormann2022security}.
Extensive research has been done in recent years in fundamental coding-theoretical properties of sum-rank metric codes, e.g., \cite{martinez2019theory,byrne2020fundamental, ott2021bounds,camps2022optimal, ott2022covering, ott2023geometrical}, constructions of perfect/optimal/systematic sum-rank metric codes \cite{martinez2020hamming, almeida2020systematic, martinez2022generalMSRD, alfarano2022sum, caruso2022duals}.

Given an ordered partition $\bn_\ell =(n_1, \dots,n_\ell)$ of $n\in\bbN$, we write a vector $\ba$ of length $n$ with respect to $\bn_{\ell}$ as
\begin{align*}
  \ba = (\ba_1,\ba_2,\dots,\ba_{\ell})\ ,
\end{align*}
where $\ba_i$ is of length $n_i$, for all $i\in[\ell]$.
\begin{definition}[Sum-rank metric]
  \label{def:sum-rank-metric}
  The \emph{sum-rank weight} on $\Fqm^n$, w.r.t.~an ordered partition $\bn_\ell =(n_1, \dots,n_\ell)$ of $n$, is defined as
\begin{align*}
    \wtSR{\bn_\ell}(\cdot)\ :\ \Fqm^n &\to \bbN\\
    \ba &\mapsto \sum_{i=1}^\ell \rank_q(\ba_i)\ .
\end{align*}
The \emph{sum-rank distance} between two vectors is defined as
\begin{align*}
    \dSR{\bn_\ell}(\cdot,\cdot)\ :\ \Fqm^n\times\Fqm^n &\to \bbN\\
    \ba,\bb &\mapsto \wtSR{\bn_\ell}(\ba-\bb)\ .
\end{align*}
For a code $\cC\subseteq\Fqm^n$, its \emph{minimum sum-rank distance} is
\begin{align*}
  \dSR{\bn_\ell}(\cC)&\defeq \min_{\substack{\bc_1,\bc_2\in\cC\\\bc_1\neq\bc_2}} \dSR{\bn_\ell}(\bc_1,\bc_2) \\
  & = \min_{\0\neq \bc\in\cC} \wtSR{\bn_\ell}(\bc)\ (\textrm{If }\cC \textrm{ is linear})\ .
\end{align*}
\end{definition}
It is known that for $\ell=1$, the sum-rank metric coincides with the rank metric, and for $\ell=n$, the sum-rank metric is the Hamming metric \cite[Proposition 1.4, 1.5]{FnTsurvey-Umberto}.
The following lemma gives a relation among the Hamming, the sum-rank and the rank weights of a fixed vector $\bx\in\Fqm^n$.
\begin{lemma}
  \label{lem:sum-rank-Hamming}
  For a vector $\bx\in\Fqm^n$ and any ordered partition $\bn_\ell =(n_1, \dots,n_\ell)$ of $n$, $\wtR(\bx)\leq \wtSR{\bn_{\ell}}(\bx)\leq \wtH(\bx)$.
\end{lemma}
\begin{proof}
  We first show that $\wtSR{\bn_{\ell}}(\bx)\leq \wtH(\bx)$. Consider $\bx=(\bx_1,\dots,\bx_\ell)\in \Fqm^n$ with $\wtH(\bx)=n-t=\sum_{i=1}^{\ell}n_i-t_i$ where $t_i$ is the number of zero entries in $\bx_i$ and $\sum_{i=1}^\ell t_i=t$.
  For any $i\in[\ell]$, $\rank_q(\bx_i)\leq \wtH(\bx_i)$ since the rank of a matrix is at most the number of its nonzero columns.
  By the definition of $\wtSR{\bn_{\ell}}$, we have $\wtSR{\bn_{\ell}}(\bx)=\sum_{i=1}^\ell\rank_q(\bx_i)\leq \sum_{i=1}^\ell n_i-t_i=n-t=\wtH(\bx)$.

  Now we show $\wtR(\bx)\leq \wtSR{\bn_{\ell}}(\bx)$. Suppose $\wtSR{\bn_{\ell}}(\bx)=t=\sum_{i=1}^\ell t_i$, which means each $\bx_i$ has $t_i$ $\Fq$-linearly independent entries for $i\in[\ell]$. Then $\bx$ has at most $t$ $\Fq$-linearly independent entries, which corresponds to the rank weight of $\bx$.
\end{proof}
It can be seen that $\Fqm^n$ is isometric to $\Fq^{m\times n}$ or $\Fq^{n\times m}$ as $\Fq$-vector spaces. The definition of the sum-rank metric on $\Fq^{m\times n}$ and $\Fq^{n\times m}$ follows naturally from \cref{def:sum-rank-metric}.
For later usage, we give the definition of sum-rank metric on these matrix spaces in the following.
For simplicity, we abuse the notation $\wtSR{\bn_\ell}$ for the matrix space.

Let $\ell\in\bbN$ and $\bn_\ell= (n_1, \dots,n_\ell)\in\bbN^\ell$ be an ordered partition of $n=\sum_{\indBlock=1}^\ell n_\indBlock$.
For matrices 
\begin{align*}
  \bA =
  \begin{pmatrix}
    \bA_1 & \bA_2 & \dots & \bA_\ell
  \end{pmatrix}
                        \in\Fq^{m\times n}
                        \quad\quad\text{ or }\quad\quad
                        \bB =
                        \begin{pmatrix}
                          \bB_1\\
                          \bB_2\\
                          \vdots\\
                          \bB_\ell
                        \end{pmatrix}
  \in\Fq^{n\times m}\ ,
\end{align*}
where $\bA_i\in\Fq^{m\times n_i}$ and $\bB_i\in\Fq^{n_i\times m}$, $i\in[\ell]$, we say $\bA$ has a \emph{column-wise} partition with respect to $\bn_{\ell}$ and $\bB$ has a \emph{row-wise} partition w.r.t.~$\bn_{\ell}$. The sum-rank weights of $\bA$ and $\bB$ w.r.t.~$\bn_\ell$ are, respectively,
\begin{align*}
  \wtSR{\bn_\ell}(\bA)\defeq \sum_{\indBlock=1}^{\ell} \rank(\bA_\indBlock)
  \quad\quad\text{ and }\quad\quad
  \wtSR{\bn_\ell}(\bB)\defeq \sum_{\indBlock=1}^{\ell} \rank(\bB_\indBlock)\ .
\end{align*}
We can find the following relation between the sum-rank weight and the rank of a matrix.
\begin{lemma}
  \label{lem:rank-leq-sumrank-mat}
  For a matrix $\bA\in\Fq^{m\times n}$ and an ordered partition $\bn_\ell= (n_1, \dots,n_\ell)$ of $n$, $\rank(\bA)\leq\wtSR{\bn_{\ell}}(\bA)\leq \ell\cdot \rank(\bA)$. Similarly, for a matrix $\bB\in\Fq^{n\times m}$, $\rank(\bB)\leq \wtSR{\bn_{\ell}}(\bB)\leq \ell\cdot \rank(\bB)$.
\end{lemma}
\begin{proof}
  Denote by $\myspan{\bA}_{\sfc}$ the column space of a matrix $\bA$. For the first inequality,
  \begin{align*}
    \rank(\bA) = \dim(\myspan{\bA}_{\sfc}) &= \dim(\myspan{\bA_1}_{\sfc}+\cdots + \myspan{\bA_\ell}_{\sfc})\\
                                       &\leq \dim(\myspan{\bA_1}_{\sfc}) + \cdots \dim(\myspan{\bA_\ell}_{\sfc})\\
    &=\rank(\bA_1) +\cdots +\rank(\bA_\ell) = \wtSR{\bn_\ell}(\bA)\ .
  \end{align*}
  For the second inequality,
  \begin{align*}
    \wtSR{\bn_\ell}(\bA)= \sum_{i=1}^{\ell}\rank(\bA_i) 
    \leq \sum_{i=1}^{\ell}\rank(\bA) =\ell\cdot \rank(\bA)\ .
  \end{align*}
  For the matrix $\bB\in\Fq^{n\times m}$ with a row-wise partition, the proof is similar, by considering the row space of $\bB$ and $\bB_i$'s.
\end{proof}
With the relation between the sum-rank metric and the Hamming metric in \cref{lem:sum-rank-Hamming}, the following Singleton bound for sum-rank-metric codes can be easily derived from the Singleton bound for Hamming-metric codes in \cref{thm:singleton-Hamming}.
\begin{theorem}[Singleton bound in the sum-rank metric {\cite[Theorem 1.4]{FnTsurvey-Umberto}}]
  Let $\bn_{\ell}$ be an ordered partition of $n\in\bbN$.
  For a code $\cC\subseteq\Fqm^n$ (linear or non-linear) with minimum sum-rank distance $\dSR{\bn_{\ell}}(\cC)=d$,
\begin{align*}
  |\cC|\leq q^{m(n-d+1)}\ .
\end{align*}
  If $\cC$ is $\Fqm$-linear, then its dimension $k$ over $\Fqm$ fulfills
  \begin{align*}
    k\leq n-d+1\ .
  \end{align*}
  The equality holds in both equations if and only if $\cC \bA\defeq \set*{\bc \bA\ |\ \bc\in\cC}$ is an MDS code, i.e., $\dH(\cC\bA)=d$, for all $\bA=\diag(\bA_1,\dots,\bA_\ell)\in\Fq^{n\times n}$ where every $\bA_i\in\Fq^{n_i\times n_i},i\in[\ell]$ is invertible.
\end{theorem}
A code whose length, cardinality and minimum sum-rank distance fulfilling the Singleton bound is a \emph{maximum sum-rank distance} (MSRD) code.
\emph{Linearized Reed-Solomon codes} \cite{martinez2018skew,caruso2018reed} are a class of MSRD codes which gained a lot of research interest. They are evaluation codes based on skew polynomials.
\begin{theorem}[Sphere-packing and Gilbert-Varshamov bounds in sum-rank metric~{\cite{byrne2020fundamental}}]
  Let $\bn_{\ell}$ be an ordered partition of $n\in\bbN$ and 
  $A^{\sfS\sfR}_{q^m}(n,d)$ be the maximum cardinality of a code 
  $\cC\subseteq\Fqm^n$ with minimum sum-rank distance $\dSR{\bn_{\ell}}(\cC)=d$. Then,
  \begin{align*}
    \frac{q^{mn}}{\card{\cB_{\sfS\sfR}^{(d-1)}}}\leq
    A^{\sfS\sfR}_{q^m}(n,d)
    \leq \frac{q^{mn}}{\card{\cB_{\sfS\sfR}^{(t)}}}\ ,
  \end{align*}
  where $t=\floor{\tfrac{d-1}{2}}$ and $\card{\cB_{\sfS\sfR}^{(\tau)}}$ is a set (often called \emph{ball}) of all the vectors of sum-rank distance at most $\tau$ to a fixed vector in $\Fqm$ (e.g., $\0$), i.e.,
  \begin{align}
    \cB_{\sfS\sfR}^{(\tau)} &\defeq \set*{\ba\in\Fqm^n\ |\ \wtSR{\bn_{\ell}}(\ba)\leq \tau}\ \textrm{and}\nonumber\\
    \card{\cB_{\sfS\sfR}^{(\tau)}} &=\sum_{s=0}^{\tau}\sum_{\substack{(s_1,\dots,s_{\ell})\in\bbN^{\ell}\\s_1+\cdots+s_{\ell}=s}}\prod_{i=1}^{\ell}\quadbinom{n_i}{s_i}_q\prod_{j=0}^{s_i-1}(q^m-q^j)\ . \label{eq:sum-rank-ball-size}
  \end{align}
\end{theorem}
Due to the second sum over all the ordered partitions of $s$ in \eqref{eq:sum-rank-ball-size}, computing the ball size is expensive. An efficient algorithm is given in \cite[Algorithm 1]{puchinger2022generic} for computing the $\card{\cB_{\sfS\sfR}^{(\tau)}(\0)}$ with complexity $O\parenv*{\tau\parenv*{\ell d^3+d^4(\ell m+n)\log(q)}}$.
Simplified forms and asymptotic behaviors of the sphere-packing bound and the Gilbert-Varshamov bound in the sum-rank metric are given in \cite{ott2021bounds}.


%% file: chap_mod_skew.tex
Finite rings are considered to be possible alphabets for linear codes first by Assmus Jr.~and Mattson \cite{assmus1963error}.
\citeauthor{blake1972codes} investigated in \cite{blake1972codes} the structure of cyclic codes over $\IntRing_m$ and studied in \cite{blake1975codes} the analogues to Hamming, Reed-Solonmon and BCH codes over $\IntRing_{p^r}$.
Spiegel \cite{spiegel1977codes,spiegel1978codes} generalized \citeauthor{blake1972codes}'s results to any integer ring $\IntRing_m$ by using the Chinese Remainder Theorem.
The interest in codes over finite rings has been evoked since the works by
Carlderbank et al.~\cite{calderbank1993linear,hammons1994z4}, which use linear codes over $\IntRing_4$ to explain the duality between the nonlinear binary Kerkock and Preparata codes.
The works by
Wood \cite{wood1999duality, wood2008code,wood2009foundations} laid a foundation of algebraic coding theory over finite rings by extending the two classical theorems by
MacWilliams \cite{macwilliams1961error, macwilliams1962combinatorial} to codes over finite rings.
The extension theorem is also known as the equivalence theorem and the MacWilliams identities deal with the relation between the Hamming weight enumerators of a linear code and its dual.
Self-dual codes (or in general, dual-containing codes) have attracted a lot of research interest since the work by Calderbank et al.~\cite{calderbank1998quantum}, which transformed the problem of constructing quantum error-correcting (QEC) codes into the problem of finding classical additive codes which are dual-containing. Several QEC codes have been constructed from classical codes, such as BCH codes \cite{aly2007quantum}, Reed-Solomon codes \cite{li2008quantum}, Reed-Muller codes \cite{steane1999quantum} and algebraic geometric codes \cite{chen2005quantum}.
Many good QEC codes have been constructed from cyclic codes over finite rings \cite{qian2009quantum,kai2011quaternary,guenda2014quantum,tang2016new, bag2019new}. Constacyclic codes and negacyclic codes, as generalizations of cyclic codes (cf.~\cref{remark:special-polycyclic}), have also been used to construct quantum codes \cite{chen2015application,gao2018uConsta, wang2020some,alahmadi2021new}.

\emph{Skew-cyclic codes} (also called $\Endom$-cyclic code) are a class of polycyclic codes $\cC(g,f)$ where $g,f$ are polynomials in a skew polynomial ring $\ring[\SkewVar;\Endom]$ and $f=\SkewVar^n-1$, so that the \emph{$\Endom$-cyclic shift} of any codeword in a skew-cyclic code $\cC$ is also a codeword, i.e.,
  \begin{align*}
    \parenv*{\Endom(c_{n-1}), \Endom(c_0), \dots, \Endom(c_{n-2})} \in\cC \textrm{ for all }(c_0,c_1,\dots, c_{n-1})\in\cC\ .
  \end{align*}
  The concept of skew-cyclic codes was introduced over finite fields by Boucher, Geiselmann and Ulmer \cite{boucher2007skew}. Analogue to \cref{def:polycyclic} of the polycyclic codes, each skew-cyclic code corresponds to a right divisor $g$ of $f=\SkewVar^n-1$.
  Since skew polynomials do not necessarily have unique irreducible factorizations, a polynomial $\SkewVar^n-1$ may have a considerable number of distinct right divisors of the same degree, which leads to many skew cyclic codes. Therefore, there is better chance to obtain codes with good parameters. This was one motivation of \cite{boucher2007skew} to introduce the notion of skew-cyclic codes. The works \cite{boucher2008skew, boucher2009coding, boucher2009modules} further investigate skew-cyclic codes over finite rings.
  The notions of \emph{skew-constacyclic} and \emph{skew-negacyclic} are the generalizations with $f=\SkewVar^n-a, a\in\ring^*$ and $f=\SkewVar^n+1$, respectively.
The skew polynomials ring $\RingSkewPolys$ with nonzero derivation was first considered by Boucher and Ulmer \cite{boucher2014linear} to construct \emph{$(\Endom,\Deriv)$-cyclic codes}, where the following \emph{$(\Endom,\Deriv)$-cyclic shift} of a codeword is also a codeword:
  \begin{align*}
    \parenv*{\Endom(c_{n-1})+ \Deriv(c_0), \Endom(c_0)+\Deriv(c_1), \dots, \Endom(c_{n-2})+\Deriv(c_{n-1})} \in\cC \textrm{ for all }(c_0,c_1,\dots, c_{n-1})\in\cC.
  \end{align*}
  Boulagouaz and Leroy \cite{boulagouaz2013delta} generalized the notion to \emph{$(\Endom,\Deriv)$-polycyclic codes}. These works both considered codes over finite fields.
  Since the work \cite{sharma2018class} by Sharma and Bhaintwal, several works also investigated the $(\Endom,\Deriv)$-polycyclic codes over finite rings \cite{ma2021sigma,suprijanto2021skew, ma2022xn, patel2022theta, suprijanto2023linear}.
  Recently, $\Endom$-cyclic codes over various rings have been used to construct quantum codes \cite{dinh2021class,verma2022new,prakash2023quantum}, as well as $(\Endom,\Deriv)$-polycyclic codes \cite{patel2022f}.

  This chapter considers constructions of dual-containing codes over finite commutative rings from skew polynomials with derivations.
  We first introduce in \cref{sec:rings-ord-four} the rings over which we search for dual-containing codes.
  In \cref{sec:skew-polycyclic-codes}, we define the $(\Endom,\Deriv)$-polycyclic codes and introduce some properties regarding the codes and their dual codes. We then present in \cref{sec:algo-search-via-GB} an algorithm using Gr\"obner bases to compute all the dual-containing $(\Endom,\Deriv)$-polycyclic codes. The resulting codes found by the algorithm are presented in \cref{sec:computation-results}.

\emph{This chapter is based on the work~\cite{liu2023gr}, submitted to Advances in Mathematics of Communications.}

\section{Base Rings, Endomorphisms and Derivations}
\label{sec:rings-ord-four}
Throughout the chapter, we denote by $\ring$ a finite commutative ring, $\Endom$ an endomorphism of $\ring$, $\Deriv$ a $\Endom$-derivation of $\ring$, and $\polyRing=\ring[X;\Endom,\Deriv]$ a skew polynomial ring.

For the base ring $\ring$ over which dual-containing codes are constructed, we consider the finite commutative ring $\ring$ that is a \emph{free $\subring$-algebra}, denoted by $\ring=\subring[\beta_1,\beta_2,\dots,\beta_s]$, where $\beta_1,\beta_2,\dots,\beta_s$ form a \emph{basis} of $\ring$ over $\subring$. This means that any element $a\in\ring$ can be written as a $\subring$-linear combinations of $\beta_1,\beta_2,\dots,\beta_s$, i.e., $a=b_1\beta_1+\dots + b_s\beta_s$ for some $b_1,\dots,b_s\in\subring$.

In this chapter we consider the following rings $\ring$ of order $4$, which are all free $\F_2$-algebras, 
to construct self-dual codes $\cC\subseteq \ring^n$:
\begin{itemize}
\item $\ring=\F_2[v]$ where $v^2=v$,
\item $\ring=\F_2[u]$ where $u^2=0$,
\item $\ring=\F_4= \F_2[\alpha]$ where $\alpha^2 = \alpha+1$. 
\end{itemize}
Here we add some insights on the notations.
For the finite field $\F_4$, it is well-known that it is also a quotient ring $\F_2[x]/\myspan{x^2-x-1}$ where $x^2-x-1$ is an irreducible polynomial in $\F_2[x]$, and $\set*{1,\alpha}$ constitutes a basis of $\F_4$ over $\F_2$ ($\alpha$ is often called the primitive element in the context of finite fields).
Analogously, the ring $\F_2[v]$ (resp.~$\F_2[u]$) is the quotient ring $\F_2[x]/\myspan{x^2-x}$ ($\F_2[x]/\myspan{x^2}$), and $v$ ($u$) constitutes a basis of the ring over $\F_2$.
Since $x^2-x$ (resp.~$x^2$) is not irreducible in $\F_2[x]$, there are zero divisors $\set*{v, v+1}$ ($\set*{u}$) in the ring.

For the endomorphisms $\Endom$ and $\Endom$-derivations $\Deriv$ of $\ring$ that we use to define the skew polynomial ring $\ring[\SkewVar; \Endom,\Deriv]$, we consider those that can be written as \emph{polynomial maps} in the subring $\subring\subseteq \ring$.
\begin{definition}[Polynomial maps]
  \label{def:poly-maps}
  A \emph{polynomial map} on a ring $\subring$ is a map
  \begin{align*}
    f\ :\ \subring & \to \subring\\
    x & \mapsto  \sum_{i=0}^s b_ix^i,
  \end{align*}
  where $s\in \bbN$ and $b_i\in \subring$.
\end{definition}
For the rings $\F_2[v],\F_2[u]$ and $\F_2[\alpha]$, the endomorphisms $\Endom$ and derivations $\Deriv$ which are polynomial maps in the subring $\subring=\F_2$ are listed in \cref{tab:endo-deri-rings}.
Note that since $\Endom(0)=0, \Endom(1)=1, \Deriv(0)=0$ and $\Deriv(1)=0$, it is sufficient to give the map on the basis ($v$, $u$ or $\alpha$) of $\ring$ over $\subring$ to determine the value $\Endom(a)$ and $\Deriv(a)$ for all $a\in\ring$.
\begin{table}[htb!]
  \centering
  \caption{Endomorphisms and derivations of the rings. The gray cells indicate the inner derivations.}
  \label{tab:endo-deri-rings}
  \begin{subtable}[h!]{\linewidth}
    \centering
    \caption{The endomorphisms $\Endom$ and derivations $\Deriv$ of the ring $\F_2[v]$.}
    \label{subtab:F2v}
    \begin{tabular}{|c||l|l||l|l|}
      \hline
       $\F_2[v]$& \multicolumn{2}{c||}{Automorphism}  & \multicolumn{2}{c|}{Endomorphism} \\ \hline
          & $\Endom_1=\Id$& $\Endom_2:v\mapsto v+1$ & $\Endom_3:v\mapsto 0$  & $\Endom_4:v\mapsto 1$  \\ \hline \hline
      $\Deriv_1=0$&  \cellcolor{gray!30}{$v\mapsto 0$} & \cellcolor{gray!30}{$v\mapsto 0$}
                                                 &  \cellcolor{gray!30}{$v\mapsto 0$}
          & \cellcolor{gray!30}{$v\mapsto 0$} \\ \hline
      $\Deriv_2$ & & \cellcolor{gray!30}{$v\mapsto 1$}
                                                 & & \\ \hline
      $\Deriv_3$ & & \cellcolor{gray!30}{$v\mapsto v$}
                                                 & \cellcolor{gray!30}{$v\mapsto v$}
          & \\ \hline
      $\Deriv_4$ & & \cellcolor{gray!30}{$v\mapsto v+1$}
                                                 & & \cellcolor{gray!30}{$v\mapsto v+1$}
      \\ \hline
    \end{tabular}
\end{subtable}
\hfil
  \begin{subtable}[h!]{0.5\linewidth}
    \centering
    \caption{The endomorphisms $\Endom$ and derivations $\Deriv$ of the ring $\F_2[u]$.}
    \label{subtab:F2u}
    \begin{tabular}{|c||l||l|l|l|}
      \hline
      $\F_2[u]$ &Automorphism &Endomorphism \\ \hline
          & $\Endom_1=\Id$& $\Endom_2 :u\mapsto 0$  \\ \hline \hline
      $\Deriv_1=0$ &  \cellcolor{gray!30}{$u\mapsto 0$} &   \cellcolor{gray!30}{$u\mapsto 0$} \\ \hline
      $\Deriv_2$ & $u\mapsto 1$ &    \\ \hline
      $\Deriv_3$ & $u\mapsto u$ & \cellcolor{gray!30}{ $u\mapsto u$}  \\ \hline
      $\Deriv_4$ & $u\mapsto u+1$ &  \\ \hline
    \end{tabular}
  \end{subtable}\hfil %
  \begin{subtable}[h!]{0.41\linewidth}
    \centering
    \caption{The endomorphisms $\Endom$ and derivations $\Deriv$ of the ring $\F_2[\alpha]$.}
    \label{subtab:F2alpha}
    \begin{tabular}{|c||l||l|l|l|}
      \hline
      $\F_2[\alpha]$ &\multicolumn{2}{c|}{Automorphism} \\ \hline
      & $\Endom_1=\Id$ & $\Endom_2:\alpha\mapsto\alpha+1$ \\ \hline \hline
      $\Deriv_1=0$&   \cellcolor{gray!30}{$\alpha\mapsto 0$} &  \cellcolor{gray!30}{$\alpha\mapsto 0$}  \\ \hline
      $\Deriv_2$ & &  \cellcolor{gray!30}{$\alpha\mapsto 1$}   \\ \hline
      $\Deriv_3$ & &  \cellcolor{gray!30}{$\alpha\mapsto \alpha$} \\ \hline
      $\Deriv_4$ & &  \cellcolor{gray!30}{$\alpha\mapsto \alpha+1$} \\ \hline
    \end{tabular}
  \end{subtable}\hfil
\end{table}

The following examples show that a map which is a polynomial map in $\subring$ is not necessarily a polynomial map in $\ring \supset \subring$.
The principle to determine whether a map $\Endom$ (or $\Deriv$) is a polynomial map in $\ring$ is to check whether $\Endom$ can be written as $\Endom(x)=\sum_{i\in\bbN}a_ix^i$ with fixed $a_i\in\ring$ for all $x\in\ring$.

\begin{example}
  \label{eg:not-poly-map-F2v}
  Consider the ring $\ring=\F_2[v]$ of order $4$.
  As listed in \cref{subtab:F2v}, there are two automorphisms $\Endom_1=\Id$ and $\Endom_2$, and two non-trivial endomorphisms $\Endom_3$ and $\Endom_4$.

  The automorphism $\Endom_1=\Id$ is trivially a polynomial map on $\ring$.
  Suppose that the automorphism $\Endom_2$ is a polynomial map on $\ring$ such that for any $x\in\ring$,
  \begin{align*}
    \Endom_2: x\mapsto \sum_{i\in\bbN} \underbrace{(b_{i,0}+b_{i,1} v)}_{\in\ring} x^i\, = \,  \sum_{i\in\bbN} b_{i,0} x^i+  \sum_{i\in\bbN} b_{i,1}v x^i  \qquad (b_{i,j}\in\F_2)\ .
  \end{align*}
  Then $\Endom_2(0)=0 \implies  b_{0,0}=0$. Since $b_{i,j}\in \set*{0,1}$, $\Endom_2(v)$ is a sum of positive powers of $v$. Since $v^2=v$, we have that $\Endom_2(v)$ is a sum of $v$, which is either $v$ or $0$ in this ring of characteristic $2$. However, since $\Endom_2(v)=v+1$, we conclude that $\Endom_2$ is not a polynomial map on $\ring$.
\end{example}
\begin{example}
  \label{eg:not-poly-map-F2u}
  Consider the ring $\ring=\F_2[u]$. As listed in \cref{subtab:F2u}, the only automorphism of $\ring$ is $\Endom_1=\Id$, which is a polynomial map in $\ring$.
  Suppose that a $\Endom_1$-derivation $\Deriv$ of $\ring$ is a polynomial map in $\ring$ such that for any $x\in\ring$,
  \begin{align*}
    \Deriv:x\mapsto\sum_{i=0}^t a_ix^i \qquad (a_i\in\ring)\ .
  \end{align*}
  Since $\Deriv(1)=0$, we must have $a_0=0$ in the polynomial map. From $u^2=0$, we obtain $\Deriv(u)= \sum_{i=1}^t a_iu^i = a_1u$. Write $a_1=b_{1,0}+b_{1,1}u\in \ring$ with some $b_{1,0},b_{1,1}\in \F_2$. Then $\Deriv(u)= b_{1,0} u+b_{1,1}u^2=b_{1,0}u$, which can never be $u+1$ or $1$.
  Hence, $\Deriv_2(u)=1$ and $\Deriv_4(u)=u+1$ are not polynomial maps on $\ring$.
\end{example}

\section{Polycyclic Codes over Rings based on Skew Polynomials}
\label{sec:skew-polycyclic-codes}
Denote by $\myspan{f}_l\subseteq\polyRing$ a left ideal in $\polyRing$ generated by $f\in\polyRing$.
The quotient ring (or factor ring) of $\polyRing$ modulo the left ideal $\myspan{f}_l$ is defined as
\begin{align*}
  \polyRing/\myspan{f}_l\defeq \set*{\left. h \Mod f\ \right|\ h\in\polyRing}\ ,
\end{align*}
where $h \Mod f$ gives the remainder of right dividing $h$ by $f$ (applying \cref{algo:EAskew}).
\begin{proposition}
  \label{prop:principle-ideal-g-rdiv-f}
  Let $\cI$ be a left ideal in $\polyRing/\myspan{f}_l$. Then
  \begin{enumerate}
  \item There is a unique monic polynomial $g\in\cI$ of minimal degree. \label{propitem:g-min-deg}
  \item $\cI$ is principle with a generator $g$. \label{propitem:I-principle}
  \item $g\divides_r f$ in $\polyRing$. \label{propitem:g-rdiv-f}
  \end{enumerate}
\end{proposition}
\begin{proof}
  \textit{\ref{propitem:g-min-deg}} Suppose that there are two monic skew polynomials with minimum degree in $\cI$, say $h$ and $g$ and $g\neq h$. Since $\cI$ forms an additive group, $r=g-h\in\cI$ and $\deg(r)<\deg(g)=\deg(h)$, which is contrary to the minimal degree assumption on $g$ and $h$.

  \textit{\ref{propitem:I-principle}} Suppose that there exists an $h\in \cI$ which is not a right multiple of $g$. Then $h=mg+r$ for some $m\in\polyRing$, $r\in\cI$ and $\deg(r)<\deg(g)$. This contradicts the degree minimality of $g$.

 \textit{\ref{propitem:g-rdiv-f}} Suppose that $g\not\divides_r f$. We can then write $f=mg+r$ for some $m,r\in\polyRing$ and $\deg(r)<\deg(g)$. Since $g$ is a generator of $\cI$, $mg\in\cI$ and $r= f-mg\equiv -mg\Mod f$, which means $r$ is an additive inverse of $mg$ in $\polyRing/\myspan{f}_l$ and therefore $r\in\cI\subseteq\polyRing/\myspan{f}_l$. However, this contradicts the degree minimality of $g$.
\end{proof}
The results above shows that every left ideal in the quotient ring $\polyRing/\myspan{f}_l$ is a principle left ideal and different left ideals are generated by different right factors of $f$.

Analogue to the polycyclic codes (\cref{def:polycyclic}) which are ideals in the quotient ring $\ring[x]/\myspan{f}$ for some $f\in\ring[x]$, for a non-commutative skew polynomial ring, we define the polycyclic codes as the left ideals in the quotient ring $\polyRing/\myspan{f}_l$ for some $f\in\polyRing$.

\begin{definition}[$(\Endom,\Deriv)$-polycyclic code]
  \label{def:skewPolycyclicCodes}
  Let $f\in \polyRing$ be a monic skew polynomial with a right divisor $g\in \polyRing$.
  A \emph{$(\Endom,\Deriv)$-polycyclic code} (in short $(\Endom,\Deriv)$-code) w.r.t.~$g,f$ is defined as a left ideal in $\polyRing/\myspan{f}_l$ generated by $g$, 
  i.e.,
  \begin{align*}
    \cC(g,f)\defeq \myspan{g}_l/\myspan{f}_l=\set*{u\cdot g \Mod f\ |\ u\in\polyRing }\ .
  \end{align*}
  The polynomial $g$ is a \emph{generator polynomial} of $\cC(g,f)$.

  Let $n=\deg(f)$ and $k=n-\deg(g)$.
  A linear block code $\cC[n,k]$ is a \emph{$(\Endom,\Deriv)$-code} if
  \begin{align*}
    \cC[n,k]\defeq\set*{\left. \bc=(c_0,c_1,\dots,c_{n-1})\ \right|\ c_0+c_1\SkewVar+\cdots+c_{n-1} \SkewVar^{n-1}\in \cC(g,f)}\ .
  \end{align*}
\end{definition}

\subsection{Module Codes}

Linear codes of length $n$ over a finite field $\F$ can be seen as subspaces of the vector spaces $\F^n$. Analogously, linear codes over a finite ring $\ring$ can be seen as \emph{submodules} of the \emph{module} $\ring^n\defeq\set*{(a_0,a_1, \dots, a_{n-1})\ |\ a_i\in\ring, \forall i\in[n]}$.
Informally speaking, the concept of \emph{module} is a generalization of the vector space, in the sense that the set of scalars is a ring instead of a field.
\begin{definition}[Module]
  \label{def:modules}
  Let $\ring$ be a ring. A \emph{left $\ring$-module} $\module$ consists of an abelian group $(\module, +)$ and a \emph{left scalar multiplication} $\cdot: \ring\times \module \to \module$ such that for all $a,b\in\ring$ and $x,y\in\module$,
  \begin{enumerate}
  \item $a\cdot(x+y) = a\cdot x$,
  \item $(a+b)\cdot x = a\cdot x+ b\cdot x$,
  \item $(ab)\cdot x = a\cdot (b\cdot x)$,
  \item $1\cdot x= x$.
  \end{enumerate}
A \emph{right $\ring$-module} $\module$ is defined similarly with a right scalar multiplication $\cdot$: $\module\times \ring \to \module$.
\end{definition}
Vector spaces over a finite field $\F$ always have a \emph{basis} (i.e., every element in the vector space is a unique $\F$-linear combination of the elements in the basis), and the dimension is unique.
However, modules do not always have a basis.
The modules that have a basis are called \emph{free} modules.

\begin{proposition}
  Let $f\in \polyRing$ be a monic skew polynomial.
  The quotient ring $\polyRing/\myspan{f}_l$ is
  a left $\polyRing$-module and
  a free left $\ring$-module. Moreover, $\polyRing/\myspan{f}_l\cong\ring^n$.
\end{proposition}
\begin{proof}
  By definition, $\polyRing/\myspan{f}_l$ is a left $\polyRing$-module.
  Since $f$ is monic, we can perform the right division (\cref{algo:EAskew}) on any element $u\in\polyRing$ by $f$ and obtain a unique remainder polynomial of degree $<n$.
  For any $r\in\polyRing$ of $\deg(r)<n$, there exist $u,q\in\polyRing$ (possibly infinite pairs) such that $u = qf+r$. Therefore, $\polyRing/\myspan{f}_l=\set*{r\in\polyRing\ |\ \deg(r)<n}$.
  It is easy to see that $\set*{r\in\polyRing\ |\ \deg(r)<n}$ is a free $\ring$-module with a basis $(1,X,\ldots,X^{n-1})$ and
  \begin{align*}
    \set*{(r_0,r_1,\dots,r_{n-1})\ |\ r_0+r_1\SkewVar+\cdots+r_{n-1} \SkewVar^{n-1}\in\polyRing/\myspan{f}_l } = \ring^n\ .
  \end{align*}
\end{proof}

\begin{proposition}
  \label{prop:module}
  Let $f\in \polyRing$ be a monic skew polynomial with a right divisor $g\in \polyRing$.
  The $(\Endom,\Deriv)$-code
  $\cC(g,f)$ is
  a left $\polyRing$-submodule of $\polyRing/\myspan{f}_l$ and
  a free left $\ring$-submodule of dimension $k=\deg(f)-\deg(g)$.
\end{proposition}
\begin{proof}

  By definition, $\cC(g,f)=\myspan{g}_l /\myspan{f}_l$ is a left $\polyRing$-submodule of $\polyRing/\myspan{f}_l$.
  Since $f$ is monic, the leading coefficient $\lc(g)$ is a right divisor of $1$ and is therefore invertible.
  For any $w\in\cC(g,f)$, $g\divides_r w$ and the quotient polynomial $q$ is unique (\cref{thm:right-division-over-ring}) and of degree at most $k$.
  This implies that the $\cC(g,f)$ is a free $\ring$-submodule of dimension $k$ with a basis $(g,Xg,\dots,X^{k-1}g)$.
\end{proof}
The relations between the quotient ring $\polyRing/\myspan{f}_l$, the $\ring$-module $\ring^n$, the left ideal $\cC(g,f)$ in $\polyRing/\myspan{f}_l$, and the block code $\cC[n,k]$ are illustrated below.
\begin{center}
\begin{tikzpicture}
  \node (Rmod) [] {$\polyRing/\myspan{f}_l$};
  \node (sub1) [below=0.5em of Rmod, rotate=90, yshift=1.5ex] {$\supseteq$};
  \node (Cgf) [below=1em of Rmod] {$\cC(g,f)$};
  \node (cong1) [right=1em of Rmod] {$\cong$};
  \node (Amod) [right=1em of cong1] {$\ring^n$};
  \node (sub2) [below=0.7em of Amod, rotate=90,yshift=1.5ex] {$\supseteq$};
  \node (cong2) [below=1.5em of cong1] {$\cong$};
  \node (Csub) [right=0.5em of cong2] {$\cC[n,k]$};
\end{tikzpicture}
\end{center}

The code $\cC[n,k]$ has a generator matrix of the following form:
\begin{align*}
  \bG=
  \begin{pmatrix}
    g_0 & g_1 & \cdots & g_{n-k} & 0& \cdots & \\
    g_0^\Deriv & g_1^\Deriv+g_0^\Endom & g_2^\Deriv+g_1^\Endom &\cdots &g_{n-k}^\Endom &  0 &\cdots \\
    \vdots & \ddots & \ddots & \cdots &\ddots &\ddots\\
    g_0^{\Deriv^{k-1}} &&\cdots &&\cdots && g_{n-k}^{\Endom^{k-1}}
  \end{pmatrix}\ .
\end{align*}
The rows are given by the coefficients of $g, \SkewVar g,\dots, \SkewVar^{k-1} g$, which can be computed using the rule $\SkewVar a =a^{\Endom} \SkewVar +a^{\Deriv} $ for $a\in \ring$. In particular, the code is completely determined by $g$, $\Endom$ and $\Deriv$.

\begin{example}
  \label{eg:genMat} Let $f\in\polyRing$ be a monic skew polynomial of degree $4$ and $g=g_0+g_1\SkewVar$ be a right divisor of $f$. The $(\Endom,\Deriv)$-code $\cC[4,3]\cong\cC(g,f)$ has a generator matrix
  \begin{align*}
    \bG=
    \begin{pmatrix}
      g_0 & g_1 & 0 & 0\\
      g_0^\Deriv & g_1^\Deriv+g_0^\Endom & g_1^\Endom &0\\
      g_0^{\Deriv^2} & g_0^{\Deriv\Endom}+g_0^{\Endom\Deriv}+g_1^{\Deriv^2} &g_0^{\Endom^2}+ g_1^{\Deriv\Endom}+g_1^{\Endom\Deriv}& g_1^{\Endom^2}
    \end{pmatrix}\ .
  \end{align*}
 \end{example}
 If $\Endom$ is of the form  $a\mapsto a^{q}$  and $\Deriv$ is an inner $\Endom$-derivation $a\mapsto \beta a -\Endom(a) \beta$, which are the only possibilities if $\ring$ is a finite field $\Fqm$, then the entries of the generator matrix $\bG$ become polynomial expressions in the coefficients of $g$, which allows sophisticated computations over $\ring$. Hence, most studies on self dual $(\Endom,\Deriv)$-code so far considered $\ring$ to be a finite field.

\subsection{Parity-Check Polynomials/Matrices of $(\Endom,\Deriv)$-Codes}
For $\ring=\Fq$, a parity-check matrix of $(\Endom,\Deriv)$-codes
has been derived for $\Deriv=0$ in \cite[Corollary 1]{boucher2009modules} for general $\Deriv\neq 0$ in \cite{boulagouaz2013delta}.
A later work \cite{boulagouaz2018characterizations} derived a parity-check matrix for $\ring$ being a finite commutative rings.
In \cite{sharma2018class}, a parity-check matrix for $(\Endom,\Deriv)$-codes $\cC(g,f)$ over the ring $\ring=\IntRing_4[u], u^2=1$ is studied when $f=hg$ is a central polynomial\footnote{A polynomial $f\in \polyRing$ is called \emph{central} if $hf=fh$, for all $h\in \polyRing$. Equivalently, a polynomial $f\in\RingSkewPolys$ is central if $\SkewVar f=f \SkewVar$.}.
The works \cite{boulagouaz2013delta, boulagouaz2018characterizations} used the framework of pseudo-linear transformations, while we only use the framework of skew polynomial rings to derive a parity-check matrix in this section.

In order to obtain a parity-check matrix for $(\Endom,\Deriv)$-codes $\cC(g,f)$, we make the additional assumption that there exists $\hbar\in \polyRing$ such that $f=hg=g\hbar$ (i.e., $g$ is a left and right divisor of $f$).
This assumption is weaker than the assumption that $f$ is central, which allows us to find more $g,f\in\polyRing$ to construct dual-containing codes (see the $[6,4]$ example in \cref{sec:all-f-are-noncentral}).

In the rest of this chapter, we also make the assumptions in following proposition on the leading coefficients of $g,h, \hbar$.
\begin{proposition}
  \label{prop:leading-coef-assumptions}
  For a $(\Endom,\Deriv)$-codes $\cC(g,f)$, where $f$ is monic and $g$ is a left and right divisor of $f$, i.e., $f=hg=g\hbar$ for some $h,\hbar\in\polyRing$, 
  \begin{enumerate}
  \item we can assume w.l.o.g.~that $g$ and $h$ are monic;
  \item if $\Endom$ is an automorphism, then we can also assume w.l.o.g.~that $\hbar$ is monic. \label{item:lc-hbar}
  \end{enumerate}
\end{proposition}
\begin{proof}
  Let $g= g_{n-k}X^{n-k}+\cdots + g_0$ and $h=h_kX^k+\cdots +h_0$.
  Since $f=hg$ and $f$ is monic, the leading coefficient of $f=hg$ is $h_k\Endom^{k}(g_{n-k})=1$, showing that $h_k$ and $\Endom^{k}(g_{n-k})$ are invertible. We can write
\begin{align*}
  f &= \underbrace{(h_kX^k+\cdots + h_0)}_{h}\cdot\underbrace{(g_{n-k}X^{n-k}+\cdots + g_0)}_{g} \\
   &= \underbrace{(h_kX^k+\cdots + h_0)\cdot g_{n-k}}_{\tilde{h}}\cdot\underbrace{g_{n-k}^{-1}\cdot(g_{n-k}X^{n-k}+\cdots + g_0)}_{\tilde{g}}\ .
\end{align*}
Note that $\tilde{g}$ is a monic polynomial. Since the endomorphism $\Endom$ maps $1$ to $1$, the leading coefficient of $\tilde{h}\tilde{g}$ is the leading coefficient of $\tilde{h}$. Because $\tilde{h}\tilde{g}=f$ is monic, we obtain that $\tilde{h}$ is also a monic polynomial.
The polynomials $g$ differ from $\tilde{g}$ by multiplying an invertible element.
It can be seen that any left multiple of $g$ by a polynomial of degree $\le k-1$ is also a left multiple of $\tilde{g}$ by another polynomial of degree $\le k-1$ and vice versa. Therefore, $\cC(g,f) = \cC(\tilde{g},f)$ and we can assume w.l.o.g.~that $g$ is a monic polynomial.
With the argument on $\tilde{h}$ above, we can assume that $h$ is also monic.

We now show that $\hbar=\hbar_kX^k+\cdots + \hbar_0$ is monic if $\Endom$ is an automorphism.
The leading coefficient $\lc(f)=\lc(g\hbar)=g_{n-k}\Endom^{n-k}(\hbar_k)=1$. Since both $f$ and $g$ are monic, $\Endom^{n-k}(\hbar_k)$ must be $1$. If $\Endom$ is an automorphism, we obtain that $\hbar_k=1$.
\end{proof}
\cref{prop:leading-coef-assumptions} \ref{item:lc-hbar} implies that if $\Endom$ is a non-trivial endomorphism (not an automorphism), the $\hbar$ in the decomposition $f=g\hbar$ may not be monic.
In the following we give an example of such case.
\begin{example} 
  \label{eg:hbar-not-monic}
  Consider the ring $\ring=\F_2[u]$ where $u^2=0$, the non-trivial endomorphism $\Endom_2(u)= 0$ of $\ring$ and the $\Endom_2$-derivation $\Deriv_3(u)=u$ (\cref{subtab:F2u}).
  Let $\polyRing=\ring[\SkewVar;\Endom_2,\Deriv_3]$. For the skew polynomials
  $g=\SkewVar^2 + u\SkewVar + u + 1$ and $f=\SkewVar^4 + (u + 1)\SkewVar^3 + \SkewVar + u + 1$, it can be verified that $g$ is a left and right divisor of $f$.
  For the case $g$ being a right divisor, we found that $f$ can be decomposed into $f=hg$ with the unique monic $h=\SkewVar^2 + (u + 1)\SkewVar + 1$.
  For the case $g$ being a left divisor, $f$ can be decomposed into $f=g\hbar$ with $\hbar = (u + 1)\SkewVar^2 + (u + 1)\SkewVar + u + 1$ or $\hbar= (u + 1)\SkewVar^2 + \SkewVar + u + 1$, neither of which is monic.
\end{example}

The following shows that the polynomial $\hbar$ plays a role as a ``parity-check polynomial''.
\begin{lemma}[Parity-check polynomial of a $(\Endom,\Deriv)$-code]
  \label{lem:chbar=0}
  Let $f\in \polyRing$ be a monic polynomial of degree $n$, $g\in\polyRing$ be a monic left and right divisor of $f$ of degree $n-k$ ($f=hg=g\hbar$ for some $h,\hbar\in\polyRing$),
  and $\cC[n,k]\cong\cC(g,f)$ be a $(\Endom,\Deriv)$-code as defined in \cref{def:skewPolycyclicCodes}.
  A word $\bc=(c_0, c_1,\dots, c_{n-1})\in \ring^n$ is a codeword of $\cC[n,k]$ if and only if the corresponding skew polynomial $c = c_0+c_1\SkewVar+c_{n-1}\SkewVar^{n-1}$ fulfills $c \hbar=0$ in $\polyRing/\myspan{f}_l$.
  \end{lemma}
\begin{proof}
  The statement is trivial for $\bc=\0$.
  Consider any nonzero $\bc\in\ring^n$.
  If $\bc\in \cC[n,k]$, then $c=wg$ for some $w\in \polyRing$. Then $c\hbar =wg\hbar=w(h g)=0\Mod f$, which is equivalent to say that $c\hbar =0$ in $\polyRing/\myspan{f}_l$.
  Conversely, if $c\hbar =0$ in $\polyRing/\myspan{f}_l$, then $c\hbar=\tilde{w} f= \tilde{w}(g \hbar)$ for some $\tilde{w}\in \polyRing$. Since $f$ is monic and $\hbar$ is a right divisor of $f$, $\hbar$ is not a zero divisor in $\polyRing$ and we obtain $c = \tilde{w} g$, showing that $\bc\in \cC[n,k]$.
\end{proof}

\begin{lemma}[Parity-check matrix of a $(\Endom,\Deriv)$-code]
  \label{lem:Mmat}
  We follow the notations in \cref{lem:chbar=0}. There is a matrix $\bM\in\ring^{n\times n}$ such that $\bc\bM=\0,\ \forall\bc \in \cC[n,k]$, where the entries of $\bM$ are images under $\Endom$ and $\Deriv$ of the coefficients of $\hbar$ and $g$.
\end{lemma}
\begin{proof}
Note that $c = c_0+c_1\SkewVar+c_{n-1}\SkewVar^{n-1}\in\cC(g,f)$ is the corresponding skew polynomial of $\bc\in\cC[n,k]$.
By \cref{lem:chbar=0}, we have $c\hbar=0$ in $\polyRing/\myspan{f}_l$, which is equivalent to
\begin{align}
  \label{eq:chbar=0}
  c\hbar=\left(\sum_{i=0}^{n-1} c_{i}X^{i}\right) \hbar \equiv 0 \Mod f\ .
\end{align}

By applying the multiplication and the right module by $f$, we get a system of linear equations
\begin{align}
  \label{eq:Mmat-entries}
  \sum_{i=0}^{n-1} c_iM_{ij}X^j =0, \ \forall j\in[0,n-1]\ ,
\end{align}
where $M_{ij}$ is the coefficient of $\SkewVar^j$ in $\SkewVar^i\cdot \hbar \Mod f$, which is the image under $\Endom$ and $\Deriv$ of the coefficients of $\hbar$ and $f$, $\forall i,j=0,\dots, n-1$.
Since $f=g\hbar$, the coefficients of $f$ can be replaced by the coefficients of $\hbar$ and $g$.
Let $\bM$ be an $n\times n$ matrix whose entries are $M_{ij}$'s.
It can be shown that
\begin{align*}
  \eqref{eq:Mmat-entries} \iff \bc\bM=\0 \ .
\end{align*}
\end{proof}
The following toy example provides some insights on how such a matrix $\bM$ may look like.
\begin{example}
  \label{eg:parity-check-mat}
  Consider a ring $\ring$, a skew polynomial ring $\polyRing=\ring[\SkewVar;\Endom,\Deriv]$, $f=\SkewVar^3+\sum_{i=0}^2f_i \SkewVar^i$, $g=\SkewVar^2+g_1\SkewVar+g_0$ and $\hbar=\hbar_1\SkewVar+\hbar_0$ in $\polyRing$ such that $f=g\hbar$.
  According to \cref{lem:chbar=0}, $\bc=(c_0,c_1,c_{2})\in  \cC[3,1]\cong\cC(g,f)$ if and only if $c\hbar \equiv 0 \Mod f$, where $c=c_0+c_1\SkewVar+c_2\SkewVar^2$.
  Since
  \begin{equation}
  \begin{split}
    c\hbar\Mod f= &\parenv*{c_2 (\hbar_1^{\Endom\Deriv}+\hbar_1^{\Deriv\Endom}+\hbar_0^{\Endom^2}{-\hbar_1^{\Endom^2}f_2})+c_1\hbar_1^\Endom}\SkewVar^2\\
           & + \parenv*{c_2(\hbar_1^{\Deriv^2}+\hbar_0^{\Endom\Deriv}+\hbar_0^{\Deriv\Endom}{-\hbar_1^{\Endom^2}f_1}) + c_1(\hbar_1^\Deriv+\hbar_0^\Endom) + c_0\hbar_1 }\SkewVar\\
                     &  + c_2(\hbar_0^{\Deriv^2}{-\hbar_1^{\Endom^2}f_0}) + c_1\hbar_0^\Deriv+c_0\hbar_0\ ,
  \end{split}
  \label{eq:ch-mod-f}
  \end{equation}
  we obtain the following matrix $\bM\in \ring^{3\times 3}$ such that $\bc \bM=\0$.
  \pgfkeys{tikz/Mmatrix/.style={matrix of math nodes,nodes in empty cells,left delimiter={(},right delimiter={)},inner sep=1pt,outer sep=3pt,column sep=4pt,row sep=4pt,nodes={minimum width=9pt,minimum height=6pt,anchor=base west,inner sep=0pt,outer sep=0pt}}}
  \begin{equation}
    \label{eq:Mn3k1}
  \bM =
  \begin{tikzpicture}
   [baseline={-0.5ex},mymatrixenv]
    \matrix [Mmatrix,inner sep=2pt] (M)
    {
      \tikzmarkin[kwad=style green]{rows}  \hbar_0 & \tikzmarkin[kwad=style orange]{cols} \hbar_1& 0 \\
      \hbar_0^{\Deriv} & \hbar_1^\Deriv+\hbar_0^\Endom&  \hbar_1^\Endom\tikzmarkend{rows} \\
      \hbar_0^{\Deriv^2}{-\hbar_1^{\Endom^2}f_0} &
      \hbar_1^{\Deriv^2}+\hbar_0^{\Endom\Deriv}+\hbar_0^{\Deriv\Endom}{-\hbar_1^{\Endom^2}f_1} &
      \hbar_1^{\Endom\Deriv}+\hbar_1^{\Deriv\Endom}+\hbar_0^{\Endom^2}{-\hbar_1^{\Endom^2}f_2} \tikzmarkend{cols} \\
    };
  \end{tikzpicture}
  \ .
  \end{equation}
  Note that the entry $M_{ij}$ in $\bM$ corresponds to the coefficient of the term $c_i\SkewVar^j$ in the polynomial $c\hbar \Mod f$, e.g., $M_{11}$ is the coefficient of $c_1\SkewVar$ in \eqref{eq:ch-mod-f}, which is $(\hbar_1^\Deriv+\hbar_0^\Endom)$.
\end{example}

\subsection{Dual Codes of $(\Endom,\Deriv)$-Codes}
\label{sec:dual-of-ideal-codes}
Most dual codes studied in the literature consider the \emph{Euclidean inner product} to define duality. In this section, we consider a more general notion of inner product and duality.
\begin{definition}[Hermitian inner product and Hermitian dual]
  \label{def:hermitian-dual}
  Let $\sigma$ be an automorphism of $\ring$ of order\footnote{The order of an automorphism $\sigma$ is the minimum positive integer $n$ such that $\sigma^n=\Id$.} at most $2$.
  The \emph{$\sigma$-Hermitian inner product} of $\bx,\by\in \ring^n$ is defined  as $ \langle \bx,\by \rangle_\sigma\defeq \sum_{i=1}^n x_i \sigma(y_i)$.
  The \emph{$\sigma$-Hermitian dual code} of a code $\cC$ is defined as
  \begin{align*}
    {\cC}^{\perp_{\sigma}}\defeq\{\bv\ |\  \myspan{ \bv,\bc }_\sigma=0,\ \forall \bc\in \cC\}\ .
  \end{align*}
  A code is \emph{$\sigma$-{dual-containing}} if ${\cC}^{\perp_{\sigma}}\subseteq \cC $, and \emph{$\sigma$-self dual} if $\cC={\cC}^{\perp_{\sigma}}$.
  In particular, if $\sigma=\Id$, we obtain the \emph{Euclidean inner product} and the \emph{Euclidean dual}. In this case, we omit the $\sigma$ in the notation.
\end{definition}

\emph{Frobenius rings} are promoted by Wood \cite{wood1999duality} as the most appropriate rings for coding theory over finite rings, because two classical theorems of MacWilliams -- the extension theorem (also known as the equivalence theorem) and the MacWilliams identities -- generalize to Frobenius rings.
Szabo and Wood \cite{szabo2017properties} showed that dual codes have complementary cardinality for codes over finite Frobenius rings.
Note that the rings $\F_2[v]$, $\F_2[u]$ and $\F_2[\alpha]$ are all Frobenius rings.

\begin{theorem}[MacWilliams theorems for codes over finite Frobenius rings {\cite{wood1999duality}}]
  Let $\cC$ be a linear code over a finite Frobenius ring $\ring$, then
  \begin{itemize}
  \item \emph{(Extension theorem)} Every isometric map $\cC\to\ring^n$ that preserves the Hamming weight of the code can be extended to a monomial transformation.
  \item \emph{(MacWilliams identities)} The weight enumerator of the dual code is determined by the weight enumerator of the code.
  \end{itemize}
\end{theorem}

\begin{lemma}[Complementary cardinality of dual codes {\cite[Theorem 3.6]{szabo2017properties}}]
\label{lem:wood-size-of-dual}
Let $\ring$ be a finite frobenius ring and $\sigma$ be an automorphism of $\ring$ of order at most $2$.
For a linear code $\cC$ over $\ring$, its $\sigma$-Hermitian dual code $C^{\perp_{\sigma}}$ has the complementary cardinality of the code, i.e.,
${|\cC|\cdot |\cC^{\perp_{\sigma}}|=|\ring|^n}$.
\end{lemma}

The following is devoted to derive a generator matrix of the Euclidean dual and $\sigma$-Hermitian dual of the $(\Endom, \Deriv)$-codes, respectively.
For the case $\Deriv=0$, the dual codes have been studied for $\ring$ being $\IntRing_4$~\cite{boucher2008skew}, $\F_4$~\cite{boucher2009modules}, $\Fq$~\cite{boucher2009modules} and general finite rings~\cite{boucher2011note,boulagouaz2018characterizations}.
For $\Deriv\neq 0$, the dual codes are much less studied.

\begin{theorem}[Generator matrix of the Euclidean dual]
  \label{thm:dual-generator}
  Let $\ring$ be a finite Frobenius ring.
  Let $\cC=\cC[n,k]$ be the linear block $(\Endom,\Deriv)$-code as defined in \cref{def:skewPolycyclicCodes}.
  Then, the Euclidean dual code $\cC^\perp$ of $\cC$ is a linear free $\ring$-module code of length $n$ and dimension $n-k$.
  A generator matrix $\bG^\perp\in \ring^{(n-k)\times n}$ of $\cC^\perp$
  is composed of the transpose of the last $n-k$ columns of the matrix $\bM$ defined in \eqref{eq:Mmat-entries}.
\end{theorem}
\begin{proof}
  It follows from \cref{lem:Mmat} that all the columns of the $n\times n$ matrix $\bM$ are orthogonal to all the codewords in $\cC$. In other words, all the columns of $\bM$ are in $\cC^\perp$.
  We denote by $\widetilde{\cC}$ the code generated by the columns of the $n\times n$ matrix $\bM$ described in \eqref{eq:Mmat-entries} (see \eqref{eq:Mn3k1} for an example).
  By construction we have that $\widetilde{\cC}\subseteq \cC^\perp$.
  In \cref{lem:Mmat} we have shown that a vector $\bc\in \cC$ if and only if $\bc \bM=\0$, which is equivalent to $\bc$ being orthogonal to all generators of $\widetilde{\cC}$. Therefore, $\cC=\widetilde{\cC}^\perp$.
  Since $\widetilde{\cC}$ is a linear code over $\ring$, by \cref{lem:wood-size-of-dual} we have that $|\widetilde{\cC}^\perp| \cdot |\widetilde{\cC}|=|\ring|^n$.
  From $\cC=\widetilde{\cC}^\perp$ and $|\cC|=|\ring|^k$ we then get $|\widetilde{\cC}|=|\ring|^{n-k}$.
  Since $\cC$ is also a linear code over $\ring$, by \cref{lem:wood-size-of-dual} we also have $|\cC^\perp|\cdot |\cC|=|\ring|^n$, which implies that $|\cC^\perp| =|\ring|^{n-k}$. Since $\widetilde{\cC}\subseteq \cC^\perp$ and both codes have the same cardinality, we conclude that $\widetilde{\cC}= \cC^\perp$, i.e., $\cC^\perp$ is generated by the columns of the matrix $\bM$.

  It can be shown from \eqref{eq:Mmat-entries} that for $i=0,\dots, n-k-1$, the $i$-th row $\bM$ correspond to $\SkewVar^{i} \hbar$ (e.g., the first two rows (the green part) in \eqref{eq:Mn3k1}). This shows that the right-upper $(n-k)\times (n-k)$ submatrix of $\bM$ is lower triangular with invertible diagonal elements $\hbar_k,\Endom(\hbar_k),\ldots,\Endom^{n-k-1}(\hbar_k)$ and the right-most $n-k$ columns of $\bM$ are therefore linearly independent. Hence, the $\ring$-submodule generated by the right-most $n-k$ columns of $\bM$ contains $|\ring|^{n-k}$ elements. This shows that $\cC^\perp$ is a free $\ring$-module generated by the right-most $n-k$ columns of $\bM$ (e.g., the last two columns (the orange part) in \eqref{eq:Mn3k1}).
\end{proof}
\begin{corollary}
  \label{cor:left-cols-dependent}
  The first $k$ columns of $\bM$ are $\ring$-linear combinations of the last $n-k$ columns of $\bM$.
\end{corollary}
\begin{proof}
  Using the fact that $|\cC^\perp|=|\ring|^{n-k}$ from \cref{lem:wood-size-of-dual} and \cref{thm:dual-generator}, the statement is proven.
\end{proof}
\begin{example}
According to \cref{thm:dual-generator}, the Euclidean dual of the code in \cref{eg:parity-check-mat} is generated by the right-most two columns of $\bM$ in \eqref{eq:Mn3k1}, i.e.,
\begin{align*}
  \bG^\perp =
  \begin{pmatrix}
{\hbar_1} &
{  \hbar_1^\Deriv+\hbar_0^\Endom} &
{\hbar_1^{\Deriv^2}+\hbar_0^{\Endom\Deriv}+\hbar_0^{\Deriv\Endom}{-\hbar_1^{\Endom^2}f_1}}\\
{0 }& {\hbar_1^\Endom} &
{  \hbar_1^{\Endom\Deriv}+\hbar_1^{\Deriv\Endom}+\hbar_0^{\Endom^2}{-\hbar_1^{\Endom^2}f_2}}
  \end{pmatrix}\ .
\end{align*}
 \end{example}
 \begin{theorem}[Generator matrix of the $\sigma$-Hermian dual]
   \label{thm:Hermitian-gen-mat}
  Let $\cC=\cC[n,k]$ be the linear block $(\Endom,\Deriv)$-code as defined in \cref{def:skewPolycyclicCodes},
  and $\bG^\perp\in \ring^{(n-k)\times n}$ be a generator matrix of the Euclidean dual code $\cC^\perp$ of the code $\cC$ as derived in \cref{thm:dual-generator}.
  Denote by $\sigma(\bG^\perp)$ the matrix after applying $\sigma$ to every entry of $\bG^\perp$.
  Then, $\sigma(\bG^\perp)$ is a generator matrix of the $\sigma$-Hermitian dual code $\cC^{\perp_\sigma}$ of $\cC$.
\end{theorem}
\begin{proof}
  For each row $\bg_i=(g_{i,0},\dots,g_{i,n-1}), i\in\{0,\dots,k-1\}$ of a generating matrix $\bG$ of $\cC$ and each row $\bg^\perp_j=(g^\perp_{i,0},\dots,g^\perp_{i,n-1}), j\in\{0,\dots,n-k-1\}$ of a generating matrix of $\bG^\perp$ of $\cC^\perp$, we have $\langle g_i,g^\perp_j \rangle=\sum_{l=0}^{n-1} g_{i,l}\,g^\perp_{j,l}=0$. Since $\sigma$ is of order at most $2$, $\langle g_i,\sigma(g^\perp_j) \rangle_\sigma=\sum_{l=0}^{n-1} g_{i,l}\, \sigma(\sigma(g^\perp_{j,l}))=\sum_{i=0}^{n-1} g_{i,l}\, g^\perp_{j,l}=0$.
  Since $\bG^\perp$ is an upper triangular matrix with $n-k$ rows as shown in the proof of \cref{thm:dual-generator}
  and $\sigma$ is an automorphism, the $\sigma(\bG^\perp)$ is also an upper triangular matrix with $n-k$ rows and therefore generates a free $\ring$-module code of dimension $|\ring|^{n-k}$. Finally, \cref{lem:wood-size-of-dual} implies that they generate $\cC^{\perp_\sigma}$.
\end{proof}
\begin{corollary}
  \label{cor:Herm-left-cols-dependent}
  Let $\sigma(\bM)$ be the matrix after applying the automorphism $\sigma$ to every entry of $\bM$ derived in \cref{lem:Mmat}. The first $k$ columns of $\sigma(\bM)$ are $\cA$-linear combinations of the last $n-k$ columns of $\sigma(\bM)$.
\end{corollary}
\begin{proof}
  Denote by $\bm_i$ the $i$-th column of $\bM$.
  It follows from \cref{cor:left-cols-dependent} that for $i\in\set*{0,1, \dots, k-1}$, we can write $\bm_i=\sum_{j=k}^{n-k-1}a_j\bm_j$ for some $a_j\in\ring$.
  By the rules of an automorphism (\cref{def:endo-deri}), we have $\sigma(\bm_i)=\sigma(\sum_{j=k}^{n-k-1}a_j\bm_j)=\sum_{j=k}^{n-k-1}\sigma(a\bm_j)=\sum_{j=k}^{n-k-1}\sigma(a_j)\sigma(\bm_j)$, which is an $\ring$-linear combination of the last $n-k$ columns of $\sigma(\bM)$.
\end{proof}
\section{Computing All $\sigma$-Dual-Containing $(\Endom,\Deriv)$-Codes}
\label{sec:algo-search-via-GB}
It follows from the definition of a $\sigma$-dual-containing code $\cC[n,k]$ in \cref{def:hermitian-dual} that, $k\geq n-k$. In the rest of this chapter, we implicitly assume this inequality holds when addressing dual-containing codes.
By \cref{def:skewPolycyclicCodes}, we can see that a $(\Endom,\Deriv)$-code $\cC(g,f)$ is determined by the polynomials $g,f\in\polyRing$. If $g$ is also a left divisor of $f$ (i.e., $f=g\hbar$), then the code $\cC(g,f)$ is determined by $g,\hbar$.

Our method to compute $\sigma$-dual-containing $(\Endom,\Deriv)$-codes considers the coefficients of $g,\hbar$ as unknowns and solves a system of equations derived from the following constraints:
\begin{enumerate}
\item $f=g\hbar$ is monic and $g\divides_r f$; \label{item:dual-containing-constraint-1}
\item (For Euclidean dual, $\sigma=\Id$) $\bM^\top\bM=\0$, where $\bM$ is described in \cref{lem:Mmat} and serves as a parity-check matrix of $\cC(g,f)$; \label{item:dual-containing-constraint-2}
\item (For $\sigma$-Hermitian dual) $\sigma(\bM)^\top\bM=\0$, where $\sigma(\bM)$ is the matrix after applying $\sigma$ to every entry of $\bM$. \label{item:Herm-dual-containing-constraint}
\end{enumerate}

\begin{lemma}
  \label{lem:dual-containing-constraints}
  Let $g, \hbar\in \polyRing$. If the coefficients of $g$ and $\hbar$ fulfill the system of polynomials equations that are derived from the constraints \ref{item:dual-containing-constraint-1} and \ref{item:dual-containing-constraint-2} (resp.~\ref{item:Herm-dual-containing-constraint}), then the $(\Endom,\Deriv)$-code $\cC(g,f)$ is dual-containing (resp.~$\sigma$-dual-containing), where $f=g\hbar$.
\end{lemma}
\begin{proof}
  The constraint \ref{item:dual-containing-constraint-1} suffices that $\cC(g,f)$ is a $(\Endom,\Deriv)$-code, $\hbar$ serves as a parity-check polynomial (\cref{lem:chbar=0}), and we can construct the matrix $\bM$ that serves as a parity-check matrix (\cref{lem:Mmat}).

  The constraint \ref{item:dual-containing-constraint-2} is equivalent to setting $\bG^\perp\cdot \bM=\0$ according to \cref{cor:left-cols-dependent}, where $\bG^\perp$ is the transpose of the last $n-k$ columns of $\bM$ (\cref{thm:dual-generator}). This ensures that $\cC^\perp\subseteq \cC$.

  Similarly, the constraint \ref{item:Herm-dual-containing-constraint} is equivalent to setting $\bG^{\perp_\sigma}\bM=\0$ according to \cref{cor:Herm-left-cols-dependent}, which ensures that $\cC^{\perp_\sigma}\subseteq \cC$.
\end{proof}

It can be seen in \cref{eg:parity-check-mat} that the entries in $\bM$ are symbolic expressions in images under compositions of $\Endom$ and $\Deriv$ of the coefficients of $\hbar$ and~$g$ (e.g., $\hbar_1^{\Endom}$, $\hbar_1^{\Deriv}$). However, since some $\Endom$ or $\Deriv$ maps are not always polynomial maps in $\ring$ (see \cref{eg:not-poly-map-F2v} and \cref{eg:not-poly-map-F2u}), it is difficult to solve for the coefficients of $g$ and $\hbar$ over $\ring$ by computers.

If $\ring=\subring[\beta_1,\dots,\beta_s]$ is a free $\subring$-algebra and the restriction of $\Deriv$ and $\Endom$ to $\subring$ are polynomial maps, then we can transform the symbolic expressions of images under $\Endom$ and $\Deriv$ (e.g., $\hbar_1^{\Endom}$) into polynomial expressions over $\subring$. By representing the unknown coefficients of $\hbar$ and $g$ as a linear combination of the algebra basis $(\beta_1,\dots,\beta_s)$ (e.g., $\hbar_1=\hbar_{1,1}\beta_1+\dots+\hbar_{1,s}\beta_s$), we obtain multivariate polynomial expressions over $\subring$ for the entries of $\bM$.

The following lemma shows that we can transform a system of symbolic equations over $\ring$
into a system of polynomial equations over $\subring$.
\begin{lemma}
  \label{lem:polymap-subring}
  Let $\cE_s$ be a finite system of symbolic equations over $\ring$ that are in the images of $\Endom$ and $\Deriv$ of a finite number of variables $y_1,\ldots,y_m$.
  If $\ring=\subring[\beta_1,\ldots,\beta_s]$ is a free $\subring$-algebra and the restriction of $\Endom$ and $\Deriv$ to $\subring$ are polynomial maps, then all solutions in $\ring^m$ to ${\cE_s}$ correspond to the solutions in $\subring^{ms}$ to a system $\cE_p$ of polynomial equations over $\subring$ in the variables $y_{1,1},\dots ,y_{1,s}, \dots,y_{m,1}, \dots, y_{m,s}$, where $y_i=y_{i,1}\beta_1+\cdots +y_{i,s}\beta_s, i\in[m]$.
\end{lemma}
\begin{proof}
The image of the basis $(\beta_1,\ldots,\beta_s)$ of $\ring$ over $\subring$ under $\Deriv$ and $\Endom$ are expressions of the form $\beta_i^{\Endom}=\gamma_{i,1}\beta_1+\cdots +\gamma_{i,s}\beta_s$ and $\beta_i^{\Deriv}=\xi_{i,1}\beta_1+\cdots +\xi_{i,s}\beta_s$ for some $\gamma_{i,j}, \xi_{i,j}\in\subring$. The symbolic expressions in the images under $\Endom$ and $\Deriv$ of the variables $y_i, i\in[m]$ can therefore be written as
\begin{align*}
  y_i^{\Endom} & = (y_{i,1}\beta_1+\cdots +y_{i,s}\beta_s)^\Endom  \\
               &=  y_{i,1}^{\Endom}\beta_1^{\Endom} + \cdots  +y_{i,s}^{\Endom}\beta_s^{\Endom}\\
               & = y_{i,1}^{\Endom}\left(\gamma_{1,1}\beta_1+\cdots +\gamma_{1,s}\beta_s\right)+\cdots +y_{i,s}^{\Endom}\left(\gamma_{s,1}\beta_1+\cdots +\gamma_{s,s}\beta_s\right)\ ,\\
  y_i^{\Deriv} &= (y_{i,1}\beta_1+\cdots +y_{i,s}\beta_s)^\Deriv \\
               &=\; (y_{i,1}\beta_1)^\Deriv+\cdots +(y_{i,s}\beta_s)^\Deriv\\
               & = y_{i,1}^\Deriv \beta_1+ y_{i,1}^\Endom \beta_1^\Deriv +\cdots +y_{i,s}^\Deriv \beta_s+y_{i,s}^\Endom \beta_s^\Deriv\\
               & = y_{i,1}^{\Deriv} \beta_1+  y_{i,1}^{\Endom}\parenv*{\xi_{1,1}\beta_1+\cdots +\xi_{1,s}\beta_s}+ \cdots +y_{i,s}^{\Deriv} \beta_s+y_{i,s}^{\Endom}\parenv*{\xi_{s,1}\beta_1+\cdots +\xi_{s,s}\beta_s} \ .
\end{align*}
Using
\begin{enumerate}
\item the algebra relations $\beta_i\beta_j=\mu_{i,j,1}\beta_1+\ldots+\mu_{i,j,s}\beta_s$ (where $\mu_{i,j,s}\in \subring$ are given),
\item the additive and multiplicative rules of $\Endom$ and $\Deriv$ (\cref{def:endo-deri}),
\item the fact that the restriction of  $\Deriv$ and $\Endom$ to $\subring$ are polynomial maps on $\subring$, so that $y_{i,j}^\Endom$ and $y_{i,j}^\Deriv$ are polynomials in $y_{i,j}$ over $\subring$,
\end{enumerate}
we can recursively transform the system $\cE_s$ of symbolic equations in variables $y_1,\dots,y_m$ into a system $\cE_p$ of polynomial equations in the variables $y_{1,1}$, $\dots$, $y_{1,s}$, $\dots$, $y_{m,1}$, $\dots$, $y_{m,s}$.
For any solution $(\hat{y}_{1,1},\dots,\hat{y}_{m,s})\in \subring^{ms}$ to $\cE_p$, we can construct the corresponding solution $(\hat{y}_1,\dots,\hat{y}_m)\in\ring^m$ to $\cE_s$ by $\hat{y}_i=\hat{y}_{i,1}\beta_1+\cdots+\hat{y}_{i,s}\beta_s, i\in[m]$.
\end{proof}

Now we can proceed by solving a system of polynomial equations over $\subring$ to get the coefficients of $g,f$ such that the code $\cC(g,f)$ over $\ring$ is a $\sigma$-dual-containing code.
We use Gr\"obner bases (see \cref{sec:grobBases}) for ideals of multivariate polynomials to solve the system of polynomial equations.
\cref{algo:dualContainingCodes-GrobBasis} summarizes our implementation in Magma~\cite{Magma} to search for $\sigma$-dual-containing $(\Endom,\Deriv)$-codes.
If a Gr\"obner basis of the ideal generated by the polynomials in the system of equations over $\subring$ can be computed, we can find all $\sigma$-dual-containing codes $\cC(g,f)$ over $\ring$ for the given parameters $[n,k]$.

\begin{algorithm}[htb!]
  \caption{Computing all $\sigma$-dual-containing $(\Endom,\Deriv)$-codes for given $n,k$.}
  \label{algo:dualContainingCodes-GrobBasis}
  \KwIn{A ring $\subring$ where Gr\"obner bases can be computed, a ring $\ring=\subring[\beta_1,\dots,\beta_s]$ which is a free $\subring$-algebra, an endomorphism $\Endom$ and a $\Endom$-derivation $\Deriv$ of $\ring$ which are polynomial maps on $\subring$, an automorphism $\sigma$ of $\ring$ of order at most $2$, and code parameters $n,k$.}
  \KwOut{A set of solutions $\cP=\{\hat{g},\hat{\hbar}, \hat{f}\ |\ \cC(\hat{g},\hat{f}) \text{ is } \sigma\text{-dual-containing}\}$}
  $P_1\gets \subring[g_{0,1},\ldots, g_{0,s},\dots,g_{n-k-1,1},\ldots,g_{n-k-1,s},\hbar_{0,1},\ldots,\hbar_{0,s},\ldots, \hbar_{k-1,1},\ldots, \hbar_{k-1,s}]$    \tcc*[r]{multivariate polynomial ring over $\subring$}
  $\cP\gets \{\}$
  \tcc*[r]{Initialize a set to collect $\sigma$-dual-containing codes}
  \ForEach{$\hbar_k=\sum_{j=1}^s\hbar_{k,j}\beta_j\in \{$Invertible element of $\ring$\} \label{line:LCchoice}}{
    $\sfL\sfS\sfE\sfs\gets \{$Constraints s.t.~$g_{i,j}, \hbar_{i,j}\in \subring$\} \tcc*[r]{$g_{i,j}^p=g_{i,j}, \hbar_{i,j}^p=\hbar_{i,j}$ if $\subring=\mathbb{F}_p$}
    $g\gets \sum_{i=0}^{n-k-1}(\sum_{j=1}^sg_{i,j}\beta_j)\SkewVar^i+ \SkewVar^{n-k}$
    \tcc*[r]{$g \in P_1[\SkewVar;\Endom,\Deriv]$}
    $\hbar\gets\sum_{i=0}^{k-1} (\sum_{j=1}^s\hbar_{i,j}\beta_j) \SkewVar^i +
    \hbar_k\SkewVar^{k}$\tcc*[r]{$\hbar \in P_1[\SkewVar;\Endom,\Deriv]$}
    $f\gets g\cdot \hbar$\tcc*[r]{$\lc(f)$ may not be monic but does not contain variable}
    $h, r\gets$
    quotient, remainder of right dividing $f$ by $g$
    \tcc*[r]{$h,r\in P_1[\SkewVar;\Endom,\Deriv]$}
    $\sfL\sfS\sfE\sfs\overset{\text{Append}}{\longleftarrow}$\{All coefficients of $r$ are $0$ \}
    \tcc*[r]{implies $g\divides_r f$\label{line:g_rdiv_f}}
    $\bM\gets$ the matrix constructed from $\hbar$ according to \cref{lem:Mmat}\;
    $\sfL\sfS\sfE\sfs\overset{\text{Append}}{\longleftarrow}$\{All entries in $\sigma(\bM)^\top\cdot \bM$ are $0$\}
    \tcc*[r]{implies $\cC^{\perp_\sigma}\subseteq \cC$}
    $\cS\gets$ \{Solutions of $g_{0,1},\ldots,g_{n-k-1,s},\hbar_{0,1},\ldots, \hbar_{k-1,s}$ from a Gr\"obner basis of $\sfL\sfS\sfE\sfs$\}\label{line:solutions}\;
    $\cP\overset{\text{Append}}{\longleftarrow}\set{\hat{g},\hat{\hbar}, \hat{f}\in \ring[\SkewVar;\Endom,\Deriv]\ |\ \forall$ solution in $\cS}$\;
    \tcc{$\hat{g},\hat{\hbar}\in \ring[\SkewVar;\Endom,\Deriv]$ are reconstructed by evaluating coefficients of $g,\hbar\in P_1[\SkewVar;\Endom,\Deriv]$ at each solution in $\cS$;
      $\hat{f}=\hat{g}\cdot \hat{\hbar}$}
  }
  \Return{$\cP$}
\end{algorithm}

  Using the results found by \cref{algo:dualContainingCodes-GrobBasis}, we can further look into the following properties of the $(\Endom,\Deriv)$-codes $\cC(g,f)$:
  \begin{enumerate}
  \item Do nonzero derivations or non-trivial endomorphisms (not automorphisms) provide new $\sigma$-dual-containing codes that could not be constructed from skew polynomials with automorphisms or zero derivations?\\
    (The answer is yes. In \cref{tab:HammingF2v2} and \cref{tab:HammingF2u2}, we give examples of Hamming weight distributions of dual-containing $(\Endom,\Deriv)$-codes that could not be found without considering either nonzero derivations or non-trivial endomorphisms.)
  \item Is existing a central $f$ such that $g\divides_r f$ a necessary condition for the $(\Endom,\Deriv)$-code $\cC(g,f)$ being dual-containing?\\
    (The answer is no. In \cref{sec:all-f-are-noncentral} we give an example of a generating polynomial $g$ of a $[6,4]$ dual-containing $(\Endom,\Deriv)$-code where all $8$ polynomials $f=hg=g\hbar$ of degree $6$ are non-central.)
  \end{enumerate}
\subsection{Is the Dual Code of a $(\Endom,\Deriv)$-Code also a $(\Endom,\Deriv)$-Code?}
\label{sec:dual_principle}
By \cref{def:skewPolycyclicCodes}, a $(\Endom,\Deriv)$-code $\cC$ is a principle left ideal $\myspan{g}_l/\myspan{f}_l$ in the quotient ring $\polyRing/\myspan{f}_l$.
The algorithm we introduce in this section allows us to test whether the dual code $\cC^{\perp}$ is also a $(\Endom,\Deriv)$-code, in other words, whether there exists a generator polynomial $g^{\perp}\in\polyRing$ such that $\cC^{\perp}=\myspan{g^\perp}_l/\myspan{f^\perp}_l$ for some $f^\perp\in\polyRing$ and $g^\perp\divides_r f^\perp$.

Note that the rows of a generator matrix $G^{\perp}$ of $\cC^{\perp}$ derived in \cref{thm:Hermitian-gen-mat} correspond to skew polynomials $p_1,\dots,p_{k}$ in $\polyRing$, and they form a basis of the code $\cC^{\perp}$ when we see it as an $\ring$-module.
If the dual code $\cC^{\perp}$ has a monic generator polynomial $g^{\perp}$, then ${g^{\perp}}$ must be a right divisor of all the polynomials $p_1,\ldots,p_{k}$.

The method is similar to \cref{algo:dualContainingCodes-GrobBasis}. We see the coefficients of $g^{\perp}$ as unknowns and translate the constraints that $g^{\perp}$ right divides all $p_1,\dots,p_{k}$ into polynomial equations.
We then compute a Gr\"obner basis of the ideal generated by these polynomials.
If the Gr\"obner basis is $\set{1}$, meaning that there is no solution for these constraints, we can conclude that $\cC^{\perp}$ is not a
$(\Endom,\Deriv)$-code.
Otherwise, the Gr\"obner basis gives the solution of the monic generator polynomial $g^{\perp}$ of the $(\Endom,\Deriv)$-code $\cC^{\perp}$.
\cref{algo:is-dual-principle} summarizes our implementation in Magma~\cite{Magma}.
The algorithm can be extended to $\sigma$-Hermitian dual of a $(\Endom,\Deriv)$-code by changing the input $\bG^\perp$ to a generator matrix $\bG^{\perp_\sigma}$ of the $\sigma$-Hermitian dual code.
\begin{algorithm}[htb!]
  \caption{Testing whether the dual code is a $(\theta,\delta)$-code.}
  \label{algo:is-dual-principle}
  \KwIn{A ring $\subring$ where Gr\"obner bases can be computed, a ring $\ring=\subring[\beta_1,\dots,\beta_s]$ which is a free $\subring$-algebra, an endomorphism $\Endom$ and a $\Endom$-derivation $\Deriv$ of $\ring$ which are polynomial maps on $\subring$, and a generator matrix $\bG^\perp\in\ring^{(n-k)\times n}$ of the dual code $\cC^\perp$.}
  \KwOut{A monic generator polynomial $g^\perp$ of degree $k$ of the dual code $\cC^\perp$ or \texttt{False}.}
  $P_1\gets \subring[g^{\perp}_{0,1},\ldots, g^{\perp}_{0,s},\ldots,g^{\perp}_{k-1,1},\ldots,g^{\perp}_{k-1,s}]$
  \tcc*[r]{multivariate ring over $\subring$}
  $g^\perp\gets \sum_{i=0}^{k-1}(\sum_{j=1}^sg^{\perp}_{i,j}\beta_j)\SkewVar^i+ \SkewVar^{k}$
  \tcc*[r]{$g^\perp \in P_1[\SkewVar;\theta,\delta]$}
  $\mathsf{LSEs}\gets \{$Constraints on $g^{\perp}_{i,j}$ such that they are in $\subring$\}
  \tcc*[r]{ $(g^{\perp}_{i,j})^p=g^{\perp}_{i,j}$ if $\subring=\F_p$}
  \ForEach{$p_l\in\RingSkewPolys$ corresponding to the $l$-th row of $\bG^\perp$ \label{line:LCchoice}}{
    $h, r\gets$ quotient, remainder of right dividing $p_l$ by $g^\perp$
    \tcc*[r]{$h,r\in P_1[\SkewVar;\theta,\delta]$}
    $\mathsf{LSEs}\overset{\text{Append}}{\longleftarrow}$\{All coefficients of $r$ are $0$ \}
    \tcc*[r]{implies $g^\perp \divides_r p_l$\label{line:g_rdiv_f}}
  }
  $\mathsf{GB}\gets$ a Gr\"oebner basis of $\mathsf{LSEs}$\;
  \If{$\mathsf{GB}=\set*{1}$}{\Return \texttt{False}}
  \Else{
    $\cS\gets$ \{Solutions of $g^{\perp}_{0,1},\dots, g^{\perp}_{0,s},\dots,g^{\perp}_{k-1,1},\ldots,g^{\perp}_{k-1,s}$ from $\mathsf{GB}$ of $\mathsf{LSEs}$\}\label{line:solutions}\;
  }
  $\cP\gets
  \{g^\perp\in \ring[\SkewVar;\theta,\delta]\ |\ \forall$ solution in $\cS\}$\;
  \tcc{the monic $\hat{g}^\perp\in \ring[\SkewVar;\theta,\delta]$ is reconstructed by evaluating coefficients of $g^\perp\in P_1[\SkewVar;\theta,\delta]$ at each solution in $\cS$} \end{algorithm}

The results of \cref{algo:is-dual-principle} show that dual codes of $(\Endom,\Deriv)$-codes are in general not $(\Endom,\Deriv)$-codes, see \cref{tab:results-is-dual-principle} and \cref{tab:exist_principle_dual_F2u2} for examples over $\F_2[v]$ and $\F_2[u]$, respectively.

\section{Computation Results on Dual-Containing $(\Endom,\Deriv)$-Codes}
\label{sec:computation-results}
In this section we present the ($\sigma$-)dual-containing codes over the rings $\F_2[v], \F_2[u]$ and $\F_2[\alpha]=\F_4$ found by \cref{algo:dualContainingCodes-GrobBasis}. 
\subsection{Results for $\ring=\F_2[v]$ with $v^2=v$}
\label{sec:results-v2-v}
We keep the notation used in \cref{subtab:F2v} and compute the ($\sigma$-)dual-containing $(\Endom,\Deriv)$-code over the ring $\ring=\F_2[v]$ with $v^2=v$ using the method given in \cref{sec:algo-search-via-GB}.
Codes of small length over $\F_2[v]$ are classified in \cite{huffman2005classification}.
We follow \cite{dougherty2001maximum} and define the Lee weight of $0,1,v,v+1$ respectively as $0,2,1,1$ and the Bachoc weight respectively as $0,1,2,2$.

\cref{tab:F2v2+v} gives an overview of the best (in terms of minimum Hamming, Lee or Bachoc distance) dual-containing codes $\cC(g,f)\subset \polyRing/\myspan{f}_l$ (\cref{algo:dualContainingCodes-GrobBasis} found all such codes).
\cref{tab:HammingF2v2} gives detailed Hamming weight distributions that could only been found by some specific $(\Endom,\Deriv)$ combinations.

    \begin{table}[htbp]
      \caption{Results on dual-containing $(\Endom,\Deriv)$-codes over $\F_2[v]$.
      The blue cells mark the code parameters or the weight distributions that could only been found with nonzero derivations. The gray cells mark the code parameters or the weight distributions that could only been found with nonzero derivations and non-trivial endomorphisms.}
    \label{tab:results-F2v2}
    \begin{subtable}[h!]{\textwidth}
      \centering
      \subcaption{The largest Hamming, Lee and Bachoc distance of dual-containing $(\Endom,\Deriv)$-codes over $\mathbb{F}_2[v]$.
        The empty set $\emptyset$ indicates that no dual-containing $(\Endom,\Deriv)$-code
        exists for the parameter $[n,k]$.
        A question mark indicates that such dual-containing codes exist, but we could not compute the minimal distance due to the computation limitation.}
      \label{tab:F2v2+v}
    \begin{tabular}{|l|c|c|c|c|c|c|c|c|c|c|c|}
      \hline
      \diagbox[height=1.5em]{$n$}{$k$}
      & $2$ & $3$ & $4$ &$5$ &$6$ & $7$ & $8$ & $9$ & $10$ & $11$& $12$\\
      \hline  \hline
      $3$ & $1,1,2$&\multicolumn{10}{c|}{} \\
      \hline
      $4$ & $2,2,4$&  $2,2,2$&\multicolumn{9}{c|}{} \\
      \hline
      $5$ &  & $\emptyset$  &$\emptyset$ &\multicolumn{8}{c|}{} \\
      \hline
      $6$ &  &  $2,2,2$  & $2,2,2$   & $2,2,2$ &\multicolumn{7}{c|}{} \\
      \hline
      $7$ &  &  &  $3,3,5$  &$\emptyset$   & $\emptyset$ &\multicolumn{6}{c|}{} \\
      \hline
      $8$ &  &  &   $4,4,7$   & $2,2,4$  & $2,2,2$  & $2,2,2$  &\multicolumn{5}{c|}{} \\
      \hline
      $9$ &  &  &    &   $\emptyset$ &   $\emptyset$  &   $\emptyset$  &\cellcolor{TUMGray!35}{$1,1,2$} &\multicolumn{4}{c|}{}\\
      \hline
      $10$ &  &  &    & $2,2,2$   &  $2,2,2$   &   $\emptyset$    & $\emptyset$ & $2,2,2$ &\multicolumn{3}{c|}{}\\
      \hline
      $11$ &  &  &    &    &    $\emptyset$  &  $\emptyset$   & $\emptyset$  & $\emptyset$  & $\emptyset$ &\multicolumn{2}{c|}{}  \\
      \hline
      $12$ &  &  &    &   &   $4,4,6$   & $3,3,4$ & $2,2,?$ & $ 2,?,?$ & $?,?,?$ & $?,?,?$ &\\
      \hline
      $13$ &  &  &    &   &    &   $\emptyset$    & $\emptyset$ & $\emptyset$ & $\emptyset$  & $\emptyset$ & $\emptyset$ \\
      \hline
    \end{tabular}
  \end{subtable}
  \begin{subtable}[h!]{\textwidth}
    \centering
  \subcaption{Hamming weight distributions of dual-containing $(\Endom,\Deriv)$-codes over $\mathbb{F}_2[v]$.
    \label{tab:HammingF2v2}}
  \begin{tabular}{|c|l|l|}
    \hline
    $[n,k]$  & {Hamming Weight Distribution}    & {Constructed with }  $(\Endom,\Deriv)$\\
    \hline \hline
    \multirow{2}{*}{$[4,2]$} & $[1,0,6,0,9]$ 
                                                     & {all combinations $(\Endom,\Deriv)$}  \\
    \cline{2-3}
             & \cellcolor{TUMBlue!25}{$[1,0,4,4,7]$} 
                                                     & $(\Endom_2,\Deriv_2), (\Endom_3,\Deriv_3),(\Endom_4,\Deriv_4)$   \\
    \hline\hline
    \multirow{1}{*}{$[6,3]$} & $[1,0,9,0,27,0,27]$ 
                             &  {all combinations  $(\Endom,\Deriv)$} \\
    \hline\hline
    \multirow{4}{*}{[6,4]} & $[ 1, 0, 9, 24, 99, 72, 51 ]$ 
                                                     &  {all combinations $(\Endom,\Deriv)$} \\
    \cline{2-3}
             & \cellcolor{TUMBlue!25}$[ 1, 0, 17, 24, 83, 72, 59 ]$ 
                                                & $(\Endom_2,\Deriv_3),(\Endom_2,\Deriv_3)$   \\
    \cline{2-3}
             & \cellcolor{TUMGray!35}$[ 1, 2, 11, 28, 87, 66, 61 ]$ 
                                                & $(\Endom_3,\Deriv_3),(\Endom_4,\Deriv_4)$    \\
    \cline{2-3}
             &\cellcolor{TUMGray!35}$[ 1, 0, 13, 24, 91, 72, 55 ]$
                                                & $(\Endom_3,\Deriv_3),(\Endom_4,\Deriv_4)$  \\
    \hline \hline
    \multirow{4}{*}{[8,4]} & $[ 1, 0, 12, 0, 54, 0, 108, 0, 81 ]$ 
                                                &  {all combinations  $(\Endom,\Deriv)$} \\
    \cline{2-3}
             & $[ 1, 0, 0, 0, 28, 56, 84, 56, 31 ]$ 
                                                &  $(\Endom_2,0)$   \\
    \cline{2-3}
             & \cellcolor{TUMBlue!25}$[ 1, 0, 4, 0, 38, 32, 100, 32, 49 ]$ 
                                                & $(\Endom_2,\Deriv_2),(\Endom_3,\Deriv_3),(\Endom_4,\Deriv_4)$   \\
    \hline
  \end{tabular}
\end{subtable}
\end{table}
\subsubsection{Testing Results by \cref{algo:is-dual-principle}}
 We apply \cref{algo:is-dual-principle} to verify whether the dual of a dual-containing $(\Endom,\Deriv)$-code over $\F_2[v]$ is also a $(\Endom,\Deriv)$-code.
 \cref{tab:results-is-dual-principle} presents the results.
 We list the following two examples to illustrate that the dual code of a dual-containing $(\Endom,\Deriv)$-code is not always a $(\Endom,\Deriv)$-code:
\begin{itemize}
\item
  For $n=4,k=3$, we found three $g\in \F_2[v][X; \Endom_2,\Deriv_2]$ that generate dual-containing $(\Endom,\Deriv)$-codes:
  $g_1=X + v + 1$, $ g_2=X + 1$, $g_3=X + v$ where only the dual of $g_2 = X + 1$ is a $(\Endom,\Deriv)$-code, with $g_2^{\perp}=X^3 + X^2 + X + 1$.
\item For $n=6, k=4$, we found four $g\in \F_2[v][X;\Endom_3,\Deriv_3]$ that generate dual-containing $(\Endom,\Deriv)$-codes:
$g_1 =X^2 + (v + 1) X + v + 1$,
$g_2=X^2 + X + 1$,
$g_3=X^2 + X + v + 1$,
$g_4=X^2 + (v + 1)X + 1$.
Only the dual codes of $g_2$ and $g_4$ are $(\Endom,\Deriv)$-codes, with $g_2^{\perp}=X^4 + X^3 + X + 1$ and $g_4^{\perp}= X^4 + (v + 1)X^3 + X + v + 1$, respectively.
 \end{itemize}

\begin{table}[htbp]
  \caption{Results over $\F_2[v]$ on whether the dual code of a dual-containing $(\Endom,\Deriv)$-code is also a $(\Endom,\Deriv)$-code.}
  \label{tab:results-is-dual-principle}
  \begin{subtable}[h!]{\textwidth}
    \centering
  \subcaption{The entries indicates that None/Some/All of the $[n,k]$ dual-containing $(\Endom,\Deriv)$-codes whose dual codes are also $(\Endom,\Deriv)$-codes.}
  \label{tab:exist_principle_dual}
  \begin{tabular}{|l|c|c|c|c|c|c|c|c|}
    \hline
    \diagbox[height=1.5em]{$n$}{$k$}
      & 2 & 3 & 4 &5 &6 &7& 8 & 9 \\   \hline  \hline
      3 & None &\multicolumn{7}{c|}{}\\ \hline
      4 & All & Some &\multicolumn{6}{c|}{}\\  \hline
      5 &  & /  & / &\multicolumn{5}{c|}{}\\  \hline
      6 &  &  All & Some & Some &\multicolumn{4}{c|}{} \\  \hline
      7 &  &  &  All & /   & / &\multicolumn{3}{c|}{}\\  \hline
             8 &  &  &  All  & Some  & Some & Some &\multicolumn{2}{c|}{} \\  \hline
             9 &  &  &    &   /   & /   &  /   & None& \\  \hline
             10 &  &  &    & All  & Some  &   /     & /  & All \\  \hline
  \end{tabular}
  \end{subtable}
  \begin{subtable}[htb!]{\textwidth}
    \centering
    \subcaption{The number of dual-containing $(\Endom,\Deriv)$-codes whose dual codes are also $(\Endom,\Deriv)$-codes. Only the parameters marked with ``Some'' in \cref{tab:exist_principle_dual} are listed.}
    \label{tab:num_cyclic_dual}
    \begin{tabular}{|l|c|c|c|c|c|c|c|c|c|}
      \hline
      \multirow{2}{*}{$[n,k]$} &  \multicolumn{9}{c|}{\# of dual-containing $(\Endom,\Deriv)$-codes for each $(\Endom,\Deriv)$} \\
                               &  \multicolumn{9}{c|}{\# of dual-containing $(\Endom,\Deriv)$-codes whose dual codes are also $(\Endom,\Deriv)$-codes}   \\
      \hline &$(\mbox{Id};0)$ & $(\Endom_2,0)$&$(\Endom_2,\Deriv_2)$&$(\Endom_2,\Deriv_3)$&$(\Endom_2,\Deriv_4)$&$  (\Endom_3,0)$&$(\Endom_3,\Deriv_3)$&$(\Endom_4,0)$&$(\Endom_4,\Deriv_4)$ \\
      \hline
      \multirow{2}{*}{$[4,3]$} &  1 & 1&  3& 1& 1& 1& 2& 1& 2 \\
                               &1& 1& 1& 1& 1& 1& 1& 1& 1 \\
      \hline
      \multirow{2}{*}{$[6,4]$} &  1& 1& 1& 2& 2& 1& 4& 1& 4  \\
                               &  1& 1& 1& 1& 1& 1& 2& 1& 2  \\
      \hline
      \multirow{2}{*}{$[6,5]$} &  1& 1& 1& 2& 2& 1& 1& 1& 1  \\
                               &  1& 1& 1& 1& 1& 1& 1& 1& 1  \\
      \hline
      \multirow{2}{*}{$[8,5]$} &  1& 3& 5& 1& 1& 1& 8& 1& 8  \\
                               & 1& 3& 1& 1& 1& 1& 1& 1& 1  \\
      \hline
      \multirow{2}{*}{$[8,6]$} &  1& 3& 5& 1& 1& 1& 4& 1& 4  \\
                               &  1& 3& 3& 1& 1& 1& 2& 1& 2  \\
      \hline
      \multirow{2}{*}{$[8,7]$} &  1& 1& 3& 1& 1& 1& 2& 1& 2  \\
                               &  1& 1& 1& 1& 1& 1& 1& 1& 1  \\
      \hline
      \multirow{2}{*}{$[10,6]$} &  1& 1& 1& 1& 1& 1& 16& 1& 16 \\
                               &  1& 1& 1& 1& 1& 1& 2& 1& 2 \\
      \hline
    \end{tabular}
  \end{subtable}
\end{table}

\subsubsection{Hermitian Dual-Containing $(\Endom,\Deriv)$-Codes over $\F_2[v]$}
It can be verified that the automorphism $\Endom_2$ of $\F_2[v]$ is of order $2$. We hence use \cref{algo:dualContainingCodes-GrobBasis} to find all $\Endom_2$-dual-containing $(\Endom,\Deriv)$-codes over $\F_2[v]$.
\cref{tab:F2v2+v_herm} gives an overview of the largest minimum Hamming, Lee or Bachoc distance of $\Endom_2$-dual-containing codes $\cC(g,f)\subset \polyRing/\myspan{f}_l$.
\cref{tab:HammingF2v2_Herm} gives some Hamming weight distributions that could only be found by some specific $(\Endom,\Deriv)$ combinations.
\begin{table}[htbp]
  \caption{Results on $\Endom_2$-dual-containing $(\Endom,\Deriv)$-codes over $\F_2[v]$. The colored cells and special symbols have the same indications as in \cref{tab:results-F2v2}.}
  \label{tab:results-hermitian-dual-F2v2}
  \begin{subtable}[htb!]{\textwidth}
    \centering
    \subcaption{The largest Hamming, Lee, Bachoc distance of $\Endom_2$-dual-containing $(\Endom,\Deriv)$-codes over $\mathbb{F}_2[v]$.
    \label{tab:F2v2+v_herm}}
  \begin{tabular}{|l|c|c|c|c|c|c|c|c|c|c|c|}
    \hline
    \diagbox[height=1.5em]{$n$}{$k$} & $2$ & $3$ & $4$&$5$ &$6$ &$7$& $8 $ & $ 9 $ \\
    \hline  \hline
    $4 $ & $2,2,4$&  $2,2,2$ &\multicolumn{6}{c|}{}\\
    \hline
    $5 $ &  & \cellcolor{TUMBlue!25}{$2,2,2$} &\cellcolor{TUMBlue!25}{$1,1,2$}&\multicolumn{5}{c|}{}\\
    \hline
    $6 $ &  & \cellcolor{TUMGray!35}{$3,3,4$}   & $ 2,2,4   $ & $ 2,2,2 $ &\multicolumn{4}{c|}{}\\
    \hline
    $7 $ & & & $3,3,5$  &\cellcolor{TUMBlue!25}{$1,1,2$}   &\cellcolor{TUMBlue!25}{$1,1,2$} &\multicolumn{3}{c|}{}\\
    \hline
    $8 $ & &  & $3,3,6$ & $ 2,2,4  $ & $ 2,2,2 $ &$2,2,2$ &\multicolumn{2}{c|}{} \\
    \hline
    $9 $ & & & & \cellcolor{TUMGray!35}{$1,1,2$}&   $\emptyset$   & $\emptyset$  &$\emptyset$& \\
    \hline
    $10 $ & &  &  & $2,2,2$ & $2,2,2$ & $\emptyset$ & $ \emptyset$ &$2,2,2$ \\  \hline
  \end{tabular}
  \end{subtable}
  \begin{subtable}[htb!]{\textwidth}
    \centering
    \subcaption{The Hamming weight distributions of $\Endom_2$-dual-containing $(\Endom,\Deriv)$-codes over $\mathbb{F}_2[v]$.}
    \label{tab:HammingF2v2_Herm}
    \begin{tabular}{|c|l|l|}
      \hline
      $[n,k]$  & \mbox{Hamming Weight}    & \mbox{Constructed with }  $(\Endom,\Deriv)$\\
      \hline \hline
      \multirow{1}{*}{$[4,2]$} & $[1,0,6,0,9]$ 
                                          & all combinations  $(\Endom,\Deriv)$\\ 
      \cline{2-3}&$[1,0,2,8,5]$ & 
                                                      $(\Endom_{2},0)$\\
      \hline
      \multirow{1}{*}{$[4,3]$} & $[1,0,18,24,21]$ 
                                          & all combinations  $(\Endom,\Deriv)$ \\ 
      \cline{2-3}&\cellcolor{TUMBlue!25}$[1,2,16,22,23]$& 
                                                                                     $(\Endom_{2},\Deriv_2),(\Endom_3,\Deriv_3),(\Endom_4,\Deriv_4)$\\
      \cline{2-3}&\cellcolor{TUMBlue!25}$[1,2,12,30,19]$ & 
                                                                 $(\Endom_{2},\Deriv_3), (\Endom_2,\Deriv_4)$\\
      \hline
      \multirow{1}{*}{$[5,3]$} & \cellcolor{TUMBlue!25}$[1,0,8,14, 23, 18]$ 
                                          & $(\Endom_2,\Deriv_3),(\Endom_2,\Deriv_4)$ \\
      \hline
      \multirow{1}{*}{$[5,4]$} &\cellcolor{TUMBlue!25}$[1,3,22,66,105,59]$ & 
                                                                                                              $(\Endom_{2},\Deriv_{3}),
                                                                                                              (\Endom_{2},\Deriv_{4})$\\
      \hline
      \multirow{1}{*}{$[6,3]$} &
                   \cellcolor{TUMGray!35}$[1,0,0,8,21, 24, 10]$ & 
                                                                   $(\Endom_{3},\Deriv_{3}),
                                                                   (\Endom_{4},\Deriv_{4})$\\
      \hline
    \end{tabular}
  \end{subtable}
\end{table}

\subsubsection[An Example where $f$ being Central is Not Necessary for Dual-Containing Polycyclic Codes]{An Example where $f$ being Central is Not Necessary for Dual-Containing $(\Endom,\Deriv)$-Codes $\cC(g,f)$}
\label{sec:all-f-are-noncentral}
In most studies on the dual-containing $(\Endom,\Deriv)$-codes $\cC(g,f)$, e.g., \cite{boucher2008skew,boucher2011note,boulagouaz2018characterizations}, $f$ is assumed to be a central polynomial, since it is easier to derive closed formulas of a generator polynomial $g$ of the code.
With the example below we intend to show that there are dual-containing codes $\cC(g,f)$ where $g$ is not a right factor of any central $f$.

Note that many $f=hg=g \hbar$ can exist for the same $g$, and all $(g,f)$ pairs lead to the same code, whose generator matrix is determined only by $g$ and the corresponding $(\Endom,\Deriv)$.
To illustrate this, we present more details of the $[6,4]$ code with the Hamming weight distribution $[ 1, 0, 13, 24, 91, 72, 55 ]$ in \cref{tab:HammingF2v2}. There are four possible generator polynomials $g$ of this code and they are presented in \cref{tab:gG_n6_k4}.
\begin{table}[htbp]
  \centering
  \caption{All possible generator polynomials/matrices of the $[6,4]$ dual-containing $(\Endom,\Deriv)$-codes over $\mathbb{F}_2[v]$ with Hamming weight distribution $[ 1, 0, 13, 24, 91, 72, 55 ]$.}
  \label{tab:gG_n6_k4}
  \begin{tabular}{l|c|c|c}
    \hline
Index & $g$& $(\Endom,\Deriv)$ &$\mathbf{G}$  \\
\hline
1& $g=X^2 + X + v + 1$ &  $(\Endom_{3},\Deriv_{3})$
  & $\begin{pmatrix}
v + 1 & 1 & 1 & 0 & 0 & 0\\
v & 1 & 1 & 1 & 0 & 0\\
v & 0 & 1 & 1 & 1 & 0\\
v & 0 & 0 & 1 & 1 & 1
\end{pmatrix}
                    $ \\
    \hline
2& $g=X^2 + (v + 1)X + 1$   & $(\Endom_{3},\Deriv_{3})$
    & $\begin{pmatrix}
1 & v + 1 & 1 & 0 & 0 & 0\\
0 & v + 1 & 1 & 1 & 0 & 0\\
0 & v & 1 & 1 & 1 & 0\\
0 & v & 0 & 1 & 1 & 1
\end{pmatrix}
$ \\
    \hline
3&    $g=X^2 + vX + 1$
&$(\Endom_{4},\Deriv_{4})$
&$\begin{pmatrix}
1 & v & 1 & 0 & 0 & 0\\
0 & v & 1 & 1 & 0 & 0\\
0 & v + 1 & 1 & 1 & 1 & 0\\
0 & v + 1 & 0 & 1 & 1 & 1
\end{pmatrix}
$ \\
    \hline
4& $g=X^2 + X + v$   &
$(\Endom_{4},\Deriv_{4})$
 &$\begin{pmatrix}
v & 1 & 1 & 0 & 0 & 0\\
v + 1 & 1 & 1 & 1 & 0 & 0\\
v + 1 & 0 & 1 & 1 & 1 & 0\\
v + 1 & 0 & 0 & 1 & 1 & 1
\end{pmatrix}
                        $\\
    \hline
  \end{tabular}
\end{table}
We consider the first $g=\SkewVar^2 + \SkewVar + v + 1\in \mathbb{F}_2[v][\SkewVar, \Endom_{3},\Deriv_{3}]$ in \cref{tab:gG_n6_k4}. There are only $8$ non-central polynomials $f$ such that $f=hg=g\hbar$ for some $h,\hbar\in\polyRing$ (i.e., $g$ is both a left and right divisor of $f$):
\begingroup
\allowdisplaybreaks
\begin{eqnarray*}
  f_1& = & X^6 + vX^4 + vX^3 + vX + v + 1 = (X^4 + X^3 + vX^2 + X + v + 1) \cdot g\\
  f_2 & = & X^6 + X^5 + (v + 1) X^4 + X^3 + v X + v + 1 = (X^4 + v X^2 + (v + 1) X + 1) \cdot g\\
  f_3& =&X^6 + (v + 1) X^4 + v X^3 + v X^2 + X + v + 1= (X^4 + X^3 + (v + 1) X^2 + 1) \cdot g\\
   f_4 &=& X^6 + X^5 + v X^4 + X^3 + v X^2 + X + v + 1 =
                                                      (X^4 + (v + 1) X^2 + v X + v + 1 ) \cdot g\\
   f_5 &=& X^6 + v X^4 + v X^3 + X^2 + (v + 1) X = (X^4 + X^3 + v X^2 + X + v) \cdot g\\
    f_6 &= &X^6 + X^5 + (v + 1) X^4 + X^3 + X^2 + (v + 1) X = (X^4 + v X^2 + (v + 1) X) \cdot g\\
   f_7 &= &X^6 + (v + 1) X^4 + v X^3 + (v + 1) X^2 = (X^4 + X^3 + (v + 1) X^2) \cdot g\\
    f_8& =& X^6 + X^5 + v X^4 + X^3 + (v + 1) X^2 = (X^4 + (v + 1) X^2 + v X + v) \cdot g
\end{eqnarray*}
\endgroup
 For each $f_i, i=1,\dots, 8$, there is a unique $h_i$ corresponding to $f_i=h_ig$ (see the decomposition above) and 16 distinct $\hbar_i$ such that $f_i=g\hbar_i$, where one of $\hbar_i$ is equal to $h_i$.
In the following we present for $f_1$ the other 15 distinct $\hbar_1\neq h_1$ such that $f_1=g\hbar_1$:
\begingroup
\allowdisplaybreaks
 \begin{align*}
   f_1 
   &= g\cdot \left(  X^4 + X^3 + X + 1\right)
   &=&\ g\cdot\left( X^4 + X^3 + v X^2 + (v + 1) X + v + 1\right)\\
        &=g\cdot\left(X^4 + (v + 1) X^3 + (v + 1) X + v + 1 \right)
   &=&\ g\cdot \left(X^4 + X^3 + (v + 1) X + v + 1\right)\\
        &=g\cdot \left(X^4 + (v + 1) X^3 + v X^2 + X + v + 1\right)
   &=&\ g\cdot \left(X^4 + (v + 1) X^3 + X + v + 1 \right)\\
        &= g\cdot \left(X^4 + X^3 + X + v + 1 \right)
   &=&\ g\cdot \left( X^4 + (v + 1) X^3 + v X^2 + (v + 1) X + 1\right)\\
        &= g\cdot \left(X^4 + X^3 + v X^2 + (v + 1) X + 1 \right)
   &=&\ g\cdot \left( X^4 + (v + 1) X^3 + (v + 1) X + 1 \right)\\
        &= g\cdot \left(X^4 + X^3 + (v + 1) X + 1 \right)
   &=&\ g\cdot \left( X^4 + (v + 1) X^3 + v X^2 + X + 1\right)\\
        &= g\cdot \left(X^4 + X^3 + v X^2 + X + 1  \right)
   &=&\ g\cdot \left(X^4 + (v + 1) X^3 + X + 1 \right)\\
        &= g\cdot \big(X^4 + (v + 1) X^3 + v X^2 + (v + 1) X + v \hspace{-10pt}&+ &1\big)\ .
 \end{align*}
 \endgroup

\subsection{Results for $\ring=\F_2[u]$ with $u^2=0$}
We keep the notations used in \cref{subtab:F2u} and compute the dual-containing $(\Endom,\Deriv)$-codes over the ring $\ring=\F_2[u]$ with $u^2=0$ using the method presented in \cref{sec:algo-search-via-GB}.
We follow \cite{dougherty2001maximum} and define the Lee weight of $0,1,u,u+1$ respectively as $0,1,2,1$ and the Euclidean weight respectively as $0,1,4,1$.

\cref{tab:F2u2} gives an overview of the best (in terms of minimum Hamming, Lee and Euclidean distance) dual-containing $(\Endom,\Deriv)$-codes (\cref{algo:dualContainingCodes-GrobBasis} found all such codes).
\cref{tab:HammingF2u2} gives some Hamming weight distributions that could only been found by some specific $(\Endom,\Deriv)$ combinations.
\begin{table}[htbp]
  \caption{Results on dual-containing $(\Endom,\Deriv)$-codes over $\F_2[u]$. The colored cells and special symbols have the same indications as in \cref{tab:results-F2v2}.}
  \label{tab:results-F2u2}
  \begin{subtable}[htb!]{\textwidth}
    \centering
    \subcaption{The best Hamming, Lee, and Euclidean distances of dual-containing $(\Endom,\Deriv)$-codes over $\F_2[u]$.
    \label{tab:F2u2}}
  \begin{tabular}{|l|c|c|c|c|c|c|c|c|}
    \hline
    \diagbox[height=1.5em]{$n$}{$k$}& $2$ & $3$ & $4$ &$5$ &$6$ &$7$& $8$ & $9$ \\
    \hline  \hline
    $4$ &  $2,4,4$& $2,2,2$ &\multicolumn{6}{c|}{}\\
    \hline
    $5$ &  & $\emptyset$  &\cellcolor{TUMBlue!25}{$1,2,2$} &\multicolumn{5}{c|}{}\\
    \hline
    $6$ &  &  $2,4,4$ & $2,2,2$ & $2,2,2$ &\multicolumn{4}{c|}{} \\
    \hline
    $7$ &  &  &  $3,3,3$ & $\emptyset$   & {$1,2,2$} &\multicolumn{3}{c|}{}
    \\
    \hline
    $8$ &  &  &  $4,4,4$  & $2,4,4$  & $2,2,2$ &$2,2,2$  &\multicolumn{2}{c|}{} \\
    \hline
    $9$ &  &  &    &   $\emptyset$   & $\emptyset$   &  $\emptyset$   &\cellcolor{TUMBlue!25}{$1,2,2$}& \\
    \hline
    $10$ &  &  &    & $2,4,6$  & $2,4,5$  &   $\emptyset$     & $\emptyset$  &$2,2,2$ \\
    \hline
  \end{tabular}
  \end{subtable}
  \begin{subtable}[htb!]{\textwidth}
    \centering
  \subcaption{Hamming weight distributions of dual-containing $(\Endom,\Deriv)$-codes over $\mathbb{F}_2[u]$.
    \label{tab:HammingF2u2}}
  \begin{tabular}{|c|l|l|} \hline
$[n,k]$  & \mbox{Hamming Weight}    & \mbox{Constructed with }  $(\Endom,\Deriv)$\\ \hline \hline
    \multirow{2}{*}{$[4,2]$} & $[1,0,2,8,5]$ 
                                    & $(\Id,0), (\Id,\Deriv_2), (\Id,\Deriv_3), (\Endom_2,\Deriv_2)$   \\
    \cline{2-3}
         &   $[1,0,6,0,9]$ 
                                    & all combinations $(\Endom,\Deriv)$ 
    \\ \hline  \hline
    \multirow{5}{*}{$[8,4]$} & $[1,0,4,0,30, 64, 52, 64, 41]$ 
                             &  $(\Id,0), (\Endom_2,\Deriv_2)$   \\
    \cline{2-3}
         & $[1,0,4,0,46, 0, 148, 0,57]$ 
                             & $(\Id,0)$   \\
    \cline{2-3}
         & $[1,0,4,16, 14, 32, 84, 80, 25]$ 
                             &   $(\Id,0)$   \\
    \cline{2-3}
         & $[1,0,12, 0,54, 0, 108, 0,81]$ 
                             &  all combinations $(\Endom,\Deriv)$ \\
    \cline{2-3}
         &\cellcolor{TUMBlue!25}$[1,0,0,0,26, 64, 72, 64, 29]$ 
                             &  $(\Id,\Deriv_2)$   \\
    \hline  \hline
    \multirow{4}{*}{$[8,5]$} & $[1,0,4,16, 94, 224, 308, 272, 105]$ 
                                    &  $(\Id,0),(\Id,\Deriv_2)$   \\
    \cline{2-3}
         & $[ 1, 0, 4, 16, 110, 160, 404, 208, 121 ]$ 
                                    &     $(\Id,0)$   \\ \cline{2-3}
         & $[ 1, 0, 12, 0, 102, 192, 396, 192, 129 ]$
                                    &   all combinations $(\Endom,\Deriv)$  \\ \cline{2-3}
         &\cellcolor{TUMBlue!25}$[ 1, 0, 16, 8, 114, 176, 360, 200, 149 ]$
                                    & $(\Id,\Deriv_2)$ \\ \hline
  \end{tabular}
\end{subtable}
\end{table}
\subsubsection{Testing Results by \cref{algo:is-dual-principle}}
We apply \cref{algo:is-dual-principle} to verify whether the dual of the a dual-containing $(\Endom,\Deriv)$-code over $\F_2[u]$ is also a $(\Endom,\Deriv)$-code.
\cref{tab:exist_principle_dual_F2u2} presents the results.
 \begin{table}[htbp]
    \centering
    \caption{Test results over $\F_2[u]$ on whether the dual of a dual-containing $(\Endom,\Deriv)$-code is also a $(\Endom,\Deriv)$-code.
    }
\label{tab:exist_principle_dual_F2u2}
       \begin{tabular}{|l|c|c|c|c|c|c|c|c|}
     \hline
    \diagbox[height=1.5em]{$n$}{$k$}& 2 & 3 & 4 &5 &6 &7& 8 & 9 \\   \hline  \hline
      4 & All & All &\multicolumn{6}{c|}{} \\  \hline
  5 &  & /  & {
              \begin{tabular}{@{}c@{}}
                $(\mbox{id},\Deriv_2)$: None\\
                $(\mbox{id},\Deriv_4)$: All
              \end{tabular}}&\multicolumn{5}{c|}{}\\
                 \hline
  6 &  &  All  & All  & All &\multicolumn{4}{c|}{}  \\  \hline
  7 &  &  &  All  & /   & All &\multicolumn{3}{c|}{} \\  \hline
  8 &  &  &  All   & All   & All  &All &\multicolumn{2}{c|}{}   \\  \hline
         9 &  &  &    &   /   & /   &  /   & {
                                             \begin{tabular}{@{}c@{}}
                                              $(\Id,\Deriv_2)$: None \\
                                              $(\Id,\Deriv_4)$: Some
                                             \end{tabular}} &\\  \hline
  10 &  &  &    & All   & All   &   /     & /  &All  \\  \hline
       \end{tabular}
\end{table}

\subsection{Results for $\ring=\F_2[\alpha]=\F_4$}
We keep the notations used in \cref{subtab:F2alpha}. Note that $\Endom_2:a\mapsto a^2$ is an automorphism of order $2$. We compute the $\Endom_2$-dual-containing $(\Endom,\Deriv)$-codes over $\F_4$ by \cref{algo:dualContainingCodes-GrobBasis}.
Following \cite{dougherty2001maximum} we define the Lee weight of $0,1,\alpha,\alpha+1$ as $0,2,1,1$, respectively, and following \cite{LingSole2001} we define the Euclidean weight as $0,1,2,1$, respectively.

  \cref{tab:F4_herm} shows the existence and the best Hamming, Lee and Euclidean distance of the $\Endom_2$-Hermitian dual-containing $(\Endom,\Deriv)$-codes over $\F_4$.
  \cref{tab:HammingF4} provides some examples of the Hamming weight distributions.
  \begin{table}[htbp]
    \caption{Results on $\Endom_2$-dual-containing $(\Endom,\Deriv)$-codes over $\F_2[\alpha]=\F_4$. The colored cells and special symbols have the same indications as in \cref{tab:results-F2v2}.}
    \label{tab:results-F4_herm}
    \begin{subtable}[htb!]{\textwidth}
      \centering
    \subcaption{The best Hamming, Lee and Euclidean distance of $\Endom_2$-dual-containing codes over $\mathbb{F}_4$.
     \label{tab:F4_herm}}
       \begin{tabular}{|l|c|c|c|c|c|c|c|c|}
         \hline
         \diagbox[height=1.5em]{$n$}{$k$} & $2$ & $3$ & $4$ &$5$ &$6$ &$7$& $8$ & $9$\\
         \hline  \hline
         $4$ & $2,2,2$&  $2,2,2$ &\multicolumn{6}{c|}{} \\
         \hline
         $5$ &  & $3,3,3$  &\cellcolor{TUMBlue!25}$1,1,1$ &\multicolumn{5}{c|}{}  \\
         \hline
         $6$ &  &  $4,4,4$  & $2,2,2$   & $2,2,2$ &\multicolumn{4}{c|}{}  \\
         \hline
         $7$ &  &  &  $3,3,3$  & $\emptyset$  & \cellcolor{TUMBlue!25}$1,1,1$ &\multicolumn{3}{c|}{}  \\
         \hline
         $8$ &  &  &   $2,2,2$  & $2,2,2$  & $2,2,2$ &$2,2,2$ &\multicolumn{2}{c|}{} \\
         \hline
         $9$ &  &  &    &   $\emptyset$&   $\emptyset$ &   $\emptyset$  &\cellcolor{TUMBlue!25}$1,1,1$ &\\
         \hline
         $10$ &  &  &    &  $(4,4,4)$  &  $(3,3,3)$ & $(2,2,2)$     & $(2,2,2)$ & $(2,2,2)$\\
         \hline
       \end{tabular}
     \end{subtable}
\begin{subtable}[htb!]{\textwidth}
      \centering
  \caption{Weight distributions of $\Endom_2$-dual-containing $(\Endom,\Deriv)$-codes over $\F_4$.
    \label{tab:HammingF4}
  }
  \begin{tabular}{|c|l|c|}
    \hline
    $[n,k]$  & \mbox{Hamming Weight Enumerator}    & \mbox{Constructed with }  $(\Endom,\Deriv)$\\
    \hline \hline
    \multirow{2}{*}{$[4,3]$}
       & $[1,0,18,24, 21]$ 
                                                   & \text{all combinations $(\Endom,\Deriv)$ maps}   \\
    \cline{2-3}
             &\cellcolor{TUMBlue!25}$[1,6,12,18, 27]$ 
                                                   & $(\Endom_2,\Deriv_2)$   \\
    \hline\hline
    \multirow{1}{*}{$[5,4]$}
       & \cellcolor{TUMBlue!25}$[1,9,30, 54, 81, 81]$ 
                                                   & $(\Endom_2,\Deriv_2 )$ \\
    \hline\hline
    \multirow{2}{*}{$[6,5]$}
             & $[1,0, 45, 120, 315, 360, 183]$ 
                                                   & all combinations $(\Endom,\Deriv)$   \\
    \cline{2-3}
             & \cellcolor{TUMBlue!25}$[1,12, 57, 144, 243, 324, 243]$ 
                                                   & $(\Endom_2,\Deriv_2 )$ \\
    \hline
 \end{tabular}
\end{subtable}
\end{table}
 \section{Summary and Outlooks}
 \label{sec:mod-skew-conclusion-outlooks}
This chapter considers $(\Endom,\Deriv)$-polycyclic codes that are constructed from principle ideals of skew polynomials with endomorphisms $\Endom$ and $\Endom$-derivations $\Deriv$.
In particular, we focused on constructing dual-containing $(\Endom, \Deriv)$-codes over rings.
As a basis, we first derived a parity-check matrix of a $(\Endom, \Deriv)$-code within the framework of skew polynomials.
For a finite commutative Frobenius ring $\ring$, we then derived generator matrices of the Euclidean dual and the $\sigma$-Hermian dual of a $(\Endom, \Deriv)$-code. This implies that the ($\sigma$-Hermitian) dual codes are $\ring$-modules.
For $\ring=\subring[\beta_1, \dots, \beta_s]$ being a free $\subring$-algebra, we developed an algorithm using Gr\"obner bases to compute all the ($\sigma$-)dual-containing $(\Endom, \Deriv)$-codes over $\ring$.
We also presented an algorithm to test whether the dual code is also a $(\Endom, \Deriv)$-code, in other words, whether there is a generator polynomial in $\RingSkewPolys$ of the dual code.

With the computational results for several rings of order $4$, we obtain the following observations:
\begin{itemize}
\item nonzero derivations and non-trivial endomorphisms (not automorphisms) do give new (Euclidean/Hermitian) dual-containing $(\Endom, \Deriv)$-codes that could not been found by zero derivations or automorphisms. See \cref{tab:results-F2v2}, \cref{tab:results-hermitian-dual-F2v2}, \cref{tab:results-F2u2}, and \cref{tab:results-F4_herm}.
\item The monic generator polynomial $g$ being a right factor of some central polynomial $f$ is not a necessary condition for the $(\Endom, \Deriv)$-code generated by $g$ to be a dual-containing code. See the example in \cref{sec:all-f-are-noncentral}.
\item The dual code of a dual-containing $(\Endom, \Deriv)$-code is in general not a $(\Endom, \Deriv)$-code. See \cref{tab:results-is-dual-principle} and \cref{tab:exist_principle_dual_F2u2}.
\end{itemize}

It can be seen that there are dual-containing $(\Endom, \Deriv)$-codes over rings with large minimum Hamming distances, e.g., the $[6,3]_4$ code over $\F_2[v]$ with $\dH=3$ in \cref{tab:results-hermitian-dual-F2v2}, and the $[8,4]_4$ code over $\F_2[u]$ with $\dH=4$ in \cref{tab:results-F2u2}. For future research, it would be helpful to compare these codes with other existing codes over rings or some bounds on the cardinality or the distance of codes over rings, such as Singleton-like bounds, sphere-packing bounds.
Moreover, fast decoding algorithms for codes based on skew polynomials with automorphisms and zero derivations has been extensively studied lately, e.g., in~\cite{bartz2021fast,Hoermann2022efficient}. Since both of the above codes with large Hamming distances can only be found by non-trivial endomorphisms or nonzero derivations, advanced decoding algorithms based on~\cite{boucher2020algorithm} for such $(\Endom, \Deriv)$-codes are relevant to be developed.

%% file: chap_eva_skew.tex
Gabidulin codes \cite{gabidulin1985theory}, introduced by
Ernest M.~Gabidulin, were the first evaluation codes from a special class of skew polynomials, namely the linearized polynomials, where the Frobenius automorphism and zero derivation are used.
Boucher and Ulmer \cite{boucher2014linear} extended the notion of evaluation codes to skew polynomials over finite fields with inner derivations.
Liu,
Manganiello and
Kschischang \cite{liu2015construction} defined \emph{generalized skew-evaluation codes} that contain Gabidulin codes as a special class, which combine the maximum distance separable (MDS) and the maximum rank distance (MRD) properties via a pasting construction.
Mart{\'\i}nez-Pe{\~n}as introduced in \cite{martinez2018skew} \emph{linearized Reed-Solomon} (LRS) codes, a class of evaluation codes based on skew polynomials that achieve the maximum sum-rank distance (MSRD) property.

This chapter investigates the support-constrained MSRD codes motivated from multi-source network coding and the advantage of vector network coding compared to scalar network coding.
We first give brief introductions to support-constrained codes and LRS codes in \cref{sec:support-constrained-codes} and \cref{sec:LRScodes}, respectively.
In \cref{sec:support-constrained-LRS} we present a necessary and sufficient condition on a generator matrix $\bG$ such that it generates an MSRD code. Moreover, if the required generator matrix $\bG$ does not satisfy the condition, we give the largest possible sum-rank distance of the code generated by $\bG$.
Using these results, we give in \cref{sec:multi-source-net-cod} a scheme to design \emph{distributed LRS codes} for distributed multi-source unicast networks.
Finally, we turn our focus to a family of \emph{multicast} networks, namely the \emph{generalized combination networks}, in \cref{sec:vec-net-cod}. We investigate the advantages of using vectors as coding coefficients at the relay nodes in the network compared to using scalars. The advantages are shown via the gap between the minimum required field size of vector coding solutions and that of scalar coding solutions.

\emph{The results in \cref{sec:support-constrained-LRS,sec:multi-source-net-cod} have been submitted to IEEE Transactions on Information Theory (TIT) and partly published in the proceeding of
  2023 IEEE Information Theory Workshop (ITW)~\cite{liu2023linearized}. The results in \cref{sec:vec-net-cod} were published in IEEE TIT~\cite{liu2021gap} and partly in the proceeding of 2020 IEEE Information Symposium on Information Theory (ISIT)~\cite{liu2020gap}. }

\section{Support-Constrained Codes}
\label{sec:support-constrained-codes}
\emph{Support-constrained codes} are codes that have some codewords having zeros at certain positions and these codewords form a basis of the code. In other words, a support-constrained code has a generator matrix having zeros at certain entries.
A formal definition is given as follows.
\begin{definition}[Support-constrained codes]
  \label{def:support-constrained-codes}
  Given the support constraints $\zeroSet_1, \dots, \zeroSet_k\subseteq \set*{1,\dots, n}$, a linear $[n,k]_q$ code is a \emph{support-constrained code} w.r.t.~$\zeroSet_1, \dots, \zeroSet_k$ if it has a generator matrix $\bG\in\Fq^{k\times n}$ fulfilling the support constraints, i.e.,
\begin{align}\label{eq:zeroConstraints}
    G_{ij}=0\ ,\quad \forall i\in\{1,\dots,k\},\ \forall j \in \zeroSet_i\ .
\end{align}
\end{definition}

Designing support-constrained error-correcting codes was motivated by its application in
weakly secure network coding for wireless cooperative data exchange \cite{yan2011weakly, song2012error, yan2014weakly, li2017cooperative}, where each node in the network stores a subset of all messages and the nodes communicate via broadcast transmissions to disseminate the messages in the presence of an eavesdropper.

From both, the theoretical and the practical point of view, the objective is to design support-constrained codes achieving the largest possible minimum distance.
In the Hamming metric, research focused on proving the following necessary and sufficient condition for the existence of MDS codes fulfilling the support constraints. The condition was first conjectured in \cite{dau2014existence} (known as the \emph{GM--MDS conjecture}), further studied in \cite{heidarzadeh2017algebraic, yildiz2018further}, and finally proven independently by
Lovett \cite{lovett2018mds} and by
Yildiz and
Hassibi \cite{yildiz2018optimum}.
\begin{theorem}[GM--MDS condition {\cite{yildiz2018optimum,lovett2018mds}}]
  \label{thm:yildizMDScondition}
  Let $\zeroSet_1,\dots,\zeroSet_k\subseteq \{1,\dots, n\}$.
  For any $q\geq n+k-1$, there exists an $[n,k]_{q}$ Reed--Solomon (RS) code with a generator matrix $\bG\in\Fq^{k\times n}$ fulfilling the support constraints in \eqref{eq:zeroConstraints},
if and only if, for any nonempty $\Omega\subseteq \{1,\dots, k\}$,
\begin{align}\label{eq:GM-MDScondition}
\left|\bigcap_{i\in\Omega}\zeroSet_i\right|+|\Omega|\leq k\ .
\end{align}
Moreover, if an MDS code has a generator matrix fulfilling the support constraint \eqref{eq:zeroConstraints}, then the sets $\zeroSet_i$’s satisfy \eqref{eq:GM-MDScondition}.
\end{theorem}

Yildiz and Hassibi adapted the approach for Gabidulin codes in \cite{yildiz2019gabidulin} and derived the following GM--MRD condition.
\begin{theorem}[GM--MRD condition {\cite[Theorem 1]{yildiz2019gabidulin}}]
  Let $\zeroSet_1,\dots,\zeroSet_k\subseteq \{1,\dots,n\}$.
  For any prime power $q$ and integer $m\geq \max \{n, k-1+\log_q k\}$, there exists an $[n,k]_{q^m}$ Gabidulin code with a generator matrix $\bG\in\Fqm^{k\times n}$ fulfilling \eqref{eq:zeroConstraints}, if and only if, for any nonempty $\Omega\subseteq \{1,\dots, k\}$, the inequality \eqref{eq:GM-MDScondition} holds.
\end{theorem}



Recently, special cases of support-constrained MDS codes have also been studied. For instance,
the work by Greaves and Syatriadi~\cite{greaves2019reed} considered the following two special cases of
$\zeroSet$'s:
\begin{enumerate}
\item $\left|\bigcap_{i=1}^s\zeroSet_i\right|=k-s$ for all $s=1,\dots, k$. Note that when $s=k-1$, it is required that $\left|\bigcap_{i=1}^{k-1}\zeroSet_i\right|=1$, which means that there is at least one column of the generator matrix $\bG$ containing $k-1$ zeros. For this case an $[n,k]_q$ RS code generated by $\bG$ exists if $q\geq n$.
\item $\left|\zeroSet_i\right|\leq i-1$ for all $i=1,\dots,k$. Note that when $i\leq k-1$, $|\zeroSet_i|\leq k-2$, which implies that less zeros are allowed in $\bG$ than in \eqref{eq:zeroConstraints}. For this case, an $[n,k]_q$ RS code generated by $\bG$ exists if $q\geq n+1$.
\end{enumerate}

Driven by the requirement of balanced computation load during the encoding process in wireless sensor networks \cite{dau2013balanced} and multiple access networks \cite{dau2014simple,halbawi2014distributed,halbawi2014distributedGab}, codes fulfilling \emph{sparse and balanced} support constraints have been proposed, e.g., in \cite{halbawi2016balanced, halbawi2018sparse}.
In \cite{song2018generalized,chen2022sparse}, the existence of an MDS code with a sparse and balanced generator matrix $\bG$ has been studied. ``Sparse'' means that each row of $\bG$ has the maximum number of zeros, i.e., $k-1$ zeros, and ``balanced'' means that the number of zeros in any two columns differs by at most one, i.e., the weight of each column is either $\ceil{k(n-k+1)/n}$ or $\floor{k(n-k+1)/n}$.
It was shown in \cite{song2018generalized} that for any $k\in[n]$, if $q\geq n+\ceil{k(k-1)/n}$, then there exists an $[n,k]_q$ generalized RS code with a sparse and balanced generator matrix.
More recently, it is shown in \cite{chen2022sparse} that for any $k\geq 3$, $\frac{k}{n}\geq \frac{1}{2}$, and $q\geq n-1$, there exists an $[n,k]_q$ MDS code with a sparse and balanced generator matrix.
\section{Linearized Reed-Solomon Codes}
\label{sec:LRScodes}
LRS codes \cite{martinez2018skew} are a class of evaluation codes based on skew polynomials \cite{ore1933theory}, achieving the Singleton bound in the sum-rank metric, and therefore known as MSRD codes.
They are the first linear MSRD codes with sub-exponential field sizes (in contrast to Gabidulin codes, which are MRD codes but require exponential field sizes in the code length).
LRS codes have been applied in network coding \cite{martinez2019reliable}, locally repairable codes \cite{martinez2019universal} and code-based cryptography \cite{hormann2022security}.
The decoding of LRS codes has been extensively studied recently, e.g., \cite{boucher2020algorithm,puchinger2021bounds,bartz2021fast,puchinger2022generic,Hoermann2022efficient,Hoermann2022errorerasure,jerkovits2022universal,bartz2022fast}.

The definition of LRS codes adopted in this chapter follows from the \emph{generalized skew evaluations codes} \cite[Section III]{liu2015construction} with particular choices of the evaluation points and column multipliers.
\begin{definition}[Linearized Reed-Solomon (LRS) code]
\label{def:LRS}
For a prime power $q$ and integers $m,n$, let $\ell\in[q-1]$ and $(n_1,\dots, n_{\ell})$ be an ordered partition of $n$ with $n_\indBlock\leq m, \forall\indBlock\in[\ell]$.
Let $a_1,\dots, a_\ell\in\Fqm$ be from distinct $\Frobaut$-conjugacy classes of $\Fqm$, called \emph{block representatives}. 
Let
\begin{align*}
\colMulVec=(\beta_{1,1},\dots,\beta_{1,n_1}\ ,\ \dots\ , \ \beta_{\ell,1},\dots,\beta_{\ell,n_\ell})\in\Fqm^n
\end{align*}
be a vector of \emph{column multipliers}, where $\betal{1},\dots,\betal{n_\indBlock}$, called the columns multipliers of the $\indBlock$-th block, are linearly independent over $\Fq, \forall \indBlock\in[\ell]$.

Let the set of \emph{code locators} be
\begin{align}\label{eq:LRSlocators}
  \locSet 
  =&\{a_1\beta_{1,1}^{q-1}, \dots, a_1\beta_{1, n_1}^{q-1}
     \ ,\ \dots\ ,\
     a_\ell\beta_{\ell,1}^{q-1}, \dots, a_\ell\beta_{\ell, n_\ell}^{q-1}\}\ .
\end{align}
An $[n,k]_{q^m}$ \emph{linearized Reed-Solomon} code is defined as
\begin{align*}
  \cC_{\locSet,\colMulVec}^{\Frobaut}[n,k]
  \defeq \set*{\left. \colMulVec \star (f(\alpha))_{\alpha\in\locSet} \  \right| \ f(\SkewVar)\in\FrobSkewPolys,\ \deg f(\SkewVar)<k},
\end{align*}
where $\FrobSkewPolys$ is the skew polynomial ring with the Frobenius automorphism $\sigma:a\mapsto a^q$ of $\Fqm$,
the evaluation $f(\alpha) = \sum_{i=0}^{\deg f}f_iN_{i}(\alpha)$ is the remainder evaluation as in \cref{thm:skewEvaNform}, and $\star$ is the entry-wise multiplication of two vectors.
\end{definition}
The code locator set $\locSet$ of LRS codes has the following properties.
\begin{proposition}[{\cite[Theorem 4.5]{lam1988vandermonde}}]
  \label{prop:blockPind}
  Since $\betal{1},\dots,\betal{n_\indBlock}$ are linearly independent over $\Fq$, the set of code locators in the $\indBlock$-th block, denoted by $\locSet^{(\indBlock)}=\set*{a_\indBlock\beta_{\indBlock,1}^{q-1},\dots, a_\indBlock\beta_{\indBlock,n_\indBlock}^{q-1}}$, is P-independent (\cref{def:PindSet}).
\end{proposition}
\begin{proposition}[{\cite[Theorem 2.11]{FnTsurvey-Umberto}}]
  \label{lem:codeLocatorsPind}
  The union of P-independent sets which are subsets of different conjugacy classes is P-independent. Hence, the code locator set $\cL$ given in \eqref{eq:LRSlocators} is P-independent.
\end{proposition}

A generator matrix of the LRS code in \cref{def:LRS} is given by
\begin{align}\label{eq:Gevaluation}
  \GLRS =&
             \left( \quad \GLRSi{1} \quad  \dots    \quad \GLRSi{\ell}\quad \right)
           \in\Fqm^{k\times n}
\end{align}
where for each $\indBlock\in[\ell]$,
\begingroup
\setlength\arraycolsep{3pt}
\allowdisplaybreaks
\begin{align}
  &\GLRSi{\indBlock}
   = \bV^\Frobaut_{k}(\locSet^{(\indBlock)})  \cdot \diag(\colMulVec^{(\indBlock)}) \nonumber \\
  =& \begin{pmatrix}
    1 & \dots & 1 \\
    N_1(a_\indBlock\betal{1}^{q-1})& \dots & N_1(a_\indBlock\betal{n_\indBlock}^{q-1})\\
    \vdots & \ddots & \vdots \\
    N_{k-1}(a_\indBlock\betal{1}^{q-1})& \dots & N_{k-1}(a_\indBlock\betal{n_l}^{q-1})
  \end{pmatrix}\cdot
  \begin{pmatrix}
    \betal{1} &&\\
    &\ddots & \\
    && \betal{n_\indBlock}
  \end{pmatrix}\nonumber \\
     =&\begin{pmatrix}
       1 &&\\
       &\ddots & \\
       && N_{k-1}(a_\indBlock)
     \end{pmatrix}\cdot
           \begin{pmatrix}
             \betal{1} & \betal{2} & \dots & \betal{n_\indBlock}\\
             \betal{1}^{q} & \betal{2}^{q} & \dots & \betal{n_\indBlock}^{q}\\
               \vdots & \vdots & \ddots & \vdots \\
             \betal{1}^{q^{k-1}} & \betal{2}^{q^{k-1}} & \dots & \betal{n_\indBlock}^{q^{k-1}}
           \end{pmatrix}\nonumber
     \ ,
\end{align}
\endgroup
where $\locSet^{(\indBlock)}=\set*{a_\indBlock\beta_{\indBlock,1}^{q-1},\dots, a_\indBlock\beta_{\indBlock,n_\indBlock}^{q-1}}$ and $\colMulVec^{(\indBlock)} =\parenv*{\betal{1},\dots,\betal{n_\indBlock}}$. The last equality
holds because for $\Frobaut(a)=a^q$, we have that $N_i(\betalt^{q-1})\cdot \betalt=\left(\betalt^{q-1}\right)^{(q^i-1)/(q-1)}\cdot \betalt= \betalt^{q^i}$.

LRS codes $\cC=\myspan{\GLRS}\subseteq\Fqm^n$ are MSRD codes \cite[Theorem 2.20]{FnTsurvey-Umberto}
and the punctured codes $\cC_{\indBlock}=\myspan{\GLRSi{\indBlock}}\subseteq \Fqm^{n_{\indBlock}}$ at any block $\indBlock=1,\dots,\ell$ are MRD codes \cite[Section III.C]{liu2015construction}.

\begin{example}[{$[12,3]$} LRS code over $\F_{4^4}$with $3$ blocks]\label{ex:GLRS}
  Let $q=4,m=4, \ell=3, n_1=n_2=n_3=4, k=3$. Denote by $\priEle$ a primitive element of $\F_{4^4}$. We choose the block representatives to be $\ba=(1, \priEle, \priEle^2)$ and the basis of each block to be $\bb_1=(1, \priEle, \priEle^2, \priEle^3), \bb_2=(\priEle, \priEle^2, \priEle^3, \priEle^4), \bb_3 = (\priEle^2, \priEle^3, \priEle^4,\priEle^5)$. Then the code locators are
  \begin{align*}
    \locSet = \{\underbrace{1, \priEle^3, \priEle^6, \priEle^9}_{\text{first block}},
    \underbrace{\priEle^4, \priEle^7,\priEle^{10}, \priEle^{13}}_{\text{second block}},
    \underbrace{\priEle^{8}, \priEle^{11}, \priEle^{14}, \priEle^{17}}_{\text{third block}}\}\ .
  \end{align*}
  The generator matrix of this LRS code is
  \begin{align*}
    \GLRS=\left(
    \begin{array}{cccc;{2pt/2pt}cccc;{2pt/2pt}cccc}
      1 & \priEle& \priEle^2& \priEle^3& \priEle& \priEle^2& \priEle^3& \priEle^4& \priEle^2& \priEle^3& \priEle^4& \priEle^5\\
      1 & \priEle^{4}& \priEle^{8}& \priEle^{12}& \priEle^5& \priEle^9& \priEle^{13}& \priEle^{17}& \priEle^{10}& \priEle^{14}& \priEle^{18}& \priEle^{22}\\
      1 & \priEle^{16}& \priEle^{32}& \priEle^{48}& \priEle^{21}& \priEle^{37}& \priEle^{53}& \priEle^{69}& \priEle^{42}& \priEle^{59}& \priEle^{74}& \priEle^{90}
    \end{array}\right)\ .
  \end{align*}
\end{example}
In \cite[Definition~31]{martinez2018skew}, LRS codes are defined using linear operator evaluations with respect to the block representatives $\ba=(a_1,\dots,a_\ell)$ and block basis $\bb_{\indBlock} = (\betal{1},\dots,\betal{n_{\indBlock}})$. It was shown in \cite[Theorem 2.18]{FnTsurvey-Umberto} that these two definitions are equivalent.


\section{GM-MSRD Condition}
\label{sec:support-constrained-LRS}
Motivated by the practical interest in support-constrained codes and the theoretical research on MSRD codes (in particular, LRS codes), we investigate the existence of support-constrained MSRD codes in this section and prove the following result.
\begin{theorem}[GM-MSRD condition]
  \label{thm:fieldSizeFromGab}
  Let $\ell,n$ be positive integers and $(n_1, \dots, n_\ell)$ be an ordered partition of $n$.
  Given $\zeroSet_1,\dots,\zeroSet_k\subset [n]$,
  for any prime power $q\geq\ell+1$ and integer $m\geq \max_{l\in[\ell]}\{k-1+\log_qk, n_l\}$, there exists an $[n,k]_{q^m}$ linearized Reed-Solomon code with $\ell$ blocks, and each block of length $n_l$, $l\in[\ell]$ such that it has a generator matrix $\bG\in\Fqm^{k\times n}$ fulfilling the support constraints
  $G_{ij}=0, \ \forall i\in[k],\ \forall j \in \zeroSet_i$,
  if and only if, for any nonempty $\Omega\subseteq [k]$,
  \begin{align*}
    \left|\bigcap_{i\in\Omega}\zeroSet_i\right|+|\Omega|\leq k\ . \tag{\textrm{Recall }\eqref{eq:GM-MDScondition}}
  \end{align*}
\end{theorem}

For the necessity, since the sum-rank weight of any vector in $\Fqm^n$ is at most its Hamming weight by \cref{lem:sum-rank-Hamming}, an MSRD code is necessarily an MDS code.
Therefore, \eqref{eq:GM-MDScondition} is also a necessary condition for $\bG$ to generate an MSRD code.

Now we proceed to show the sufficiency of \eqref{eq:GM-MDScondition} for MSRD codes, in particular, via support-constrained LRS codes with sufficiently large alphabet size.
Note that for any $\Omega = \{i\}$, we have $|\zeroSet_i|\leq k-1$. One can add elements from $[n]$ to each $\zeroSet_i$ until $|\zeroSet_i|$ reaches $k-1$ while preserving \eqref{eq:GM-MDScondition} \cite[Corollary 3]{yildiz2019gabidulin}. This operation will only put more zero constraints on $\bG$ but not remove any. This means that the code we design under the new $\zeroSet_i$'s of size $k-1$ will also satisfy the original constraints. Therefore, without loss of generality, along with \eqref{eq:GM-MDScondition}, we can assume that
\begin{align}
  \label{eq:maxRowZeros}
    |\zeroSet_i|= k-1,\ \forall i\in[k]\ .
\end{align}

Let $\GLRS$ be a generator matrix of an LRS code as in \eqref{eq:Gevaluation}. Given the following matrix
\begin{align}
  \label{eq:defTmat}
  \bG = \bT\cdot \GLRS\ ,
\end{align}
if $\bT\in\Fqm^{k\times k}$ has full rank, then $\bG$ is another generator matrix of the LRS code.
Recall that $a_1,\dots, a_\ell\in\Fqm$ are the block representatives, $\beta_{1,1},\dots,\beta_{1,n_1}, \dots, \beta_{\ell,1},\dots,\beta_{\ell,n_\ell}\in\Fqm$ are the column multipliers, and $\locSet = \{\alphas{n}\}$ is the code locator set as defined in \eqref{eq:LRSlocators}.

Let $n_0=0$. Define the following bijective map between the indices,
\begin{equation}
  \label{eq:indMap}
  \begin{split}
  \indMap:  \bbN\times \bbN & \to \bbN\\
    l,t & \mapsto j=t+\sum_{r=0}^{l-1}n_r\ .
  \end{split}
\end{equation}
Then $\alpha_j = a_l\beta^{q-1}_{t}$ for $j=\indMap(l,t)$.
The inverse map $\indMap^{-1}: \bbN \to \bbN\times \bbN$ is
$j \mapsto (l, t)$,
where $l=\max\set*{\left. i\in[\ell] \ \right|\ \sum_{r=0}^i n_r\leq j}$ and $t=j-\sum_{r=0}^{l-1}n_r$.

For all $i\in[k]$, define the skew polynomials
\begin{align}\label{eq:skewPolyEachRow}
  \rowpoly_i(\SkewVar) \defeq \sum_{j=1}^{k} T_{i,j}\SkewVar^{j-1}\ \in\FrobSkewPolys\ ,
\end{align}
where $T_{i,j+1}$ is the entry at $i$-th ($i\in[k]$) row, $j$-th ($j\in[k]$) column in $\bT$.
The entries of $\bG$ are 
$G_{ij}=\rowpoly_i(a_l\betalt^{q-1})\betalt,i\in[k],j=\indMap(l,t)\in[n]$. Then, the zero constraints in \eqref{eq:zeroConstraints} become the following root constraints on $\rowpoly_i$'s:
\begin{align}\label{eq:rootConstraints}
   \rowpoly_i(a_l\betalt^{q-1})=0, \quad\forall i \in[k],\ \forall j=\indMap(l,t)\in \zeroSet_i\ .
\end{align}

For brevity, denote by
\begin{align}
  \label{eq:rootSet}
  \rootSet_i \defeq\{\alpha_j=a_l\betalt^{q-1}\ |\ j=\indMap(l,t)\in \zeroSet_i\}
\end{align}
the corresponding set of code locators to the zero set $\zeroSet_i$.
Since $\locSet$ is P-independent, any subset $\rootSet_i\subset\locSet$ is also P-independent \cite[Theorem 23]{lam1986general}.
Then the minimal polynomial $\rowpoly_{\rootSet_i}(\SkewVar)$ of $\rootSet_i$ is of degree $|\rootSet_i|=k-1$. By the properties of $f_i(\SkewVar)$ in \eqref{eq:skewPolyEachRow} and \eqref{eq:rootConstraints}, it can be seen that $\rowpoly_{i}(\SkewVar)=\rowpoly_{\rootSet_i}(\SkewVar)$ by setting $T_{i, k}=1$ (as minimal polynomials are monic polynomials).
By the computation of the minimal polynomial in \eqref{eq:lclm}, the skew polynomials $\rowpoly_i(\SkewVar)$ fulfilling \eqref{eq:rootConstraints} can be written as
\begin{align}\label{eq:lclmMinPoly}
  \rowpoly_i(\SkewVar) = \rowpoly_{\rootSet_i}(\SkewVar) = \lclm_{\alpha\in\rootSet_i} \set*{\SkewVar-\alpha} \ .
\end{align}
Since all $\rootSet_i\subset \cL, i\in[k]$ are P-independent, it follows from \cref{lem:noMoreZero} that $\rowpoly_i(\alpha)\neq 0$, for all $\alpha\in\cL\setminus\rootSet_i$. Hence, there is no other zero entry in $\bG$ than the required positions in $\zeroSet_i$'s.
Moreover, with $k-1$ P-independent roots of $\rowpoly_i(\SkewVar)$ and setting $T_{i, k}=1$,
the coefficients $T_{i,j}$ of $\rowpoly_i(\SkewVar)$ in \eqref{eq:skewPolyEachRow} are uniquely determined in terms of $a_1\beta_{1,1}^{q-1}, \dots,a_\ell\beta_{\ell,n_\ell}^{q-1}$.

In the following, we assume that $a_1,\dots,a_\ell$ are fixed, nonzero, and from distinct $\Frobaut$-conjugacy classes. We see $\betalt$'s as variables
of the following commutative multivariate polynomial ring
\begin{align}
\label{eq:multiVarRing}
    R_{\numMulPolyVar}\defeq &\Fqm[\beta_{1,1},\dots, \beta_{\ell,n_\ell}]\ ,
\end{align}
and the coefficients $T_{i,j}$ of $\rowpoly_i(\SkewVar)$ are polynomials in $R_{\numMulPolyVar}$.
Then, the problem of finding $\betalt$'s such that $\bG$ generates the same LRS code as $\GLRS$ becomes finding $\betalt$'s
such that
\begin{align}
  \label{eq:defPl}
  P(\beta_{1,1},\dots, \beta_{\ell,n_\ell})\defeq P_{\bT}(\beta_{1,1},\dots, \beta_{\ell,n_\ell})\cdot \prod_{l=1}^{\ell} P_{\bM_l}(\beta_{l,1},\dots, \beta_{l,n_l}) \neq 0\ ,
\end{align}
where $P_{\bT}$ is the determinant of $\bT$ in \eqref{eq:defTmat}, whose entries are the coefficients of $f_i$'s and therefore determined by the roots of $f_i$'s, and
\begin{align*}
  P_{\bM_l}
  \defeq & \det
          \begin{pmatrix}
             \betal{1} & \betal{2} & \dots & \betal{n_l}\\
             \betal{1}^{q} & \betal{2}^{q} & \dots & \betal{n_l}^{q}\\
              \vdots & \vdots & \ddots & \vdots \\
             \betal{1}^{q^{n_l-1}} & \betal{2}^{q^{n_l-1}} & \dots & \betal{n_l}^{q^{n_l-1}}
          \end{pmatrix}\ . 
\end{align*}
Since the coefficient of the monomial $\prod_{i=1}^{n_l} \beta_{l,i}^{q^{i-1}}$ in  $P_{\bM_l}$ is $1$, $P_{\bM_l}$ is a nonzero polynomial in $R_{\numMulPolyVar}$.

\begin{claim}\label{claim}
If the condition in \eqref{eq:GM-MDScondition} is satisfied, then $P_{\bT}$ is a nonzero polynomial in $R_\numMulPolyVar$.
\end{claim}
With \cref{claim}, we can conclude that $P(\beta_{1,1},\dots, \beta_{\ell,n_\ell})$ is a nonzero polynomial in $R_{\numMulPolyVar}$.
We now proceed to prove \cref{thm:fieldSizeFromGab} assuming that \cref{claim} is true.
A more general statement (\cref{thm:sufficientCond}) of the claim is given and proven in \cref{sec:ind_proof}.

For a fixed $l\in[\ell],t\in[n_\ell]$, the $\betalt$-degree of $P_{\bM_l}$ is $\deg_{\betalt}P_{\bM_l} = q^{n_l-1}$ \cite[Lemma 3.51]{lidl1997finite}. Moreover, $\deg_{\betalt}P_{\bT} \leq (k-1)(q-1)\cdot q^{k-2}$, which is proven in \cref{apendix:fieldSizeSkewPoly}.
Then, the $\betalt$-degree of $P(\beta_{1,1},\dots, \beta_{\ell,n_\ell})$ in \eqref{eq:defPl} is
\begin{align*}
  \deg_{\betalt}P &\leq (k-1)(q-1)\cdot q^{k-2}+q^{n_l-1}\ .
\end{align*}
\begin{proof}[Proof of \cref{thm:fieldSizeFromGab}]
  \cref{claim} implies that $P(\beta_{1,1},\dots, \beta_{\ell,n_\ell})$ is a nonzero polynomial.
  By the Combinatorial Nullstellensatz \cite[Theorem 1.2]{alon1999combinatorial}(see also in \cref{lem:nullstellensatz}),
  there exist $\evapt{\beta}_{1,1}, \dots, \evapt{\beta}_{\ell, n_\ell}$ in $\Fqm$ such that
  \begin{align*}
    P(\evapt{\beta}_{1,1}, \dots, \evapt{\beta}_{\ell, n_\ell})\neq 0\ ,
  \end{align*}
  if
  \begin{align}
    q^m &>\max_{l\in[\ell], t\in[n_l]}\{\deg_{\betalt}P\}  \nonumber\\
        &=\max_{l\in[\ell]}\{ (k-1)(q-1)\cdot q^{k-2}+q^{n_l-1}\}\ .\label{eq:condition_on_qm_2}
  \end{align}
  If $m\geq \max_{l\in[\ell]}\{k-1+\log_qk\ ,\ n_l\}$, we have
  \begin{align*}
    q^m =& (q-1)q^{m-1}+q^{m-1}\\
    \geq& \max_{l\in[\ell]}\{k(q-1)\cdot q^{k-2} + q^{n_l-1}\} >\eqref{eq:condition_on_qm_2}\ .
  \end{align*}
  To have $a_1,\dots,a_\ell$ from different non-trivial $\Frobaut$-conjugacy class of $\Fqm$, by \cref{thm:conjugacyClasses}, we require $q-1\geq \ell$.
\end{proof}
\begin{remark}
Consider the extreme cases:
\begin{enumerate}
\item
For $\ell=1$, the sum-rank metric is the rank metric and LRS codes are Gabidulin codes.
In this case, the field size in \cref{thm:fieldSizeFromGab} coincides with \cite[Theorem 1]{yildiz2019gabidulin}.
\item
For $\ell=n$ and $n_l=1,\forall l\in[\ell]$, the sum-rank metric is the Hamming metric. In addition, with $\Frobaut=\Id$, LRS codes are GRS codes with distinct nonzero $a_1,\dots,a_\ell$ as code locators and nonzero $\betalt$'s as column multipliers (see \cite[Theorem 2.17]{FnTsurvey-Umberto}, \cite[Table II]{martinez2019reliable}).
In this case, by adapting the setup in \eqref{eq:skewPolyEachRow}
to $\sigma=\Id$, and the proof in \cref{apendix:fieldSizeSkewPoly} with the usual evaluation of commutative polynomials, one can obtain the same results as in \cite[Theorem~2]{yildiz2018optimum}.
\end{enumerate}
\end{remark}
If the necessary and sufficient condition on $\zeroSet_1,\dots,\zeroSet_k$ in \eqref{eq:GM-MDScondition} is not satisfied, then we cannot obtain an MSRD code fulfilling the support constraints.
The following result derives the largest possible sum-rank distance that can be achieved with the given constraints.
In fact, the largest sum-rank distance can be achieved by subcodes of LRS codes.
This result is an analogue to those for MDS codes \cite{yildiz2018optimum} and MSRD codes \cite{yildiz2019gabidulin}.
The following upper bound on the minimum Hamming distance of a support-constrained code $\cC$ is given in \cite[Theorem 1]{yildiz2018optimum},
\begin{align*}
  \dH(\cC) \leq n-\widetilde{k}+1
\end{align*}
where
\begin{align}
  \label{eq:kSuperCode}
  \widetilde{k} \defeq \max_{\varnothing \neq \Omega\subseteq [k]}  \left|\bigcap_{i\in\Omega} \zeroSet_i\right|+|\Omega|\ .
\end{align}
Note that $\widetilde{k}>k$ if $\zeroSet_1,\dots,\zeroSet_k$ do not satisfy the condition in \eqref{eq:GM-MDScondition}.
For any ordered partition $\bn_\ell =(n_1, \dots,n_\ell)$ of $n$, according to \cref{lem:sum-rank-Hamming}, we have
\begin{align}\label{eq:distanceSubcode}
  \dSR{\bn_{\ell}}(\cC)\leq \dH(\cC) \leq n-\widetilde{k}+1\ .
\end{align}

\begin{theorem}\label{thm:subcode}
  Given $\zeroSet_1,\dots,\zeroSet_k\subseteq[n]$, let $\widetilde{k}$ be as in \eqref{eq:kSuperCode}. For any prime power $q\geq\ell+1$ and integer $m\geq \max_{l\in[\ell]}\{\widetilde{k}-1+\log_q\widetilde{k}, n_l\}$, there exists a subcode of an $[n,\widetilde{k}]_{q^m}$ linearized Reed-Solomon code with $\ell$ blocks, and each block of length $n_l$, $l\in[\ell]$ such that it has a generator matrix $\bG\in\Fqm^{k\times n}$ fulfilling the support constraints
  $G_{ij}=0, \ \forall i\in[k],\ \forall j \in \zeroSet_i$.
\end{theorem}
\begin{proof}
  Let $\zeroSet_{k+1}=\cdots = \zeroSet_{\widetilde{k}} = \varnothing$. For any nonempty $\Omega\subseteq [\widetilde{k}]$, we have
  \begin{align*}
    \left|\bigcap_{i\in\Omega}\zeroSet_i\right|+|\Omega|\leq \widetilde{k}\ .
  \end{align*}
  Then, by \cref{thm:fieldSizeFromGab}, there exists an LRS code of dimension $\widetilde{k}$ with a generator matrix $\widetilde{\bG}\in\Fqm^{\widetilde{k}\times n}$ having zeros at the positions specified by $\zeroSet_1, \dots,\zeroSet_{\widetilde{k}}$. Since it is an MSRD code, its sum-rank distance is $n-\widetilde{k}+1$.
  The first $k$ rows of $\widetilde{\bG}$ will generate a subcode $\cC$ whose sum-rank distance is as good as the LRS code, i.e., $\dSR{\bn_{\ell}}(\cC)\geq n-\widetilde{k}+1$, where $\bn_{\ell}=(n_1,\dots, n_\ell)$. Hence, the subcode achieves the largest possible distance given in \eqref{eq:distanceSubcode}.
\end{proof}

\subsection{A More General Result of Claim 1}
\label{sec:ind_proof}

Let $R_\numMulPolyVar$
be the commutative multivariate polynomial ring defined in \eqref{eq:multiVarRing}. Note that $R_0=\Fqm$.
Let $\Frobaut$ be the Frobenius automorphism of $R_0$, which we extend to any $a=\sum_{\bi\in\bbN^n} a_{\bi}\cdot \beta_{1,1}^{i_1}\cdots \beta_{\ell, n_\ell}^{i_\numMulPolyVar}\in R_\numMulPolyVar$ by
\begin{align*}
    \Frobaut\ :\ R_\numMulPolyVar &\to R_\numMulPolyVar\\
  \sum_{\bi\in\bbN^n} a_{\bi}\cdot \beta_{1,1}^{i_1}\cdots \beta_{\ell, n_\ell}^{i_\numMulPolyVar} &\mapsto \sum_{\bi\in\bbN^n}
                                                                                                       a_{\bi}^q\cdot (\beta_{1,1}^{i_1})^q\cdots (\beta_{\ell, n_\ell}^{i_\numMulPolyVar})^q\ .
\end{align*}
Let $R_\numMulPolyVar[\SkewVar;\Frobaut]$ be the univariate skew polynomial ring with indeterminate $X$, whose coefficients are from $R_\numMulPolyVar$, i.e.,
\begin{align*}
  R_\numMulPolyVar[\SkewVar;\Frobaut] \coloneqq \left\{ \left. \sum_{i\in\bbN} c_iX^i\ \right|\  c_i \in R_\numMulPolyVar \right\}\ .
\end{align*}
The degree of $f=\sum_{i\in\bbN} c_iX^i\in\FrobPolysn$ is $\deg f\defeq\max\set*{ i\in\bbN\ |\ c_i\neq 0}$ and
$\deg 0 \defeq -\infty$ by convention.

Similar to skew polynomials over a finite field, addition is commutative and multiplication is defined using the commutation rule
\begin{align}\label{eq:commutationRule}
  \SkewVar\cdot a = \Frobaut(a)\cdot \SkewVar, \ \forall a\in R_\numMulPolyVar\ ,
\end{align}
and naturally extended by distributivity and associativity. As in \eqref{eq:prod-skew-polys}, the product of $f,g\in\FrobPolysn$ with $\deg f=d_f$ and $\deg g=d_g$ is
\begin{align}
f\cdot g=\sum_{i=0}^{d_f}\sum_{j=0}^{d_g} f_i\Frobaut^{i}(g_j) \SkewVar^{i+j}\ ,\label{eq:prodRn}
\end{align}
and the degree of the product is $\deg \left(f\cdot g\right) = d_f + d_g$. Note that
in general, $f\cdot g\neq g\cdot f$. 

With a bit abuse of the notation, in the following, we also denote by
\begin{align*}
\locSet
  =&\{a_1\beta_{1,1}^{q-1}, \dots, a_1\beta_{1, n_1}^{q-1}\ ,\ \dots \ , \ a_\ell\beta_{\ell,1}^{q-1}, \dots, a_\ell\beta_{\ell, n_\ell}^{q-1}\} \subseteq \mulVarRng
  \end{align*}
  the P-independent set as a subset of $\mulVarRng$.
  Let $\rootSet_i\subseteq \locSet$ be the set as in \eqref{eq:rootSet} corresponding to $\zeroSet_i$ and $f_{\rootSet_i}\in\FrobPolysn$ be the minimal polynomial of $\rootSet_i$ as in \eqref{eq:lclmMinPoly}.

We note the following properties of $\FrobPolysn$, which will be useful for the proof of the more general result of \cref{claim} in \cref{thm:sufficientCond}. The proofs of these properties can be found in \cref{appendix:SkewPolyProperty}.
\begin{enumerate}[label={\bfseries P\arabic*}]
\item \label{p:noZeroDiv}
  $\FrobPolysn$ is a ring without zero divisors.
\item \label{p:gcrd}
  For any sets 
  $\rootSet_1, \rootSet_2\subseteq \mulVarRng$ s.t.~$\rootSet_1\cup\rootSet_2$ is P-independent,
  $\gcrd(f_{\rootSet_1},f_{\rootSet_2})=f_{\rootSet_1\cap\rootSet_2}$.
  In particular, $\rootSet_1\cap\rootSet_2=\varnothing \iff \gcrd(f_{\rootSet_1},f_{\rootSet_2})=1$.
\item \label{p:xlrdivision}
  For $t\in\bbN$ and any $f\in\FrobPolysn$, $\SkewVar^t \divides_l f \iff \SkewVar^t\divides_r f$.
  In this case, we write $\SkewVar^t| f$.
\item \label{p:xdivision}
  For $t\in\bbN$ and any $f_1, f_2 \in\FrobPolysn$ such that $\SkewVar\not\divides f_2$, then $\SkewVar^t\divides (f_1\cdot f_2)\iff\SkewVar^t\divides f_1$.
\end{enumerate}
In the general result in \cref{thm:sufficientCond}, we are interested in skew polynomials in the following form: for any $\zeroSet\subseteq[n], \tfzt\geq 0$
\begin{align}\label{eq:fZt}
  \fZt(\zeroSet,\tfzt)\coloneqq \SkewVar^{\tfzt}\cdot \underset{\substack{\alpha\in\set*{\left. a_l\betalt^{q-1} \ \right|\ \indMap(l,t)\in\zeroSet}}}{\lclm}\set*{\SkewVar-\alpha} \in \FrobPolysn\ ,
\end{align}
where $\indMap(l,t)$ is defined in \eqref{eq:indMap}.

Define the set of skew polynomials of the following form:
\begin{equation}
  \label{eq:Skewnk}
  \Skewnk \coloneqq \set*{\left. \fZt(\zeroSet,\tfzt)\ \right|\ \tfzt\geq 0, \zeroSet\subseteq [n]~\text{s.t.}~|\zeroSet|+\tfzt\leq k-1}\subseteq \FrobPolysn\ .
\end{equation}
Note that $\deg f \leq k-1, \forall f\in\Skewnk$.
We also note the following properties of polynomials in $\Skewnk$, whose proofs are given in \cref{appendix:SkewPolyProperty}.
\begin{enumerate}[resume,label={\bfseries P\arabic*}]
\item \label{p:gcrd2polys}
  For any $f_1 = \fZt(\zeroSet_1,\tfzt_1), f_2 = \fZt(\zeroSet_2, \tfzt_2)\in\Skewnk$, we have
  \begin{align*}
    \gcrd(f_1,f_2)=\fZt(\zeroSet_1\cap \zeroSet_2, \min\{\tfzt_1,\tfzt_2\})\in\Skewnk\ .
  \end{align*}
\item \label{p:reducevar}
  Let $f=\fZt(\zeroSet,\tfzt)\in\Skewnk$ and let $f'=f|_{\beta_{\ell, n_\ell}=0}\in R_{\numMulPolyVar-1}[\SkewVar;\Frobaut]$ (namely, we substitute $\beta_{\ell, n_\ell}=0$ in each coefficient of $f$). Then $f'\in\cS_{n-1,k}$ and
\begin{align*}
    f'= \begin{cases}
    \fZt(\zeroSet,\tfzt) & n\not\in \zeroSet\ ,\\
    \fZt(\zeroSet\setminus\{n\}, \tfzt+1) & n\in \zeroSet\ .
    \end{cases}
\end{align*}
\end{enumerate}

The following theorem is a more general statement of \cref{claim} and it is the analogue of \cite[Theorem 3.A]{yildiz2019gabidulin} for skew polynomials.
\begin{theorem}\label{thm:sufficientCond}
  Let $k\geq \numRows\geq 1$ and $n\geq 0$. For any $f_1, f_2,\dots, f_\numRows\in \Skewnk$, the following are equivalent:
  \begin{enumerate}[label={\textrm(\roman*)}]
  \item For any $g_1,g_2,\dots, g_\numRows\in\FrobPolysn$ such that $\deg (g_i\cdot f_i)\leq k-1$, we have \label{item:equiv1}
    \begin{align*}
      \sum_{i=1}^\numRows g_i\cdot f_i = 0\quad \implies \quad g_1=g_2=\dots = g_\numRows =0\ .
    \end{align*}
  \item For all nonempty $\Omega\subseteq [\numRows]$, we have \label{item:equiv2}
    \begin{align}\label{eq:equiv2}
      k-\deg (\gcrd_{i\in\Omega} f_i) \geq \sum_{i\in\Omega} (k-\deg f_i)\ .
    \end{align}
  \end{enumerate}
\end{theorem}

Before proving \cref{thm:sufficientCond}, we first show in \cref{cor:claim} that \cref{claim} is a special case of \cref{thm:sufficientCond}.
For this purpose, we give an equivalence of \cref{thm:sufficientCond} in terms of matrices with entries from $R_\numMulPolyVar$.

We first describe the multiplication between skew polynomials in matrix language.
Let $\up=\sum_{i\in\bbN} \upc_{i} \SkewVar^{i}\in\FrobPolysn$. For $b-a\geq \deg \up$, define the following matrix
\begin{align*}
  \bS_{a\times b}(\up)\coloneq
  \begin{pmatrix}
    \upc_0 & \cdots & \upc_{b-a} &\vphantom{v} \\
    & \Frobaut(\upc_0) & \cdots & \Frobaut(\upc_{b-a})\\
    &\hphantom{MM} \ddots & \hphantom{MM} \ddots &\hphantom{MM} \ddots\\
    &&\Frobaut^{a-1}(\upc_0) &  \cdots  &\Frobaut^{a-1}(\upc_{b-a})
 \end{pmatrix}\in R_n^{a\times b}\ .
\end{align*}
  In particular, for $a=1$, denote by $\FrobPolysn_{<b}$ the set of skew polynomials of degree strictly less than $b$. The map
  \begin{equation}
    \label{eq:polyMatMap}
    \begin{split}
      \bS_{1\times b}(\cdot)\ : \ \FrobPolysn_{<b} &\to \mulVarRng^b\\
      u &\mapsto (u_0, \dots, u_{b-1})
    \end{split}
  \end{equation}
  is bijective and $\bS_{1\times b}(0)=\0, \forall b\in\bbN$.
For any skew polynomial $\vp=\sum_i\vpc_i\SkewVar^{i}\in\FrobPolysn$, we have
\begin{align}\label{eq:matOfProduct}
  \bS_{a\times b}(\vp\cdot \up) = \bS_{a\times c}(\vp)\cdot \bS_{c\times b}(\up)\ ,
\end{align}
where $a,b,c\in\bbN$ are such that $c-a\geq \deg  \vp, b-c\geq \deg  \up$.
As a special case, when $\vp = \SkewVar^\tfzt, \tfzt\in\bbN$, we can write
\begin{align}
  \bS_{a\times(b+\tfzt)}(\SkewVar^\tfzt\cdot \up)=& \bS_{a\times (a+\tfzt)}(\SkewVar^\tfzt)\cdot \bS_{(a+\tfzt)\times (b+\tfzt)}(\up)\nonumber\\
  =&
      \left( \0_{a\times \tfzt}\quad \bI_{a \times a}\right)
                           \cdot \bS_{(a+\tfzt)\times (b+\tfzt)}(\up)\ .  \label{eq:matfZt} 
\end{align}
By the definition in \eqref{eq:fZt}, we have $\fZt(\zeroSet,\tfzt)=\SkewVar^\tfzt\cdot \up$ for some $\up\in\FrobPolysn$. It can be readily seen from \eqref{eq:matfZt} that the first $\tfzt$ columns of $\bS_{a\times (b+\tfzt)}(\fZt(\zeroSet,\tfzt))$ are all zero.

For $s\in[k]$, $i\in [\numRows]$, let $f_i=\fZt(\zeroSet_i, \tfzt_i)\in\Skewnk$.
We write $\bS(f_i)$ instead of $\bS_{(k-\tfzt_i-|\zeroSet_i|)\times k}(f_i)$ for ease of notation. By \eqref{eq:matfZt}, $\bS(f_i)$ looks like
\begingroup
\setlength\arraycolsep{3pt}
\begin{align*}
\bS(f_i)
	 &=\begin{pmatrix}
        0 & \cdots & 0 & \times & \times & \cdots & \times\\
        0 & \cdots & 0 &  & \times & \times & \cdots & \times\\
        \vdots && \vdots &&&\ddots & \ddots & &\ddots\\
        \makebox[0pt][l]{$\smash{\underbrace{\phantom{\begin{matrix}0&\cdots&0\end{matrix}}}_{\tfzt_i}}$}0 & \cdots & 0&
        \makebox[0pt][l]{$\smash{\hspace{-5pt}\underbrace{\phantom{\begin{matrix}\times&\times&\cdots\end{matrix}}}_{k-1-\tfzt_i-|\zeroSet_i|}}$}
        &&&
        \makebox[0pt][l]{$\smash{\underbrace{\phantom{\begin{matrix}\times&\times&\cdots&\times\end{matrix}}}_{|\zeroSet_i|+1}}$}
        \times & \times & \cdots & \times\\
    \end{pmatrix}\!\!\!\!\left.\vphantom{\begin{pmatrix}0\\0\\\vdots\\0\end{pmatrix}}
				\right\} {\scriptstyle k-\tfzt_i-|\zeroSet_i|}\ ,\\
\end{align*}
\endgroup
where the $\times$'s represent possibly nonzero entries.
Then, applying \eqref{eq:matOfProduct} to the expression $g_i\cdot f_i$ in \cref{thm:sufficientCond} yields
\begin{align}
  \bS_{1\times k}(g_i\cdot f_i) = \bu_i \cdot \bS(f_i)\ ,\nonumber
\end{align} 
where $\bu_i=\bS_{1\times (k-\tfzt_i-|\zeroSet_i|)}(g_i)$ is a row vector.
Therefore, we can write
\begin{align}
  \bS_{1\times k}(\sum_{i=1}^\numRows g_i\cdot f_i) =
    (\bu_1, \cdots , \bu_\numRows)
                     \cdot
                     \underbrace{\begin{pmatrix}
                       \bS(f_1)\\ \vdots \\ \bS(f_\numRows)
                     \end{pmatrix}}_{=:\bM(f_1,\dots, f_\numRows)}\ ,\label{eq:defMmat}
\end{align}
which is a linear combination of the rows of $\bM(f_1,\dots,f\numRows)$.

The following theorem is equivalent to \cref{thm:sufficientCond} in matrix language and is analogous to \cite[Theorem 3.B]{yildiz2019gabidulin}.
\begin{theorem}\label{thm:sufficientMat}
  Let $k\geq s\geq 1$ and $n\geq 0$. For $i\in [\numRows]$, let $\zeroSet_i\in[n], \tfzt_i\geq 0$ such that $\tfzt_i+|\zeroSet_i|\leq k-1$ and $f_i=\fZt(\zeroSet_i, \tfzt_i)\in\Skewnk$. The matrix $\bM(f_1,\dots, f_\numRows)$ defined in \eqref{eq:defMmat} has full row rank if and only if, for all nonempty $\Omega\subseteq[\numRows]$,
  \begin{align}
    k-\left|\bigcap_{i\in\Omega}\zeroSet_i\right| - \min_{i\in\Omega} \tfzt_i \geq \sum_{i\in\Omega}(k-\tfzt_i-|\zeroSet_i|)\ .
    \label{eq:matEquiv2}
  \end{align}
\end{theorem}
\begin{proof}
    For brevity, we write $\bM$ instead of $\bM(f_1,\dots, f_\numRows)$. 
    The logic of the proof is as follows
    \begin{align*}
      \bM \textrm{ has full row rank }\overset{\matProne}{\iff} \ref{item:equiv1} \overset{\textrm{\cref{thm:sufficientCond}}}{\iff} \ref{item:equiv2} \overset{\matPrtwo}{\iff} \eqref{eq:matEquiv2}\textrm{ holds}
    \end{align*}
    where \ref{item:equiv1} and \ref{item:equiv2} are shown to be equivalent in \cref{thm:sufficientCond}.
    We only need to show the equivalence \matProne{} and \matPrtwo.

    \matProne: Assuming $\bM$ has full row rank, it is equivalent to writing
  \begin{equation}
    \begin{split}
      \forall &\bu
      \in\mulVarRng^{1\times \sum_{i=1}^s(k-\tfzt_i-|\zeroSet_i|)},\
      \bu\cdot \bM=\0 \Longrightarrow \bu=\0\ .
  \end{split}
    \label{eq:matEquiv1}
  \end{equation}
  Partition $\bu$ into $s$ blocks $(\bu_1,\dots, \bu_\numRows)$, where $\bu_i\in\mulVarRng^{1\times (k-\tfzt_i-|\zeroSet_i|)}$. Note that $\bu=\0\iff \forall i\in[s], \bu_i =\0$.
  For each $i\in[s]$, the set $\{g_i\ |\ \bS_{1\times (k-\tfzt_i-|\zeroSet_i|)}(g_i)=\bu_i,\ \forall \bu_i\in\mulVarRng^{1\times (k-\tfzt_i-|\zeroSet_i|)}\}$ is $\FrobPolysn_{<(k-\tfzt_i-|\zeroSet_i|)}$, which is the set of skew polynomials of degree less than $ (k-\tfzt_i-|\zeroSet_i|)$, since the map $\bS_{1\times *}(\cdot)$ defined in \eqref{eq:polyMatMap} is bijective.
  Therefore, $\bu_i=\0\iff g_i=0,\ \forall i\in[\numRows]$. It can be further inferred that every $\bu\in \mulVarRng^{1\times \sum_{i=1}^{\numRows}(k-\tfzt_i-|\zeroSet_i|)}$ corresponds to a unique tuple $(g_1, \dots, g_\numRows)\in \FrobPolysn_{<(k-\tfzt_1-|\zeroSet_1|)}\times \dots\times\FrobPolysn_{< (k-\tfzt_\numRows-|\zeroSet_\numRows|)}$. We denote the Cartesian product by $\cG$.
  Since $\deg f_i= \tfzt_i+|\zeroSet_i|,\ \forall i\in[\numRows]$, for any tuple $(g_1,\dots, g_\numRows)\in\cG$, $\deg (g_i\cdot f_i) \leq k-1,\ \forall i\in[\numRows]$.

  By the equality in \eqref{eq:defMmat}, $\bu\cdot \bM=\bS_{1\times k}(\sum_{i=1}^\numRows g_i\cdot f_i)$ and $\bS_{1\times k}(\sum_{i=1}^\numRows g_i\cdot f_i)=\0 \iff \sum_{i=1}^\numRows g_i\cdot f_i=0$.
  Hence \eqref{eq:matEquiv1} can be equivalently written as
  \begin{align*}
    \forall g_1,\dots,g_\numRows\in\FrobPolysn &\text{ such that }\deg (g_i\cdot f_i)\leq k-1\ ,\\
    &  \sum_{i=1}^\numRows g_i\cdot f_i=0 \implies g_i=0,\ \forall i\in[\numRows]\ ,
  \end{align*}
  which is exactly the statement \ref{item:equiv1}.

\matPrtwo:
  It follows from \ref{p:gcrd2polys} that for any nonempty set $\Omega\subseteq [\numRows]$,
  \begin{align*}
    \deg (\gcrd_{i\in\Omega} f_i) = \left|\bigcap_{i\in\Omega}\zeroSet_i\right| + \min_{i\in\Omega} \tfzt_i\ .
  \end{align*}
  Then the left hand side of \eqref{eq:equiv2},
  $k-\deg (\gcrd_{i\in\Omega} f_i)$, is equal to $k- |\bigcap_{i\in\Omega}\zeroSet_i| - \min_{i\in\Omega} \tfzt_i$, which is the left hand side of \eqref{eq:matEquiv2}. By the definition of $f_i=\fZt(\zeroSet_i,\tfzt_i)$ in \eqref{eq:fZt}, the right hand side of \eqref{eq:equiv2},
  $\sum_{i\in\Omega}(k-\deg f_i)$, is equal to $\sum_{i\in\Omega}(k-(|\zeroSet_i|+\tfzt_i))$, which is the right hand side of \eqref{eq:matEquiv2}.
\end{proof}
As a special case, when $s=k$, $\tfzt_i=0$ and $|\zeroSet_i| = k-1,\ \forall i\in[k]$, each block $\bS(f_i)$ becomes a row vector with entries being the coefficients of $f_i=\fZt(\zeroSet_i,0)=\sum_{j=1}^{k} f_{i,j}\SkewVar^{j-1}\in\FrobPolysn$ and
\begin{align}
  \bM(f_1,\dots,f_k) =
  \begin{pmatrix}
    f_{11} & f_{12} & \cdots & f_{1k}\\
    f_{21} & f_{22} & \cdots & f_{2k}\\
    \vdots & \vdots & \ddots & \vdots\\
    f_{k1} & f_{k2} & \cdots & f_{kk}
  \end{pmatrix}\in\mulVarRng^{k\times k}\ .\label{eq:defMk}
\end{align}
Note that $\bM(f_1, \dots,f_k)$ coincides with the matrix $\bT$ in \eqref{eq:defTmat}.
Hence, we have \cref{cor:claim} below, which is exactly \cref{claim}.

\begin{corollary}\label{cor:claim}
For $i\in[k]$, let $\zeroSet_i\subseteq [n]$ with $|\zeroSet_i|=k-1$. Then, $\det \bM(f_1,\dots,f_k)$ is a nonzero polynomial in $\mulVarRng$, if and only if, for all nonempty $\Omega\subseteq[k]$,
  $k-\left|\bigcap_{i\in\Omega} \zeroSet_i \right| \geq |\Omega|. $
\end{corollary}
\subsubsection{Proof of \cref{thm:sufficientCond}}
\label{sec:proofMainResult}

Denote $f_{\Omega}\defeq \gcrd_{i\in\Omega} f_i$. By \ref{p:gcrd}, $f_{\Omega}$ is equal to the minimal polynomial of the set $Z_\Omega \defeq \bigcap_{i\in\Omega}Z_i$.

We first show the direction \ref{item:equiv1}$\implies$\ref{item:equiv2}. Suppose \ref{item:equiv2} does not hold and w.l.o.g., assume that for $\Omega=\{1,2,\dots,\subRows\}\subseteq[k]$, $k-\deg f_\Omega<\sum_{i\in\Omega} (k-\deg f_i)$. For $i\in \Omega$, let $f_i=q_i\cdot f_\Omega$ for some $q_i\in \FrobPolysn$. Then, for $g_1, \dots, g_\subRows\in\FrobPolysn$ such that $\deg(g_i\cdot f_i)\leq k-1$, the equation $\sum_{i\in \Omega} g_i\cdot q_i = 0$ gives a homogeneous linear system of equations in the unknowns which are the coefficients of the $g_i$'s. Since the $g_i$'s are such that $\deg(g_i\cdot f_i)\leq k-1$, the number of unknowns is at least $\sum_{i\in\Omega}(k-\deg f_i)$. The number of equations is at most $k-\deg f_\Omega$, which is smaller than the number of unknowns by the assumption. Therefore, one can find $g_1,\dots, g_\subRows$, not all zero, solving the linear system of equations, which contradicts \ref{item:equiv1}.

  We then show the direction \ref{item:equiv2}$\implies$\ref{item:equiv1} by induction.
  We do induction on the parameters $(k,\numRows,n)$ considered in the lexicographical order $\prec$ (page~\pageref{item:lexicographic-order}).

  For the induction basis, when $(k\geq \numRows=1,n\geq 0)$, \ref{item:equiv1} always holds due to \ref{p:noZeroDiv}, i.e., $g_1\cdot f_1 = 0 $ implies $g_1=0$.

For $(k\geq \numRows\geq 2, n=0)$, both \ref{item:equiv1} and \ref{item:equiv2} never hold therefore they are equivalent. Note that $n=0\implies \rowpoly_i=\SkewVar^{\tfzt_i}$ for all $i\in[k]$. For any $f_i=\SkewVar^{\tfzt_i}$ and $f_j=\SkewVar^{\tfzt_j}$ with $\tfzt_i\neq \tfzt_j$ (w.l.o.g.~assuming $\tfzt_i>\tfzt_j$), there exist $g_i=1$ and $g_j=-\SkewVar^{\tfzt_i-\tfzt_j}$ such that $g_if_i+g_jf_j=0$ and hence \ref{item:equiv1} never holds. Suppose $\tfzt_1\leq\tfzt_2$, then for $\Omega=\{1,2\}$, \eqref{eq:equiv2} becomes $k-\tfzt_1\geq (k-\tfzt_1)+(k-\tfzt_2)$, which contradicts that $\deg \rowpoly_i=|\zeroSet_i|+\tfzt_i\leq k-1$. Hence, \ref{item:equiv2} never holds.

  For $(k\geq\numRows\geq 2, n\geq 1)$, we do the induction with the following hypotheses:
  \begin{enumerate}[label={\bfseries H\arabic*}]
  \item \label{H:fromPre}
    Assume that \ref{item:equiv2}$\implies$\ref{item:equiv1} is true for all parameters $(k',\numRows',n')\prec(k,\numRows,n)$.
  \item \label{H:fromHere}
    Take any $f_1,\dots, f_\numRows\in\Skewnk$ for which \ref{item:equiv2} is true for $(k,\numRows,n)$.
  \end{enumerate}

  The logic of the proof is summarized in Figure \ref{fig:proofLogic}.
  \begin{figure}[h]
    \centering
    \input{figs/GM-MSRD/inductionProofLogic}
    \caption{Proof logic for $\ref{item:equiv2}\implies\ref{item:equiv1}$ with initial hypothesis \ref{H:fromPre} and \ref{H:fromHere}.}
    \label{fig:proofLogic}
  \end{figure}
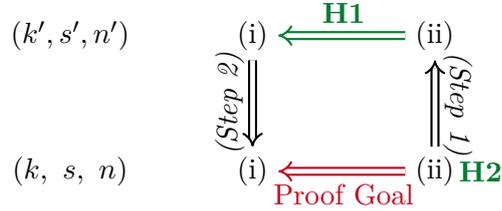

  Starting from \ref{H:fromHere}, we have that for all the subsets $\varnothing\neq \Omega\subseteq[\numRows]$, the inequality \eqref{eq:equiv2} in \ref{item:equiv2} holds.
  We prove that \ref{item:equiv1} is true for $(k,\numRows, n)$ via \stepone $\rightarrow$ \ref{H:fromPre} $\rightarrow$ \steptwo under different cases:

  \begin{enumerate}[leftmargin=3.4em,label={\bfseries \textit{Case \arabic*}}]
  \item \label{case_s3n2} For $\numRows\geq 3$ and $n\geq 2$,
    \begin{enumerate}[leftmargin=2em, label={\bfseries \textit{Case 1\alph*}}]
    \item \label{case_s3n2_a} $\forall i\in[\numRows]$, $\tfzt_i\geq 1$ (i.e., $|Z_i|\leq k-2$).
      \textit{(In this case, we do induction by reducing $k$.)} 
    \item \label{case_s3n2_b} $\exists$ a unique $i\in[\numRows]$ such that $\tfzt_i=0$.
      \textit{(In this case, we do induction by reducing $k$. We may need to reduce $s$ as well.)} 
    \item \label{case_s3n2_c} $\exists\ \Omega \subset [\numRows]$ with $2\leq |\Omega|\leq \numRows-1$ such that \eqref{eq:equiv2} holds with equality.
      \textit{(In this case, we do induction by reducing $s$.)} 
    \item \label{case_s3n2_d} $\forall\ \Omega\subset [\numRows]$ with $2\leq |\Omega|\leq \numRows-1$, \eqref{eq:equiv2} holds strictly and $\exists$ at least two $i\in[\numRows]$ such that $\tfzt_i=0$.
      \textit{(In this case, we do induction by reducing $n$.)} 
    \end{enumerate}
  \item \label{case_s2n2} For $\numRows=2$ and $n\geq 2$,
    \begin{enumerate}[leftmargin=2em, label={\bfseries \textit{Case 2\alph*}}]
    \item \label{case_s2n2_a} $\forall i\in\{1,2\}$, $\tfzt_i\geq 1$ (i.e., $|Z_i|\leq k-2$).
      \textit{(The same as \ref{case_s3n2_a}.)}
    \item \label{case_s2n2_b} $\exists$ a unique $i\in\{1,2\}$ such that $\tfzt_i=0$.
      \textit{(The same as \ref{case_s3n2_b}.)}
    \item \label{case_s2n2_c} $\forall i\in\{1,2\}$, $\tfzt_i=0$.
      \textit{(In this case, we do induction by reducing $n$.)} 
    \end{enumerate}
  \item \label{case_s2n1} For $\numRows\geq 2$ and $n=1$,
    \begin{enumerate}[leftmargin=2em, label={\bfseries \textit{Case 3\alph*}}]
    \item \label{case_s2n1_a} $\forall i\in[\numRows]$, $\tfzt_i\geq 1$ (i.e., $|Z_i|\leq k-2$).
      \textit{(The same as \ref{case_s3n2_a}.)}
    \item \label{case_s2n1_b} $\exists$ a unique $i\in\{1,2\}$ such that $\tfzt_i=0$.
      \textit{(The same as \ref{case_s3n2_b}.)}
    \item \label{case_s2n1_c} $\exists$ at least two $i\in[\numRows]$, $\tfzt_i=0$.
      \textit{(We show that this case cannot happen if \ref{item:equiv2} is true for $(k\geq s\geq 2, n=1)$.)}
    \end{enumerate}
  \end{enumerate}
  We illustrate the reduction of $s$ and $n$ in the induction under these cases in \cref{fig:induction_reducing}. We omitted the parameter $k$ for clarity and simplicity, since only $s,n$ are essential in classifying the different cases. The elaborated proofs for each case are presented in \cref{appendix:induction-proof-general-result}.
  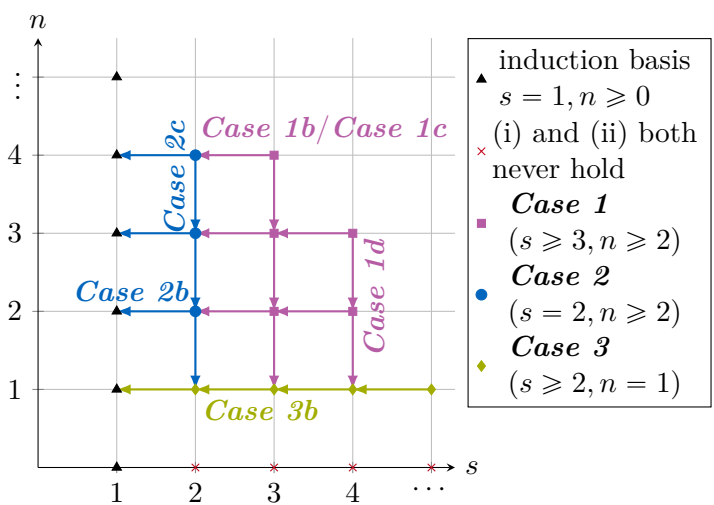
\begin{figure}[htb]
    \centering
    \input{figs/GM-MSRD/induc_pic}
    \caption{Illustration of the induction for \ref{item:equiv2}$\implies$\ref{item:equiv1} under difference cases.}
    \label{fig:induction_reducing}
  \end{figure}

\section{Applications in Multi-Source Network Coding}
\label{sec:multi-source-net-cod}
\emph{Distributed multi-source networks} were studied in \cite{halbawi2014distributed, halbawi2014distributedGab}, where supported-constrained error-correcting codes were used to achieve reliable communication against malicious (or failed) nodes in the network. In this section, we introduce a scheme to design \emph{distributed LRS codes} for any such network instance. The scheme illustrates how the necessary and sufficient conditions derived in \cref{sec:support-constrained-LRS} can be used as constraints in a linear programming problem to design the parameters of desired distributed LRS codes.

  Consider a distributed multi-source network as illustrated in \cref{fig:RLNmodel}. The receiver at the sink intends to obtain all the messages in a set $\cM$ by downloading through an $\Fq$-linear network from multiple source nodes. Each source node has access to only a few messages in $\cM$. This access is assumed to have unlimited link capacity (e.g., the source nodes store the subset of $\cM$ locally).
The topology of the $\Fq$-linear network is not known to the source nodes nor to the sink; therefore, it is a \emph{non-coherent} communication scenario.
This model can find its applications in data sharing platforms, sensor networks, satellite communication networks and MIMO (multiple-input multiple-output) antenna communication systems, etc.
  \begin{figure}[h]
    \centering
    \input{figs/GM-MSRD/MultiSourceNetwork}
    \caption{Illustration of the distributed multi-source network model.}
    \label{fig:RLNmodel}
  \end{figure}
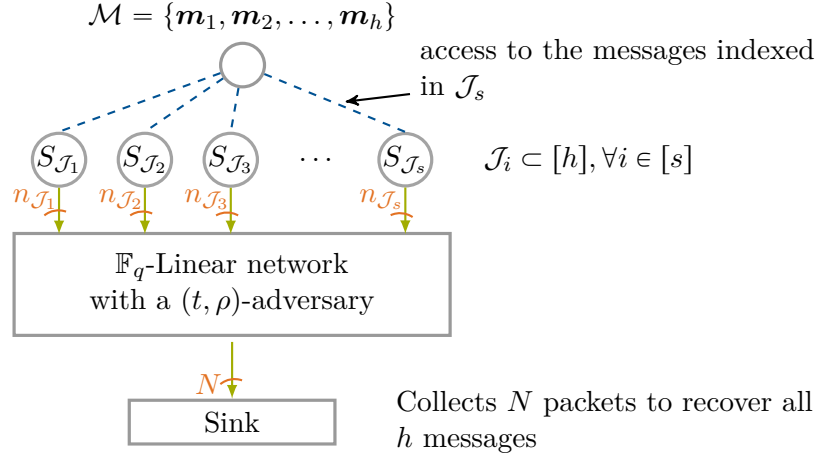

  The set $\cM$ contains $h$ messages. The message $\bm_j, j\in[h]$ is composed of $r_j$ symbols over $\Fqm$, i.e., $\bm_j\in\Fqm^{r_j}$.
  The source node $S_{\cJ_i}, i\in[s]$ has access only to the messages indexed in $\cJ_i$, e.g., if $\cJ_2=\{3,6\}$, then $S_{\cJ_2}$ only has the access to the messages $\bm_3$ and $\bm_6$.
  Let $\cS=\{\cJ_1,\dots, \cJ_s\}$. For any $\cJ\in\cS$, the source node $S_{\cJ}$ encodes the messages $\bm_j, j\in\cJ$, into $n_{\cJ}$ symbols over $\Fqm$, denoted by $\bc_{\cJ}\in\Fqm^{n_\cJ}$. It then extends them to their matrix representation over $\Fq$, denoted by $\bC_{\cJ}\in\Fq^{m\times n_{\cJ}}$, and then generates $\bX_{\cJ}= (\0\ \cdots \ \bI_{n_{\cJ}}\ \0 \ \cdots\ \bC_{\cJ}^{\top})\in\Fq^{n_{\cJ}\times (n+m)}$, where $n=\sum_{\cJ\in\cS}n_{\cJ}$. We call each row of $\bX_{\cJ}$ a \emph{packet}.
  Denote $\bX=
  \begin{pmatrix}
    \bX_{\cJ_1}\\
    \vdots\\
    \bX_{\cJ_s}
  \end{pmatrix}\in\Fq^{n\times (n+m)}
  $, where the rows are the packets transmitted by all the source nodes into the $\Fq$-linear network.
  The task is to design $n_\cJ$ for all $\cJ\in\cS$ such that the sink can recover all the messages $\bm_i$. The goal of the design is to minimize the total number of packets $n$.
  A concrete example is given in \cref{sec:ex-dist-LRS}.

  In the $\Fq$-linear network, whenever there is a transmission opportunity, a relay node in the network produces and sends an arbitrary $\Fq$-linear combination of all the incoming packets they have received.
  Suppose that there are at most $t$ \emph{malicious} nodes that inject erroneous packets and at most $\rho$ \emph{frozen} nodes that do not send any packet, which we refer as a \emph{$(t,\rho)$-adversary}.
 The sink collects $N\geq n-\rho$ packets, which are represented by the rows of $\bY\in\Fq^{N\times (n+m)}$.
 The transmitted packets (rows of $\bX$) and the received packets (rows of $\bY$) can be related via the following network equation:
  \begin{align}
    \label{eq:net-equation}
    \bY=\bA\bX+\bE\ ,
  \end{align}
  where $\bA\in\Fq^{N\times n}$ is the \emph{transfer} matrix of the network and the difference between the number of columns and its row-rank is at most $\rho$. In other words, $n-\rank(\bA)\leq \rho$.
  $\bE\in\Fq^{N\times M}$ is an error matrix of $\rank(\bE)\leq t$.
  Note that the matrices $\bA$ and $\bE$ are not known to any of the source nodes or the sink since we consider a non-coherent communication scenario.

  The \emph{capacity region} of a multi-source network with $h$ messages is a set $\set*{(r_1,\dots, r_h)}\subseteq\bbN^h$ such that the receiver at the sink can recover all the messages $\bm_j\in\Fqm^{r_j}, j\in[h]$.
  The capacity region of a multi-source network against a $(t,\rho)$-adversary has been given in \cite[Theorem 2]{dikaliotis2011multiple} (for $\rho=0$) and \cite[Corollary 66]{ravagnani2018adversarial}.
  To present the result, we require the following definitions of \emph{min-cut}.
  \begin{definition}[Min-cut between a set of nodes and another node]
    For a directed graph $\cG(\cV,\cE)$ composed of a set of nodes $\cV$ and a set of edges $\cE$, a \emph{cut} between a set of nodes $\cV'\subset\cV$ and another node $t\in\cV\setminus\cV'$ is a subset of edges $\cE_{\cV',t}\subseteq\cE$ such that, after removing the edges in $\cE_{\cV',t}$, there is no path from any of the nodes in $\cV'$ to $t$. The \emph{min-cut} between $\cV'$ and $t$ is the smallest cardinality of a cut between $\cV'$ and $t$.
  \end{definition}
  \begin{definition}[Min-cut between a subset of messages and the sink]
    Consider the distributed multi-source network with $h$ messages as above. Given a subset of messages, $\cJ'\subseteq[h]$, consider the set of source nodes $\cV'_{\cJ'}$ that contain messages in $\cJ'$, namely,
    \begin{align*}
      \cV'_{\cJ'} = \set*{\left. S_\cJ \in \cS \ \right|\ \cJ\cap\cJ'\neq\varnothing}\ .
    \end{align*}
    We define the \emph{min-cut between $\cJ'$ and the sink} as the min-cut between $\cV'_{\cJ'}$ and the sink, and denote it by $\mincut_{\cJ'}$.
  \end{definition}
  \begin{theorem}[\hspace{1pt}\cite{dikaliotis2011multiple, ravagnani2018adversarial}]
    \label{thm:capacity-distributed-network}
    Consider a multi-source network with $h$ messages. For any $(r_1,\dots, r_h)\in\bbN^h$ in the capacity region against a $(t,\rho)$-adversary, we have
    \begin{align}
      \label{eq:multi-source-capacity}
      \forall \varnothing \neq \cJ'\subseteq [h],\ \sum_{i\in\cJ'} r_i\leq \mincut_{\cJ'} - 2t-\rho\ ,
    \end{align}
    where $\mincut_{\cJ'}$ is the min-cut between the set ${\cJ'}$ of messages and the sink.
    \end{theorem}
  In addition to the general settings, we further assume the following setup of the non-coherent network:
  \begin{itemize}
  \item
    The communication capacity of the non-coherent linear network is large enough so that the min-cut $\mincut_{\cJ'}$ for all $\cJ'\subseteq[h]$ is determined by the number of encoded symbols $n_{\cJ}$ sent by the source node $S_{\cJ}$ for all $\cJ\in\cS$, i.e.,
      \begin{align*}
        \mincut_{\cJ'}= n - \sum_{\substack{\cJ\in\cS\\\cJ\subseteq [h]\setminus\cJ'}}n_{\cJ}\ .
      \end{align*}
      Note that the term $\sum_{\substack{\cJ\in\cS\\\cJ\subseteq [h]\setminus\cJ'}} n_{\cJ}$ is the total number of encoded symbols that do not contain any information about the messages in $\cJ'$.
  \item Although the encoding is distributed (since each source node may access only a few messages), there is a centralized coordination unit designing the overall code, and the sink knows the distributed code.
  \end{itemize}

  \subsection{Sum-Rank Weight of Error and Erasure with Constrained Rank Weight}
  \label{sec:sum-rank-weight-rank-error}
  In the following, we intend to use LRS codes for the distributed multi-source linear network model.
  Note that the errors and erasures in the $(t,\rho)$-adversarial model are measured in the rank metric. However, LRS codes are used to deal with errors and erasures in the sum-rank metric.
  Hence, we first look into the sum-rank deficiency of the network transfer matrix $\bA\in\Fq^{ N\times n}$ and the sum-rank weight of the error matrix $\bE\in\Fq^{N\times M}$.

  Let $\ell\in\bbN$ and $\bn_\ell=(n_1,\dots, n_\ell)$ be an ordered partition of $n$.
  By \cref{lem:rank-leq-sumrank-mat}, we have
  \begin{align}
    \label{eq:sumrank-weight-erasure}
    \wtSR{\bn_\ell}(\bA)\geq \rank(\bA)\geq n-\rho\ .
  \end{align}
  Hence the sum-rank weight of the erasure induced by the rank-deficient $\bA$ is at most $\rho$.

  For the error $\bE$, consider an ordered partition $\bN_{\ell}=(N_1,\dots, N_\ell)$ of $N$ such that

  \input{figs/GM-MSRD/E_partition.tex}

  Given $\rank(\bE)\leq t$, by \cref{lem:rank-leq-sumrank-mat}, we have
  \begin{align}
    \label{eq:sumrank-weight-E}
    \wtSR{\bN_\ell}(\bE) = \sum_{i=1}^{\ell} \rank(\bE_i)
    \leq \sum_{i=1}^{\ell} \rank(\bE) = \ell t\ .
  \end{align}
  This upper bound holds for any arbitrary ordered partition $\bN_{\ell}$ of $N$.
  A lower bound on $\Pr[\left. \wtSR{\bN_\ell}(\bE)=\ell t\ \right|\ \rank(\bE)=t]$ (i.e., the probability that \eqref{eq:sumrank-weight-E} is tight) for small $t$ ($t\leq N_i, \forall i\in[\ell]$) is given in \cite[Theorem 1]{couvee2023notes}. In particular, if $q\geq \ell+1$, then $\Pr[\wtSR{\bN_\ell}(\bE)=\ell t\ |\ \rank(\bE)=t]>1/4$ \cite[Corollary 1]{couvee2023notes}.

  It can been seen from \eqref{eq:sumrank-weight-erasure} and \eqref{eq:sumrank-weight-E} that the network model in \eqref{eq:net-equation} results in an erasure of sum-rank weight at most $\rho$ and an error of sum-rank weight at most $\ell t$.
  It has been shown in \cite[Theorem 1, Eq.(4), Proposition 2]{martinez2019reliable} that a code with sum-rank distance $d$ can guarantee reliable communication against errors of sum-rank weight at most $\ell t$ and erasures with sum-rank weight at most $\rho$ in the non-coherent communication if $d \geq 2\ell t+\rho+1$.
  Therefore, an LRS code with sum-rank distance $d\geq 2\ell t+\rho+1$ can correct any error and erasure in the $(t,\rho)$-adversarial model.
  \subsection{Example of Distributed LRS codes}
  \label{sec:ex-dist-LRS}
  We first give a toy example to show the usage of LRS codes in a distributed multi-source network, and then provide a general scheme to design \emph{distributed LRS codes} for arbitrary distributed multi-source networks in \cref{sec:distributed-LRS}.

  Consider the following toy example of the network illustrated in \cref{fig:RLNmodel}: There are $h=4$ messages in $\cM$. The lengths of messages are $(r_1, r_2, r_3, r_4)=(1,3,2,3)$. There are $4$ source nodes and each can access to only $3$ messages, i.e., $\cJ_1= \{1,2,3\}, \cJ_2= \{1,2,4\}, \cJ_3 = \{1,3,4\}, \cJ_4 = \{2,3,4\}$.
  Suppose there is a $(2,2)$-adversary in the $\Fq$-linear network.

  The number of encoded packets from each source node is $(n_{\cJ_1}, n_{\cJ_2}, n_{\cJ_3}, n_{\cJ_4})= (6,7,2,8)$ (see \ref{item:dist-LRS-LP} in \cref{sec:distributed-LRS} for the computation of these values) and $n=\sum_{i=1}^4 n_{\cJ_i} = 23$.
  Let $\bm=(\bm_1,\bm_2,\bm_3,\bm_4)$ be a concatenated vector of all the messages. Some entries in a encoding matrix $\bG$ are forced to be $0$, as shown in \cref{fig:eg-support-constrained-LRS}, so that $\bm\cdot \bG$ represents the overall encoding at all source nodes.
  \begin{figure}[h]
    \centering
    \input{figs/GM-MSRD/dis_LRS_gen.tex}
    \caption{Illustration of the required encoding matrix for the instance of distributed multi-source networks with $h=4$ messages. This support-constrained matrix is a generator matrix of a $[23,9]$ LRS over $\F_{4^9}$ with $\ell=3$ blocks.}
    \label{fig:eg-support-constrained-LRS}
  \end{figure}
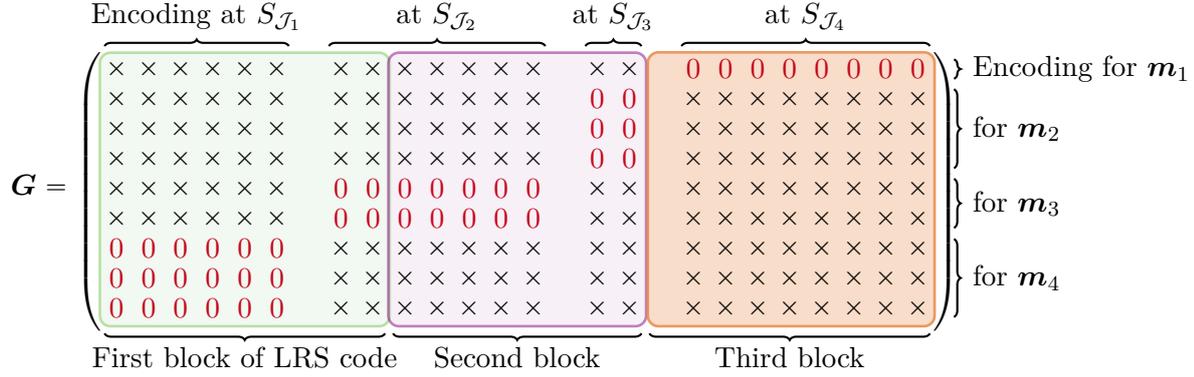
  For example, the first $6$ columns of $\bG$, corresponding to the encoding at $S_{\cJ_1}$, have zero entries in the last $3$ rows. This indicates that $S_{\cJ_1}$ does not encode $\bm_4$ since it does not have the access to $\bm_4$.

  We can obtain the support-constrained encoding matrix $\bG$ from a generator matrix of a $[23, 9]_{4^9}$ LRS code with $\ell=3$ blocks. The lengths of the blocks are $(n_1, n_2, n_3)=(8,7,8)$ (see \ref{item:dist-LRS-LP} in \cref{sec:distributed-LRS} for the computation of these parameters).
  Let $\priEle$ be a primitive element of $\F_{4^{9}}$. The block representatives of the LRS code are $(a_1, a_2,a_3)=(1,\priEle,\priEle^2)$ and the column multipliers are $\bb=( 1,\priEle,\dots,\priEle^7, \priEle, \priEle^2,\dots, \priEle^7,\priEle^2,$ $ \priEle^3,\dots, \priEle^9)$.
  We construct a generator matrix $\GLRS$ of the LRS code according to \eqref{eq:Gevaluation} and find a full-rank matrix $\bT\in\F_{4^9}^{9\times 9}$ such that the support-constrained encoding matrix $\bG$ is given by $\bG=\bT\cdot \GLRS$.
  It can be verified by \cref{thm:fieldSizeFromGab} that such a matrix $\bT$ exists over $\F_{4^9}$ and it can be found by solving a linear system of equations.
  For this example, we found the following $\bT$ as a solution (see \cref{sec:distributed-LRS} \ref{step:dist-LRS-compute-T} for the computation method that we used here).
  \begin{align*}
    \bT=
    \begin{pmatrix}
\priEle^{ 29883 }& \priEle^{       208968 }& \priEle^{         19488  }& \priEle^{        27791    }& \priEle^{   137529   }& \priEle^{    135128   }& \priEle^{ 142532  }& \priEle^{ 123564   }& \priEle^{ 199506 }\\
\priEle^{ 8272  }& \priEle^{      137891  }& \priEle^{        117682  }& \priEle^{        134830   }& \priEle^{    175546  }& \priEle^{     199273  }& \priEle^{  233167 }& \priEle^{  13175   }& \priEle^{ 75587} \\
\priEle^{ 171183}& \priEle^{        60863 }& \priEle^{         88547  }& \priEle^{        152810   }& \priEle^{    183852  }& \priEle^{     129008  }& \priEle^{  223733 }& \priEle^{  220778  }& \priEle^{  215911} \\
\priEle^{ 136657}& \priEle^{        53725 }& \priEle^{         187129 }& \priEle^{         236279  }& \priEle^{     244758 }& \priEle^{      124656 }& \priEle^{   163100}& \priEle^{   222367 }& \priEle^{   245041}\\
\priEle^{ 67172 }& \priEle^{       31331  }& \priEle^{        217264  }& \priEle^{        133630   }& \priEle^{    190037  }& \priEle^{     228340  }& \priEle^{  210873 }& \priEle^{  222699  }& \priEle^{  102082}\\
\priEle^{ 180377}& \priEle^{        78748 }& \priEle^{         71136  }& \priEle^{        170404   }& \priEle^{    251773  }& \priEle^{     44364   }& \priEle^{ 188627  }& \priEle^{ 44347    }& \priEle^{145983}\\
\priEle^{ 82368 }& \priEle^{       167072 }& \priEle^{         210000 }& \priEle^{         110692  }& \priEle^{     24773  }& \priEle^{     69984   }& \priEle^{ 182180  }& \priEle^{ 211569   }& \priEle^{ 24237}\\
\priEle^{ 78461 }& \priEle^{       249391 }& \priEle^{         68483  }& \priEle^{        120459   }& \priEle^{    140206  }& \priEle^{     243029  }& \priEle^{  126875 }& \priEle^{  75641   }& \priEle^{ 12289}\\
\priEle^{ 33368 }& \priEle^{       98307  }& \priEle^{        247550  }& \priEle^{        210053   }& \priEle^{    223247  }& \priEle^{     103052  }& \priEle^{  160318 }& \priEle^{  69947   }& \priEle^{ 42305}
    \end{pmatrix}
  \end{align*}
\begin{remark}
  With the choice of $\ell$ and $(n_1,\dots,n_\ell)$ for this toy LRS code we intend to show that the number of blocks $\ell$ does not need to be the same as the number of source nodes $s$.
  The value of $\ell$ determines the upper bound in \eqref{eq:sumrank-weight-E} on the sum-rank weight of $\bE$.
  We listed several other parameters of the LRS codes in \cref{tab:ell_grows} that can be used for this network example.
  It can be seen that, the larger $\ell$ is, the larger error-correction capability is required, which results in larger sum-rank distance of the LRS code and hence larger total length $n$ and field size $q^m$.
  However, larger $\ell$ may result in a smaller field size. For instance, suppose that the messages $\bm_i$ are over $\F_{3^{11}}$. According to \cref{tab:ell_grows}, setting $\ell=1$ (i.e., using a distribued Gabidulin code \cite{halbawi2014distributedGab}) requires a field size $q^m=3^{15}$ while using the distributed LRS codes with $\ell=2$ requires a field size $q^m=3^{11}$ (note that the field size of the messages is $3^{11}$ hence the code should be over $\F_{3^{11}}$).

  \begin{table}[h!]
    \centering
    \caption{Parameters of distributed LRS codes for the toy network example while increasing $\ell$. The $q$ and $m$ are the minimal parameters of the required field over which the $[n,\widetilde{k},d]$ distributed LRS code can be constructed, where $d=2\ell t+\rho+1$ is the sum-rank distance of the distributed LRS code.}
    \label{tab:ell_grows}
    \begin{tabular}{c|c|c|c|c|c}
      $\ell$ & $q$ & $m$ & $[n,\widetilde{k}, d]$ & $(n_1,\dots, n_{\ell})$ & $(n_{\cJ_1},n_{\cJ_2},n_{\cJ_3},n_{\cJ_4})$ 
      \\
      \hline
      $1$ (Gabidulin code) & $2$ & $15$ & $[15, 9, 7]$ & $(15)$ & $(6,1,0,8)$
      \\
      \hline
      $2$ & $3$ & $10$ & $[19, 9 , 11]$ & $(10, 9)$ & $(6,5,0,8)$
      \\
      \hline
      $3$ (\cref{fig:eg-support-constrained-LRS}) & $4$ & $9$ & $[23,9,15]$ & $(8,7,8)$ & $(6,7,2,8)$\\
      \hline
      $4$ & $5$ & $9$ & $[27, 9, 19]$ & $(7,7,7,6)$ & $(6,7,6,8)$\\
      \hline
      $5$ & $7$ & $11$ & $[33, 11, 23]$ & $(7,6,7,7,6)$ & $(8,9,6,10)$\\
      \hline
      $6$ & $7$ & $12$ & $[38, 12, 27]$ & $(7,6,6,7,6,6)$ & $(9,10,8,11)$\\
      \hline
      $7$ & $8$ & $13$ & $[43, 13, 31]$ & $(6,6,6,7,6,6,6)$ & $(10,11,10,12)$\\
    \end{tabular}
  \end{table}
  In \cref{tab:s_changes}, we list the parameters of LRS codes for several different $\cS=\{\cJ_1,\cJ_2,\cJ_3,\cJ_4\}$. It can be seen that encoding each message independently requires a longer code (hence, a larger alphabet size) than jointly encoding subsets of messages.
  \begin{table}[h!]
    \centering
    \caption{Parameters of distributed LRS codes for the toy example while changing $\cS$. }
    \label{tab:s_changes}
    \begin{tabular}{c|c|c|c|c|c|c}
      $\cS$ &$\ell$ & $q$ & $m$ & $[n,\widetilde{k}, d]$ & $(n_1,\dots, n_{\ell})$ & $(n_{\cJ_1},n_{\cJ_2},n_{\cJ_3},n_{\cJ_4})$
      \\
      \hline
      \multirow{3}{*}{$\{\{1\},\{2\},\{3\},\{4\}\}$} &$1$& $2$ & $33$ & $[33, 27, 7]$ & $(33)$ & $(7,9,8,9)$\\
      \cline{2-7}
            &$2$& $3$ & $39$ & $[49, 39, 11]$ & $(25,24)$ & $(11,13,12,13)$\\
      \cline{2-7}
            &$3$& $4$ & $51$ & $[65, 51, 15]$ & $(22,22,21)$ & $(15,17,16,17)$\\
      \hline
      \multirow{3}{*}{
      \begin{tabular}{@{}c@{}}
        $\{\set*{1,2},\set*{1,3},$\\
        $\set*{2,4},\set*{3,4}\}$
      \end{tabular}}
            & $1$ & $2$ & $17$ & $[17, 11, 7]$ & $(17)$ & $(6,1,3,7)$\\
      \cline{2-7}
            & $2$ & $3$ & $15$ & $[25, 15, 11]$ & $(13,12)$ & $(10,1,3,11)$\\
      \cline{2-7}
            & $3$ & $4$ & $19$ & $[33, 19, 15]$ & $(11,11,11)$ & $(14,1,3,15)$\\
    \end{tabular}
  \end{table}
\end{remark}

Now we proceed to apply the \emph{lifting} technique \cite{silva2008rank} to deal with the non-coherent situation.
Supposing $(\bc_{\cJ_1},\bc_{\cJ_2},\bc_{\cJ_3},\bc_{\cJ_4})= (\bm_1,\bm_2,\bm_3,\bm_4)\cdot \bG$. Each source node $S_{\cJ_i}$ generates $\bC_{\cJ_i}=\extbasis{\bbeta}(\bc_{\cJ_i})\in\Fq^{m\times n_{\cJ_i} }$ by the map defined below
and lifts the $\bC_{\cJ_i}^\top$ by adding the identity and zero matrices as in \eqref{eq:lifting-X} to obtain the transmitted packets (rows of $\bX$).
  \input{figs/GM-MSRD/lifting}
  Each row is a packet of length $n+m$ ($=23+9=31$ for the toy example) over $\Fq$ ($\F_4$) transmitted into the network.
  Note that for the lifting step, the centralized coordination unit is also needed to instruct the source nodes where to put the identity matrix in their packets.
  \subsection{The General Scheme: Distributed LRS Codes}
  \label{sec:distributed-LRS}
  In the following we present the general scheme at the centralized coordination unit to design the overall distributed LRS codes, given:
  \begin{itemize}
  \item the total number of messages $h$ and their lengths $r_1,\dots, r_h$;
  \item the set $\cS=\{\cJ_1, \dots,\cJ_s\}$, where each $\cJ_i\subset [h]$ contains the indices of the messages that the source node $S_{\cJ_i}$ has access to;
  \item the $(t,\rho)$-adversarial model: the maximum number $t$ of malicious nodes and the maximum number $\rho$ of frozen nodes in the network;
  \item the number of blocks $\ell$ of the LRS code.
  \end{itemize}
  The task is to design the $n_{\cJ}$, for all $\cJ\in\cS$, such that the sink can recover all $h$ messages. The goal of the design is to minimize $n$, the total number of the encoded symbols.

  The general scheme contains the following steps:
  \begin{enumerate}
  \item \label{item:dist-LRS-LP}
    Solving the following integer linear programming problem for $(n_{\cJ_1},\dots,n_{\cJ_s})$
    \begingroup
    \setlength\arraycolsep{3pt}
    \allowdisplaybreaks
    \begin{align}
      \text{minimize } \quad & n= n_{\cJ_1}+\dots+n_{\cJ_s}\nonumber \\
      \text{subject to } \quad & \forall \varnothing\neq \cJ'\subseteq[h],\ \sum_{i\in\cJ'} r_i + 2t+\rho\leq n - \sum_{\substack{\cJ\in\cS\\\cJ\subseteq [h]\setminus\cJ'}}n_{\cJ}  \label{eq:capacity-constraints}\ ,\\
                             & \forall \varnothing\neq \Omega\subseteq[h],\ \sum_{\substack{\cJ\in\cS\\ [h]\setminus \cJ\supseteq\Omega} }  n_{\cJ} + \sum_{i\in\Omega} r_i \leq n-2\ell t-\rho \label{eq:dis-zero-constraints} \ ,\\
                             &\forall \cJ\in\cS, \quad\quad\quad n_{\cJ}\geq 0\ .\nonumber
    \end{align}
    \endgroup
    \textit{Remark:
      Recall that we assume that the min-cut $\mincut_{\cJ'}= n - \sum_{\substack{\cJ\in\cS\\\cJ\subseteq [h]\setminus\cJ'}}n_{\cJ}$, for all $\varnothing\neq\cJ'\subseteq[h]$. With the constraints in \eqref{eq:capacity-constraints}, the choice of $(n_{\cJ_1},\dots, n_{\cJ_s})$ guarantees that the message lengths $(r_1,\dots,r_h)$ are in the capacity region given in \cref{thm:capacity-distributed-network}. \newline
    Let $$\widetilde{k}:=\underset{\varnothing\neq\Omega\subseteq[h]}{\max} \sum\limits_{\substack{\cJ\in\cS\\ [h]\setminus \cJ\supseteq\Omega} }  n_{\cJ} + \sum\limits_{i\in\Omega} r_i\ .$$
    By \cref{thm:subcode}, there exist a subcode of an $[n,\widetilde{k}]$ LRS code whose generator matrix fulfills the support constraints of the encoding matrix $\bG$ for the distributed multi-source network.
    The constraints in \eqref{eq:dis-zero-constraints} guarantee that $\widetilde{k}\leq n-2\ell t-\rho$, which ensures that the $[n,\widetilde{k}]$ LRS code can decode the rank-metric errors and erasures induced by the $(t,\rho)$-adversarial model (see \cref{sec:sum-rank-weight-rank-error}).
    }
  \item Determine the field size $q^m$ required for the $[n, \widetilde{k}]$ LRS code with $\ell$ blocks according to \cref{thm:subcode}.

    \noindent
    \textit{Remark: The total length should be distributed as evenly as possible into $\ell$ blocks so that the extension degree $m$ is minimized.}
  \item Construct a generator matrix $\GLRS$ of the $[n, \widetilde{k}]_{q^m}$ LRS code according to \eqref{eq:Gevaluation}.
  \item Find a full-rank matrix $\bT\in\Fqm^{k\times \widetilde{k}}$ (where $k=\sum_{i=1}^h r_i$) such that the support-constrained encoding matrix $\bG$ can be obtained from $\bG=\bT\cdot \GLRS$. \label{step:dist-LRS-compute-T}

    \noindent
    \textit{Remark: This can be done by solving a linear system of equations for the entries of $\bT$. For example, in our implementation, we see the entries of $\bT$ as variables in a multivariate polynomial ring $\polyRing=\Fqm[T_{11},\dots,T_{kk}]$ and translate the constraints (the zero entries in $\bG$) into a system of linear equations. We use the facilities (Gr\"obner bases, variate, etc.) for multivariate polynomials embedded in \href{https://doc.sagemath.org/html/en/reference/index.html}{SageMath} \cite{sagemath2022} to solve the system. Note that the \texttt{variate()} function in SageMath avoids computing the whole solution space when the system is underdetermined. If this is the case, let $\lambda$ be the degree of freedom of the system. We assign random values in $\Fqm$ to $\lambda$ variables, so that it becomes a determined system that is solvable by \texttt{variate()}.}
  \end{enumerate}

\section{Vector Network Coding for Generalized Combination Networks}
\label{sec:vec-net-cod}
A \emph{multicast network} is a network with exactly one source and multiple receivers demanding all the messages from the source. In networks that apply routing, every relay node can only pass on their received data. \emph{Network coding} has been attracting increasing attention since the seminal paper by
Ahlswede, et.~al.~\cite{ahlswede2000network}, which showed that the throughput of the network can be increased significantly by not just forwarding packets but also performing operations on them.
We formulate the \emph{network coding problem} as follows: for each node in the network,
find an encoding function of its incoming messages for each of its outgoing links.
A \emph{solution} is a set of the encoding functions at all the nodes in the network, such that each receiver can recover all (or a predefined subset of all) the messages.
A network is \emph{solvable} if a solution exists. 

Although network coding has the advantage of good throughput, the encoding at relay nodes incurs extra delay and memory occupation than routing.
This section considers these costs from the aspect of the required alphabet size to utilize network coding.
Reducing the alphabet size of the coding operations results in less complexity, hence less delay, and less memory occupation
for practical implementations of network coding~\cite{langberg2006encoding,LS2009,GSRMsep2019}.




In the following, we first formally introduce the concepts that are considered in this section.
\subsubsection{Generalized Combination Networks}


The main object that we study in this section is the class of \emph{generalized combination networks}.
An $(\eps,\ell)-\cN_{h,r,\alpha\ell+\eps}$ generalized combination network is illustrated in Figure~\ref{fig:Network} (see also~\cite{EW18}).
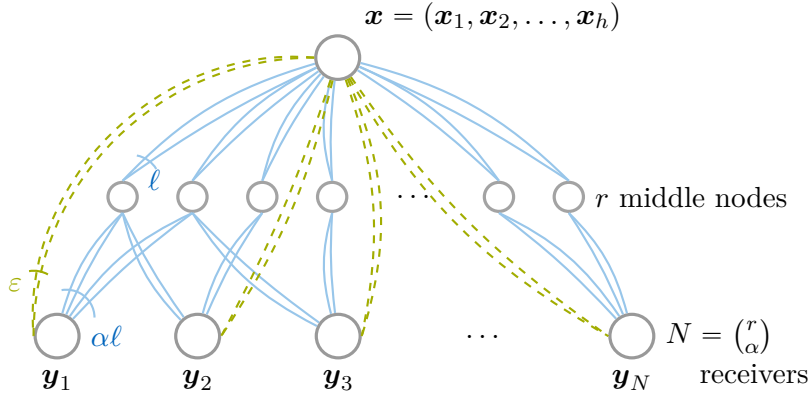
\begin{figure}[!htb]
  \centering
  \input{./figs/VecNetCod/GeneralizedNetwork.tex}
  \caption{Illustration of an $(\eps,\ell)-\mathcal{N}_{h,r,\alpha\ell+\eps}$ network.}
  \label{fig:Network}
\end{figure}

The network has three layers. The first layer consists of a source with $h$ source messages. The source is connected to each of $r$ middle nodes in the second layer
via $\ell$ parallel links (solid lines).
Any $\alpha$ middle nodes are connected to a unique receiver in the third layer, each via $\ell$ parallel links.
This implies that there are $N=\binom{r}{\alpha}$ receivers that receive distinct packets from the second layer.
In addition, each receiver is also connected to the source via $\eps$ direct links (dashed lines).
It was shown in~\cite[Theorem 8]{EW18} that the $(\eps,\ell)-\cN_{h,r,\alpha\ell+\eps}$ network has a trivial solution if $h\leq\ell+\eps$ and it has no solution if $h>\alpha\ell+\eps$.
We only consider non-trivially solvable networks, hence we assume that $\ell+\eps< h \leq \alpha\ell+\eps$ throughout the section.

\subsubsection{Vector/Scalar Linear Solutions of Network Coding}
In linear network coding, the outgoing links of each relay node carry linear functions of the messages from the incoming links. The linear functions are called \emph{coding coefficients}. The set of the coding coefficients is the solution of the linear network.

If the messages
are scalars in $\F_{q}$ and the coding coefficients are vectors over $\F_{q}$, then a solution is called a \emph{scalar linear solution}, denoted by $(q,1)$-{linear solution}.
If the messages are vectors in $\Fq^t$, and the coding coefficients are matrices over $\F_q$, then a solution is called a \emph{vector linear solution}, denoted by $(q,t)$-linear solution.

We now formulate a solution to the $(\eps,\ell)-\cN_{h,r,\alpha\ell+\eps}$ generalized combination network illustrated in \cref{fig:Network}.
W.l.o.g., we only formulate it with the notations for a vector linear solution;
a scalar linear solution can be obtained by simply setting $t=1$.

Denote by $\bx_1,\dots,\bx_h\in\mathbb{F}^t_{q}$ the $h$ source messages and by $\by_1,\dots,\by_N\in\mathbb{F}^{(\eps+\alpha\ell)t}_q$ the packets received by each receiver.
For each $i\in[N]$, $\by_i$
is the concatenation of all the packets that the $i$-th receiver gets from the $\alpha$ middle nodes and the source node.
Since each middle node has $\ell$ incoming links and $\alpha\ell$ outgoing links,
we assume w.l.o.g.~that the middle nodes just forward their incoming packets and the encoding is done at the source node.

We denote by $\bA_1,\dots,\bA_r \in \Fq^{\ell t \times ht}$ the coding coefficients used by the source node for the messages transmitted to the $r$ middle nodes,
and by $\bB_1,\dots,\bB_N\in\F_q^{\eps t\times ht}$ the coding coefficients used by the source node for the messages transmitted directly to the receivers.
Then, for each $i\in[N]$,
\begin{equation*}
\by_i
=\underbrace{\begin{pmatrix}
    \hphantom{w} & \bA_{i_1}& \hphantom{w}\\
    & \vdots&\\
    & \bA_{i_\alpha}& \\
    & \bB_i&
\end{pmatrix}}_{(\eps+\alpha\ell)t\times ht}
\cdot
\underbrace{\begin{pmatrix}
\bx_1\\
\vdots\\
\bx_{h}
\end{pmatrix}}_{ht\times 1}\ , 
\end{equation*}
where $\{\bA_{i_1},\dots, \bA_{i_\alpha}\}\subset \{\bA_1,\dots,\bA_r \}$.

Note that the receivers can recover the $h$ source messages $\bx_1,\dots,\bx_h$ if and only if
\begin{equation}\label{eq:NC_sol}
  \rank
  \begin{pmatrix}
    \hphantom{w} & \bA_{i_1}& \hphantom{w}\\
    & \vdots&\\
    & \bA_{i_\alpha}&
\end{pmatrix}
\geq (h-\eps)t,\ \forall i \in[N]\ .
\end{equation}
Hence a {solution} to the $(\eps,\ell)-\mathcal{N}_{h,r,\alpha\ell+\eps}$ network is a set of the coding coefficients $\set*{\bA_1,\dots,\bA_r }$ such that \eqref{eq:NC_sol} holds. The coding coefficients for the direct links $\bB_1,\dots,\bB_N$ can be determined once $\set*{\bA_1,\dots,\bA_r }$ is given.

\subsubsection{Gap between Required Alphabet Sizes for Scalar  and Vector Solutions}
The goal of this section is to investigate the gap between the minimum required alphabet size for scalar and vector solutions of the generalized combination networks.
This gap
was shown to be positive for generalized combination networks~\cite{EW18}. We further quantify the advantage of vector linear solutions versus scalar linear solutions of the generalized combination networks.
For this purpose, we need to fix a metric.

We follow the notations from~\cite{cai2020network} to distinguish between optimal scalar and vector solutions. Given a generalized combination network $\cN$, let
$$q_s(\cN):=\min\{q\ |\ \cN\textrm{ has a } (q,1)\textrm{-linear solution}\}\ .$$

The $(q_s(\cN),1)$-linear solution is said to be \emph{scalar-optimal}.
Similarly, let
$$q_v(\cN)\defeq\min\{q^t\ |\ \cN\textrm{ has a } (q,t)\textrm{-linear solution}\}\ .$$
Note that $q_v(\cN)$ is defined by the size of the vector space, rather than the field size.
For $q^t = q_v(\cN)$, the $(q,t)$-linear solution is called \emph{vector-optimal}.
We define the \emph{gap} as
\begin{align*}
 \gap (\cN):= \log_2 (q_s(\cN)) -\log_2(q_v(\cN))\ ,
\end{align*}
which intuitively measures the advantage of vector network coding by the amount of extra bits per transmitted symbol that an optimal scalar linear solution has to pay compared to an optimal vector linear solution.

\subsubsection{Overview of Results}
  In \cref{sec:ub_rs}, we give two upper bounds on $r_{\max}$, the maximal number of nodes in the middle layer of a generalized combination network (\cref{cor:imupperbound-N} (valid for $h\geq 2\ell+\eps$) and \cref{cor:imupperbound-2-Network} (a better bound for $\alpha=2$)).
  In \cref{sec:lb_rv}, we give two lower bounds on $r_{\max}$ (\cref{thm:LLL_bound} and \cref{cor:EK19_lb} ($h\leq 2\ell+\eps$)).
  In \cref{sec:bound_gap}, we provide an upper bound on $\gap (\cN)$
  for any fixed generalized combination network $\cN$ (\cref{thm:gap_ub}),
  and a lower bound on $\gap (\cN)$ (\cref{thm:gap_lb}).
  We compare the new bounds with some existing bounds on $r_{\max}$ in \cref{sec:discussion} and summarize the best known bound on $r_{\max}$ in \cref{tab:r_bound_summary}.

\subsection{Upper Bounds on the Maximum Number of Middle Layer Nodes}
\label{sec:ub_rs}

The Grassmannian of dimension $k$ is a set of all $k$-dimensional subspaces of $\mathbb{F}_q^n$. Recall that its cardinality is the well-known $q$-binomial coefficient:
\begin{equation*}
  |\cG_q(n,k)|=\quadbinom{n}{k}_q := \prod\limits_{i=0}^{k-1} \frac{q^n-q^i}{q^k-q^i}=\prod\limits_{i=0}^{k-1} \frac{q^{n-i}-1}{q^{k-i}-1}\ .
\end{equation*}
A good approximation of the $q$-binomial coefficient can be found in~\cite[Lemma 4]{koetter2008coding}:
\begin{equation}
    \label{eq:gauss}
q^{k(n-k)}\leq \quadbinom{n}{k}_q < \gamma \cdot q^{k(n-k)}\ ,
\end{equation}
where $\gamma\approx 3.48$.

\begin{lemma}\label{lm-full-t2}
Let $\alpha\geq 2$, $h,\ell,t\geq 1$, $\eps\geq 0$, $h-\eps \geq 2 \ell$, and let $\cT$ be a collection of subspaces  of $\F_q^{(h-\eps) t}$ such that
\begin{enumerate}
\item[(i)] each subspace has dimension at most $\ell t$, and
\item[(ii)] any subset of $\alpha$ subspaces spans $\F_q^{(h-\eps)t}$.
\end{enumerate}
Then, we have $\alpha \ell  \geq h-\eps$ and
$$|\cT| \leq \parenv*{\floor{ \frac{h-\eps}{\ell} }-2} +\parenv*{\alpha- \floor{ \frac{h-\eps}{\ell} }+1} \quadbinom{\ell t+1}{1}_q\ .$$
\end{lemma}

\begin{proof}
Take arbitrarily $\floor{\frac{h-\eps}{\ell}}-2$ subspaces from $\cT$ and a subspace  $W\subset\F_q^{(h-\eps)t}$ of dimension $(h-\eps)t-\ell t-1$ which contains all these $\lfloor \frac{h-\eps}{\ell} \rfloor-2$  subspaces. Then,
for any subspace $T\in \cT$, there is a \emph{hyperplane} (an $((h-\eps) t-1)$-dimensional subspace) of $\F_q^{(h-\eps)t}$ containing both $W$ and $T$.
Note that there are $\quadbinom{\ell t+1}{\ell t}=\quadbinom{\ell t +1}{1}$ hyperplanes of $\F_q^{(h-\eps)t}$ containing $W$ and each of them contains at most $\alpha-1$ subspaces from $\cT$. Thus,
\begin{align*}
 |\cT| &\leq  \parenv*{\floor{ \frac{h-\eps}{\ell} }-2} + \quadbinom{\ell t+1}{\ell t}_q\parenv*{\alpha-1 - \parenv*{\floor{ \frac{h-\eps}{\ell} }-2}} \\
  &=  \parenv*{\floor{ \frac{h-\eps}{\ell} }-2} +\parenv*{\alpha- \floor{ \frac{h-\eps}{\ell} }+1}  \quadbinom{\ell t+1}{1}_q\ .
\end{align*}
\end{proof}

\begin{theorem}\label{thm-upbound-v2}
Let $\alpha\geq 2$, $h,\ell,t\geq 1$, $\eps\geq 0$, $h-\eps\geq 2\ell$, and let $\cS$ be a collection of  subspaces of $\F_q^{ht}$ such that
\begin{enumerate}
\item[(i)] each subspace has dimension at most $\ell t$, and
\item[(ii)] any subset of $\alpha$ subspaces spans a subspace of dimension at least $(h-\eps)t$.
\end{enumerate}
Then, we have $\alpha \ell  \geq h-\eps$ and
\begin{align}
  |\cS| &\leq \quadbinom{(\eps+\ell)t}{\eps t}_q \parenv*{ \parenv*{\alpha- \floor{ \frac{h-\eps}{\ell} }+1}  \frac{q^{\ell t+1}-1}{q-1}-1} + \floor{\frac{h-\eps}{\ell}}-1\nonumber \\
        &<
          \gamma\parenv*{\alpha- \floor{ \frac{h-\eps}{\ell} }+1} q^{\ell t (\eps t+1)} {+ \floor{ \frac{h-\eps}{\ell}}-1}\label{eq:ub-strict-ineq}\ .
\end{align}
\end{theorem}

\begin{proof}
Take arbitrarily $\floor{\frac{h-\eps}{\ell}}-1$ subspaces from $\cS$ and a subspace $W \subset \F_q^{ht}$ of dimension $(h-\eps)t-\ell t$ such that $W$ contains all these  $\floor{\frac{h-\eps}{\ell}}-1$ subspaces. Then, for any subspace $S \in \cS$ there is a subspace of dimension $(h-\eps)t$ containing both $W$ and $S$.

Let $m\defeq \quadbinom{(\eps+\ell )t}{\eps t}_q$. Then, there are $m$ subspaces of dimension $(h-\eps)t$ containing $W$, say $W_1, W_2, \ldots, W_m$. Note that every subset of $\alpha$ subspaces in $W_i \cap \cS$ span the subspace $W_i$. According to Lemma~\ref{lm-full-t2}, we have
\begin{align*}
    |W_i \cap \cS| &\leq \parenv*{\floor{ \frac{h-\eps}{\ell} }-2}
    +\parenv*{\alpha- \floor{ \frac{h-\eps}{\ell} }+1}  \quadbinom{\ell t +1}{1}_q\ .
\end{align*}
Hence,
\begin{align*}
|\cS|  &\leq \sum_{i=1}^{m}\parenv*{|W_i\cap \cS|- \parenv*{\floor{ \frac{h-\eps}{\ell}}-1  }} + \floor{ \frac{h-\eps}{\ell}}-1 \\
 &\leq  \quadbinom{(\eps+\ell)t}{\eps t}_q \parenv*{ \parenv*{\alpha- \floor{ \frac{h-\eps}{\ell} }+1}  \frac{q^{\ell t+1}-1}{q-1}  -1   }+ \floor{ \frac{h-\eps}{\ell}}-1 \ .
 \end{align*}
The inequality \eqref{eq:ub-strict-ineq} is derived from \eqref{eq:gauss}.
\end{proof}
The following corollary rephrases Theorem~\ref{thm-upbound-v2} with network parameters.

\begin{corollary}\label{cor:imupperbound-N}
Let $\alpha\geq 2$, $h,\ell,t\geq 1$, $\eps\geq 0$, and $h-\eps \geq 2\ell$. If $(\eps,\ell)-\mathcal{N}_{h,r,\alpha\ell+\eps}$ has a $(q,t)$-linear solution then
\begin{align*}
r \leq r_{\max}
< \gamma\theta q^{\ell t (\eps t+1)} +\alpha-\theta\ ,
\end{align*}
where $\theta\defeq \alpha- \floor{ \frac{h-\eps}{\ell} }+1$ and $\gamma\approx 3.48$.
\end{corollary}
\begin{proof}
If a $(q,t)$-linear solution exists, then each of the $r$ nodes in the middle layer gets a subspace of dimension $\ell t$ of the source messages space. Since all receivers are able to recover the entire source message space, every $\alpha$-subset of the middle nodes span a subspace of dimension at least $(h-\eps)t$. The statement then follows from Theorem~\ref{thm-upbound-v2}.
\end{proof}

Theorem~\ref{thm-upbound-v2} and Corollary~\ref{cor:imupperbound-N} are valid for all $\alpha\geq 2$. However, we derive a tighter upper bound for $\alpha = 2$, as shown in the following theorem.

\begin{theorem}\label{thm:imupbound-2}
Let $\alpha=2,h,\ell,t\geq 1$, $\eps\geq 0$, and let $\cS$ be a collection of  subspaces of $\F_q^{ht}$ such that
\begin{enumerate}
\item[(i)] each subspace has dimension at most $\ell t$, and
\item[(ii)] the sum of any two subspaces has dimension at least $(h-\eps)t$.
\end{enumerate}
Then, we have
\begin{align*}
  |\cS|  \leq \frac{\quadbinom{ht}{2\ell t - (h-\eps)t+1}_q}{ \quadbinom{\ell t}{2\ell t - (h-\eps)t+1}_q}
  < \gamma \cdot q^{(h-\ell)(2\ell+\eps-h)t^2+(h-\ell)t}\ .
\end{align*}
\end{theorem}

\begin{proof}
  We may assume that each subspace has dimension  $\ell t$. Since the sum of every two subspaces has dimension at least $(h-\eps)t$, their intersection has dimension at most $2\ell t - (h-\eps)t$. It follows that any subspace of dimension $2\ell t - (h-\eps)t+1$ is contained in at most one subspace of $\cS$. Note that there are $\quadbinom{ht}{2\ell t - (h-\eps)t+1}_q$ subspaces of dimension  $2\ell t - (h-\eps)t+1$  and each   subspace of dimension $\ell t$ contains $\quadbinom{\ell t}{2\ell t - (h-\eps)t+1}_q$ such spaces.
  We then have
  \begin{align*}
  |\cS| \leq \quadbinom{ht}{2\ell t - (h-\eps)t+1}_q \Big/ \quadbinom{\ell t}{2\ell t - (h-\eps)t+1}_q\ .
  \end{align*}
\end{proof}
The following corollary rephrases Theorem~\ref{thm:imupbound-2} with network parameters.

\begin{corollary}\label{cor:imupperbound-2-Network}
Let $\alpha= 2$, $h,\ell,t\geq 1$, $\eps\geq 0$. If $(\eps,\ell)-\mathcal{N}_{h,r,\alpha\ell+\eps}$ has a $(q,t)$-linear solution then
\begin{align*}
r \leq r_{\max}
< \gamma \cdot q^{(h-\ell)(2\ell+\eps-h)t^2+(h-\ell)t}\ ,
\end{align*}
where $\gamma\approx 3.48$.
\end{corollary}
\begin{proof}
If a $(q,t)$-linear solution exists, then each of the $r$ nodes in the middle layer gets a subspace of dimension $\ell t$ of the source messages space. Since all receivers are able to recover the entire source message space, any two subset of the middle nodes span a subspace of dimension at least $(h-\eps)t$. We then use Theorem~\ref{thm:imupbound-2}.
\end{proof}

\subsection{Lower Bounds on the Maximum Number of Middle Layer Nodes}
\label{sec:lb_rv}
We now turn to study a lower bound on $r_{\max}$ {with the parameters $\alpha,\ell,\eps,h$ being fixed}. The main results are summarized in Theorem~\ref{thm:LLL_bound} and Corollary~\ref{cor:EK19_lb}.

\subsubsection{A Lower Bound by Lov\'asz-Local Lemma}
\begin{lemma}[Lov\'asz-Local-Lemma~{\cite[Ch.~5]{TheProbMethod}}, \cite{LLLbeck1991}]
  \label{lem:LLL}
Let $\mathcal{E}_1, \mathcal{E}_2, \hdots,\mathcal{E}_k$ be a sequence of events. Each event occurs with probability at most $p$ and each event is independent of all the other events except for at most $d$ of them.
If $epd\leq 1$, where $e\approx 2.718$ is the base of natural logarithms, then there is a nonzero probability that none of the events occurs.
\end{lemma}

Recall that a solution to the $(\eps,\ell)-\mathcal{N}_{h,r,\alpha\ell+\eps}$ network is a set of the coding coefficients $\set*{\bA_1,\dots,\bA_r }$ such that~\eqref{eq:NC_sol} holds.
We choose the matrices $\bA_1,\dots,\bA_r \in \Fq^{\ell t \times ht}$ independently and uniformly at random.
For $1 \leq i_1 <  \dots<i_\alpha \leq r$, we define the event
\begin{align*}
\mathcal{E}_{i_1,\dots,i_\alpha} \defeq \set*{\left. (\bA_{i_1},\dots, \bA_{i_\alpha})\ \right|\ \rank \begin{pmatrix}
\bA_{i_1} \\ \vdots \\ \bA_{i_\alpha}
\end{pmatrix} < (h-\eps)t }\ .
\end{align*}

\begin{lemma}
  \label{lem:upper_bound_on_p}
Let $\alpha\geq 2$, $h,\ell,t\geq 1$, $\eps\geq 0$. Fixing $1 \leq i_1 < \dots < i_\alpha \leq r$, we have
\begin{align*}
\Pr(\mathcal{E}_{i_1,\dots,i_{\alpha}}) \leq 2\gamma \cdot q^{(h-\alpha\ell-\eps)\eps t^2+(h-\alpha\ell-2\eps)t-1}\ ,
\end{align*}
where $\gamma\approx3.48$.
\end{lemma}
\begin{proof}
The number of matrices $\bA\in\mathbb{F}_q^{m\times n}$ of rank $s$ is
\begin{equation}
  M(m,n,s):=\prod\limits^{s-1}_{j=0}\frac{(q^m-q^j)(q^n-q^j)}{q^s-q^j}
  \leq \gamma\cdot q^{(m+n)s-s^2} \label{eq:number_matrices_upper_bound}.
\end{equation}
Then,
\begin{align}
  \Pr(\mathcal{E}_{i_1,\dots,i_{\alpha}}) &=\frac{\sum\limits^{(h-\eps)t-1}_{i=0}M(\alpha \ell t,ht,i)}{q^{\alpha \ell h t^2}}\nonumber\\
    & \leq \frac{\sum\limits^{(h-\eps)t-1}_{i=0}\gamma \cdot q^{(h+\alpha\ell)ti-i^2}}{q^{\alpha\ell h t^2}}\label{eq:NM_upper_bound} \\
    & \leq \gamma \cdot \frac{q}{q-1}\cdot q^{\max_i\{(h+\alpha\ell)ti-i^2\} - \alpha\ell h t^2} \label{eq:summation_upper_bound}\\
    & = \gamma \cdot \frac{q}{q-1}\cdot q^{(h+\alpha\ell)ti-i^2|_{i=(h-\eps)t-1} - \alpha\ell h t^2} \label{eq:quadratic_maximum}\\
    & \leq \gamma\cdot 2\cdot q^{(h-\alpha\ell-\eps)\eps t^2+(h-\alpha\ell-2\eps)t-1} \ ,\nonumber
\end{align}
where (\ref{eq:NM_upper_bound}) holds due to~\eqref{eq:number_matrices_upper_bound}, (\ref{eq:summation_upper_bound}) follows from a geometric sum, and \eqref{eq:quadratic_maximum} follows by maximizing $(h+\alpha\ell)ti-i^2$.
\end{proof}

\begin{lemma}\label{lem:upper_bound_on_d}
Let $\alpha\geq 2$, $h,\ell,t\geq 1$, $\eps\geq 0$. Fixing $1 \leq i_1 <\dots < i_\alpha \leq r$, the event $\mathcal{E}_{i_1,\dots,i_\alpha}$ is statistically independent of all the other events $\mathcal{E}_{i_1',\dots,i_\alpha'}$ ($1 \leq i_1' < \dots < i_\alpha' \leq r$), except for at most $\alpha\binom{r-1}{\alpha-1}$
of them.
\end{lemma}

\begin{proof}
  For $1 \leq i_1 < \dots < i_\alpha \leq r$ and $1 \leq i_1' <\dots < i_\alpha' \leq r$, the events $\mathcal{E}_{i_1,\dots,i_\alpha}$ and $\mathcal{E}_{i_1',\dots,i_\alpha'}$ are statistically independent if and only if $\{i_1,\dots,i_\alpha\} \cap \{i_1',\dots,i_\alpha'\} = \emptyset$. Thus, having chosen $1 \leq i_1 < \dots < i_\alpha \leq r$, there are
  at most
  $\binom{\alpha}{1}\binom{r-1}{\alpha-1}$
  ways of choosing $\{i_1',\dots,i_\alpha'\}$ such that it is not independent from $\{i_1,\dots,i_\alpha\} $ (including the case $\{i_1',\dots,i_\alpha'\}=\{i_1,\dots,i_\alpha\}$).
\end{proof}
\begin{remark}
Lemma~\ref{lem:upper_bound_on_d} is a union-bound argument on the number of dependent events. The exact number is $\binom{r}{\alpha}-\binom{r-\alpha}{\alpha}$.
However the exact expression makes it harder to resolve for $r$ later thus we use the bound instead.
\end{remark}
\begin{theorem}\label{thm:LLL_bound}
  Let $\alpha\geq 2$, $\eps \geq 0$, $\ell,t \geq 1$, and $1\leq h \leq \alpha\ell+\eps$ be fixed integers.
  If
  \begin{equation}\label{eq:r_LLL_lb}
      r \leq \beta \cdot
      q^{\frac{f(t)}{\alpha-1}}\ ,
  \end{equation}
where $\beta := \parenv*{\frac{(\alpha-1)!}{2e\gamma\alpha}}^{\frac{1}{\alpha-1}}, \gamma\approx3.48$ and $f(t):=(\alpha\ell+\eps-h)\eps t^2+(\alpha\ell+2\eps-h)t +{1}$, then $(\eps,\ell)-\mathcal{N}_{h,r,\alpha\ell+\eps}$ has a $(q,t)$-linear solution.

Namely, for an $(\eps,\ell)-\mathcal{N}_{h,r,\alpha\ell+\eps}$ that has a $(q,t)$-linear solution, the maximum number of middle nodes satisfies
\begin{align*}
  r_{\max}\geq \beta \cdot q^{\frac{f(t)}{\alpha-1}}\ .
\end{align*}
\end{theorem}

\begin{proof}
Let $p=\Pr(\mathcal{E}_{i_1,\dots,i_\alpha})$ and denote by $d$ the number of other events $\mathcal{E}_{i'_1,\dots,i'_\alpha}$ that are dependent on $\mathcal{E}_{i_1,\dots,i_\alpha}$. We have shown that $p \leq 2\gamma \cdot q^{(h-\alpha\ell-\eps)\eps t^2+(h-\alpha\ell-2\eps)t-1}$ in \cref{lem:upper_bound_on_p} and $d\leq \alpha\binom{r-1}{\alpha-1}$ in \cref{lem:upper_bound_on_d}.
By the Lov\'asz Local Lemma, it suffices to show that $epd\leq 1$.
Noting that $d\leq \alpha\binom{r-1}{\alpha-1}\leq \alpha\cdot \frac{(r-1)^{\alpha-1}}{(\alpha-1)!}$, we shall require
\begin{align*}
    e\cdot2\gamma q^{(h-\alpha\ell-\eps)\eps t^2+(h-\alpha\ell-2\eps)t-1}
    \cdot \alpha\frac{(r-1)^{\alpha-1}}{(\alpha-1)!}\leq 1\ .
\end{align*}

Namely, if
$ r \leq {\beta}\cdot
q^{\frac{(\alpha\ell+\eps-h)\eps}{\alpha-1}t^2+\frac{\alpha\ell+2\eps-h}{\alpha-1}t+\frac{1}{\alpha-1}}+1$, then $(\eps,\ell)-\mathcal{N}_{h,r,\alpha\ell+\eps}$ has a $(q,t)$-linear solution.
We omit the plus one for simplicity.
\end{proof}
\begin{remark}
 For any $\alpha\geq 7$,
 \eqref{eq:r_LLL_lb} can be simplified to
 \begin{equation*}
    r \leq
    q^{\frac{f(t)}{\alpha-1}}\ ,
 \end{equation*}
 since the prefactor $\beta>1$ for all $\alpha\geq 7$.
\end{remark}

\begin{remark}
 For $t\geq 3$, $\alpha \geq 5$ or $q\geq 4$, it can be seen from numerical analysis that
 $\beta \cdot q^{\frac{\alpha\ell+2\eps-h}{\alpha-1}t +\frac{1}{\alpha-1}} \geq 1$. Thus, \eqref{eq:r_LLL_lb} can be simplified to a looser upper bound
 \begin{align*}
    r \leq q^{\frac{(\alpha\ell+\eps-h)\eps}{\alpha-1}t^2}\ .
 \end{align*}
 However, omitting the term $\beta \cdot q^{\frac{\alpha\ell+2\eps-h}{\alpha-1}t+\frac{1}{\alpha-1}} $ will cause a loss in estimating the maximum achievable number of middle nodes. Nevertheless, the loss is negligible when $t\to \infty$.
\end{remark}

\subsubsection{A Lower Bound by $\alpha$-Covering Grassmannian Codes}
\begin{definition}[Covering Grassmannian Codes~\cite{EZjul2019}]
  An $\alpha$-$(n,k,\delta)_q^c$ \emph{covering Grassmannian code} $\cC$ is a subset of $\mathcal{G}_q(n,k)$ such that each subset with $\alpha$ codewords of $\cC$ spans a subspace whose dimension is at least $\delta+k$ in $\mathbb{F}_q^n$.
\end{definition}
The following theorem from~\cite{EZjul2019} shows the connection between covering Grassmannian codes and linear network coding solutions.
\begin{theorem}[{\cite[Thm.~4]{EZjul2019}}]\label{thm:cover_scalar_sol}
  The $(\eps,\ell)-\mathcal{N}_{h,r,\alpha\ell+\eps}$ network is solvable with a $(q,t)$-linear solution if and only if there exists an $\alpha$-$(ht,\ell t,ht-\ell t-\eps t)_{q}^c$ code with $r$ codewords.
\end{theorem}

Let ${\cal B}_q(n,k,\delta;\alpha)$ denote the maximum possible size of an $\alpha$-$(n,k,\delta)_q^c$ covering Grassmannian code.
Let $\bA$ be a $k\times (n-k)$ matrix, and let $\bI_k$ be a $k\times k$ identity matrix. The matrix $[\bI_k\ \bA]$ can be viewed as a generator matrix of a $k$-dimensional subspace of $\Fq^{n}$, and it is called the \emph{lifting} of $\bA$.
When all the codewords of an MRD code $\cC$ are lifted to $k$-dimensional subspaces, the result is
called \emph{lifted MRD code}, denoted by $\cC^{\mathrm{lifted}}$.
\begin{theorem}\label{thm:EK_1_ext}
Let $n,k,\delta$ and $\alpha$ be positive integers such that $1\leq \delta \leq k$, $\delta+k\leq n$ and $\alpha\geq 2$. Then
$${\cal B}_q(n,k,\delta;\alpha) \geq (\alpha -1) q^{\max\{k,n-k\}(\min\{k,n-k\}-\delta+1)}\ .$$
\end{theorem}

\begin{proof} Let $m=n-k$ and $K=\max\{m,n-m\}(\min\{m,n-m\}-\delta+1)$.
Since $\delta \leq \min\{m,n-m\}$,  an $[m\times (n-m), K, \delta]_q$ MRD code $\cC$ exists.
Let $\cC^{\mathrm{lifted}}$ be the lifted code of $\cC$. Then $\cC^{\mathrm{lifted}}$ is a subspace code of $\Fq^n$, which contains $q^K$ $m$-dimensional subspaces as codewords and its minimum subspace distance is $2\delta$~\cite{silva2008rank}.

Hence, for any two different codewords $C_1,C_2 \in \cC^{\mathrm{lifted}}$ we have
$$\dim (C_1 \cap C_2)\leq m-\delta\ .$$

Now, let $\cD\defeq\set*{C^{\perp} \ |\ C \in \cC^{\mathrm{lifted}}}$. Take $\alpha -1$ copies of $\cD$ and denote their multiset union by $\cD^{\alpha-1}$.
We show that $\cD^{\alpha-1}$ is an $\alpha$-$(n,k,\delta)_q^c$ covering Grassmannian code.
Each subspace $D\in\cD^{\alpha-1}$ has dimension $n-m$, since it is the dual of a codeword in $\cC^{\mathrm{lifted}}$.
For any $\alpha$ subspaces $D_1,D_2,\ldots, D_\alpha\in\cD^{\alpha-1}$, there exist $1\leq i < j\leq \alpha$ such that $D_i \not= D_j$. Let $C_i=D_i^{\perp}$ and $C_j=D_j^\perp$. By definition, $C_i$ and $C_j$ are two distinct codewords of $\cC^{\mathrm{lifted}}$. We then have
\begin{align*}
    \dim\parenv*{\sum_{\ell=1}^\alpha D_\ell} & \geq \dim \parenv*{D_i +D_j }
    = n - \dim \parenv*{D_i^{\perp} \cap D_j^{\perp}}\\
    & = n - \dim \parenv*{C_i \cap C_j}
    \geq n-m+\delta = k+\delta\ .
\end{align*}
So far we have shown that $\cD^{\alpha-1}$ is an $\alpha$-$(n,k,\delta)_q^c$ covering Grassmannian code. Then the statement follows by
\begin{align*}
 {\cal B}_q(n,k,\delta;\alpha)\geq|\cD^{\alpha-1}|=  (\alpha-1)|\cD|=(\alpha-1)|\cC^{\mathrm{lifted}}|
=  (\alpha-1) q^{\max\{k,n-k\}(\min\{k,n-k\}-\delta+1)}\ .
\end{align*}
\end{proof}

The following corollary rephrases \cref{thm:cover_scalar_sol} using the result in \cref{thm:EK_1_ext}.


\begin{corollary}\label{cor:EK19_lb}
Let $\alpha\geq 2$, $h,\ell,t\geq 1$, $\eps\geq 0$, $h\leq 2\ell+\eps$.
For an $(\eps,\ell)-\mathcal{N}_{h,r,\alpha\ell+\eps}$ which has a $(q,t)$-linear solution, the maximum number of middle nodes is
\begin{align*}
  r_{\max}\geq (\alpha-1)q^{g(t)}\ ,
\end{align*}
where
 \begin{align*}
     g(t)&:=\max\{\ell t,(h-\ell)t\}
     \cdot(\min\{\ell t, (h-\ell)t\}-(h-\ell-\eps)t+1) \\
     &=\begin{cases}\ell\eps t^2 +\ell t & h\leq 2\ell \\ (h-\ell)(2\ell+\eps-h)t^2+(h-\ell)t & \text{otherwise} \end{cases}\ .
 \end{align*}
\end{corollary}

\subsection{Bounds on the Gap on Field Size}
\label{sec:bound_gap}
In the last section, we presented bounds on $r_{\max}$. The main results in this section are the upper and lower bounds on $\gap(\cN)$ in \cref{thm:gap_ub} and \cref{thm:gap_lb}, respectively.
To discuss $\gap (\cN)$, we first need the following conditions on the smallest field size $q_s(\cN)$ or $q_v(\cN)$, under which a network $\cN$ is solvable.

\begin{lemma}\label{lem:lb_q}
Let $\alpha\geq 2$, $r,h,\ell,t\geq 1$, $\eps\geq 0$. If $(\eps,\ell)-\mathcal{N}_{h,r,\alpha\ell+\eps}$ has a $(q,t)$-linear solution then
    \begin{align*}
        q^t \geq
        \begin{cases}
        \parenv*{ \frac{r+\theta-\alpha}{\gamma\cdot\theta}}^{\frac{1}{\ell (\eps t+1)}} & h\geq 2\ell+\eps\\
        \parenv*{ \frac{r}{\gamma(\alpha-1)}}^{\frac{1}{\ell(\eps t+1)}} & \text{otherwise}
        \end{cases}\ ,
    \end{align*}
    where $\theta := \alpha-\floor{\frac{h-\eps}{\ell}}+1$ and $\gamma\approx 3.48$.
\end{lemma}
\begin{proof}
    The first case follows from Corollary~\ref{cor:imupperbound-N} that for $h\geq 2\ell+\eps$,
    $q^t\geq \parenv*{ \frac{r+\theta-\alpha}{\gamma\cdot\theta}}^{\frac{1}{\ell (\eps t+1)}}$. The second case is derived from an upper bound on $r$ in~\cite{EZjul2019} (recalled in Corollary~\ref{cor:EZ19_vector}) in a similar manner.
\end{proof}

\begin{lemma}\label{lem:ub_q}
Let $\alpha\geq 2$, $r,h,\ell,t\geq 1$, $\eps\geq 0$. There exists a $(q,t)$-linear solution to $(\eps,\ell)-\mathcal{N}_{h,r,\alpha\ell+\eps}$ when
\[
    q^t \geq
    \begin{cases}
    \parenv*{\frac{r}{\beta}}^{\frac{(\alpha-1)t}{f(t)}} & h\geq2\ell+\eps \\
    \parenv*{\frac{r}{\alpha-1}}^{\frac{t}{g(t)}}   &\text{otherwise}
    \end{cases}\ ,
    \]
    where $\beta$ and $f(t)$ are defined as in Theorem~\ref{thm:LLL_bound}, and $g(t)$ is defined as in Corollary~\ref{cor:EK19_lb}.
\end{lemma}
\begin{proof}
    The proof is similar to that in Lemma~\ref{lem:lb_q} and the cases follow from Theorem~\ref{thm:LLL_bound} and Corollary~\ref{cor:EK19_lb}, respectively.
\end{proof}
\cref{lem:lb_q,lem:ub_q} can be seen as the necessary and the sufficient conditions respectively on the pair $(q,t)$ such that~a $(q,t)$-linear solution exists.

In the following, we use the lemmas above to derive  bounds on the $\gap(\cN)$ for a given network $\cN$. The bounds are determined only by the network parameters.

\begin{theorem}\label{thm:gap_ub}
 Let $\alpha\geq 2$, $r,h,\ell\geq 1$, $\eps\geq 0$. Then for the $(\eps,\ell)-\mathcal{N}_{h,r,\alpha\ell+\eps}$ network,
  \begin{align*}
      \gap(\cN)\leq
      \begin{cases}
      \frac{\alpha-1}{f(1)}\log_2\parenv*{\frac{r}{\beta}}-A & h\geq 2\ell+\eps\\
      \frac{1}{g(1)}\log_2\parenv*{\frac{r}{\alpha-1}}- B & \text{otherwise}
      \end{cases}\ ,
  \end{align*}
where $\theta = \alpha-\floor{\frac{h-\eps}{\ell}}+1$, $\beta$ and $f(t)$ are defined as in Theorem~\ref{thm:LLL_bound}, $g(t)$ is defined as in Corollary~\ref{cor:EK19_lb}, and
\begin{align*}
A &\defeq \min \set*{\left. \log_2\parenv*{q^t} \ \right|\ q^t \geq \parenv*{\tfrac{r+\theta-\alpha}{\gamma \theta}}^{\frac{1}{\ell(\varepsilon t + 1)}} }\ ,\\
B &\defeq \min \set*{\left. \log_2\parenv*{q^t} \ \right|\ q^t \geq \parenv*{\tfrac{r}{\gamma (\alpha-1)}}^{\frac{1}{\ell(\varepsilon t + 1)}} }\ .
\end{align*}
Furthermore, for $t_A \defeq \min \set*{ t \ |\ 2^t \geq \parenv*{\tfrac{r+\theta-\alpha}{\gamma \theta}}^{\frac{1}{\ell(\varepsilon t + 1)}} } > 2$, we have
\begin{align*}
A \geq \min\set*{t_A,\frac{1}{\ell(\varepsilon (t_A-2) + 1)} \log_2\parenv*{\tfrac{r+\theta-\alpha}{\gamma \theta}}} \geq t_A-1\ ,
\end{align*}
and for $t_B \defeq \min \set*{ t \ |\  2^t \geq \parenv*{\tfrac{r}{\gamma (\alpha-1)}}^{\frac{1}{\ell(\varepsilon t + 1)}} } > 2$, we have
\begin{align*}
B \geq \min\set*{t_B,\frac{1}{\ell(\varepsilon (t_B-2) + 1)} \log_2\parenv*{\tfrac{r+\theta-\alpha}{\gamma \theta}}} \geq t_B-1\ .
\end{align*}

\end{theorem}

\begin{proof}
We only prove the bound for the case $h\geq 2\ell+\varepsilon$.
The other case follows analogously.
Lemma~\ref{lem:ub_q} implies that
\[q_s(\cN)\leq \parenv*{\frac{r}{\beta}}^{\frac{\alpha-1}{f(1)}}\ .\]
By the definition of $q_v(\cN)$ and Lemma~\ref{lem:lb_q}, $q^t = q_v(\cN)$ must fulfill
\begin{align}
  \label{eq:qt-lb}
q^t \geq \parenv*{\tfrac{r+\theta-\alpha}{\gamma \theta}}^{\frac{1}{\ell(\varepsilon t + 1)}}\ .
\end{align}
Hence, we get a lower bound on $q_v(\cN)$ by determining the smallest $q^t$ that fulfills \eqref{eq:qt-lb}, i.e., the constraint in $A$.
Note that the left-hand side of the inequality is a strictly monotonically increasing function in $t$ (for a fixed prime power $q$), and the right side is monotonically decreasing in $t$, which imply
that $A$ and $t_A$ are well-defined.

For the lower bound on $A$ for $t_A>2$, consider the case that there is a prime power $q>2$ and a positive integer $t$ with $2^{t_A} \geq q^t \geq
\parenv*{\tfrac{r+\theta-\alpha}{\gamma \theta}}^{\frac{1}{\ell(\varepsilon t + 1)}}$.
Then we have $t \leq t_A-2$ since $q \geq 3$ and $t_A \geq 3$.
Hence,
\begin{align*}
q^t \geq \parenv*{\tfrac{r+\theta-\alpha}{\gamma \theta}}^{\frac{1}{\ell(\varepsilon (t_A-2) + 1)}} \geq \parenv*{\tfrac{r+\theta-\alpha}{\gamma \theta}}^{\frac{1}{\ell(\varepsilon (t_A-1) + 1)}} \geq 2^{t_A-1}\ ,
\end{align*}
which proves the claim.
\end{proof}

\begin{corollary}\label{cor:gap_ub}
Let $\alpha\geq 2$, $r,h,\ell\geq 1$, $\eps\geq 1$. Then for the $(\eps,\ell)-\mathcal{N}_{h,r,\alpha\ell+\eps}$ network,
  \begin{align*}
      \gap(\cN)\leq
      \begin{cases}
      \frac{\alpha-1}{f(1)}\log_2\parenv*{\frac{r}{\beta}}- \max\set*{\sqrt{\frac{1}{\ell\varepsilon}\log_2\parenv*{\frac{r+\theta-\alpha}{\gamma\theta}}+\frac{1}{4 \varepsilon^2}}-\frac{2\varepsilon+1}{2 \varepsilon},1} & h\geq 2\ell+\eps\\
      \frac{1}{g(1)}\log_2\parenv*{\frac{r}{\alpha-1}}- \max\set*{\sqrt{\frac{1}{\ell\varepsilon}\log_2\parenv*{\tfrac{r}{\gamma (\alpha-1)}}+\frac{1}{4 \varepsilon^2}}-\frac{2\varepsilon+1}{2 \varepsilon},1} & \text{otherwise}
      \end{cases}\ .
  \end{align*}
In particular, if all parameters are constants except for $r\to\infty$, then $\gap(\cN)\in O(\log r)$.
\end{corollary}
\begin{proof}
We only prove the bound for the case $h\geq 2\ell+\varepsilon$.
The other case follows analogously.
We determine $t_A$ as defined in Theorem~\ref{thm:gap_ub}.
Note that $2^t$ is strictly monotonically increasing in $t$ and $\parenv*{\tfrac{r+\theta-\alpha}{\gamma \theta}}^{\frac{1}{\ell(\varepsilon t + 1)}}$ is strictly monotonically decreasing.
Hence, we have $t_A = \lceil t' \rceil$, where $t'$ is the unique (positive) solution of
\[2^{t'} = \parenv*{\tfrac{r+\theta-\alpha}{\gamma \theta}}^{\frac{1}{\ell(\varepsilon t' + 1)}}\ .\]
By rewriting this equation into a quadratic equation in $t'$, we obtain the following positive solution for $\varepsilon>0$:
\begin{equation*}
t' = \sqrt{\frac{1}{\ell\varepsilon}\log_2\parenv*{\frac{r+\theta-\alpha}{\gamma\theta}}+\frac{1}{4 \varepsilon^2}}-\frac{1}{2 \varepsilon}.
\end{equation*}
Using the bound $A \geq t_A-1$ for $t_A >2$ (Theorem~\ref{thm:gap_ub}) and the trivial bound $A \geq 1$ otherwise, the claim follows.
The asymptotic statement is an immediate consequence.
\end{proof}

\begin{theorem}
\label{thm:gap_lb}
  Let $\alpha\geq 2$, $r,h,\ell\geq 1$, $\eps\geq 0$. Then for the $(\eps,\ell)-\mathcal{N}_{h,r,\alpha\ell+\eps}$ network,
  \begin{align*}
      \gap(\cN)\geq
      \begin{cases}
      \frac{1}{\ell(\eps+1)}\log_2\parenv*{\frac{r+\theta-\alpha}{\gamma\theta}}-t_{\Delta} & h\geq 2\ell+\eps\\
      \frac{1}{\ell(\eps+1)}\log_2\parenv*{\frac{r}{\gamma(\alpha-1)}}-t_{\star} & \text{otherwise}
      \end{cases}\ ,
  \end{align*}
  where $t_{\Delta}$ is the smallest positive integer such that~$2^{\frac{f(t_{\Delta})}{\alpha-1}}\geq \frac{r}{\beta}$ and $t_{\star}$ is the smallest positive integer such that~$2^{g(t_{\star})}\geq \frac{r}{\alpha-1}$.
  Here, $\beta$ and $f(t)$ are defined as in Theorem~\ref{thm:LLL_bound}, and $g(t)$ is defined as in Corollary~\ref{cor:EK19_lb}.
\end{theorem}
\begin{proof}
 Let us only consider the first case $h\geq 2\ell+\eps$. The other case can be proved in the same manner. According to Lemma~\ref{lem:lb_q}, we have the lower bound on the smallest field size of a scalar solution,
 \[q_s(\cN)\geq \parenv*{ \frac{r+\theta-\alpha}{\gamma\cdot\theta}}^{\frac{1}{\ell (\eps +1)}}\ .\]
For vector solutions, according to Lemma~\ref{lem:ub_q}, we want to find $(q,t)$ such that~$q^{\frac{f(t)}{\alpha-1}}\geq \frac{r}{\beta}.$
Since $t_{\Delta}$ is the smallest positive integer $t$ such that~$2^{\frac{f(t_{\Delta})}{\alpha-1}}\geq \frac{r}{\beta}$, it is guaranteed that a $(2,t_{\Delta})$-linear solution exists.
Therefore, $q_v(\cN)$ (the smallest value of $q^t$) should be at most
$q_v(\cN)\leq 2^{t_{\Delta}}$.
The lower bound then follows directly from the definition of $\gap(\cN)$.
\end{proof}

By carefully bounding $t_{\star}$ and $t_{\Delta}$, the following result is obtained.

\begin{corollary}
\label{cor:gap}
Let $\alpha\geq 2$, $r,h,\ell,\eps\geq 1$. Then, for the $(\eps,\ell)-\mathcal{N}_{h,r,\alpha\ell+\eps}$ network,
\begin{align*}
\gap(\cN) \geq
\begin{cases}
\frac{\log_2\parenv*{\frac{r+\theta-\alpha}{\gamma \theta}}}{\ell(\eps+1)} - \sqrt{\frac{(\alpha-1)\log_2(\frac{r}{\beta})}{(\alpha\ell+\eps-h)\eps}}  & h\geq 2\ell+\eps,\\
\frac{\log_2\parenv*{\frac{r}{\alpha-1}}-2}{\ell(\eps+1)}-\sqrt{\frac{\log_2(\frac{r}{\alpha-1})}{\ell\eps}}& \text{otherwise.}
\end{cases}
\end{align*}
In particular, if all parameters are constants except for $r\to\infty$, then $\gap(\cN) \in \Omega(\log r)$.
\end{corollary}
\begin{proof}
When $h \geq 2\ell+\eps$, noting that $\alpha \ell+2\eps-h>0$, we may choose
$$t=\parenv*{\frac{(\alpha-1)\log_2(\frac{r}{\beta})}{(\alpha\ell+\eps-h)\eps} }^{1/2}\ ,$$
such that $2^{f(t)}\geq 2^{(\alpha\ell+\eps-h)\eps t^2}=(\frac{r}{\beta})^{\alpha-1}.$
Then we have that
\begin{align*}
\gap(\cN)\geq &
\frac{\log_2\parenv*{\frac{r+\theta-\alpha}{\gamma \theta}}}{\ell(\eps+1)} - \parenv*{\frac{(\alpha-1)\log_2(\frac{r}{\beta})}{(\alpha\ell+\eps-h)\eps} }^{1/2}
\\ \geq & \frac{\log_2 (r+\theta-\alpha)-\log_2\theta-2}{\ell(\eps+1)} -\parenv*{\frac{\log_2 r-\log_2 \beta}{ (\ell-\frac{h-\ell-\eps}{\alpha-1})\eps}}^{1/2}\ .
\end{align*}
Recall that $\beta$ and $\theta$ are determined by $\alpha, h, \eps,$ and $\ell$. Thus, if $\alpha, h, \eps,$ and $\ell$ are fixed,  $\gap(\cN)\in\Omega(\log r)$.

When $h< 2\ell+\eps$, we may choose $$t=\parenv*{\frac{\log_2(\frac{r}{\alpha-1})}{\ell\eps}}^{1/2}\ ,$$
such that $2^{g(t)}\geq 2^{\ell \eps t^2}=\frac{r}{\alpha-1}.$ It follows that
\begin{align*}
\gap(\cN)&\geq \frac{\log_2\parenv*{\frac{r}{\gamma(\alpha-1)}}}{\ell(\eps+1)}-\parenv*{\frac{\log_2(\frac{r}{\alpha-1})}{\ell\eps}}^{1/2}\\
& \geq \frac{\log_2\parenv*{\frac{r}{\alpha-1}}-2}{\ell(\eps+1)}-\parenv*{\frac{\log_2(\frac{r}{\alpha-1})}{\ell\eps}}^{1/2}\ .
\end{align*}
This shows that $\gap(\cN) \in \Omega(\log r)$.
\end{proof}

\cref{cor:gap_ub} and \cref{cor:gap} show that for fixed network parameters,
the gap size grows as
\begin{equation*}
\gap(\cN) = \Theta(\log r) \quad (r \to \infty)\ .
\end{equation*}

\begin{example}
  We illustrate the proof of Theorem~\ref{thm:gap_ub} and Theorem~\ref{thm:gap_lb} by two network examples with $r=8\times 10^5$ in Figure~\ref{fig:q_bound} and $r=8\times 10^6$ in Figure~\ref{fig:q_bound2}. Note that the curves in the figures are not bounds on the gap size.
  They are the necessary (blue curve) and the sufficient (green curve) condition on $q^t$ such that a $(q,t)$-linear solution exists.
  Namely, there is no $(q,t)$-linear solution in the region below the blue curve and there must be a $(q,t)$-linear solution in the region above the green curve.
  Thus, the \emph{minimum} gap of the network $(2,1)-\cN_{12,r,20}$ is determined by the difference between the necessary condition with $t=1$ and the minimum $2^t$ that is in the region above the sufficient condition.
  Similarly, the \emph{maximum} gap of the network is determined by the difference between the sufficient condition with $t=1$ and the minimum $2^t$ that is in the region above the necessary condition.

  By comparing the two plots it can be seen that the gap increases as the number of middle node in the network increases.
\begin{figure}[!htb]
\begin{subfigure}{.5\textwidth}
    \centering
    \input{./figs/VecNetCod/h12l1e2r8e5a20}
    \caption{For the network $(2,1)-\cN_{12,8e5,20}$.}
    \label{fig:q_bound}
\end{subfigure}
\begin{subfigure}{.5\textwidth}
    \input{./figs/VecNetCod/h12l1e2r8e6a20}
    \caption{For the network $(2,1)-\cN_{12,8e6,20}$.}
    \label{fig:q_bound2}
\end{subfigure}
\caption{An illustration of proofs of Theorem~\ref{thm:gap_ub} and Theorem~\ref{thm:gap_lb}.}
\label{fig:gap-bound}
\end{figure}
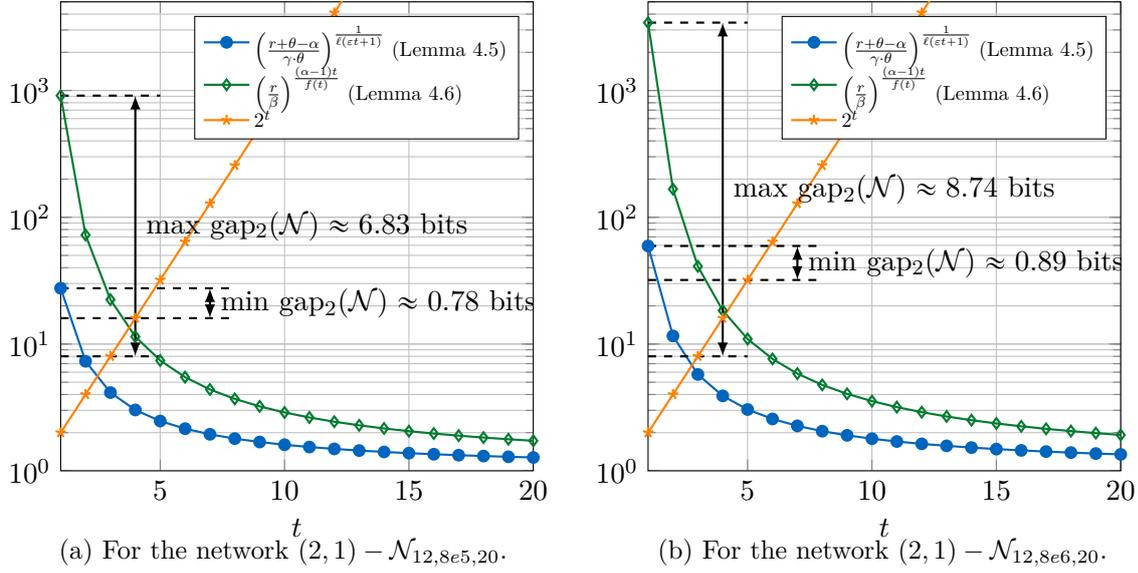
\end{example}

\subsection{Comparisons of Bounds on $r_{\max}$}\label{sec:discussion}
In the following we compare our upper and lower bound on $r_{\max}$ with previously known bounds.
\subsubsection{Other Upper Bound on $r_{\max}$}
We recall the result from~\cite[Corollary 3]{EZjul2019} and compare it with our upper bound in Corollary~\ref{cor:imupperbound-2-Network}.

\begin{theorem}[{\cite[Corollary~3]{EZjul2019}}]\label{thm:codesize}
  If $n$, $k$, $\delta$, and $\alpha$, are positive integers such that $1<k<n,\ 1\leq\delta\leq n-k$ and $2\leq\alpha\leq\quadbinom{k+\delta-1}{k}_{q}+1$, then for an $\alpha-(n,k,\delta)_q^c$ {covering Grassmannian code} $\cC$, we have
$$|\cC|\leq\floor{(\alpha-1)\frac{\quadbinom{n}{\delta+k-1}_{q}}{{\quadbinom{n-k}{\delta-1}_{q}}}}\ .$$
\end{theorem}

By combining Theorem~\ref{thm:codesize} and Theorem~\ref{thm:cover_scalar_sol}, the following corollary can be derived.
\begin{corollary}\label{cor:EZ19_vector}
  If the $(\eps,\ell)-\mathcal{N}_{h,r,\alpha\ell+\eps}$ network has a $(q,t)$-linear solution, then
  \begingroup
  \allowdisplaybreaks
  \begin{align*}
    r\leq r_{\max}&\leq\floor{(\alpha-1)\frac{\quadbinom{ht}{ht-\eps t-1}_{q}}{{\quadbinom{ht-\ell t}{ht -\ell t-\eps t -1}_{q}}}}\\
    & < (\alpha-1)\frac{\gamma q^{(\eps t+1)(ht -\eps t-1)}}{q^{(\eps t+1)(ht -\ell t -\eps t-1)}}\\
    & =  \gamma (\alpha-1)q^{\ell t(\eps t+1)}\ ,
  \end{align*}
  \endgroup
  with $1<\ell t<h t$, $0\leq \eps\leq h-\ell-\frac{1}{t}$, $2\leq \alpha\leq\quadbinom{ht-\eps t-1}{\ell t}_{q}+1$.\\
\end{corollary}
\subsubsection{Comparison Between the Upper Bounds}
We first show that for some parameters, the upper bound in Corollary~\ref{cor:imupperbound-N} can be tighter than that in Corollary~\ref{cor:EZ19_vector}.
The upper bounds in Corollary~\ref{cor:imupperbound-N} and Corollary~\ref{cor:EZ19_vector} can be respectively written as
$$U_A := \quadbinom{(\eps+\ell) t}{\eps t}_q \parenv*{\theta \cdot \frac{q^{\ell t+1}-1}{q-1}  -1   }+ \alpha -\theta\ ,$$
where $\theta = \parenv*{\alpha- \floor{ \frac{h-\eps}{\ell} }+1}$,
and
$$U_B:= (\alpha-1)\frac{\quadbinom{ht}{ht-\eps t-1}_{q}}{{\quadbinom{ht-\ell t}{ht -\ell t-\eps t -1}_{q}}}=(\alpha-1)q^{\ell t(\eps t +1)}\prod\limits^{\eps t}_{i=0}\frac{q^{ht-i}-1}{q^{ht-i}-q^{\ell t}}\ .$$

\begin{lemma}\label{lem:compare-2}
 Let $h \geq 2\ell+\eps$ and $2\leq \alpha \leq \quadbinom{ht-\eps t-1}{\ell t}_q+1$. Assume $\quadbinom{\eps+\ell)t}{\eps t}_q\leq \alpha$,
then
\[
  \log_q U_A - \log_q U_B < \log_q\frac{2\theta\alpha}{\alpha-1}- \ell \eps t^2\ .
\]
Particularly, if $\frac{2\theta \alpha}{\alpha-1}\leq q^{\ell\eps t^2}$,
then
$U_A<U_B$ (the upper bound in Corollary~\ref{cor:imupperbound-N} is tighter than that in  Corollary~\ref{cor:EZ19_vector}).
 \end{lemma}

\begin{proof}
Under the assumption $\quadbinom{(\eps+\ell)t }{\eps t}_q\leq {\alpha}$, we have
  \begingroup
  \allowdisplaybreaks
\begin{align*}
\log_q U_A & \leq \log_q \parenv*{\alpha \parenv*{ \theta\cdot \frac{q^{\ell t+1}-1}{q-1}  -1   }+ \alpha-\theta } \\
&=\log_q\parenv*{ \alpha \theta \cdot\frac{q^{\ell t+1}-1}{q-1} -\alpha+\alpha-\theta}\\
&=\log_q \theta+\log_q\parenv*{\alpha \cdot \frac{q^{\ell t+1}-1}{q-1} -1}\\
&<\log_q \theta+\log_q\parenv*{\alpha \cdot \frac{q^{\ell t+1}-1}{q-1}}\\
&\overset{(*)}{<}\log_q \theta +\log_q \alpha + \log_q\parenv*{2\cdot q^{\ell t}}\\
& = \log_q \theta +\log_q \alpha +  \ell t + \log_q 2\ .
\end{align*}
\endgroup
The inequality $(*)$ holds because $\frac{q^{\ell t+1}-1}{q-1}=\sum\limits^{\ell t}_{i=0}q^i<2\cdot q^{\ell t}$.
With the bounds on the $q$-binomial coefficient in \eqref{eq:gauss}, we have
\begin{align*}
  \log_q U_B >  \log(\alpha-1) + \ell t (\eps t+1)\ ,
\end{align*}
and therefore
\begin{align*}
  \log_q U_A - \log_q U_B < \log_q\frac{2\theta\alpha}{\alpha-1}- \ell \eps t^2\ .
\end{align*}
Together with the assumption $\frac{2\theta \alpha}{\alpha-1}\leq q^{\ell\eps t^2}$, the statement follows.
 \end{proof}

 \begin{lemma}\label{lem:compare}Let $h \geq 2\ell+\eps$ and $2\leq\alpha \leq \quadbinom{ht-\eps t-1}{\ell t}_q+1$.
Assume $\quadbinom{(\eps+\ell)t}{\eps t}_q\geq \alpha$. If
$h\geq  2\eps$, then
$$\frac{U_A}{U_B}\leq \frac{8\theta}{\alpha-1}\ .$$
Particularly, if ${8\theta} < {\alpha-1}$,
we have
${U_A}<{U_B}$ (the upper bound in Corollary~\ref{cor:imupperbound-N} is tighter than that in  Corollary~\ref{cor:EZ19_vector}).
\end{lemma}
\begin{proof}
  Since $\quadbinom{\eps+\ell)t}{\eps t}_q\geq \alpha$,  we have that
  $$ U_A \leq  \theta\cdot \quadbinom{(\eps+\ell) t}{\eps t}_q \frac{q^{\ell t+1}-1}{q-1}\ .$$
  Then,
  \begingroup
  \allowdisplaybreaks
\begin{align*}
  \frac{U_A}{U_B}  &\leq  \frac{\theta}{\alpha-1}\cdot \frac{q^{\ell t+1}-1}{q-1}\quadbinom{(\eps+\ell)t}{\eps t}_q
  \cdot\quadbinom{(h-\ell)t}{(h-\ell-\eps)t-1}_q\quadbinom{ht}{(h-\eps)t-1}_q^{-1}  \\
  &= \frac{\theta}{\alpha-1}\cdot\frac{q^{\ell t+1}-1}{q-1}\quadbinom{(\eps+\ell)t}{\eps t}_q\quadbinom{(h-\ell)t}{\eps t+1}_q\quadbinom{ht}{\eps t+1}_q^{-1}    \\
  &= \frac{\theta}{\alpha-1}\cdot \frac{q^{\ell t+1}-1}{q-1}\cdot \frac{(q^{(\eps+\ell)t}-1)\cdots (q^{\ell t+1}-1)}{(q^{\eps t}-1)\cdots(q-1)}
  \cdot \frac{(q^{(h-\ell)t}-1)\cdots(q^{(h-\ell-\eps)t}-1)}{(q^{ht}-1)\cdots (q^{(h-\eps)t}-1)}\\
  &< \frac{\theta}{\alpha-1}\cdot \frac{q^{\ell t+1}}{q-1}\cdot \frac{q^{(\eps+\ell)t}\cdots q^{\ell t+1}}{(q^{\eps t}-1)\cdots(q-1)}
  \cdot \frac{q^{(h-\ell)t}\cdots q^{(h-\ell-\eps)t}}{(q^{ht}-1)\cdots (q^{(h-\eps)t}-1)}\\
  &=\frac{\theta}{\alpha-1}\cdot \frac{q}{q-1}\cdot
  \prod_{i=1}^{\eps t}  \parenv*{1-\frac{1}{q^i}}^{-1}  \cdot\prod_{i=h t -\eps t}^{h t}  \parenv*{1-\frac{1}{q^i}}^{-1}\\
  &\leq \frac{\theta}{\alpha-1}\cdot\parenv*{1+\frac{1}{q-1}}
  \prod_{i=1}^{h t}  \parenv*{1-\frac{1}{q^i}}^{-1}\quad \textrm{(assume $2\eps\leq h$)}\\
  &<  \frac{8\cdot \theta }{\alpha-1}\ ,
\end{align*}
\endgroup
and the statement follows.
\end{proof}

Now,  we compare the upper bound  in Corollary~\ref{cor:imupperbound-2-Network} with that in Corollary~\ref{cor:EZ19_vector} for $\alpha=2$.
\begin{lemma}\label{lem:compare-alpha2}
Denote $U_C:= \gamma q^{(h-\ell)(2\ell+\eps-h)t^2+(h-\ell)t}$
and $U_D:= \gamma q^{\ell t(\eps t+1)}$. Then,
\begin{align*}
  \log_q U_C -\log_q U_D  =   [(h-\ell)(2\ell+\eps-h)-\eps \ell]t^2 +(h-2\ell)t\ .
\end{align*}

Particularly, if one of the following three conditions is satisfied,
\begin{itemize}
    \item $\eps t +1<\ell t$, and either $h>2\ell$ or $h<\ell+\eps+\frac{1}{t}$;
    \item $\eps t +1> \ell t$, and either $h>\ell+\eps+\frac{1}{t}$ or $h<2\ell$;
    \item $\eps t +1=\ell t$ and $h\not=2\ell$,
\end{itemize}
then,
\begin{align*}
   \log_q U_C -\log_q U_D   <  0\ .
\end{align*}
In other words, the upper bound in Corollary~\ref{cor:imupperbound-2-Network} is tighter than the upper bound in Corollary~\ref{cor:EZ19_vector} for $\alpha=2$, if one of the conditions above holds.
\end{lemma}

\begin{proof}
  Denote $C=(h-\ell)(2\ell+\eps-h)t+(h-\eps)$ and $D=\ell(\eps t+1)$. Then $\log_q U_C -\log_q U_D=Ct-Dt$.
  It suffices to show that $C < D$.
Note that $C=-th^2 +{3\ell +\eps }t h + h + \cdots$ is a quadratic function in $h$ which is symmetric about $h=\frac{(3\ell+\eps)t+1}{2t}$. We proceed in three cases, according to the position of the axis of symmetry.

\begin{enumerate}
    \item If $\eps t+1 < \ell t$, then $\frac{(3\ell+\eps)t+1}{2t}< 2\ell$, i.e., the axis of symmetry is on the left of $h=2\ell$. In this case, $C$ is decreasing when $h \geq 2\ell$. It follows that $C < D$ for $h> 2\ell$ as $C=D$ when $h=2\ell$.
    Furthermore, according to the symmetry, $C < D$ also holds for $h<\ell +\eps +\frac{1}{t}$.
    \item If $\eps t+1 > \ell t$,  then $\frac{(3\ell+\eps)t+1}{2t}> 2\ell$. Using the same argument, we can see that $C < D$ holds for $h< 2\ell$ and $h>\ell+\eps +\frac{1}{t}$.
    \item If $\eps t+1=\ell t$,  then $\frac{(3\ell+\eps)t+1}{2t}= 2\ell$. The maximal value of $C-D$ is taken at $h=2\ell$, which is $0$. So $C<D$ for all $h\not=2\ell$.
\end{enumerate}
\end{proof}

The following example shows that, in some cases, the upper bound in Corollary~\ref{cor:imupperbound-2-Network} matches a lower bound from~\cite{EKOOfeb2020} within a factor of $\gamma\approx 3.48$.
\begin{example}
Let $\alpha=2$, $\eps=\ell$, and $h=2\ell+1$.
A lower bound from \cite{EKOOfeb2020} is
$$q^{(\ell^2-1)t^2+(\ell+1) t} \leq r\ .$$
For the upper bound,
Corollary~\ref{cor:imupperbound-2-Network} shows that
$$r \leq \gamma q^{(\ell^2-1)t^2+(\ell+1)t}\ ,$$
agreeing with the lower bound up to a factor of $\gamma$. In contrast, Corollary~\ref{cor:EZ19_vector} shows that
$$r \leq \gamma q^{\ell^2 t^2+\ell t}\ ,$$
which differs from the lower bound by a factor of $\gamma q^{t^2-t}$.
\end{example}

\subsubsection{Other Lower Bounds on $r_{\max}$}
Let ${\cB}_q(n,k,\delta;\alpha)$ denote the maximum possible size of an $\alpha$-$(n,k,\delta)_q^c$ covering Grassmannian code. The following lower bounds were proposed on ${\cB}_q(n,k,\delta;\alpha)$ for $\delta \leq k$ in~\cite{EKOOfeb2020}.

\begin{theorem}[{\cite[Theorem 21]{EKOOfeb2020}}]
  \label{thm:EK_lb}
Let $1\leq \delta \leq k$, $k+\delta \leq n$ and $2\leq \alpha \leq q^k+1$ be integers.
\begin{enumerate}
    \item If $n<k+2\delta$, then
    $${\cal B}_q(n,k,\delta;\alpha) \geq (\alpha -1) q^{\max\{k,n-k\}(\min\{k,n-k\}-\delta+1)}.$$

    \item If $n\geq k+2\delta$, then for each $t$ such that $\delta\leq t\leq n-k-\delta$, we have \label{case:EK_2}
    \begin{enumerate}
        \item If $t<k$, then
        $${\cal B}_q(n,k,\delta;\alpha) \geq (\alpha-1)q^{k(t-\delta+1)}{\cal B}_q(n-t,k,\delta;\alpha).$$
        \item If $t\geq k$, then
        \begin{align*}{\cal B}_q(n,k,\delta;\alpha) \geq & \ (\alpha-1)q^{t(k-\delta+1)}{\cal B}_q(n-t,k,\delta;\alpha)
         +{\cal B}_q(t+k-\delta,k,\delta;\alpha).
        \end{align*}
    \end{enumerate}
\end{enumerate}
\end{theorem}
\subsubsection{Discussion of Lower Bounds}
Theorem~\ref{thm:EK_1_ext} improves the lower bounds in \cref{thm:EK_lb} \cite{EKOOfeb2020} by removing the conditions $\alpha \leq q^k+1$ and $n<k+2\delta$.
For $n\geq k+2\delta$, the numerical results show that either could be tighter, depending on the parameters. The theoretical comparison between the two lower bounds is complicated due to the recursive function.

In the following, we compare the lower bound on $r_{\max}$ in Corollary~\ref{cor:EK19_lb} with the upper bounds in the previous sections.
\begin{itemize}
\item When $h \leq 2\ell$, Corollary~\ref{cor:EK19_lb} gives
$$r_{\max}\geq (\alpha-1)q^{\ell t(\eps t+1)}\ ,$$
which coincides with the upper bound (up to a constant factor of $\gamma\approx 3.48$) in \cref{cor:EZ19_vector}, $r_{\max}< \gamma(\alpha-1)q^{\ell t (\eps t+1)}$.
\item When $h \geq 2\ell$ and $\alpha =2$,
Corollary~\ref{cor:EK19_lb} gives
$$r_{\max}\geq q^{(h-\ell)(2\ell+\eps-h)t^2 +(h-\ell)t}\ ,$$
which coincides with the upper bound (up to a constant factor of $\gamma$) in
\cref{cor:imupperbound-2-Network},
$r_{\max}< \gamma q^{(h-\ell)(2h+\eps-h)t^2 +(h-\ell)t}$.
\item The upper bound in Corollary~\ref{cor:imupperbound-N} cannot be applied here as $(h-\eps)/\ell \leq 2$.
\end{itemize}

Based on the comparisons above, we summarize the best known bounds on $r_{\max}$ for different parameter ranges, in Table~\ref{tab:r_bound_summary}.

\begin{table*}[h!]
  \caption{Upper bounds (UBs) and lower bounds (LBs) on $r_{\max}$ of the $(\eps,\ell)-\cN_{\alpha,r,\alpha\ell+\eps}$ network with $(q,t)$-linear solutions.
  The bounds are valid for $\alpha\geq 2,h,\ell\geq 1,\eps\geq 0$. For non-trivially solvable generalized combination networks, one should consider $\ell+\eps\leq h\leq \alpha\ell+\eps$. The other parameters are $\gamma\approx3.48$, $\beta = \parenv*{(\alpha-1)!/(2e\gamma\alpha)}^{1/(\alpha-1)}$, $f(t)=(\alpha\ell+\eps-h)\eps t^2+(\alpha\ell+2\eps-h)t +{1}$, $\theta = \alpha-\floor{(h-\eps)/\ell}+1$, and $g(t)=\max\{\ell t,(h-\ell)t\}\cdot(\min\{\ell t, (h-\ell)t\}-(h-\ell-\eps)t+1)$.}
  \label{tab:r_bound_summary}
  \begin{center}
  \begin{tabular}[l]{|l|l|l|l|l|}
    \hline
    UB & $h< 2\ell+\eps$& Reference &$h\geq 2\ell+\eps$ & Reference\\
    \hline
    $\alpha> 2$&
                 \begin{tabular}{@{}r@{}}
                   $r_{\max}<\gamma(\alpha-1)$\\
                   $\cdot q^{\ell t(\eps t+1)}$
                 \end{tabular}
       &
         \begin{tabular}{@{}l@{}}
           \cite{EZjul2019} (cf.\\
           \cref{cor:EZ19_vector})
         \end{tabular}
       &
         \begin{tabular}{@{}r@{}}
           $r_{\max}<\gamma\theta q^{\ell t(\eps t+1)}$\\
           $+\alpha-\theta$
         \end{tabular}
       &\cref{cor:imupperbound-N}\\
    \hline
    $\alpha=2$
       &\multicolumn{3}{c|}{$r_{\max}<\gamma q^{\min \{\ell t(\eps t+1),(h-\ell)(2\ell+\eps-h)t^2+(h-\ell)t\}}$}&
         \begin{tabular}{@{}l@{}}
           \cite{EZjul2019} \& \cref{cor:imupperbound-2-Network}\\
           (Comparison in \\
           \cref{lem:compare-alpha2})
         \end{tabular}
    \\
    \hline
    \hline
    LB& $h< 2\ell+\eps$& Reference & $h\geq 2\ell+\eps$ & Reference\\
    \hline
    $\alpha\geq 2$ & $r_{\max}\geq (\alpha-1)q^{g(t)}$& Corollary~\ref{cor:EK19_lb}& $r_{\max}\geq \beta \cdot q^{\frac{f(t)}{\alpha-1}}$ &  Theorem~\ref{thm:LLL_bound} \\
    \hline
  \end{tabular}
  \end{center}
\end{table*}

\section{Summary and Outlooks}
The contributions in this chapter are of two-fold.
We first investigated the minimum required field size to construct an MSRD code (in particular, LRS codes) from a support-constrained generator matrix. For this purpose, we proved that the condition on the support constraints such that a support-constrained MDS/MRD code exists is also a necessary and sufficient condition for a support-constrained MSRD code, via the framework of skew polynomials.
Given support constraints fulfilling this condition, an $[n,k]_{q^m}$ support-constrained LRS code exists for any prime power $q\geq \ell+1$ and integer $m\geq \max_{l\in[\ell]}\{k-1+\log_qk, n_l\}$, where $\ell$ is the number of blocks and $n_l$ is the length of the $l$-th block of the LRS code.
If the desired support constraints do not fulfill the necessary condition, the maximum sum-rank distance of a code fulfilling these constraints is given.
With these results, we proposed a network coding scheme using support-constrained LRS codes for the distributed multi-source networks.
The key of the scheme is to formulate all the technical requirements into an integer linear programming (ILP) problem.
However, the ILP problem has $\Omega(2^h)$ constraints, where $h$ is the number of messages to be cast. For large $h$, solving (even constructing the constraints) the ILP problem is computationally expensive.
In future research, more specific distributed networks should be investigated so that the scheme may become more practical while considering other properties of the networks.

In the second part of this chapter, we quantified the advantage of vector network coding compared to scalar network coding in a family of multicast networks -- generalized combination networks.
By studying necessary and sufficient conditions for the existence of $(q,t)$-linear solutions to the generalized combination network
$(\eps,\ell)$-$\mathcal{N}_{h,r,\alpha\ell+\eps}$.
We derived upper and lower bounds on $r_{\max}$, the maximum number of nodes in the middle layer.
The lower bounds coincide (up to a constant factor of $\gamma\approx 3.48$) with the upper bounds for $h\leq 2\ell$ or $h\geq 2\ell,\alpha=2$.
With these results, we obtained upper and lower bounds on $\gap(\cN)$,
which is the number of extra bits that a scalar solution has to pay compared to a vector solution of a generalized combination network $\cN$.
The asymptotic behavior of the upper and lower bound shows that $\gap(\cN) = \Theta\parenv*{\log(r)}$.
Namely, for large generalized combination networks, using a scalar linear solution over-pays an order of $\log(r)$ extra bits per symbol, than using a vector linear solution.
A notable observation is, the novel upper and lower bounds on $\gap(\cN)$ holds for all parameters range of the generalized combination network, except $\eps=0$.
This may imply that the direct links between the source and the receivers are crucial for vector network coding to have an advantage in generalized combination networks.
For future research, the role of the direct links for a nonzero $\gap(\cN)$ can be further investigated.


%% file: figs/GM-MSRD/inductionProofLogic.tex
\def\x{0.6}
\begin{tikzpicture}
  \node (i) at (0,0) {(i)};
  \node (ii) at ($(i)+(\x*4, 0)$) {(ii)};
  \node (pre) at ($(i)-(\x*4, 0)$) {$(k',s',n')$};
  \node (H1) at ($(i)+(\x*2, \x*0.5)$) {\textcolor{TUMGreenDark}{\ref{H:fromPre}}};

  \node (i2) at ($(i)-(0, \x*3)$) {(i)};
  \node (ii2) at ($(i2)+(\x*4, 0)$) {(ii)};
  \node (here) at ($(i2)-(\x*4, 0)$) {$(k,\ s,\ n)$};
  \node (H2) at ($(ii2)+(\x*1, 0)$) {\textcolor{TUMGreenDark}{\ref{H:fromHere}}};
  \node (Goal) at ($(i2)+(\x*2, -\x*0.5)$) {\textcolor{TUMRed}{Proof Goal}};

  \draw[{Implies}-, thick, double, double distance=\x*3pt, TUMGreenDark] (i.east) -- (ii.west);
  \draw[{Implies}-, thick, double, double distance=\x*3pt] (ii.south) -- (ii2.north);
  \draw[-{Implies}, thick, double, double distance=\x*3pt] (i.south) -- (i2.north);
  \draw[{Implies}-, thick, double, double distance=\x*3pt, TUMRed] (i2.east) -- (ii2.west);

  \node[rotate=-90] (step1) at ($(ii2)+(\x*0.5, \x*1.5)$) {\stepone};
  \node[rotate=90] (step2) at ($(i2)+(-\x*0.5, \x*1.5)$) {\steptwo};
  
\end{tikzpicture}

%% file: figs/GM-MSRD/induc_pic.tex
\begin{tikzpicture}

\begin{axis}[
    xmin=0, xmax=5.3,
    ymin=0, ymax=5.5,
    xlabel=$s$,
    ylabel=$n$,
    grid=major,
    xtick={0,1,2,3,4,5},
    xticklabels = {$0$,$1$,$2$, $3$, $4$, $\cdots$},
    ytick={0,1,2,3,4,5,6,7},
    yticklabels={$0$,$1$,$2$, $3$, $4$, $\vdots$},
    unit vector ratio=1 1,
    legend pos=outer north east,
    legend style={cells={align=left}},
    axis lines=middle,
    xlabel style={right},
    ylabel style={above}
]

\addplot[only marks, mark=triangle*, mark size=2pt, color=black] table {
    1 0
    1 1
    1 2
    1 3
    1 4
    1 5
};
\addlegendentry{induction basis\\$s=1,n\geq 0$}

\addplot[only marks, mark=x, mark size=2pt, color=TUMRed] table {
    2 0
    3 0
    4 0
    5 0
};
\addlegendentry{\ref{item:equiv1} and \ref{item:equiv2} both\\never hold}

\addplot[only marks, mark=square*, mark size=1.5pt,color=TUMPink] table {
    3 2
    4 2
    3 3
    4 3
    3 4
};
\addlegendentry{\ref{case_s3n2}\\($s\geq 3,n\geq 2$)}

\addplot[only marks, mark=*, mark size=2pt,color=TUMBlue] table {
    2 2
    2 3
    2 4
};
\addlegendentry{\ref{case_s2n2}\\($s=2, n\geq 2$)}
\addplot[only marks, mark=diamond*, mark size=2pt,color=TUMGreen] table {
    2 1
    3 1
    4 1
    5 1
};
\addlegendentry{\ref{case_s2n1}\\($s\geq 2, n=1$)}

\draw[-latex, thick, color=TUMPink] (3,4) -- (3,3);
\draw[-latex, thick, color=TUMPink] (3,3) -- (3,2);
\draw[-latex, thick, color=TUMPink] (3,2) -- (3,1);
\draw[-latex, thick, color=TUMPink] (4,3) -- node [below, yshift=-2ex, rotate=90] {\ref{case_s3n2_d}} (4,2);
\draw[-latex, thick, color=TUMPink] (4,2) -- (4,1);

\draw[-latex, thick, color=TUMPink] (3,4) -- node [above, xshift=7ex] {\ref{case_s3n2_b}/\ref{case_s3n2_c}} (2,4);
\draw[-latex, thick, color=TUMPink] (4,2) -- (3,2);
\draw[-latex, thick, color=TUMPink] (3,2) -- (2,2);
\draw[-latex, thick, color=TUMPink] (4,3) -- (3,3);
\draw[-latex, thick, color=TUMPink] (3,3) -- (2,3);

\draw[-latex, thick, color=TUMBlue] (2,2) -- node [above, xshift=-2ex] {\ref{case_s2n2_b}} (1,2);
\draw[-latex, thick, color=TUMBlue] (2,3) -- (1,3);
\draw[-latex, thick, color=TUMBlue] (2,4) -- (1,4);
\draw[-latex, thick, color=TUMBlue] (2,4) -- node [above, yshift=1.3ex, rotate=90] {\ref{case_s2n2_c}} (2,3);
\draw[-latex, thick, color=TUMBlue] (2,3) -- (2,2);
\draw[-latex, thick, color=TUMBlue] (2,2) -- (2,1);

\draw[-latex, thick, color=TUMGreen] (2,1) -- (1,1);
\draw[-latex, thick, color=TUMGreen] (3,1) -- node [below, xshift=2ex] {\ref{case_s2n1_b}} (2,1);
\draw[-latex, thick, color=TUMGreen] (4,1) -- (3,1);
\draw[-latex, thick, color=TUMGreen] (5,1) -- (4,1);

\end{axis}

\end{tikzpicture}


%% file: figs/GM-MSRD/MultiSourceNetwork.tex
\def\x{0.55}

\begin{tikzpicture}
  [font=\normalsize,>=stealth',
  mycircle/.style={circle, draw=TUMGray, very thick, text width=.1em, minimum height=1.5em, text centered},
  myrectangle/.style={rectangle, draw=TUMGray, very thick, text width=\linewidth, minimum height=1.5em, text centered},
  mycircle_small/.style={circle,draw=TUMGray!90,very thick, inner sep=0,minimum size=1em,text centered},
  mylink/.style={color=TUMBlueDark, thick, dashed},
  mylink_net/.style={color=TUMGreen, thick},
  myarrow/.style={->, color=black, thick},
  myarc/.style={color=TUMOrange, thick},
  ]
  \coordinate (Mset) at (0*\x,4*\x);
  {\node[mycircle,label=above:{$\cM=\{\bm_1,\bm_2,\dots,\bm_h\}$}] (Mset) {};}
  \node[mycircle_small,below left = \x*40pt and \x*100pt of Mset] (M0) {$S_{\cJ_1}$};
  \node[mycircle_small,right = \x*20pt of M0] (M1) {$S_{\cJ_2}$};
  \node[mycircle_small,right = \x*20pt of M1] (M2) {$S_{\cJ_3}$};
  \node[draw=none,right = \x*20pt of M2] (M3) {$\dots$};
  \node[mycircle_small,right = \x*20pt of M3] (M4) {$S_{\cJ_s}$};

  \node[draw=none,right= 15pt of M4] (cJ) {$\cJ_i\subset [h], \forall i \in[s]$};
  \node[below left= \x*30pt and \x*10pt of M0] (R0) {};
  \node[below right= \x*100pt and \x*10pt of M4] (R4) {};
  \draw[myrectangle] (R0) rectangle (R4) node[pos=.5, text width=5cm] (RLN) {$\Fq$-Linear network with a $(t,\rho)$-adversary};

  \path[] (Mset) edge[mylink] (M0.north)
  edge[mylink] node [midway] (link1) {} (M1.north)
  edge[mylink] (M2.north)
  edge[mylink] node [midway] (links) {} (M4.north);

  \node[draw=none,above right= -\x*12pt and \x*4em of links, text width=\x*9.5cm] (T1) {access to the messages indexed in $\cJ_s$};
  \path[] (T1) edge[myarrow] (links);

  \path[-latex] (M0.south) edge[mylink_net] node (S1N) {} ($(M0.south)-(0, \x*32pt)$);
  \path[-latex] (M1.south) edge[mylink_net] node (S2N) {} ($(M1.south)-(0, \x*32pt)$);
  \path[-latex] (M2.south) edge[mylink_net] node (S3N) {} ($(M2.south)-(0, \x*32pt)$);
  \path[-latex] (M4.south) edge[mylink_net] node (SsN) {} ($(M4.south)-(0, \x*32pt)$);

  \draw [myarc] ($(S1N)-(\x*10pt,\x*3pt)$) arc (130:70:\x*15pt);
  \node[draw=none, myarc] at ($(M0.south)-(\x*16pt, \x*8pt)$) (c1) {$n_{\cJ_1}$};
  \draw [myarc] ($(S2N)-(\x*10pt,\x*3pt)$) arc (125:65:\x*15pt);
  \node[draw=none, myarc] at ($(M1.south)-(\x*16pt, \x*8pt)$) (c2) {$n_{\cJ_2}$};
  \draw [myarc] ($(S3N)-(\x*10pt,\x*3pt)$) arc (120:60:\x*15pt);
  \node[draw=none, myarc] at ($(M2.south)-(\x*16pt, \x*8pt)$) (c3) {$n_{\cJ_3}$};
  \draw [myarc] ($(SsN)-(\x*10pt,\x*3pt)$) arc (115:55:\x*15pt);
  \node[draw=none, myarc] at ($(M4.south)-(\x*16pt, \x*8pt)$) (cJ) {$n_{\cJ_s}$};

  \node[draw=none] at ($(RLN.south)-(-\x*70pt,\x*50pt)$) (sink1) {};
  \node[draw=none] at ($(RLN.south)-(\x*70pt,\x*80pt)$) (sink2) {};
  \draw[myrectangle] (sink1) rectangle (sink2) node[pos=.5] (sink) {Sink};
  \node[right=-\x*28em of sink,text width=5.7cm] (recover) {Collects $N$ packets to recover all $h$ messages};

  \path[-latex] ($(RLN.south)-(0, \x*10pt)$) edge[mylink_net] (sink.north);
  \draw [myarc] ($(sink.north)-(\x*8pt,-\x*10pt)$) arc (130:70:\x*15pt);
  \node[draw=none, myarc] at ($(sink.north)-(\x*16pt, -\x*10pt)$) (sinkN) {$N$};




\end{tikzpicture}


%% file: figs/GM-MSRD/E_partition.tex
\pgfkeys{tikz/mymatrixenv/.style={decoration={brace},every left delimiter/.style={xshift=5pt, yshift=-1pt},every right delimiter/.style={xshift=-5pt, yshift=-1pt}}}

\pgfkeys{tikz/mymatrix/.style={matrix of math nodes,nodes in empty cells,left delimiter={(},right delimiter={)},inner sep=2pt,outer sep=2pt,column sep=4pt,row sep=4pt,nodes={minimum width=3pt,minimum height=2pt,anchor=center,inner sep=0pt,outer sep=0pt}}}

\pgfkeys{tikz/mymatrixbrace/.style={decorate,thick}}

\tikzset{style green/.style={
    set fill color=TUMGreenDark!80!lime!20,fill opacity=0.3,
    set border color=TUMGreenDark!60!lime!40,draw opacity=1.0,
  },
  style cyan/.style={
    set fill color=cyan!90!blue!60, draw opacity=0.4,
    set border color=blue!70!cyan!30,fill opacity=0.1,
  },
  style orange/.style={
    set fill color=TUMOrange!80,fill opacity=0.3,
    set border color=TUMOrange!90,  draw opacity=0.8,
  },
  style brown/.style={
    set fill color=brown!70!orange!40, draw opacity=0.4,
    set border color=brown, fill opacity=0.3,
  },
  style purple/.style={
    set fill color=violet!90!pink!20, draw opacity=0.5,
    set border color=violet, fill opacity=0.3,
  },
  style black/.style={
    set fill color=none, fill opacity=0.3,
    set border color=black, draw opacity=0.5,
  },
  kwad/.style={
    above left offset={-0.07,0.23},
    below right offset={0.07,-0.23},
    #1
  },
  pion/.style={
    above left offset={-0.07,0.2},
    below right offset={0.07,-0.32},
    #1
  },
  border/.style={
    above left offset={-0.03,0.18},
    below right offset={0.03,-0.3},
    #1
  },
  set fill color/.code={\pgfkeysalso{fill=#1}},
  set border color/.style={draw=#1}
}

\[  
  \begin{tikzpicture}[baseline={-0.5ex},mymatrixenv]
    \matrix [mymatrix,row sep=7pt] (A) 
    {
      \bA_{1,1} & \bA_{1,2} & \cdots & \bA_{1,\ell} \\
      \bA_{2,1} & \bA_{2,2} & \cdots & \bA_{2,\ell} \\
      \vdots & \vdots & \vdots & \vdots \\
      \bA_{\ell,1} & \bA_{\ell,2} & \cdots & \bA_{\ell,\ell} \\
    };
    \matrix [right= 8pt of A, mymatrix, inner sep=2pt] (X)
    {
      \bX_{1} \\
      \bX_{2} \\
      \vdots \\
      \bX_{\ell} \\
    };
    \node[right =30pt of X] (plus) {$+$};

    \matrix [mymatrix, right = 5pt of plus, row sep=8pt, column sep=6pt] (E)
    {
      \bE_1\\
      \bE_2\\
      \vdots\\
      \bE_{\ell}\\
    };
    \node[right=30pt of E] (period) {.};
    \node[left =30pt of A] (eq) {$\bA\bX+\bE =$};
        \mymatrixbracebottom{A}{1}{1}{$n_1$}
        \mymatrixbracebottom{A}{2}{2}{$n_2$}
        \mymatrixbracebottom{A}{4}{4}{$n_{\ell}$}
        \mymatrixbraceleft{A}{1}{1}{$N_1$}
        \mymatrixbraceleft{A}{2}{2}{$N_2$}
        \mymatrixbraceleft{A}{4}{4}{$N_{\ell}$}

        \mymatrixbraceright{X}{1}{1}{$n_1$}
        \mymatrixbraceright{X}{2}{2}{$n_2$}
        \mymatrixbraceright{X}{4}{4}{$n_4$}
        \mymatrixbracebottom{X}{1}{1}{$M$}

        \mymatrixbracebottom{E}{1}{1}{$M$}
        \mymatrixbraceright{E}{1}{1}{$N_1$}
        \mymatrixbraceright{E}{2}{2}{$N_2$}
        \mymatrixbraceright{E}{4}{4}{$N_{\ell}$}

    \end{tikzpicture}
\]

%% file: figs/GM-MSRD/dis_LRS_gen.tex
\pgfkeys{tikz/mymatrixenv/.style={decoration={brace},every left delimiter/.style={xshift=5pt, yshift=-1pt},every right delimiter/.style={xshift=-5pt, yshift=-1pt}}}

\pgfkeys{tikz/mymatrix/.style={matrix of math nodes,nodes in empty cells,left delimiter={(},right delimiter={)},inner sep=1pt,outer sep=3pt,column sep=4pt,row sep=4pt,nodes={minimum width=9pt,minimum height=6pt,anchor=center,inner sep=0pt,outer sep=0pt}}}

\pgfkeys{tikz/mymatrixbrace/.style={decorate,thick}}

\tikzset{style green/.style={
    set fill color=TUMGreenDark!80!lime!20,fill opacity=0.3,
    set border color=TUMGreenDark!60!lime!40,draw opacity=1.0,
  },
  style cyan/.style={
    set fill color=cyan!90!blue!60, draw opacity=0.4,
    set border color=blue!70!cyan!30,fill opacity=0.1,
  },
  style orange/.style={
    set fill color=TUMOrange!80,fill opacity=0.3,
    set border color=TUMOrange!90,  draw opacity=0.8,
  },
  style brown/.style={
    set fill color=brown!70!orange!40, draw opacity=0.4,
    set border color=brown, fill opacity=0.3,
  },
  style purple/.style={
    set fill color=violet!90!pink!20, draw opacity=0.5,
    set border color=violet, fill opacity=0.3,
  },
  kwad/.style={
    above left offset={-0.07,0.23},
    below right offset={0.07,-0.23},
    #1
  },
  pion/.style={
    above left offset={-0.03,0.18},
    below right offset={0.07,-0.23},
    #1
  },
  poz/.style={
    above left offset={-0.03,0.18},
    below right offset={0.03,-0.3},
    #1
  },set fill color/.code={\pgfkeysalso{fill=#1}},
  set border color/.style={draw=#1}
}

\[
   \bG =
    \begin{tikzpicture}[baseline={-0.5ex},mymatrixenv]
        \matrix [mymatrix,column sep=3pt,inner sep=2pt] (G)
        {
          \tikzmarkin[kwad=style green]{LRS-first} \times  &  \times &  \times&  \times&  \times&  \times&
          \hphantom{\times} &  \times&  \times& \tikzmarkin[kwad=style purple]{LRS-second} \times &  \times&  \times&  \times & \times&
          \hphantom{\times} &  \times&  \times&
          \tikzmarkin[pion=style orange]{LRS-third} \hphantom{\times} & \textcolor{TUMRed}{0} & \textcolor{TUMRed}{0} & \textcolor{TUMRed}{0} & \textcolor{TUMRed}{0} & \textcolor{TUMRed}{0} & \textcolor{TUMRed}{0} & \textcolor{TUMRed}{0} & \textcolor{TUMRed}{0}\\
    \times&  \times& \times&  \times& \times&  \times&  \hphantom{\times} &
    \times&  \times& \times&  \times& \times&  \times&  \times & \hphantom{\times} &
    \textcolor{TUMRed}{0} & \textcolor{TUMRed}{0} & \hphantom{\times}&
    \times&  \times& \times&  \times& \times & \times&  \times& \times\\
    \times&  \times& \times&  \times& \times&  \times&  \hphantom{\times} &
    \times&  \times& \times&  \times& \times&  \times&  \times & \hphantom{\times} &
    \textcolor{TUMRed}{0} & \textcolor{TUMRed}{0} & \hphantom{\times}&
    \times&  \times& \times&  \times& \times & \times&  \times& \times\\
    \times&  \times& \times&  \times& \times&  \times&  \hphantom{\times} &
    \times&  \times& \times&  \times& \times&  \times&  \times & \hphantom{\times} &
    \textcolor{TUMRed}{0} & \textcolor{TUMRed}{0} & \hphantom{\times}&
    \times&  \times& \times&  \times& \times & \times&  \times& \times\\
    \times&  \times& \times&  \times& \times&  \times&  \hphantom{\times} &
    \textcolor{TUMRed}{0} & \textcolor{TUMRed}{0} & \textcolor{TUMRed}{0} & \textcolor{TUMRed}{0} & \textcolor{TUMRed}{0} & \textcolor{TUMRed}{0} & \textcolor{TUMRed}{0} &  \hphantom{\times} &
    \times&  \times& \hphantom{\times} &
    \times&  \times& \times&  \times& \times&  \times&  \times&  \times\\
    \times&  \times& \times&  \times& \times&  \times&  \hphantom{\times} &
    \textcolor{TUMRed}{0} & \textcolor{TUMRed}{0} & \textcolor{TUMRed}{0} & \textcolor{TUMRed}{0} & \textcolor{TUMRed}{0} & \textcolor{TUMRed}{0} & \textcolor{TUMRed}{0} &  \hphantom{\times} &
    \times&  \times& \hphantom{\times} &
    \times&  \times& \times&  \times& \times&  \times&  \times&  \times\\
    \textcolor{TUMRed}{0}  & \textcolor{TUMRed}{0} & \textcolor{TUMRed}{0} & \textcolor{TUMRed}{0} & \textcolor{TUMRed}{0} & \textcolor{TUMRed}{0} & \hphantom{\times} &
    \times & \times & \times & \times & \times & \times & \times &\hphantom{\times} &
    \times & \times &\hphantom{\times} &
    \times  & \times & \times & \times & \times & \times & \times & \times \\
    \textcolor{TUMRed}{0}  & \textcolor{TUMRed}{0} & \textcolor{TUMRed}{0} & \textcolor{TUMRed}{0} & \textcolor{TUMRed}{0} & \textcolor{TUMRed}{0} & \hphantom{\times} &
    \times & \times & \times & \times & \times & \times & \times &\hphantom{\times} &
    \times & \times &\hphantom{\times} &
    \times  & \times & \times & \times & \times & \times & \times & \times \\
    \textcolor{TUMRed}{0}  & \textcolor{TUMRed}{0} & \textcolor{TUMRed}{0} & \textcolor{TUMRed}{0} & \textcolor{TUMRed}{0} & \textcolor{TUMRed}{0} & \hphantom{\times}&
    \times & \times \tikzmarkend{LRS-first} & \times & \times & \times & \times & \times &\hphantom{\times} &
    \times & \times \tikzmarkend{LRS-second}&\hphantom{\times} &
    \times  & \times & \times & \times & \times & \times & \times & \times \tikzmarkend{LRS-third} \\
    };

    \mymatrixbraceright{G}{1}{1}{Encoding for $\bm_1$}
    \mymatrixbraceright{G}{2}{4}{for $\bm_2$}
    \mymatrixbraceright{G}{5}{6}{for $\bm_3$}
    \mymatrixbraceright{G}{7}{9}{for $\bm_4$}
    \mymatrixbracetop{G}{1}{6}{Encoding at $S_{\cJ_1}$}
    \mymatrixbracetop{G}{8}{14}{at $S_{\cJ_2}$}
    \mymatrixbracetop{G}{16}{17}{at $S_{\cJ_3}$}
    \mymatrixbracetop{G}{19}{26}{at $S_{\cJ_4}$}
    \mymatrixbracebottom{G}{1}{9}{First block of LRS code}
    \mymatrixbracebottom{G}{10}{17}{Second block}
    \mymatrixbracebottom{G}{18}{26}{Third block}
  \end{tikzpicture}
\]

%% file: figs/GM-MSRD/lifting.tex



\renewcommand*\mymatrixbracebottom[4]{
    \draw[mymatrixbrace, decoration={mirror, raise=1pt}] (#1.south-|#1-1-#2.south west) -- node[below=2pt] {#4} (#1.south-|#1-1-#3.north east);
}
\begin{align}
  \label{eq:lifting-X}
  \bX =
  \begin{tikzpicture}[baseline={-0.5ex},mymatrixenv]
    \matrix [mymatrix,inner sep=1pt, ampersand replacement=\&] (LX)
    {
      \bI_{n_{\cJ_1}} \&\&\&\&  \hphantom{holder} \bC_{\cJ_1}^\top \hphantom{holder}  \\
      \& \bI_{n_{\cJ_2}} \&\&\& \hphantom{holder}  \bC_{\cJ_2}^\top \hphantom{holder}  \\
      \&\& \bI_{n_{\cJ_3}} \&\& \hphantom{holder}  \bC_{\cJ_3}^\top \hphantom{holder}  \\
      \&\&\& \bI_{n_{\cJ_4}} \& \hphantom{holder}  \bC_{\cJ_4}^\top \hphantom{holder}  \\
    };
        \mymatrixbraceright{LX}{1}{4}{$n=\sum_{i=1}^4n_{\cJ_i}$}
        \mymatrixbracebottom{LX}{5}{5}{$m$}
        \mymatrixbracebottom{LX}{1}{4}{$n$}
    \end{tikzpicture}
\end{align}

%% file: figs/VecNetCod/GeneralizedNetwork.tex
\def\x{0.7}

\begin{tikzpicture}
  [font=\normalsize,>=stealth',
  mycircle/.style={circle, draw=TUMGray, very thick, text width=.1em, minimum height=1.5em, text centered},
  mycircle_small/.style={circle,draw=TUMGray!90,very thick, inner sep=0,minimum size=1em,text centered},
  mylink/.style={color=TUMBlueLight, thick},
  mylink_dir/.style={color=TUMGreen, thick, dashed}]

  \coordinate (Source) at (0*\x,4*\x);
  {\node[mycircle,label=above right:{$\bx=({\bx}_1,{\bx}_2,\dots,{\bx}_h)$}] (Source) {};}

  \node[mycircle_small,below left = \x*60pt and \x*100pt of Source] (M0) {};
  \node[mycircle_small,right = \x*20pt of M0] (M1) {};
  \node[mycircle_small,right = \x*20pt of M1] (M2) {};
  \node[mycircle_small,right = \x*20pt of M2] (M3) {};
  \node[draw=none,right = \x*20pt of M3] (M4) {$\dots$};
  \node[mycircle_small,right = \x*20pt of M4] (M5) {};
  \node[mycircle_small,right = \x*20pt of M5] (M6) {};

  \node[mycircle,below left = \x*60pt and \x*20pt of M0,label=below:{$\by_1$}] (R0) {};
  \node[mycircle,right = \x*50pt of R0,label=below:{$\by_2$}] (R1) {};
  \node[mycircle,right = \x*50pt of R1,label=below:{$\by_3$}] (R2) {};
  \node[draw=none,right = \x*50pt of R2] (R3) {$\dots$};
  \node[mycircle,right = \x*50pt of R3,label=below:{$\by_{N}$}] (R4) {};

  \path[] (Source) edge[mylink,bend right=15] (M0.north)
  edge[mylink,bend right=15] (M1.north)
  edge[mylink,bend right=15] (M2.north)
  edge[mylink,bend right=15] (M3.north)
  edge[mylink,bend left=15] (M5.north)
  edge[mylink,bend left=15] (M6.north);

  \path[] (R0) edge[mylink,bend left=15] (M0.south)
  edge[mylink,bend left=15] (M1.south);
  \path[] (R1) edge[mylink,bend left=15] (M0.south)
  edge[mylink,bend left=15] (M2.south);
  \path[] (R2) edge[mylink,bend left=15] (M1.south)
  edge[mylink,bend left=15] (M3.south);
  \path[] (R4) edge[mylink,bend right=15] (M5.south)
  edge[mylink,bend right=15] (M6.south);

  \draw [mylink,solid] ($(M0)+(\x*18pt,\x*15pt)$) arc (30:80:\x*15pt); 
  \draw [mylink,-] ($(R0)+(\x*20pt,\x*10pt)$) arc (0:100:\x*15pt); 
  \node[draw=none,above right = \x*-8pt and \x*2pt of M0] (ell) {\textcolor{TUMBlue}{$\ell$}};
  \node[draw=none,right=0pt of M6] (rs) {$r$ middle nodes};
  \node[draw=none,right=0pt of R4] (N) {$N=\binom{r}{\alpha}$};
  \node[draw=none,below right=\x*-5pt and \x*-50pt of N] (recNodes) {receivers};

  \path[] (Source) edge[mylink,bend right=5] (M0.north)
  edge[mylink,bend right=5] (M1.north)
  edge[mylink,bend right=5] (M2.north)
  edge[mylink,bend right=5] (M3.north)
  edge[mylink,bend left=5] (M5.north)
  edge[mylink,bend left=5] (M6.north);

  \path[] (R0) edge[mylink,bend left=5] (M0.south)
  edge[mylink,bend left=5] (M1.south);
  \path[] (R1) edge[mylink,bend left=5] (M0.south)
  edge[mylink,bend left=5] (M2.south);
  \path[] (R2) edge[mylink,bend left=5] (M1.south)
  edge[mylink,bend left=5] (M3.south);
  \path[] (R4) edge[mylink,bend right=5] (M5.south)
  edge[mylink,bend right=5] (M6.south);

  \path[] (Source.west) edge[mylink_dir, bend right=45] (R0.west)
  edge[mylink_dir, bend right=50] (R0.west)
  (Source) edge[mylink_dir,bend left=5] (R1.east)
  edge[mylink_dir,bend left=10] (R1.east)
  (Source) edge[mylink_dir,bend left=15] (R2.east)
  edge[mylink_dir,bend left=20] (R2.east)
  (Source) edge[mylink_dir,bend right=15] (R4.west)
  edge[mylink_dir,bend right=10] (R4.west);

  \draw [mylink_dir,solid] ($(R0)+(-\x*15pt,\x*35pt)$) arc (100:60:\x*15pt); 
  \node[draw=none,above left = \x*10pt and 3pt of R0] (alphal) {\textcolor{TUMGreen}{$\varepsilon$}};
  \node[draw=none,right = 0pt of R0] (alphal) {\textcolor{TUMBlue}{$\alpha\ell$}};

\end{tikzpicture}

%% file: figs/VecNetCod/h12l1e2r8e5a20.tex
\definecolor{darkgreen}{rgb}{0,0.7,0}
\begin{tikzpicture}
\pgfplotsset{compat = 1.3}
\begin{axis}[
	legend style={nodes={scale=0.7, transform shape}},
	width = \linewidth,
	height = \linewidth,
	title = {$h=12$, $\varepsilon=2$, $\ell=1$, $r=8\times 10^5$, $\alpha=20$},
	xlabel = {{$t$}},
	xmin = 1,
	xmax = 20,
	ymin = 1,
	ymax = 5000,
	legend pos = north east,
	legend cell align=left,
	ymode=log,
	grid=both,
	ytick={1,10,100,1000,10000},
	yminorticks=true]

\addplot [solid, color=TUMBlue, thick, mark=*] table[row sep=\\] {
1.000000 27.544680 \\
2.000000 7.311772 \\
3.000000 4.141513 \\
4.000000 3.020039 \\
5.000000 2.470231 \\
6.000000 2.149396 \\
7.000000 1.940920 \\
8.000000 1.795247 \\
9.000000 1.688008 \\
10.000000 1.605904 \\
11.000000 1.541101 \\
12.000000 1.488691 \\
13.000000 1.445454 \\
14.000000 1.409190 \\
15.000000 1.378347 \\
16.000000 1.351800 \\
17.000000 1.328715 \\
18.000000 1.308458 \\
19.000000 1.290542 \\
20.000000 1.274584 \\
};
\addlegendentry{{$\left(\frac{r+\theta-\alpha}{\gamma\cdot\theta}\right)^{\frac{1}{\ell (\varepsilon t+1)}}$ (\cref{lem:lb_q})}};

\addplot [solid, color=TUMGreenDark, thick, mark=diamond] table[row sep=\\] {
1.000000 910.202123 \\
2.000000 72.448517 \\
3.000000 22.388315 \\
4.000000 11.443355 \\
5.000000 7.418817 \\
6.000000 5.480896 \\
7.000000 4.383703 \\
8.000000 3.692574 \\
9.000000 3.223351 \\
10.000000 2.886711 \\
11.000000 2.634796 \\
12.000000 2.439934 \\
13.000000 2.285133 \\
14.000000 2.159441 \\
15.000000 2.055507 \\
16.000000 1.968230 \\
17.000000 1.893969 \\
18.000000 1.830058 \\
19.000000 1.774507 \\
20.000000 1.725796 \\
};
\addlegendentry{{$\left(\frac{r}{\beta}\right)^{\frac{(\alpha-1)t}{f(t)}}$ (\cref{lem:ub_q})}};

\addplot [solid, color=orange, thick, mark=star] table[row sep=\\] {
1.000000 2.000000 \\
2.000000 4.000000 \\
3.000000 8.000000 \\
4.000000 16.000000 \\
5.000000 32.000000 \\
6.000000 64.000000 \\
7.000000 128.000000 \\
8.000000 256.000000 \\
9.000000 512.000000 \\
10.000000 1024.000000 \\
11.000000 2048.000000 \\
12.000000 4096.000000 \\
13.000000 8192.000000 \\
14.000000 16384.000000 \\
15.000000 32768.000000 \\
16.000000 65536.000000 \\
17.000000 131072.000000 \\
18.000000 262144.000000 \\
19.000000 524288.000000 \\
20.000000 1048576.000000 \\
};
\addlegendentry{{$2^t$}};

\addplot [dashed, color=black, thick, forget plot] table[row sep=\\] {
1.000000 910.202123 \\
5.000000 910.202123 \\
};
\addlegendentry{{$2^t$}};

\addplot [dashed, color=black, thick, forget plot] table[row sep=\\] {
1.000000 8.000000 \\
5.000000 8.000000 \\
};
\addlegendentry{{$2^t$}};

\draw[<->,>=latex, thick] (axis cs: 4.000000, 910.202123) -- node[right] {max $\gap(\cN) \approx 6.83$ bits} (axis cs: 4.000000, 8.000000);

\addplot [dashed, color=black, thick, forget plot] table[row sep=\\] {
1.000000 27.544680 \\
8.000000 27.544680 \\
};
\addlegendentry{{$2^t$}};

\addplot [dashed, color=black, thick, forget plot] table[row sep=\\] {
1.000000 16.000000 \\
8.000000 16.000000 \\
};
\addlegendentry{{$2^t$}};

\draw[<->,>=latex, thick] (axis cs: 7.000000, 27.544680) -- node[right] {min $\gap(\cN) \approx 0.78$ bits} (axis cs: 7.000000, 16.000000);

\end{axis}
\end{tikzpicture}


%% file: figs/VecNetCod/h12l1e2r8e6a20.tex
\begin{tikzpicture}
\pgfplotsset{compat = 1.3}
\begin{axis}[
	legend style={nodes={scale=0.7, transform shape}},
	width = \linewidth,
	height = \linewidth,
	title = {$h=12$, $\varepsilon=2$, $\ell=1$, $r=8\times 10^6$, $\alpha=20$},
	xlabel = {{$t$}},
	xmin = 1,
	xmax = 20,
	ymin = 1,
	ymax = 5000,
	legend pos = north east,
	legend cell align=left,
	ymode=log,
	grid=both,
	ytick={1,10,100,1000,10000},
	yminorticks=true]

\addplot [solid, color=TUMBlue, thick, mark=*] table[row sep=\\] {
1.000000 59.343413 \\
2.000000 11.588401 \\
3.000000 5.754622 \\
4.000000 3.900535 \\
5.000000 3.045419 \\
6.000000 2.565901 \\
7.000000 2.262948 \\
8.000000 2.055645 \\
9.000000 1.905488 \\
10.000000 1.792004 \\
11.000000 1.703372 \\
12.000000 1.632318 \\
13.000000 1.574133 \\
14.000000 1.525641 \\
15.000000 1.484625 \\
16.000000 1.449492 \\
17.000000 1.419068 \\
18.000000 1.392474 \\
19.000000 1.369031 \\
20.000000 1.348214 \\
};
\addlegendentry{{$\left(\frac{r+\theta-\alpha}{\gamma\cdot\theta}\right)^{\frac{1}{\ell (\varepsilon t+1)}}$ (Lemma~\ref{lem:lb_q})}};

\addplot [solid, color=TUMGreenDark, thick, mark=diamond] table[row sep=\\] {
1.000000 3426.852562 \\
2.000000 166.699487 \\
3.000000 40.991539 \\
4.000000 18.387193 \\
5.000000 10.956583 \\
6.000000 7.631492 \\
7.000000 5.844183 \\
8.000000 4.761175 \\
9.000000 4.047704 \\
10.000000 3.548001 \\
11.000000 3.181350 \\
12.000000 2.902351 \\
13.000000 2.683767 \\
14.000000 2.508383 \\
15.000000 2.364848 \\
16.000000 2.245401 \\
17.000000 2.144574 \\
18.000000 2.058413 \\
19.000000 1.983995 \\
20.000000 1.919113 \\
};
\addlegendentry{{$\left(\frac{r}{\beta}\right)^{\frac{(\alpha-1)t}{f(t)}}$ (Lemma~\ref{lem:ub_q})}};

\addplot [solid, color=orange, thick, mark=star] table[row sep=\\] {
1.000000 2.000000 \\
2.000000 4.000000 \\
3.000000 8.000000 \\
4.000000 16.000000 \\
5.000000 32.000000 \\
6.000000 64.000000 \\
7.000000 128.000000 \\
8.000000 256.000000 \\
9.000000 512.000000 \\
10.000000 1024.000000 \\
11.000000 2048.000000 \\
12.000000 4096.000000 \\
13.000000 8192.000000 \\
14.000000 16384.000000 \\
15.000000 32768.000000 \\
16.000000 65536.000000 \\
17.000000 131072.000000 \\
18.000000 262144.000000 \\
19.000000 524288.000000 \\
20.000000 1048576.000000 \\
};
\addlegendentry{{$2^t$}};

\addplot [dashed, color=black, thick, forget plot] table[row sep=\\] {
1.000000 3426.852562 \\
5.000000 3426.852562 \\
};
\addlegendentry{{$2^t$}};

\addplot [dashed, color=black, thick, forget plot] table[row sep=\\] {
1.000000 8.000000 \\
5.000000 8.000000 \\
};
\addlegendentry{{$2^t$}};

\draw[<->,>=latex, thick] (axis cs: 4.000000, 3426.852562) -- node[right] {max $\gap(\cN) \approx 8.74$ bits} (axis cs: 4.000000, 8.000000);

\addplot [dashed, color=black, thick, forget plot] table[row sep=\\] {
1.000000 59.343413 \\
8.000000 59.343413 \\
};
\addlegendentry{{$2^t$}};

\addplot [dashed, color=black, thick, forget plot] table[row sep=\\] {
1.000000 32.000000 \\
8.000000 32.000000 \\
};
\addlegendentry{{$2^t$}};

\draw[<->,>=latex, thick] (axis cs: 7.000000, 59.343413) -- node[right] {min $\gap(\cN) \approx 0.89$ bits} (axis cs: 7.000000, 32.000000);

\end{axis}
\end{tikzpicture}


%% file: chap_mul_poly.tex
\noindent
Reed-Muller (RM) codes are a class of well-studied evaluation codes of low-degree multivariate polynomials.
The restrictions to the evaluation points that fall on one line in the evaluation space can be readily seen to be equivalent to the evaluation of a low-degree univariate polynomial.
This property gives RM codes the desired properties of being \emph{locally testable}\footnote{Locally testable codes are codes where the membership in the code is verifiable with a constant number of queries (independent from the code length).} \cite{rubinfeld1996robust} and \emph{locally decodable}\footnote{
Locally decodable codes are codes that allow a single bit of the original message to be decoded with high probability by only examining (or querying) a small number of bits of a possible corrupted codeword.
} \cite{lipton1990efficient,babai1991checking}, which have been a subject of extensive studies, e.g.,~\cite{arora1997improved,alon2005testing,bhattacharyya2010optimal,ron2012new,minzer2023optimal}, in the last years.
However, the obvious drawback of RM codes with the nice local properties is their rather low rate. Concretely, for an RM code based on $m$-variate polynomials, the rate is $\leq\frac{1}{m!}$.
In recent years, several new families of codes were proposed to overcome the rate bottleneck of RM codes while preserving the local properties, such as,
multiplicity codes \cite{kopparty2014high,kopparty2015list}, lifted codes \cite{guo2013new,kopparty2014high}, expander codes \cite{hemenway2015local,hemenway2018linear} and tensor product codes \cite{viderman2010note,kopparty2020list}. Among these new code constructions, multiplicity codes and lifted codes are also evaluation codes based on multivariate polynomials.
This chapter concerns constructions of codes with local properties based on evaluation codes.
We give a brief introduction of lifted codes based on multivariate polynomials in \cref{sec:lift-multivariate}.
In \cref{sec:QCLRS}, we introduce a new class bivariate evaluation codes, the \emph{quadratic-lifted Reed-Solomon} (QLRS) codes and study their dimension, distance and local recovery property.
Motivated by a class codes with local properties -- \emph{batch codes}, we introduce a family of subspaces that can be used to construct such codes and propose a construction of such subspaces based on
the best-known evaluation codes -- Reed-Solomon (RS) codes, in \cref{sec:AAD-subspaces}.
Aside from the explicit construction, we give an upper and a lower bound on the growth of the cardinally of such family of subspaces.

\emph{The results in \cref{sec:QCLRS} have been published in the proceedings of 12th International Workshop on Coding and Cryptography (WCC)~\cite{liu2022quadratic} and the results in \cref{sec:AAD-subspaces} have been published in Finite Fields and Their Applications~\cite{liu2021almost}.}

\section{Lifted Codes on Multivariate Polynomials}
\label{sec:lift-multivariate}
Lifted codes were introduced by Guo, Kopparty and Sudan~\cite{guo2013new} as evaluation codes obtained from multivariate polynomials over finite fields. It is required that the restrictions of every codeword on subsets of coordinates are codewords of a \emph{base} code, e.g., an RS code.
By construction, lifted codes are equipped with local properties.
A setting of particular interest, the \emph{lifted Reed-Solomon} codes \cite{guo2013new}, is the case where the restrictions of every codeword on all the lines in the evaluation space form codewords of an RS code.
A surprising advantage of lifted Reed-Solomon codes is that they achieve much larger asymptotic rate compared to RM codes as the field size grows.
The work \cite{guo2013new} derived good bounds on the rate of lifted Reed-Solomon for the bi-variate case.
Tight asymptotic bounds for general cases were derived in \cite{polyanskii2019trivariate,holzbaur2020lifted}.


\emph{Degree-lifted} codes are a class of evaluation codes introduced in \cite{eli2013degree-lifted}. The codes are composed of low-weighted-degree polynomials and codewords are evaluations of such polynomials at the rational points of an algebraic curve.
A class of lifted codes based on code automorphisms was introduced in \cite{guo2016high-rate}.
The works \cite{Li2019lifted-multiplicity,wu2015revisiting,holzbaur2021lifted} studied lifted multiplicity codes.
\emph{Hermitian-lifted codes}, proposed in \cite{lopez2021hermitian}, are constructed from evaluating bivariate polynomials at points on Hermitian curves. The restriction of the codewords on any line for a codeword of an RS code.
\emph{Wedge-lifted codes}~\cite{hastings2021wedge} utilize the trace operation to obtain binary codes with good locality properties.
\emph{Weighted $\eta$-lifted codes} introduced in
\cite{lavauzelle2020weighted} are a class of multivariate evaluation codes where the restriction on the points on a higher degree curve (instead of a line) forms a codeword of an RS code. This gives a more general definition of QLRS codes investigated in \cref{sec:QCLRS}.

\section{Quadratic-Lifted Reed-Solomon Codes}
\label{sec:QCLRS}
We first give some notations and preliminaries needed to define the code and to present the results.
Our goal is to study the properties of QLRS codes, specifically, dimension, minimum Hamming distance and local recovery capability, which are respectively investigated in \cref{QCLRS:dimension}, \cref{QCLRS:distance} and \cref{QCLRS:local-recovery}.

For non-negative integers $a,b\in\bbN$ with \emph{binary representations}
$a=(a_1,\dots, a_{\ell})_2$, $b=(b_1,\dots, b_{\ell})_2$,
we say that \emph{$a$ lies in the 2-shadow of $b$}, denoted by $a\leq_2 b$, if $a_i\leq b_i,\ \forall i \in
[\ell]$.
The bit
$a_{\ell}$
is the most significant bit in the binary representation of $a$, i.e., $a=a_1+a_2\cdot 2+a_3\cdot 2^2+\cdots +a_{\ell}\cdot 2^{\ell-1}$.

  We call a \emph{quadratic curve} on $\Fq^2$ the set of zeros of a bivariate polynomial $p(x,y)=y-\phi(x)$, where $\phi\in\Fq[x]$ with $\deg \phi\leq 2$ is a \emph{quadratic function}.
  We denote by $\Phi$ the set of all quadratic functions over $\Fq$,
  \begin{align}\label{eq:Phi}
  \Phi\defeq \set*{\left.\phi(x)=\alpha x^2+\beta x+\gamma\ \right|\ \forall \alpha,\beta,\gamma \in \Fq}\ .
  \end{align}
For a bivariate polynomial $f\in\Fq[x,y]$ and a quadratic function $\phi\in\Fq[x]$, we define the \emph{restriction} of $f$ on $\phi$ as
\begin{align*}
  f|_{\phi}\defeq f(x, \phi(x))\in\Fq[x]\ .
\end{align*}
Note that the definition of the \emph{restriction} here can be also written as the restriction of the vector $\parenv*{f(a)}_{a\in\Fq}\in\Fq^q$ on the quadratic curve corresponding to $\phi$, so that this is the same underlying technique as used for the other lifted codes mentioned in \cref{sec:lift-multivariate}.

Define an operation $\modstarq$ that takes a non-negative integer and maps it to an element in $[0, q-1]$ as follows:
\begin{align*}
  a\modstarq :=
  \begin{cases}
    a, & \text{if } a\leq q-1\\
    q-1, & \text{if } a\ (\modns\ q-1) =0,\ a\neq 0\\
    a \ (\modns\ q-1), & \text{else}
  \end{cases}
\end{align*}
It can be readily seen that if $a\modstarq=b$, then $x^a=x^b\ (\modns\ x^q-x)$ in $\Fq[x]$.

\begin{lemma}[Lucas' Theorem~\cite{LucasTheorem}]
  \label{lem:lucasThm}
  Let $p$ be a prime and $a,b\in\bbN$ be written in $p$-ary representations
  $a=(a_1,\dots, a_{\ell})_p$, $b=(b_1,\dots, b_{\ell})_p$.
  Then
  \begin{align*}
    \binom{a}{b}=\prod_{i=1}^\ell\binom{a_i}{b_i}\Mod p\ .
  \end{align*}
  If $p=2$, then $\binom{a}{b}=1$ if and only if $b\leq_2 a$.
\end{lemma}


QLRS codes generalize lifted Reed-Solomon codes by considering the restriction of bivariate polynomials to quadratic functions rather than only linear functions.
This definition is coincidentally identical to the weighted $\eta$-lifted Reed-Solomon codes \cite[Def.~IV.1]{lavauzelle2020weighted} with $\eta=2$.
The formal definition is the following.
\begin{definition}[Quadratic-lifted Reed-Solomon (QLRS) codes]
  \label{def:curve-liftedRScodes}
   Let $q$ be a power of $2$, $r\in[q-1]$ and $\Phi$ be a set of quadratic functions as given in \cref{eq:Phi}. A quadratic-lifted Reed-Solomon (QLRS) code is defined as
  \begin{align*}
    \cC_{q}(\Phi,q-r):=\{f\in\Fq[x,y]\ |\ \deg (f|_{\phi})<q-r,\ \forall \phi\in\Phi \}\ .
  \end{align*}
\end{definition}
The integer $r$ can be seen as the \emph{local redundancy} since it is the redundancy of the Reed-Solomon code $\RS_q(q-r)$ such that $f|_\phi$ is a codeword polynomial in $\RS_q(q-r)$.
\subsection{Dimension of Quadratic-Lifted Reed-Solomon Codes}
\label{QCLRS:dimension}

The dimension of the lifted Reed-Solomon codes is analyzed via the number of \emph{good monomials}, which are the multivariate monomials whose restriction to a line results in a codeword of an RS code.
The linear span of these good monomials is shown to generate the lifted Reed-Solomon code \cite{guo2013new,holzbaur2020lifted}.
Similarly, we investigate the dimension of QLRS codes using the \emph{good monomials} as the tool.
We first derive a necessary and sufficient condition (\cref{lem:goodCond}) on the good monomials for any QLRS code via similar approaches as in~\cite{hastings2021wedge}. Then we show in \cref{thm:dim_good} that these good monomials form a basis of the QLRS code as in \cite{Li2019lifted-multiplicity}.
\begin{definition}[{$(\Phi,q-r)^*$}-good monomial]
  \label{def:goodMonomials}
  Given a set $\Phi$ of quadratic functions, a monomial $m(x,y)=x^ay^b$ is \emph{$(\Phi,q-r)^*$-good} if $\deg (m|_{\phi})<q-r, \forall \phi \in\Phi$. The monomial is \emph{$(\Phi,q-r)^*$-bad} otherwise.
\end{definition}

\begin{lemma}
  \label{lem:goodCond}
  A monomial $m(x,y)=x^ay^b$ is $(\Phi,q-r)^*$-good if and only if
  \begin{align}\label{eq:good_mono_suf_nec_cond}
    2i+j+a \ (\modstar\ q) < q-r,\ \forall i\leq_2 b,\ j\leq_2 b-i\ .
  \end{align}
\end{lemma}
\begin{proof}
  Let $\phi(x)=\alpha x^2+\beta x+\gamma\in\Phi$. The restriction of the monomial $m(x,y)=x^ay^b$ on $\phi$ is
  \begingroup
  \allowdisplaybreaks
  \begin{align*}
    m|_{\phi}(x)&=x^a(\alpha x^2+\beta x+\gamma)^b\\
             &=x^a\sum_{i=0}^b\binom{b}{i}\alpha^ix^{2i}\cdot (\beta x+\gamma)^{b-i}\\
             &=\sum_{i=0}^b\binom{b}{i}\alpha^ix^{2i+a}\cdot \sum_{j=0}^{b-i}\binom{b-i}{j} \beta^jx^j\cdot \gamma^{b-i-j}\\
                &\overset{(*)}{=} \sum_{i\leq_2 b}\alpha^ix^{2i+a}\cdot \sum_{j\leq_2 b-i} \beta^jx^j\cdot \gamma^{b-i-j}\\
    &=\sum_{i\leq_2 b}\ \sum_{j\leq_2 b-i}\alpha^i\cdot \beta^j\cdot \gamma^{b-i-j}\cdot x^{2i+j+a}\ ,
  \end{align*}
  \endgroup
  where the equality $(*)$ follows from the Lucas' Theorem (\cref{lem:lucasThm}). If the condition~\eqref{eq:good_mono_suf_nec_cond} in the statement holds, then $\deg m_\phi(x)<q-r$. Therefore the sufficiency is proven.

  Let $m|_\phi^*(x)\defeq m|_{\phi}(x)\ \modns\ x^q-x$. The coefficient of $x^s$ in $m|_\phi^*(x)$ is
  \begin{align*}
    [x^s]m|_\phi^* = \sum_{\substack{i\leq_2 b,\ j\leq_2 b-i\\2i+j+a\modstarq =s}}\alpha^i\cdot \beta^j\cdot \gamma^{b-i-j}\ ,
  \end{align*}
  which can be seen as a multivariate polynomial in $\Fq[\alpha,\beta,\gamma]$.
  Assume the condition~\eqref{eq:good_mono_suf_nec_cond} does not hold but $m(x,y)$ is $(\Phi,q-r)^*$-good, i.e., for any $s\geq q-r$, $[x^s]m|_\phi^*$ is not a zero polynomial but the evaluations at all $(\alpha,\beta,\gamma)\in\Fq^3$ equal to $0$.
  However, by~\cref{lem:nullstellensatz}, since the exponents $i,j,b-i-j<q$, there exists some $(\alpha_0,\beta_0,\gamma_0)\in\Fq^3$ such that the evaluation of $[x^s]m|_\phi^*$ at $(\alpha_0,\beta_0,\gamma_0)$ is nonzero.
  By contradiction we have proven that the condition~\eqref{eq:good_mono_suf_nec_cond} is also a necessary condition.
\end{proof}

The first important result is that the dimension of the code $\cC_q(\Phi,q-r)$ is exactly the number of $(\Phi,q-r)^*$-good monomials, which is presented in~\cref{thm:dim_good}.
In order to prove this, we first discuss in the following lemma a special case that will be excluded in the proof of~\cref{thm:dim_good}.
\begin{lemma}\label{lem:difference_q-1}
  Consider two monomials $m_1(x,y)=x^{q-1}y^b$ and $m_2(x,y)=y^b$ with $b\in[0, q-1]$ and a polynomial $P(x,y)$ containing $m_1$ and $m_2$, i.e.,
  \begin{align*}
    P(x,y) &= (\xi_1 x^{q-1}y^{b}+ \xi_2 y^b) + P'(x,y)
  \end{align*}
  where $\xi_1,\xi_2\neq 0$ and $P'(x,y)$ does not contain $m_1$ or $m_2$. Then, $P$ is $(\Phi,q-r)^*$-bad for any $r\in[q-1]$. 
\end{lemma}
\begin{proof}
  Consider the restriction of $P$ on the function $\phi(x)=\gamma$ for some $\gamma\in \Fq$ and $\alpha=\beta=0$,
\begin{align*}
    P|_\phi(x) &=P(x,y=\gamma)= \gamma^b (\xi_1 x^{q-1}+ \xi_2) + P'(x,y=\gamma) \ .
\end{align*}
First, observe that for this choice of $\alpha,\beta$, $P|_\phi(x)$ is of degree at most $q-1$ and we are only interested in the coefficient of $x^{q-1}$. Further, the only monomials of $P'(x,y=\gamma)$ that contribute to this coefficient are of the form $\xi' x^{q-1} \gamma^{b'}$ with $\xi'\neq 0$.
Since $P'(x,y)$ does not contain the monomials $m_1, m_2$ by definition, we conclude that $b' \neq b$.
Now consider the coefficient of $x^{q-1}$ in $P|_\phi(x)$:
\begin{align*}
    [x^{q-1}]P|_\phi = \underbrace{\gamma^b \xi_1}_{\text{from }m_1+m_2} + \underbrace{\gamma^{b'} \xi' + \ldots}_{\text{from } P'(x,y)} \ .
\end{align*}
We view this as a polynomial in $\Fq[\gamma]$.
Since $b\neq b'$ and $\xi_1,\xi'\neq 0$, this is not a zero polynomial. Also, as $b,b'\in[0,q-1]$ this is a polynomial in $\Fq[\gamma]$ of degree $\leq q-1$.
By~\cref{lem:nullstellensatz}, there exists $\gamma\in\Fq$ such that $[x^{q-1}]P|_{\phi}\neq 0$, which means $P|_{\phi}(x)$ is of degree $q-1$ for some $\gamma$.
Therefore, $P$ is $(\Phi,q-r)^*$-bad according to~\cref{def:goodMonomials} for any $q-r\leq q-1$ (equivalently, for any $r\in[q-1]$).
\end{proof}
\begin{theorem}
  \label{thm:dim_good}
  Let $q$ be a power of $2$, $r\in[q-1]$ and $\Phi$ be the set of all quadratic functions.
  The dimension of the QLRS code $\cC_{q}(\Phi,q-r)$ is the number of $(\Phi,q-r)^*$-good monomials.
\end{theorem}

\begin{proof}
Assume a polynomial $P\in\Fq[x,y]$ containing $(\Phi,q-r)^*$-bad monomials is $(\Phi,q-r)^*$-good.
Let $\cG$ and $\cB$ be subsets of indices of all $(\Phi,q-r)^*$-good and -bad monomials, respectively (assuming the monomials are ordered according to some order).
We can write $P$ as
\begin{align*}
    P=\sum_{c\in\cG} \xi_c x^{a_c}y^{b_c}+\sum_{c\in\cB} \xi_c x^{a_c}y^{b_c}\ ,
\end{align*}
with $\xi_c\in \F_q^*$. Restricting $P$ on the quadratic function $\phi(x)=\alpha x^2+\beta x +\gamma$ gives the following univariate polynomial
\begin{align*}
    P|_{\phi}&=\sum_{c\in\cG\cup\cB} \xi_c x^{a_c}(\alpha x^2+\beta x+\gamma)^{b_c}\\
    &=\sum_{c\in\cG\cup\cB} \xi_c \sum_{i=0}^{b_c}\sum_{j=0}^{b_c-i}\binom{b_c}{i}\binom{b_c-i}{j}\alpha^i\cdot \beta^j\cdot \gamma^{b_c-i-j}\cdot x^{2i+j+a_c}\ .
\end{align*}

Let $P|_{\phi}^*\defeq P|_{\phi} \Mod (x^{q}-x)$. Denote by $[x^s]P|_{\phi}^*$ the coefficient of $x^s$ in $P|_{\phi}^*$. By \cref{lem:lucasThm}, we have
\begin{align}
    [x^{s}]P|_{\phi}^*&=\sum_{c\in\cG\cup\cB}\sum_{\substack{i\leq_2 b_c,\ j\leq_2 b_c-i\\2i+j+a_c\modstarq = s}}\xi_c\cdot  \alpha^i\cdot \beta^j\cdot \gamma^{b_c-i-j}\nonumber\ .
\end{align}
The $(\Phi,q-r)^*$-good monomials do not contribute to the coefficients for $s\geq q-r$ (see~\cref{def:goodMonomials}), therefore,
\begin{align}
    [x^{s}]P|_{\phi}^*&=\sum_{c\in\cB} \sum_{\substack{i\leq_2 b_c,j\leq_2 b_c-i\\2i+j+a_c\modstarq = s}}\xi_c\cdot \alpha^i\cdot \beta^j\cdot \gamma^{b_c-i-j}\quad \text{for }s\geq q-r\ .\label{eq:poly_alpha_beta_gamma}
\end{align}

We view $[x^{s}]P|_{\phi}^*$ as a trivariate polynomial in $\F_q[\alpha,\beta,\gamma]$.
Note that $P$ is $(\Phi,q-r)^*$-good only if
\begin{align}
\label{eq:P-good-cond}
    [x^{s}]P|_{\phi}^*\ (\alpha,\beta,\gamma)&=0\ ,\quad \forall \alpha,\beta,\gamma\in\Fq, \forall s\geq q-r\ .
\end{align}

Now consider two bad monomials $x^{a_c}y^{b_c}$ and $x^{a_d}y^{b_d}$ with $c,d\in\cB$. 
Then the corresponding terms in~\eqref{eq:poly_alpha_beta_gamma} contributed by them can be added up
only if $\alpha^{i_c}\beta^{j_c}\gamma^{b_c-i_c-j_c}=\alpha^{i_d}\beta^{j_d}\gamma^{b_d-i_d-j_d}$, which is true if and only if
\begin{align*}
    \iff &\begin{cases}
      &i_c=i_d\\
      &j_c=j_d\\
      &b_c-i_c-j_c=b_d-i_d-j_d\\
      &2i_c+j_c+a_c \modstarq = 2i_d+j_d+a_d\modstarq
    \end{cases}\\
     \Longrightarrow &\begin{cases}
      &b_c=b_d\\
      &|a_c-a_d|=0 \textrm{ or } q-1\ .
    \end{cases}
\end{align*}
For the case $|a_c-a_d|=q-1$, such polynomials are bad according to~\cref{lem:difference_q-1}.
For the case $|a_c-a_d|=0$, we can conclude that
\eqref{eq:poly_alpha_beta_gamma} is in its simplest form\footnote{A polynomial is in its simplest form if no terms can be further combined.}.

Assume $\cB$ is non-empty. Since $\xi_c\neq 0$ for all $c$, \eqref{eq:poly_alpha_beta_gamma} is a nonzero polynomial in $\Fq[\alpha, \beta,\gamma]$.
By~\cref{lem:nullstellensatz}, since
the exponents $i,j,b_c-i-j<q$, there exists some $\alpha_0,\beta_0,\gamma_0\in\Fq$, such that $[x^{s}]P|_{\phi}^*(\alpha_0,\beta_0,\gamma_0)\neq 0$. This contradicts the assumption that $P$ is $(\Phi,q-r)^*$-good and implies that \eqref{eq:P-good-cond} can be fulfilled only if $[x^{s}]P|_{\phi}^*$ is a zero polynomial, i.e., $\cB$ is empty. Hence, a polynomial $P$ is $(\Phi,q-r)^*$-good if and only if it only consists of good monomials.
\end{proof}

\subsubsection{Counting $(\Phi,q-r)^*$-Bad Monomials}
By \cref{thm:dim_good}, we can calculate the dimension by
\begin{align*}
  k&=\textrm{the number of }(\Phi,q-r)^*\textrm{-good monomials}\\
   &=q^2-\textrm{the number of }(\Phi,q-r)^*\textrm{-bad monomials.}
\end{align*}
Since it is hard to directly analyze the $(\Phi,q-r)^*$-bad monomials,
we first consider a slightly different notion of \emph{$(\Phi,q-r)$-bad monomials} as given in \cref{def:bad-monomials-exponents}.
Then, we derive upper and lower bounds on the number of $(\Phi,q-r)^*$-bad monomials in \cref{thm:exact_bounds_bad} and finally establish the results on the rate of QLRS codes in \cref{thm:asym_rate}.


\begin{definition}[$(\Phi,q-r)$-bad monomials]
  \label{def:bad-monomials-exponents}
  Let $q=2^\ell$ and $r\in[q-1]$.
  A monomial $m(x,y)=x^ay^b$ (or the exponents $(a,b)$) is \emph{$(\Phi,q-r)$-bad} if
  there exist $i\leq_2 b$ and $j\leq_2 b-i$ such that $2i+j+a \pmod{q} \geq q-r$.
  For an integer $t\geq 0$, we define
  \begin{align}
    S_t(\ell)&\defeq \bigg\{\left. (a,b)\in \mathbb{Z}_q^2\quad \right|\hspace{-3em}
               \begin{split}
                 &\exists\ i\leq_2 b,\ j\leq_2 b-i ,\ \textrm{such that}\\
                 &2i+j+a=q-r'+tq, \textrm{ for some }r'\in[r]
               \end{split}\ \bigg\}\ .
                   \label{eq:def_S_t}
  \end{align}
\end{definition}
For $r\in[q-1]$ and $t\geq 3$, the set $S_t(\ell)$ is empty as $2i+j+a\leq i+b+a\leq 2b+a\leq 3(q-1)< q-r+tq$. Hence, if $x^ay^b$ is $(\Phi,q-r)$-bad, then $(a,b)\in S_0(\ell)\cup S_1(\ell)\cup S_2(\ell)$.

In the following,
we attempt to derive some recursive relations on $S_0(\ell)$, $S_1(\ell)$ and $S_2(\ell)$ for $r\in[q-1]$.
We have two observations in \cref{lem:go-down} and \cref{lem:down-by-one}.
\begin{lemma}\label{lem:go-down}
  Consider $q=2^\ell$, $r<\frac{q}{2}$,
  $a=(a_1,\ldots,a_{\ell})_2$ and $b=(b_1,\ldots, b_{\ell})_2$.
  Let $a'\defeq(a_1,\ldots,a_{\ell-1})_2$ and $b'\defeq(b_1,\ldots, b_{\ell-1})_2$.
  If $(a,b)\in S_0(\ell)\cup S_{1}(\ell)\cup S_{2}(\ell)$, then $(a',b')\in S_0(\ell-1)\cup S_{1}(\ell-1)\cup S_2(\ell-1)$.
\end{lemma}
\begin{proof}
  The condition $(a,b)\in S_0(\ell)\cup S_1(\ell)\cup S_{2}(\ell)$ implies that there exist
  $i=(i_1,\ldots,i_{\ell})_2\le_2 b$ and $j=(j_1,\ldots,j_{\ell})_2\le_2 b-i$
  such that $2i+j+a= q-r' \pmod{q}$ for some $r'\in[r]$.
  Let $i':=(i_1,\ldots, i_{\ell-1})_2$ and $j':=(j_1,\ldots,j_{\ell-1})_2$.
  Clearly $i'\le_2 b'$ and $j'\le_2 b'-i'$ and  $2i'+j'+a' = \frac{q}{2}-r' \pmod{\frac{q}{2}}$.
\end{proof}

\begin{lemma}\label{lem:down-by-one}
  For $t=1,2$, if $(a,b)\in S_t(\ell)$, then $(a,b)\in S_{t-1}(\ell)$. 
\end{lemma}
\begin{proof}
  We first prove for $t=1$. The condition $(a,b)\in S_1(\ell)$ implies that there exists an $i\leq_2 b$ and an $j\leq_2 b-i$ with $2i+j+a = 2q-r'$ for some $r'\in[r]$. The statement $(a,b)\in S_0(\ell)$ means that there exists $i'\leq_2 i$ and $j'\leq_2 j$ such that $2i'+j'+a \in[ q-r,q-1]$.
  Note that for $q-r\leq a\leq q-1$ the statement holds with $i'=j'=0$. Assuming $a< q-r$, we claim that there exists of a pair $(i'\leq_2 i, j'\leq_2 j)$ such that $2i'+j'+a = q-r'$. This claim would imply the required statement.
  Such $i',j'$ can be found by \cref{algo:ij-reduction} where we replace some ones in the binary representations of $i$ and $j$ by zeros to get $i'$ and $j'$ so that $(2i+j)-(2i'+j')=q$.
  Note that the algorithm outputs the correct $i',j'$ for $2i+j>q$ if it enters the $\delta\leq 0$ else-part (\cref{line:deltaleq0}) at some point. Assume the contrary that this does not happen, resulting in that the algorithm output the all-zero $i',j'$ at the end. However, this contrary implies that $\delta=q-(2i+j)>0$ which contradicts the condition that $2i+j+a\geq 2q-r$ while $a<q-r$.

  For $t=2$, given $i,j$ such that $2i+j+a=3q-r'$, which implies that $2i+j>q$, we can find $i',j'$ by \cref{algo:ij-reduction} such that $2i'+j'+a=2q-r'$. This completes the proof.
\end{proof}
\begin{algorithm}
\caption{Deduct $q=2^\ell$ from $2i+j$}\label{algo:ij-reduction}
\DontPrintSemicolon
\SetKw{Init}{Init:}
\KwIn{$\ell\geq 1, i,j$}
\KwOut{$i',j'$ such that $i'\leq_2,j'\leq_2$}
\Init{$i'\leftarrow i,j'\leftarrow j, h\leftarrow \ell,\Delta\leftarrow 1$}\;
\uIf{$h=1$\label{line:h_equal_1}}{ \Return{$i',j'$}}
Let $h\leftarrow h-1$ and $\Delta\leftarrow 2\Delta$\label{line:Delta}\;
Compute $\delta\leftarrow \Delta-i'_{h}-j'_{h+1}$\label{line:compute-delta}\;
\uIf{$\delta>0$}{$i'_{h}\leftarrow 0,\ j'_{h+1}\leftarrow 0$\;
\textbf{Go back to Line}~\ref{line:h_equal_1}}
\uElse{\label{line:deltaleq0}
Let $
                \begin{cases}
                    i'_{h}\leftarrow 0,\quad & \textbf{if }\Delta-i'_{h}=0\\
                    j'_{h+1}\leftarrow 0,\quad & \textbf{if }\Delta-j'_{h+1}=0\\
                    i'_{h}\leftarrow 0,j'_{h+1}\leftarrow 0, \quad & \textbf{if }\Delta-i'_{h}-j_{h+1}=0
                \end{cases}
$\label{line:assign-i-j}\;
\Return $i',j'$}
\end{algorithm}
\begin{example}[A toy example of \cref{lem:down-by-one}]
  Consider the parameters $q=2^\ell=2^4, r=2$.
  It can be seen that $(a,b)=(12,14)=((0011)_2,(0111)_2)$ is in $S_1(4)$, since with $i=2=(0010)_2,j=5=(0101)_2$ we have $i\leq_2 b,j\leq_2 b-i$ and $2i+j+a=2q-r'=30=(01111)_2$ with $r'=2\in[r]$.
  We now apply \cref{algo:ij-reduction} to find $i'\leq_2 i,j'\leq_2 j$ such that $2i'+j'+a=q-r'$ for some $r'\in[r]$:
\begin{enumerate}
    \item[]\textbf{Init:} $i'\leftarrow (0010)_2,j'\leftarrow (0101)_2$, $\Delta=1$ and $h\leftarrow 4$.
    \item[]\textbf{\cref{line:Delta}}: Let $h\leftarrow 3, \Delta\leftarrow 2$.
    \item[]\textbf{\cref{line:compute-delta}}: Compute $\delta\leftarrow \Delta-i'_3-j'_4=2-1-1=0$.
    \item[]\textbf{\cref{line:deltaleq0}}: Since $\delta\not>0$ and $\Delta-i'_3-j'_4=0$,
    \item[]\textbf{\cref{line:assign-i-j}}: $i'_3\leftarrow 0$, $j'_4\leftarrow 0$ and output $i'=(0000)_2=0,j'=(0100)_2=2$.
\end{enumerate}
Note that $i'\leq_2i\leq_2 b$, $j'\leq_2 j\leq_2 b-i$ and $2i'+j'+a=(0111)_2=14=q-r'$ with $r'=2\in[r]$. Hence, $(a,b)$ is also in $S_0(4)$.
\end{example}
It follows from Lemma~\ref{lem:down-by-one} that $x^ay^b$ is $(\Phi,q-r)$-bad if and only if $(a,b)\in S_0(\ell)$.
Together with the observation in \cref{lem:go-down}, we can provide a recursive formula for computing the size of $S_t(\ell)$ for $t=0,1,2$.
\begin{lemma}\label{lem: key ingredient}
For $1\leq r< \frac{q}{2}$, it holds that
\begin{align*}
    |S_0(\ell)|&=3|S_0(\ell-1)|+|S_1(\ell-1)|,\\
    |S_1(\ell)|&=|S_0(\ell-1)|+|S_1(\ell-1)| + |S_2(\ell-1)|,\\
    |S_2(\ell)|&=|S_2(\ell-1)|.
\end{align*}
\end{lemma}
\begin{proof}
  We follow the notations in \cref{lem:down-by-one}. To obtain a valid $S_t(\ell-1)$, we require $r<q^{\ell-1}=\frac{q}{2}$ due to the condition in \cref{lem:go-down}.
 According to \cref{lem:go-down} and \cref{lem:down-by-one}, we know that if $(a,b)\in S_0(\ell)$, then $(a',b')\in S_0(\ell-1)\cup S_1(\ell-1)\cup S_2(\ell-1)$. The statement can be proven by counting how many ways to add the most significant bits $a_\ell$ and $b_\ell$ for $a'$ and $b'$ to obtain $a$ and $b$. Denote them by $a=[a',a_\ell],b=[b',b_\ell]$. Recall the definition in~\eqref{eq:def_S_t}, given $(a',b')\in S_t(\ell-1)$, there exist $i'\leq_2 b',j'\leq_2 b'-i'$ such that $2i'+j'+a'=\frac{q}{2}-r'+t\frac{q}{2}$ with some $r'\in[r]$.
 Construct $i,j$ by appending one most significant bit to $i',j'$, i.e., $i=[i',i_\ell]$ and $j=[j',j_\ell]$ with $i_\ell\leq_2b_\ell$ and $j_\ell\leq_2b_\ell$. To obtain $(a,b)\in S_t(\ell)$, we require $2i+j+a=q-r'' + tq$ with some $r''\in[r]$.
We can write
\begin{align*}
  2i+j+a=2i'+j'+a' + (2i_\ell+j_\ell+a_\ell)\frac{q}{2}\ .
\end{align*}
Since the difference between $2i+j+a$ and $2i'+j'+a'$ is always some multiple of $\frac{q}{2}$, we obtain $r''=r'$.
By \cref{lem:down-by-one}, we have $S_2(\ell-1)\subset S_1(\ell-1)\subset S_0(\ell-1)$.

We first prove $|S_0(\ell)|$. To have $(a,b)\in S_0(\ell)$, we require $2i+j+a=q-r'$. Consider three cases,
\begin{itemize}
\item
  Given $(a',b')\in S_0(\ell-1) \setminus S_1(\ell-1)$, it means 
  $2i'+j'+a'=\frac{q}{2}-r'$. To obtain $2i+j+a=q-r'$, we require $2i_\ell+j_\ell+a_\ell=1$. There are three options of $(a_\ell,b_\ell)$ that this can be fulfilled, i.e., $(a_\ell,b_\ell)=(1,0), (0,1)$ or $(1,1)$.
\item
  Given $(a',b')\in S_1(\ell-1)\setminus S_2(\ell-1)$, we have $2i'+j'+a'=q-r'$. The option $(a_\ell,b_\ell)=(0,0)$
  makes $(a,b)\in S_0(\ell)$. Since $S_1(\ell-1)\subset S_0(\ell-1)$, we can find $i''\leq_2 i'\leq_2 b'$ and $j''\leq_2 j'\leq_2 b-i$ such that $2i''+j''+a'=\frac{q}{2}-r'$ (by~\cref{algo:ij-reduction}). So all the other three options in the first case are also valid for this case.
\item
  Given $(a',b')\in S_2(\ell-1)$, we have $2i'+j'+a'=\frac{3}{2}q-r'$. Since $S_2(\ell-1)\subset S_1(\ell-1)$, all four options of $(a_\ell,b_\ell)$ allow to get $(a,b)\in S_0(\ell)$.
\end{itemize}

Then we show $|S_1(\ell)|$. We require $2i+j+a=2q-r'$. Again, consider the three cases,
\begin{itemize}
\item
  Given $(a',b')\in S_0(\ell-1) \setminus S_1(\ell-1)$, we have $2i'+j'+a'=\frac{q}{2}-r'$. This means that $2i_\ell+j_\ell+a_\ell=3$ is required. $(a_\ell,b_\ell)=(1,1)$ is the only way to add the most significant bits. 
\item
  Given $(a',b')\in S_1(\ell-1)\setminus S_2(\ell-1)$, we have $2i'+j'+a'=q-r'$. We need $2i_\ell+ j_\ell+ a_\ell=2$ to obtain $2i+j+a=2q-r'$. The two options $(0,1)$ and $(1,1)$ allow this.
\item
  Given $(a',b')\in S_2(\ell-1)$, we have $2i'+j'+a'=\frac{3}{2}q-r'$. We require $2i_\ell+j_\ell+a_\ell=1$ to obtain $2i+j+a=2q-r'$. The three options $(a_\ell,b_\ell)=(1,0), (0,1)$ or $(1,1)$ can fulfill this.
\end{itemize}
Now we show $|S_2(\ell)|$. We require $2i+j+a=3q-r'$. Consider the three cases,
\begin{itemize}
\item
  Given $(a',b')\in S_0(\ell-1) \setminus S_1(\ell-1)$, we have $2i'+j'+a'=\frac{q}{2}-r'$. This means that $2i_\ell+j_\ell+a_\ell=5$ is required. However,
  this cannot happen since all $i_\ell,j_\ell,a_\ell$ are in $\F_2$.
\item
  Given $(a',b')\in S_1(\ell-1)\setminus S_2(\ell-1)$, we have $2i'+j'+a'=q-r'$. We need $2i_\ell+j_\ell+a_\ell=4$. However, due to $j_\ell\leq_2 b_\ell -i_\ell$, this cannot happen since $i_\ell$ and $j_\ell$ cannot be one at the same time.
\item
  Given $(a',b')\in S_2(\ell-1)$, we have $2i'+j'+a'=\frac{3}{2}q-r'$. We require $2i_\ell+j_\ell+a_\ell=3$. $(a_\ell,b_\ell)=(1,1)$ is the only option.
\end{itemize}
To sum up, the statements follow from
\begin{align*}
  |S_0(\ell)|&=3(|S_0(\ell-1)\setminus S_1(\ell-1)|)+4(|S_1(\ell-1)\setminus S_2(\ell-1)|)+ 4|S_2(\ell-1)|\\
  |S_1(\ell)|&=|S_0(\ell-1)\setminus S_1(\ell-1)|+2|S_1(\ell-1)\setminus S_2(\ell-1)| +3|S_2(\ell-1)|\\
  |S_2(\ell)|&=|S_2(\ell-1)|\ .
\end{align*}
\end{proof}
Lemma~\ref{lem: key ingredient} yields a recurrence relation for $|S_0(\ell)|$, $|S_1(\ell)|$ and $|S_2(\ell)|$. For a given $r$, the initial value $\ell_0$ should be chosen such that $S_i(\ell_0), i=0,1,2$ is a valid set according to the definition in~\eqref{eq:def_S_t}.
Denote by $\bs(\ell) = (|S_0(\ell)|,|S_1(\ell)|,|S_2(\ell)|)^\top$.
We then have
\begin{align}\label{eq:matA}
  \bs(\ell)= \bA^{\ell-\ell_0}\cdot \bs(\ell_0),\ \text{ where } \bA =
  \begin{pmatrix}
    3 & 1 & 0\\
    1 & 1 & 1\\
    0 & 0 & 1
  \end{pmatrix}\ .
\end{align}
The recursion enables us to find the asymptotic behavior of the number of $(\Phi,q-r)$-bad monomials, which is exactly $|S_0(\ell)|$.
Note that the order of $|S_j(\ell)|, j=0,1,2$ is controlled by $\lambda_1^\ell$, where $\lambda_1=2+\sqrt{2}$ is the largest eigenvalue of $\bA$ in~\eqref{eq:matA}.
Hence,
\begin{align}\label{eq:S0_asym}
  |S_0(\ell)|=\Theta ((2+\sqrt{2})^{\ell})\ .
\end{align}
For different $r$, the exact values of $|S_0(\ell)|$ can be different, since the initial value $|S_0(\ell_0)|$ depends on $r$. However, the asymptotic behavior is the same for any fixed $r$.

We provide the exact expressions of $|S_0(\ell)|$ for $r=1$ and $r=3$, denoted by $|S_0^{(1)}(\ell)|$ and $|S_0^{(3)}(\ell)|$ respectively, which we use later to derive upper and lower bound on the number of $(\Phi,q-r)^*$-bad monomials:
\begin{align}
    |S_0^{(1)}(\ell)|= & \frac{5\sqrt{2}+7}{2(3\sqrt{2}+4)} \cdot \lambda_1^\ell + \frac{5\sqrt{2}-7}{2(3\sqrt{2}-4)} \cdot \lambda_2^\ell \nonumber \\
    \approx & 0.8536 \cdot \lambda_1^\ell +  0.1464 \cdot \lambda_2^\ell \label{eq:S0r1}\\
    |S_0^{(3)}(\ell)|= & \frac{65\sqrt{2}+92}{4(12\sqrt{2}+17)}\cdot \lambda_1^\ell + \frac{65\sqrt{2}-92}{4(12\sqrt{2}-17)} \cdot \lambda_2^\ell - \lambda_3^\ell \nonumber \\
    \approx & 1.3536 \cdot \lambda_1^\ell +  0.6465\cdot \lambda_2^\ell -1 \label{eq:S0r3}
\end{align}
where $\lambda_1=2+\sqrt{2},\lambda_2=2-\sqrt{2},\lambda_3=1$ are the three distinct eigenvalues of the matrix $\bA$.


Recall from Lemma~\ref{lem:goodCond} that a monomial $m(x,y)=x^ay^b$ is $(\Phi,q-r)^*$-bad if and only if there exist $i\leq_2 b$ and $j\leq_2 b-i$ such that $2i+j+a \ (\modstar\ q)\geq q-r$.
For $q=2^\ell$ and $r\in[q-1]$, we define the following set 
\begin{align}
  \label{eq:star-bad-set}
  S^*(\ell):=\bigg\{\left.\vphantom{S^{q^9}}(a,b)\in \mathbb{Z}_q^2\ \right|\ \begin{split}
        \exists\ i\leq_2 b, j\leq_2 b-i ,
        \textrm{ s.t.}&\ 2i+j+a=q-r'+t(q-1),\\& \textrm{ for some }r'\in[r], t\geq 0
    \end{split}\ \bigg\}\ .
\end{align}
It is clear that $(a,b)\in S^*(\ell)$ if and only if $x^ay^b$ is $(\Phi,q-r)^*$-bad.

We first relate $|S^*(\ell)|$ with $|S_0(\ell)|$ in Lemma~\ref{lem:upper_bad} and Lemma~\ref{lem:lower_bad}.
\begin{lemma}\label{lem:upper_bad}
  Let $\ell\geq 2, q=2^\ell$,
  $1\leq r\leq \frac{q}{4}$,
  $s=\ceil{\log_2 (r)}$ and $q'=2^{\ell-s}$. Denote by $S^{(3)}_0(\ell-s)$  the set of $(a,b)$ such that $x^ay^b$ is $(\Phi,q'-3)$-bad. Then
    \begin{align}
      |S^*(\ell)|< 4r^2\cdot |S_0^{(3)}(\ell-s)|\ . \nonumber
    \end{align}
    If $r$ is a power of $2$, then
    \begin{align}
    |S^*(\ell)|\leq r^2\cdot |S_0^{(3)}(\ell-s)| \ . \nonumber
  \end{align}
\end{lemma}
\begin{proof}
By definition, we require $\ell-s\geq 2$ to have a valid $S_0^{(3)}(\ell-s)$. Therefore, we require $\ell\geq 2$ and $r\leq \frac{q}{4}$.
Let $x^a y^b$ be an arbitrary $(\Phi,q-r)^*$-bad monomial.
By definition, this means that there exist $i\le_2 b$ and $j\le_2 b-i$ such that $2i+j+a = q-r' + (q-1)t$ for some $r'\in[r]$ and $t\in[0,2]$.\footnote{Note that $2i+j+a\leq 2b+a\leq 3(q-1)<q-r+(q-1)t$ for any $t\geq 3$ and $r<q$.} We drop $s=\lceil \log (r)\rceil$ least significant bits in $a$, $b$, $i$ and $j$ to obtain $a'$, $b'$, $i'$ and $j'$,
and write
\begin{align*}
    i&=i'\cdot 2^s+r_{i}\\
    2i&=2i'\cdot 2^s + 2r_i\\
    j&=j'\cdot 2^s + r_j\\
    a&=a'\cdot 2^s + r_a
\end{align*}
where the remainders $0\leq r_{i},r_j,r_a<2^s$.
Recall that $q'=q/{2^s}=2^{\ell-s}$, it is clear that
\begin{align}
    2i'+j'+a'&= \frac{2i+j+a}{2^s}-\frac{2r_{i}+r_j+r_a}{2^s}\nonumber\\
    &=\frac{q-r'+(q-1)t}{2^s}-\frac{2r_{i}+r_j+r_a}{2^s}\nonumber \\
    &=q'(t+1)-\frac{r'+t}{2^s}-\frac{2r_{i}+r_j+r_a}{2^s}\nonumber\ .
\end{align}
Since the bits in $i,j$ cannot be both one at the same position, $2r_i+r_j\leq 2(2^s-1)$ and therefore $0\leq 2r_i+r_j+r_a\leq 3(2^s-1)$. In addition, since $1\leq r'+t\leq r+2\leq 2^s+2$, we have
$$
q'(t+1)-4< 2i'+j'+a'\leq q'(t+1) - \frac{1}{2^s}\ .
$$
As $2i'+j'+a'$ can only be an integer, we have
\begin{align*}
    q'(t+1)-3\leq 2i'+j'+a'\leq q'(t+1) - 1\ .
\end{align*}
This implies that $(a',b')$ is $(\Phi,q'-3)$-bad.
Therefore, adding arbitrary $s$ least significant bits to a pair $(a',b')\in  S_0^{(3)}(\ell-s)$ results in a pair $(a,b)$ that may be $(\Phi,q-r)^*$-bad. The number of $(\Phi,q-r)^*$-bad monomials is therefore bounded from above by
\begin{align*}
2^s\cdot  2^s\cdot |S_0^{(3)}(\ell-s)|&= (2^{\ceil{\log_2(r)}})^2  \cdot |S_0^{(3)}(\ell-s)|< (2r)^2 \cdot |S_0^{(3)}(\ell-s)|\ .
\end{align*}

If $r$ is a power of $2$, we can set $s=\log_2r$ and obtain the tighter bound.
\end{proof}
\begin{lemma}\label{lem:lower_bad}
  Let $\ell\geq 1, q=2^\ell$,
  $1\leq r\leq \frac{q}{2}$,
  $s=\floor{\log_2 r}$ and $q'=2^{\ell-s}$. Denote by $S_0^{(1)}(\ell-s)$ the set of $(a,b)$ such that $x^ay^b$ is $(\Phi,q'-1)$-bad. Then
  \begin{align}
    |S^*(\ell)|> \frac{r^2}{4}\cdot |S_0^{(1)}(\ell-s)|\ . \nonumber
  \end{align}
  If $r$ is a power of $2$, then
  \begin{align}
    |S^*(\ell)|\geq  r^2\cdot |S_0^{(1)}(\ell-s)|\ . \nonumber
  \end{align}
\end{lemma}
\begin{proof}
  By definition, we require $\ell-s\geq 1$ to have a valid $S_0^{(1)}(\ell-s)$. Therefore, we require $\ell\geq 1$ and $r\leq \frac{q}{2}$. Consider a pair $(a',b')\in S_0^{(1)}(\ell-s)$.
  According to the definition in~\eqref{eq:def_S_t}, there exist $i'\leq_2 b,j'\leq_2 b-i$ such that $2i'+j'+a'=q'-1$.

  Construct integers $a,b,i,j$ as
  $a=a''+2^s\cdot a'$, $b=b''+2^s\cdot b'$, $i=2^s\cdot i'$ and $j=2^s\cdot j'$, where $a'',b''\in[0, 2^s-1]$.
  Note that this is equivalent to appending
  the binary representation of $a''$ (resp.~$b''$) on the left to the binary representations of $a'$ (resp.~$b'$), and appending $s$ zeros on the left to the binary representations of $i'$ and $j'$.
  It can be seen that $i\leq_2 b$, $j\leq_2 b-i$ by construction, and
  \begingroup
  \allowdisplaybreaks
\begin{align}
    2i+j+a&=(2i'+j'+a')2^s + a'' \nonumber \\
    &=(q'-1)\cdot 2^s+a''\nonumber\\
    &= q- 2^s + a''\ . \label{eq:2i+j+a}
\end{align}
\endgroup
Since $s=\floor{\log_2 r}$, $q-r\leq\eqref{eq:2i+j+a}\leq q-1$. Hence any choice of $a'',b''\in[0,2^s-1]$ results in a pair $(a,b)$ such that $x^ay^b$ is $(\Phi,q-r)^*$-bad.
In total there are $(2^s)^2 > \left(\frac{r}{2}\right)^2$ ways of choosing $a'',b''$.
If $r$ is a power of $2$, we can set $s=\log_2r$ and obtain a tighter lower bound.
\end{proof}
In the following theorem we provide
upper and lower bounds on the number $|S^*(\ell)|$ of $(\Phi,q-r)^*$-bad monomials in terms of $\ell$ and $r$.
\begin{theorem}\label{thm:exact_bounds_bad}
  Let $\ell\geq 2, q=2^\ell$,
  $1\leq r\leq \frac{q}{4}$,
  $s = \log_2r$ and $S^*(\ell)$ be the set of
  $(\Phi,q-r)^*$-bad monomials as in \eqref{eq:star-bad-set}. Then
    \begin{align}
        \frac{0.8536 \cdot \lambda_1^{\ell-\floor{s}} +  0.1464 \cdot \lambda_2^{\ell-\floor{s}}}{4}< \frac{|S^*(\ell)|}{r^2}< 4(1.3536 \cdot \lambda_1^{\ell-\ceil{s}} +  0.6465\cdot \lambda_2^{\ell-\ceil{s}} -1)\ , \nonumber
    \end{align}
    where $\lambda_1 = 2+\sqrt{2}$ and $\lambda_2 = 2-\sqrt{2}$.

    If $r$ is a power of $2$, we obtain
    \begin{align}
        0.8536 \cdot \lambda_1^{\ell-s} +  0.1464 \cdot \lambda_2^{\ell-s}
        \leq \frac{|S^*(\ell)|}{r^2}\leq 1.3536 \cdot \lambda_1^{\ell-s} +  0.6465\cdot \lambda_2^{\ell-s} -1\ .
        \nonumber
    \end{align}
\end{theorem}
\begin{proof}
It follows directly from the estimation of $|S_0(\ell)|$ in \eqref{eq:S0r1}--\eqref{eq:S0r3} and the bounds in \cref{lem:upper_bad} and \cref{lem:lower_bad}.
\end{proof}

For an illustration, we plot in~\cref{fig:dimPlot1} the rate of the code $\cC_q(\Phi,q-r)$ with $q=2^5$, which is done by computer-search according to the necessary and sufficient condition in~\cref{lem:goodCond}. The lower and upper bounds on the rates for $r\in[1,\tfrac{q}{4}]$ are calculated from the bounds on $|S^*(\ell)|$ in \cref{thm:exact_bounds_bad}.
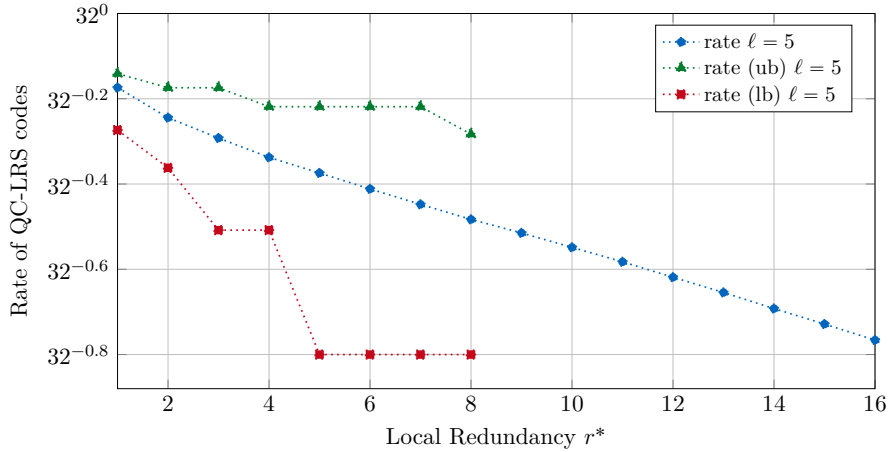
\begin{figure}[h]
  \centering
  \input{figs/QC-LRS/ratePlot_l=5_log.tex}
  \caption{The dimension of QLRS code $\cC_q(\Phi,q-r)$ with $q=2^5$ along with the corresponding upper bound (ub) and lower bound (lb) for $r\in[1, \tfrac{q}{4}]$ calculated by $1-|S^*(\ell)|/q^2$. The lower and upper bound on $|S^*(\ell)|$ are given in \cref{thm:exact_bounds_bad}.}
  \label{fig:dimPlot1}
\end{figure}

\begin{corollary}\label{thm:asym_rate}
  Let $\mu = \log_2(2+\sqrt{2})$. For $q\to\infty$ and
  $1\leq r\leq \frac{q}{4}$,
  the number of $(\Phi,q-r)^*$-bad monomials is 
  \begin{align}
  |S^*(\ell)|=\Theta(r^{2-\mu} q^{\mu})\ .\nonumber
  \end{align}
  Moreover, the QLRS code $\cC_q(\Phi,q-r)$ has rate
  \begin{align*}
    R= 1-\Theta\left((q/r)^{\mu-2}\right)=1-\Theta\left((q/r)^{-0.2284}\right)\ .
  \end{align*}
\end{corollary}
\begin{proof}
It can be seen from \cref{thm:exact_bounds_bad} that the order of $|S^*(\ell)|$ is controlled by $\lambda_1^\ell$. Its asymptotic estimation is obtained by neglecting the other terms and the constant coefficients.
The rate is calculated by dividing the number of good monomials, $q^2-|S^*(\ell)|$ by the number of all bi-variate monomials, $q^2$.
\end{proof}
\begin{remark}\label{rem:LRSasymRate}
  Recall that the rate of bivariate lifted Reed-Solomon codes given in~\cite{holzbaur2020lifted} is
  $$R= 1- \Theta((q/r)^{\log_2 3 -2}=1-\Theta((q/r)^{-0.4150})\ .$$
  We compare the performance of QLRS codes with lifted Reed-Solomon codes in terms of local recovery in~\cref{QCLRS:local-recovery}.
\end{remark}
\subsection{Minimum Hamming Distance of Quadratic-Lifted Reed-Solomon Codes}
\label{QCLRS:distance}
Similar to RS codes and RM codes, a QLRS code $\cC_{q}(\Phi,q-r)$ can be written as a block code with length $n=q^2$ and dimension $k=q^2-|S^*(\ell)|$, where $\ell=\log_2q$, via evaluations:
\begin{align*}
  \QC_q[n,k]\defeq \set*{\left.\parenv*{f(\bv)}_{\bv\in\Fq^2}\ \right|\ f\in \cC_q(\Phi,q-r)}\ .
\end{align*}
The minimum Hamming distance of a block code is given in \cref{def:Hamming-metric}.
We define the minimum Hamming distance of $\cC_{q}(\Phi,q-r)$ as the minimum Hamming distance of $\QC_q[n,k]$, i.e.,
\begin{align*}
  \dH(\cC_{q}(\Phi,q-r))\defeq \dH(\QC_q[n,k])\ .
\end{align*}
We provide upper and lower bounds on the minimum Hamming distance of QLRS codes 
in the following theorem.
\begin{theorem}[Bounds on minimum Hamming distance]
  Let $q$ be a power of $2$, $r\in[q-1]$ and $\Phi$ be the set of all quadratic functions. The QLRS code $\cC_{q}(\Phi,q-r)$ has minimum Hamming distance
  \begin{align*}
    qr+1\leq\dH(\cC_{q}(\Phi,q-r))\leq qr+q\ .
  \end{align*}
\end{theorem}
\begin{proof}
  We first show the upper bound. Let $\cA\subset \Fq$ be a subset with $|\cA|=q-r-1$. Consider a bivariate polynomial $f(x,y)=\prod_{\alpha\in\cA} (x-\alpha)\in\Fq[x,y]$. It can be seen that $\deg(f|_\phi)=q-r-1$ for any $\phi\in\Phi$ therefore $f(x,y)$ is in the code $\cC_{q}(\Phi,q-r)$.
  The zeros of $f(x,y)$ in $\Fq^2$ are $\{(x,y) \ |\ x\in\cA,y\in\Fq\}$. Therefore, the evaluations of $f(x,y)$ in $\Fq^2$ are of weight $q^2-q(q-r-1)=qr+q$. Due to the linearity of the code, the upper bound on the minimum distance is proven.\\
  Now we prove the lower bound. For any nonzero $f\in\cC_{q}(\Phi,q-r)$ consider a point $\bp=(a,b)\in\Fq^2$ such that $f(a,b)\neq 0$. Denote by $\cL_{\bp}\subset\Phi$ the set of lines (i.e., quadratic functions with $\alpha=0$) in $\Phi$ intersecting with each other only at $\bp$. It can be seen that $|\cL_{\bp}|=q$. By definition, for any line $L\in\cL_{\bp}$, $\deg \parenv*{f|_L}<q-r$. Therefore there are at least $r+1$ nonzero
  evaluations of $f$ at the points on $L$.
  Denote by $\wtH(f)$ the number of nonzero evaluations of $f$ on $\Fq^2$ and by $\wtH(f|_L)$ the number of nonzero evaluations of $f$ on $L$, then
\begin{align*}
    \wtH(f)&\geq \sum_{L\in\cL_{\bp}}\big(\wtH(f|_L)\underbrace{-1}_{\textrm{excluding }f(\bp)}\big) \underbrace{+ 1}_{\textrm{including }f(\bp)} \geq qr+1
\end{align*}
Note that the bounds are derived in a similar manner as for lifted Reed-Solomon codes in \cite[Theorem 5.1]{guo2013new}.
\end{proof}

\subsection{Local Recovery of Erasures}
\label{QCLRS:local-recovery}
Consider an erasure channel with erasure probability $\tau$.
A \emph{local recovery set} of an erasure (i.e., a missing codeword symbol) in a codeword is
a set of indices where the erasure can be recovered by accessing only the other codeword symbols whose indices are in the set.
Given a QLRS code over $\Fq$, the number of local recovery sets of any codeword symbol is the number of quadratic functions over $\Fq$ passing through the evaluation point of that codeword symbol, which is $q^2$. For a lifted Reed-Solomon code, the number of local recovery sets of any codeword symbol is $q+1$.

In this section, we are interested in correcting a certain erasure within any local recovery set by QLRS codes.
We say a local recovery fails if the erasure cannot be recovered from any of its local recovery sets.
By the structure of a QLRS code $\cC_q(\Phi,q-1)$, this happens if and only if
there are at least $r$ other erasures in each local recovery set of the erased symbol.

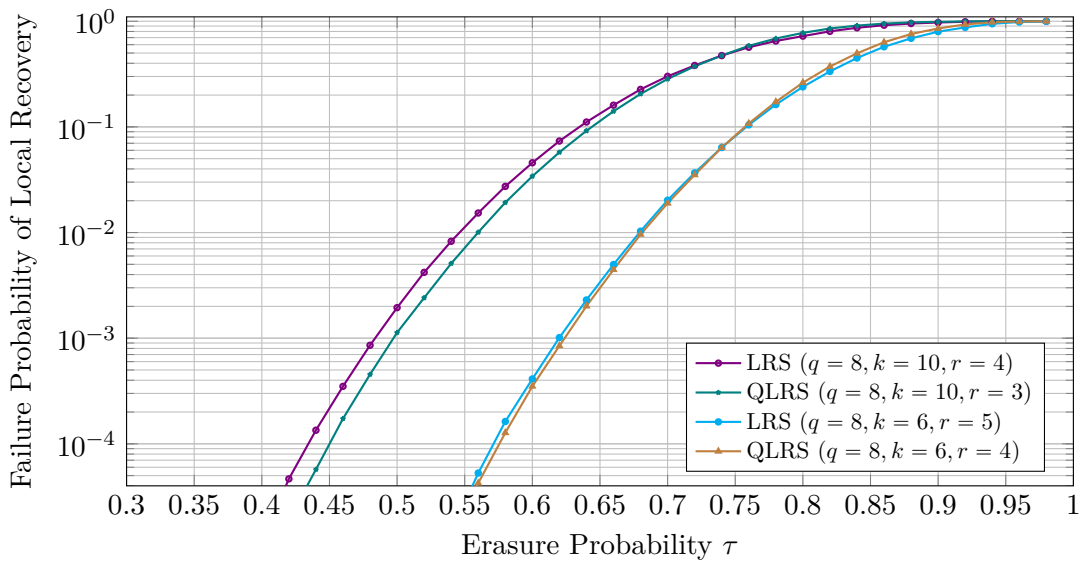
\begin{figure}[bh]
  \centering
    \input{figs/QC-LRS/global_l=3_dim10.tex}
  \caption{Local recovery performance of lifted Reed-Solomon (LRS) codes and QLRS codes of length $n=q^2=64$ and dimension $k=10$ (rate $=k/n=0.15625$) or $k=6$ (rate $=k/n=0.09375$).}
  \label{fig:dim10}
\end{figure}
For a lifted Reed-Solomon code, since all the local recovery sets of a certain codeword symbol are disjoint, the failure probability of a local recovery for an erasure, denoted by $P_{f,LRS}$, is
$$P_{f,LRS}=\left(\sum_{i=r}^{q-1} \binom{q-1}{i} \tau^i(1-\tau)^{q-1-i} \right)^{q+1}\ .$$
For QC-LRS codes, since the local recovery sets may intersect with each other, a closed form of the failure probability is still an open problem.

In order to compare the performance of these two codes, we run simulations with both codes of length $n=64$, dimension $k=10$ and $k=6$, respectively.
It can be seen from \cref{thm:asym_rate} and \cref{rem:LRSasymRate} that for the same local redundancy $r$, the dimension of the lifted Reed-Solomon code is larger than that of QLRS. To have a fair comparison, we choose different local redundancy $r$ for lifted Reed-Solomon codes and QLRS codes such that their dimensions are the same.
For instance, to obtain the same dimension $k=10$, we need to set the local redundancy $r=4$ for the lifted Reed-Solomon code and $r=3$ for the QLRS code.

The simulation results are presented in \cref{fig:dim10}. We can see that for both $k=10$ and $k=6$, the failure probability of a local recovery with QLRS is smaller than or similar to that with lifted Reed-Solomon codes for $\tau\leq 0.7$.
\section{Almost Affinely Disjoint Subspace Design based on Reed-Solomon Codes}
\label{sec:AAD-subspaces}
In this section we show an application of the most well-known evaluation codes -- Reed-Solomon codes -- in constructing a family of $k$-dimensional subspaces of $\Fq^n$. This family of subspaces was motivated by \emph{batch codes}, which is a class of local recovery codes with \emph{availability} for distributed storage systems. We point out the connection after introducing the necessary notations.

\begin{definition}[Almost affinely disjoint (AAD) subspace family]\label{def:AAD-family}
  Given positive integers $k$ and $n$ such that $n>2k$, let $\cF$ be a family of $k$-dimensional linear subspaces in $\F_q^{n}$. This family is said to be \emph{$L$-almost affinely disjoint}, denoted by \emph{$[n,k,L]_q$-AAD}, if the following two properties hold:
  \begin{enumerate}
  \item The family $\cF$ is a partial $k$-spread of $\F_q^n$, i.e., a collection of $k$-dimensional subspaces with pairwise trivial intersection.
  \item \label{item:AAD-2}
    For any $\cS\in \cF$ and $\bu\in\F_q^n\setminus \cS$, the affine subspace of $\cS$ w.r.t.~$\bu$,
    $$\bu+\cS\defeq\{\bu+\bv\ |\ \bv\in S\}$$
    intersects at most $L$ subspaces from the family $\cF$.
  \end{enumerate}
\end{definition}

We denote the \textit{maximal size} of an $[n,k,L]_q$-AAD
family by $m^{AAD}_q(n,k,L)$ and
define the \emph{polynomial growth} of the maximal size of an AAD family as
\begin{align}\label{eq:polynomial-growth}
p^{AAD}(n,k,L)&\defeq\limsup_{q\to\infty} \log_q(m^{AAD}_q(n,k,L))\ .
\end{align}

AAD subspace families with $n=2k+1$ were first introduced by Polyanskii and Vorobyev in \cite{polyanskii2019constructions} to construct \emph{primitive batch codes}~\cite[Definition 2.3]{ishai2004batch}. A binary primitive $[N,K,s]_q$-batch code encodes an information vector $\bx\in\F_2^K$
into a codeword $\bc\in\F_2^N$,
such that for any multiset $\set*{\set*{i_1,\dots, i_s}}\subseteq[K]$, there exists $s$ mutually disjoint sets $\cS_1,\dots,\cS_s\subseteq[N]$ (referred to as \emph{recovering sets}) such that each $x_{i_j},j\in[s]$ can be recovered by the bits in $\bc$ whose indices are in $\cS_j$.
The parameter $s$ is usually called \emph{availability} and it plays an important role in supporting high throughput of the distributed storage system.

It has been shown in \cite[Lemma 2]{polyanskii2019constructions} that a systematic $[N,K,s]_q$-batch code can be constructed from an $[n,k,L]_q$-AAD subspace family $\cF$, where $N=q^n+|\cF|q^{n-k}$, $K=q^n$ and $s=\floor{|\cF|/{L}}$.
The explicit encoding procedure is as follows:
we associate $K=q^n$ information bits with $K$ points in $\F_q^n$ and let the first $K$ bits in $\bc$ equal to $\bx$.
For every affine subspace $\bu+\cS$ with $\cS\in \cF$ and $\bu\in\F_q^n$, we compute a parity-check bit as a sum of information bits associated with the points lying in this affine subspace and append it to $\bc$.
As the number of distinct affine subspaces of such a form is $|\cF|q^{n-k}$, the constructed systematic code has length $N=q^n+|\cF|q^{n-k}$.
By the definition of the $[n,k,L]_q$-AAD family, it can be seen that for every bit in $\bx$, each of its recovery sets (composed of a parity-check bit $c_i$ for some $i>K$ and the information bits that are the other sumands of $c_i$) intersects with at most $L$ recovery sets of any other bit. Hence, $s=\floor{|\cF|/{L}}$.

A naive way to construct AAD families is by exploiting constructions of long linear codes $\cC$ with fixed $\dH(\cC)$.
Let $\bH$ be a parity-check matrix of a linear $[N,K]_q$ code $\cC$ with $\dH(\cC)=3k+1$.
Let the subspace $\cS_i$ be the $\Fq$-linear span of $k$ consecutive columns, from the $(ik+1)$-th to the $(i+1)k$-th column, of $\bH$. Then $\cF=\{\cS_1,\ldots,\cS_{\floor{N/k}}\}$ is an $[N-K,k,1]_q$-AAD family.
Thus, for a fixed minimum Hamming distance, the longer the code $\cC$, the larger the constructed AAD family.
Yekhanin and Dumer have developed a class of long non-binary codes with fixed Hamming distance in~\cite{yekhanin2004long}.
For $k=1$, linear $[N,K]_q$ codes with $\dH=3k+1=4$ are known to be equivalent to caps in projective geometries and have been studied extensively under this name, e.g., in~\cite{mukhopadhyay1978lower,edel1999recursive,hirschfeld2001packing}. From the results in~\cite{edel1999recursive,yekhanin2004long}, for fixed $k$ and large enough $n$, it holds that $p^{ADD}(n,k,1)\ge (3k-1)(n+1)/(9k^2-9k+1)$.

In the rest of the section, we present a construction of $[n,k,L]_q$-AAD families based on Reed-Solomon codes for $k=1,2$ in \cref{sec:AAD-RS-construction}, and new upper and lower bounds on the polynomial growth $p^{ADD}(n,k,L)$ of the maximal size of an AAD family for general $L\geq 1$ in \cref{sec:AAD-upper-bound}.

\subsection{Explicit Constructions}\label{sec:AAD-RS-construction}
\begin{construction}\label{construction:AAD-RS}
  Let $q\ge nk$, $m=q^{n-2k}$ and $\gamma$ be a primitive element of $\F_q$. For $i\in[m]$, let $\cS_i$ be a subspace spanned by the vectors $\{\bv_{i,1},\ldots, \bv_{i,k}\}$ with
  \begin{equation}
    \label{eq:def-v-vec}
    \bv_{i,t} \defeq \begin{pmatrix} \be_t & \Gamma_{t}(\bc_{i}) &  h_t(\bc_i)\end{pmatrix}\in \Fq^n, \ t\in[k]\ ,
  \end{equation}
where $\be_t$ is a unit vector $\in\Fq^k$ with a one at the $t$-th position, $\bc_i$ is a codeword of an $\RS_q[n-k-1,n-2k]$ code having a parity-check matrix as the following
\begin{equation}\label{eq::parity-check matrix}
  \bH_{\RS}:=\begin{pmatrix}
    1 & 1 & 1 & \cdots & 1\\
    1 & \gamma & \gamma^2 & \cdots & \gamma^{n-k-2}\\
    \vdots & \vdots & \vdots & \ddots & \vdots \\
    1 & \gamma^{k-2} & \gamma^{2(k-2)} & \cdots & \gamma^{(n-k-2)(k-2)}
  \end{pmatrix}\ ,
\end{equation}
$\Gamma_t(\cdot)$ is a map
\begin{align*}
  \Gamma_t(\cdot): \F_q^{n-k-1}&\to \F_q^{n-k-1}\\
  \bx & \mapsto
        \begin{pmatrix}
          \gamma^{t-1}x_1 &\gamma^{2(t-1)}x_2 &\gamma^{3(t-1)}x_3 & \cdots & \gamma^{(n-k-1)(t-1)}x_{n-k-1}
        \end{pmatrix}\ ,
\end{align*}
and $h_t(\cdot)$ is a function $h_t(\bx)\defeq \sum_{p=1}^{n-k-1}x_{p}^{(t-1)(n-k-1)+p+1}$.
Then let $\cF_{n,k}\defeq \set*{\cS_1,\dots,\cS_m}$.
\end{construction}
\begin{theorem}\label{th:general construction}
The family $\cF_{n,k}$ from Construction~\ref{construction:AAD-RS} is a partial $k$-spread in $\F_q^n$. Moreover, for $k=1$ and $k=2$, $\cF_{n,k}$ is $[n,k,L({n,k})]_q$-AAD with $L({n,1})=n-1$ and $L({n,2})=1+2(n-2)(2n-6)$.
\end{theorem}
\begin{proof}
  The vectors $\bv_{i,1},\ldots, \bv_{i,k}\in\Fq^n$ are linearly independent as the restriction to the first $k$ coordinates are unit vectors. Hence, their span defines a $k$-dimensional subspace in $\F_q^n$.
  Suppose that $\cS_i$ and $\cS_j$ have a non-trivial intersection.
  Then, 
  \begin{align*}
    \rk
    \begin{pmatrix}
      \bv_{i_1}\\
      \vdots\\
      \bv_{i,k}\\
      \bv_{j,1}\\
      \vdots\\
      \bv_{j,k}
    \end{pmatrix}=\rk
    \begin{pmatrix}
      \be_1 & \Gamma_{1}(\bc_{i}) &  h_1(\bc_i)\\
      \vdots\\
      \be_k & \Gamma_{k}(\bc_{i}) &  h_k(\bc_i)\\
      \be_1 & \Gamma_{1}(\bc_{j}) &  h_1(\bc_j)\\
      \vdots\\
      \be_k & \Gamma_{k}(\bc_{j}) &  h_k(\bc_j)
    \end{pmatrix}\leq 2k-1\ ,
  \end{align*}
  which yields that 
  \begin{align}
    \label{eq:Gamma-two-cws}
    \rk
    \begin{pmatrix}
      \Gamma_1(\bc_i-\bc_j)\\
      \Gamma_2(\bc_i-\bc_j)\\
      \vdots\\
      \Gamma_{k}(\bc_i-\bc_j)
    \end{pmatrix}<k\ .
  \end{align}
  Denote by $c_{i,j}$ the $j$th entry of $\bc_i$. Since $\bc_i-\bc_j$ is a nonzero codeword of the $\RS_q[n-k-1,n-2k]$ code with $\dH=k$, there exist $k$ coordinates $p_1,\ldots,p_{k}\in[n-k-1]$ such that $u_{t}\defeq c_{i,p_t}-c_{j,p_t}\neq 0$ for $t\in[k]$. Thus, after restricting each row of the matrix in \eqref{eq:Gamma-two-cws} to coordinates $p_1,\ldots, p_k$, the rank deficiency of the matrix in \eqref{eq:Gamma-two-cws} is equivalent to
    $$
    \det \begin{pmatrix}
    u_1 & u_2 & \cdots & u_k\\
    \gamma^{p_1}u_1 & \gamma^{p_2} u_2& \cdots & \gamma^{p_k} u_k \\
    \vdots & \vdots & \ddots & \vdots \\
    \gamma^{p_1(k-1)}u_1 & \gamma^{p_2(k-1)} u_2& \cdots & \gamma^{p_k(k-1)} u_k
    \end{pmatrix}
    = \prod_{t=1}^{k} u_t \prod_{1\le s<r\le k} (\gamma^{p_r}-\gamma^{p_s})=0\ .
    $$
    However, since $\gamma$ is a primitive element of the field $\Fq$ with $q\geq nk$ and all $u_t$'s for $t\in[k]$ are nonzero, the determinant cannot be zero, which contradicts the assumption that $S_i$ and $S_j$ intersect non-trivially. Hence, the family $\cF_{n,k}$ is a partial $k$-spread.

    Suppose that for the $\cS_i\in\cF_{n,k}$ and some $\bv\in\F_q^n\not\in \cS_i$, the affine space $\cV\defeq\bv+\cS_i$
    intersects more than $L_{n,k}$ subspaces from the family $\cF_{n,k}$.
    Let $\bv$ be a nonzero vector from
    one of the $L_{n,k}+1$ subspaces intersecting $\cV$. E.g., assume $\bv\in\cS_j\in\cF_{n,k}$, for some $j\in[m]\setminus\set*{i}$. Then we can write $\bv$ as a linear combination of the basis of $\cS_j$. Namely,
    $$\bv=\bv_j(\balpha):=\sum_{t=1}^k \alpha_t \bv_{j,t} \textrm{ for some } \balpha\in\F_q^{k}\setminus\{\0\}\ ,$$
    where $\bv_{j,t}$'s are given in \eqref{eq:def-v-vec}.
    In what follows, we estimate
    the number of other subspaces in $\cF_{n,k}$ such that there exists a linear combination of their basis in $\cV$, i.e.,
    \begin{align}
      \label{eq:num-bad-ell}
      \left|\set*{\ell\in[m]\setminus\{i,j\}\ |\  \bv_{\ell}(\bbeta)\in\cV, \textrm{ for some }\0\neq\bbeta\in\Fq^k}\right|\ .
    \end{align}
    It can be seen that this is the number of subspaces in $\cF_{n,k}$ that intersect $\cV$.
    Note that $\bv_{\ell}(\bbeta)\in\cV$ is is equivalent to the property that
    \begin{align*}
      \rk
      \begin{pmatrix}
        \bv_{i,1}\\
        \vdots\\
        \bv_{i,k}\\
        \bv_j(\balpha)\\
        \bv_\ell(\bbeta)
      \end{pmatrix}
      \leq k+1\ .
    \end{align*}
    By the structure of the vectors $\bv_{i,t}$'s in \eqref{eq:def-v-vec}, the rank deficiency above implies that 
    \begin{align*}
      \rk
      \underbrace{\begin{pmatrix}\sum_{t=1}^{k}\alpha_t \Gamma_t(\bc_j-\bc_i) & \sum_{t=1}^{k}\alpha_t (h_t(\bc_j)-h_t(\bc_i))\\
    \sum_{t=1}^{k}\beta_t \Gamma_t(\bc_\ell-\bc_i) & \sum_{t=1}^{k}\beta_t(h_t(\bc_\ell)-h_t(\bc_i))
    \end{pmatrix}}_{\defeqrev \bR_{\bbeta,\ell}\in\Fq^{2\times (n-k)}}
                                                     =1\ .
    \end{align*}
    Denote
    \begin{align}
      \label{eq:set-bad-cw}
      \cT\defeq\set*{\left.\bc_{\ell}\in\RS_q[n-k-1,n-2k]\ \right|\ \rk(\bR_{\bbeta,\ell})=1 \textrm{ for some }\0\neq\bbeta\in\Fq^k}\ .
    \end{align}
    Then, \eqref{eq:num-bad-ell}$=|\cT|$.

    Now we show that at least one entry in the first $n-k-1$ entries of the first row of $\bR_{\bbeta,\ell}$
    is nonzero.
    Observe that for each $p\in[n-k-1]$, the $p$-th entry in $\sum_{t=1}^{k}\alpha_t \Gamma_t(\bc_j-\bc_i)$ has the form $(c_{j,p}-c_{i,p})\sum_{t=1}^{k}\alpha_t\gamma^{(t-1)p}$.
    Since $\bc_i$ and $\bc_j$ are codewords of the $\RS_q[n-k-1,n-2k]$ code with minimum Hamming distance $k$, there are at least $k$ indices $p\in[n-k-1]$ such that $c_{j,p}-c_{i,p}\neq 0$.
    We see $\sum_{t=1}^{k}\alpha_t\gamma^{(t-1)p}$ as a polynomial of degree $k-1$
    in the unknown $x=\gamma^p$.
    Then the polynomial has at most $k-1$ distinct roots in $\F_q$.
    For each root $x_0$, there is at most one $p\in[n-k-1]$ such that $\gamma^p=x_0$, since $\gamma$ is a primitive element in $\Fq$ with $q\geq nk$.
    Then for any nonzero $\balpha$, there are at most $k-1$ distinct $p$ so that $\sum_{t=1}^{k}\alpha_t\gamma^{(t-1)p}=0$.
    Hence, there is at least one entry $(c_{j,p_*}-c_{i,p_*})\sum_{t=1}^{k}\alpha_t\gamma^{(t-1)p_*}$ in the vector $\sum_{t=1}^{k}\alpha_t \Gamma_t(\bc_j-\bc_i)$ being nonzero.

    To continue, we need the following lemma, whose proof is given later.
    \begin{lemma}\label{lem:given-beta-num-of-solutions}
      Given a nonzero vector $\bbeta\in\F_q^k$, there are at most $k(n-k)$ distinct codeword $\bc_\ell\in\RS_q[n-k-1, n-2k]$
      such that $\rk(\bR_{\bbeta,\ell})=1$. 
    \end{lemma}
    Let us proceed with proving the remaining statement of this theorem, i.e., the value of $L_{n,k}$ for $k=1,2$.
    For this purpose, we estimate the number of possible $\bbeta$'s such that the first $n-k-1$ columns of $\bR_{\beta,\ell}$,
    are collinear to the $p_*$-th column and then apply Lemma~\ref{lem:given-beta-num-of-solutions}.

    For $k=1$, note that for any multiple of $\bbeta=(1)$, the set of $\bc_{\ell}$ such that $\rk(\bR_{\bbeta,\ell})=1$ is the same.
    By Lemma~\ref{lem:given-beta-num-of-solutions}, the number of distinct appropriate $\ell$'s is at most $n-1$. Thus, $\cF_{n,1}$ is an $[n,1,n-1]_q$-AAD family.

    Now we discuss the case $k=2$.
    Note that any nonzero vector $\bbeta \in \Fq^2$ is either collinear to $(1,q-1)$ or $(\beta_1,1-\beta_1)$, where $\beta_1\in\Fq$.
    Define a set
    \begin{align*}
      \cB\defeq \set*{(\beta_1,1-\beta_1)\in\Fq^2\ |\ \beta_1+(1-\beta_1) \gamma^{p}=0, \textrm{ for some }p\in[n-3]}\ ,
    \end{align*}
    and it can be readily seen that $|\cB|\leq n-3$.
    We assume that for $\bbeta=(1,q-1)$ or $\bbeta\in\cB$, the set $\cT$ in \eqref{eq:set-bad-cw} is not empty.
    We now estimate the number of \emph{suspicious} $\widetilde{\bbeta}=(\widetilde{\beta}_1,1-\widetilde{\beta}_1)\not\in \cB$ such that the set $\cT$ may not be empty.
    If $\bR_{\widetilde{\bbeta},\ell}$ has rank $1$, then two rows of  $\bR_{\widetilde{\bbeta},\ell}$
    are collinear and there exists some $\lambda\in\F_q^*$ such that
    \begin{align}\label{eq:colinear-constraint}
      \parenv*{\widetilde{\beta}_1+(1-\widetilde{\beta}_1)\gamma^p}\parenv*{c_{\ell,p}-c_{i,p}}=\underbrace{\lambda(\alpha_1+\alpha_2\gamma^{p})(c_{j,p}-c_{i,p})}_{\defeqrev w_p},\ \forall p\in[n-3]\ .
    \end{align}
    From the parity-check equation $\sum_{p=1}^{n-3}(c_{\ell,p}-c_{i,p})=0$ imposed by the first column of~\eqref{eq::parity-check matrix}, we have
    $$
    \sum_{p=1}^{n-3} \frac{w_p}{\widetilde{\beta}_1+(1-\widetilde{\beta}_1)\gamma^p} = 0\quad \iff  \quad \sum_{p=1}^{n-3}w_p \prod_{\substack{t=1\\t\neq p}}^{n-3}(\widetilde{\beta}_1+(1-\widetilde{\beta}_1)\gamma^t)=0\ .
    $$
    Since there is some $p\in[n-3]$ such that $w_{p}\neq 0$ and $\widetilde{\beta}_1 +(1-\widetilde{\beta}_1)\gamma^{p}\neq 0$ for all $p\in[n-3]$,
    the left-hand side of the above equation is a nonzero polynomial in $\Fq[\widetilde{\beta}_1]$ of degree at most $n-4$.
    Therefore, there are at most $n-4$ suspicious $\widetilde{\bbeta}$'s
    such that the vector $\sum_{t=1}^{2}\widetilde{\beta}_t \Gamma_t(\bc_\ell-\bc_i)$ is collinear to $\sum_{t=1}^{2}\alpha_t \Gamma_t(\bc_j-\bc_i)$.
    Let $\cD$ be the union of the suspicious $\widetilde{\bbeta}$'s, the set $\cB$ and the set $\set*{(1,q-1)}$.
    It can be seen that
    \begin{align*}
      |\cD|&=\left|\set*{\widetilde{\bbeta}}\right|+|\cB|+|\set*{(1,q-1)}|\\
           &\leq (n-4) + (n-3) + 1 \leq 2n-6\ .
    \end{align*}
    Hence, by Lemma~\ref{lem:given-beta-num-of-solutions}, $\cF_{n,2}$ is an $[n,2,L({n,2})]_q$-AAD family with $L({n,2})=1+2(n-2)(2n-6)$.
    \end{proof}
    \begin{proof}[Proof of \cref{lem:given-beta-num-of-solutions}]
      Recall that the entry at the first row, $p_*$-th column of $\bR_{\bbeta,\ell}$ is nonzero.
    If $\bR_{\bbeta,\ell}$ has rank $1$, then each of the first $n-k-1$ columns of $\bR_{\bbeta,\ell}$
    is linearly dependent on the $p_*$-th column. Moreover, the dependency is determined by the first row of $\bR_{\bbeta,\ell}$.
    Fix a $c_{\ell,*}\in\Fq$. For any $p\in[n-k-1]$, let
    $$
    \phi_p= \frac{(c_{j,p} - c_{i,p})\sum_{t=1}^{k}\alpha_t \gamma^{(t-1)p}}{(c_{j,p_*} - c_{i,p_*})\sum_{t=1}^{k}\alpha_t \gamma^{(t-1)p_*}}\ .
    $$
    Then, having the $p$-th column collinear to the $p_*$-th column gives the following system of equations on the unknowns $c_{\ell,p},p\in[n-k-1]\setminus\{p_*\}$:
    \begin{align}\label{eq:first-system}
      (c_{\ell,p} - c_{i,p})\sum_{t=1}^{k}\beta_t \gamma^{(t-1)p} =  \phi_p(c_{\ell,p_*} - c_{i,p_*})\sum_{t=1}^{k}\beta_t \gamma^{(t-1)p_*}\ .
    \end{align}
    Note that $\sum_{t=1}^{k}\beta_t \gamma^{(t-1)p}=0$ for at most $k-1$ distinct $p\in[n-k-1]\setminus\set*{p_*}$. Therefore, \eqref{eq:first-system} provides at least $n-k-2-(k-1)=n-2k-1$ equations on the unknown $c_{\ell,p}$'s.

    Since $c_{\ell,p}$ are entries of a codeword of the $\RS_q[n-k-1,n-2k]$ code with a parity-check matrix \eqref{eq::parity-check matrix}, we also have the following equations:
    \begin{equation}\label{eq:second-system}
      \sum_{p=1}^{n-k-1} \gamma^{(p-1)(t-1)}(c_{\ell,p} - c_{i,p}) = 0,\ \forall t\in[k-1]\ .
    \end{equation}
    Thus, the system of equations~\eqref{eq:first-system}-\eqref{eq:second-system} gives at least $n-k-2$ linearly independent equations on the $n-k-2$ unknowns $\{c_{\ell,p},p\in[n-k-1]\setminus\{p_*\}\}$.
    This system of equations has at most one solution. 
    W.l.o.g., for any $p\in[n-k-1]\setminus\set*{p_*}$, we write $c_{\ell,p}=a_p c_{\ell,p_*}+b_p$ with some $a_p,b_p\in\F_q$ for later use.
    So far, we have shown that given the $c_{\ell,p_*}$, the $\bc_{\ell}$ is uniquely determined if $\rk(\bR_{\bbeta,\ell})=1$.

    To have $\bR_{\bbeta,\ell}$ has rank $1$, we also require that the last column of $\bR_{\bbeta,\ell}$ is collinear to the $p_*$-th column, which implies that 
    $$
    \det
    \begin{pmatrix}\sum_{t=1}^k \alpha_t \gamma^{(t-1)p_*}(c_{j,p_*}-c_{i,p_*}) & \sum_{t=1}^{k}\alpha_t\parenv*{h(\bc_j)-h(\bc_i)}\\
      \sum_{t=1}^k \beta_t \gamma^{(t-1)p_*}(c_{\ell,p_*}-c_{i,p_*}) & \sum_{t=1}^{k}\beta_t(h(\bc_\ell)-h(\bc_i))
    \end{pmatrix}=0\ .
    $$
    Note that the right-bottom entry
    \begin{align*}
      \sum_{t=1}^{k}\beta_t(h(\bc_\ell)-h(\bc_i))
                  &= \sum_{t=1}^{k}\beta_t \sum_{p=1}^{n-k-1}\left((a_p c_{\ell,p_*}+b_p)^{(t-1)(n-k-1)+p+1}-c_{i,p}^{(t-1)(n-k-1)+p+1}\right)
    \end{align*}
    is a polynomial in $c_{\ell,p_*}$ of degree
    at least $p_*+1$ and
    at most $k(n-k-1)+1\le k(n-k)$.
    Therefore, the determinant is a nonzero polynomial in $c_{\ell,p_*}$ of degree
    at least $p_*+1$ and
    at most $k(n-k)$.
    Since $q\ge nk$, there are at most $k(n-k)$ solutions for $c_{\ell,p_*}$ resulting in a zero determinant.
    \end{proof}
    \begin{corollary}
      For $k=1,2, n>2k$ and $L=L(n,k)$, the polynomial growth of the maximum cardinality of an $[n,k,L]_q$-AAD family is
      $$p^{AAD}(n,k,L)\geq n-2k\ .$$
    \end{corollary}
    \begin{proof}
      The statement follows from \cref{thm:AAD-upper-bound} and the cardinality of the family $\cF_{n,k}$ given in \cref{construction:AAD-RS}.
    \end{proof}
\subsection{Bounds on Polynomial Growth of the Cardinality}
\label{sec:AAD-upper-bound}
In this section, we give an upper bound and a lower bound on the polynomial growth $p^{AAD}(n,k,L)$ of the maximal size of an AAD family, which is defined as in \eqref{eq:polynomial-growth}.
\subsubsection{An Upper Bound}
\begin{theorem}[Upper bound]
  \label{thm:AAD-upper-bound}
Fix arbitrary positive integers $L,k,n$ such that $2k< n$ and a prime power $q$. Let $\cF$ be an $[n,k,L]_q$-AAD family. Then
\begin{equation}\label{eq:AAD-upper-bound}
|\cF|\le 1+  L\frac{q^{n-k}-1}{q^{k}-1}\ .
\end{equation}
For $L=q^{o(1)}$, it follows that $p^{AAD}(n,k,L)\le n- 2k$.
\end{theorem}
\begin{proof}
  Let $m\defeq|\cF|$
  and $\cS_i$ be the $i$-th subspace in $\cF$.
  For all $i\in[m]$, let $\bG_i\in \F^{k\times n}_q$ be a matrix such that $\myspan{\bG_i}_r=\cS_i$ and let $\bH_i\in \F^{(n-k)\times n}_q$ be a matrix such that $\myspan{\bH_i}_r=\cS_i^{\perp}$, where $\myspan{\cdot}_r$ denotes the row span.
  Hence, for all $i\in[m]$, $\bH_i\bG_i^\top=\0$.
  Note that for any $j\in[m-1]$, $\bH_m\bG_j^\top\in \F^{(n-k)\times k}_q$ has full column rank because $\cS_m$ and $\cS_j$ have only trivial intersection by \cref{def:AAD-family} of the AAD family $\cF$. 
  For any $j\in[m-1]$, let $\hat{\bG}_j\in\Fq^{(n-2k)\times(n-k)}$ be a full-rank matrix
  such that
  \begin{align}
    \label{eq:def-hat-G}
    \hat{\bG}_j\cdot(\bH_m\bG_j^\top)=\0\ .
  \end{align}

  Now, we prove via contradiction that for any nonzero vector $\bw^\top\in \F^{(n-k)\times 1}_q$, $\hat{\bG}_j\bw^\top=\0$ holds for at most $L$ different $j\in[m-1]$.
  Suppose that for some set $\cJ\defeq\{j_1,\ldots, j_{L+1}\}\subset [m-1]$, we have $\hat \bG_{j_t} \bw^\top = \0$ for every $j_t \in \cJ$.
  This implies that $\bw^\top$ is in the column span of $\bH_m\bG_{j_t}^\top$, i.e., $\bw^\top = \bH_m\bG_{j_t}^\top\by_{j_t}^\top$ for some nonzero $\by_{j_t}^\top\in \F^{k\times 1}_q$.

  Let $\bv^\top\defeq\bG^\top_{j_1} \by^\top_{j_1}\in\Fq^{n\times 1}$. Then,
  $$\forall j_t\in\cJ,\ \bH_m\bG_{j_t}^\top\by_{j_t}^\top=\bw^\top=\bH_m\bG_{j_1}^\top\by_{j_1}^\top = \bH_m\bv^\top\ ,$$
  which means
  $$\forall j_t\in\cJ,\ \bG_{j_t}^\top\by_{j_t}^\top = \bv^\top + \bG_m^\top\bx_{j_t}^\top, \textrm{ for some }\bx_{j_t}^\top\in\F_q^{k\times 1}\ .$$
  But this implies that $\bv+\cS_m$ and $\cS_{j_t}$ intersect non-trivially, for all $j_t\in\cJ$. By \cref{def:AAD-family}, there are at most $L$ different $j$'s so that $\bv+\cS_m$ and $\cS_j$ intersect. This leads to a contradiction.

  We have derived above a necessary condition for
  a collection of
  subspaces to form an $[n,k,L]_q$-AAD family $\cF$, that is,
  for any nonzero vector $\bw^\top\in \F^{(n-k)\times 1}_q$, $\hat{\bG}_j\bw^\top=\0$ holds for at most $L$ different $j\in[m-1]$, where the $\hat{\bG}_j$'s are defined in \eqref{eq:def-hat-G} and $m=|\cF|$.
  It can be seen that the expectation
  \begin{align*}
    \mathbb{E}\left|\set*{\left. j\in[m-1]\ \right|\ \hat \bG_j \bw^\top = \0}\right| =(m-1) \Pr \left\{\hat \bG_j\bw^\top = \0\right\} =(m-1)\frac{q^{k}-1}{q^{n-k}-1}\ .
  \end{align*}
  If the condition in \eqref{eq:AAD-upper-bound} is not fulfilled, i.e., $m > L \frac{q^{n-k}-1}{q^k-1} + 1$, then the above expectation is at least $L+1$, violating the necessary condition.
\end{proof}
\begin{remark}
  Note that if we change the definition of an AAD subspace family by dropping the first property (being a partial spread) in \cref{def:AAD-family}, then the matrices $\bH_m \bG_j^\top$ would have full rank for at least $m-L-1$ different $j$'s. This results in the bound $m \leq  L + 1 +L \frac{q^{n-k}-1}{q^k-1}.$
\end{remark}

\subsubsection{A Lower Bound by Random Construction}
The construction proposed in \cref{sec:AAD-RS-construction} gives a lower bound for $k=1,2$.
For general $k\geq q$, we give a lower bound
via another subspace family, the \emph{almost sparse subspace family} defined below.
\begin{definition}[Almost sparse subspace family]\label{def:AS-family}
  Given positive integers $k$ and $n$ such that $2k<n$, let $\cF$ be a family of $k$-dimensional linear subspaces in $\F_q^{n}$. This family is said to be \emph{$L$-almost sparse}, denoted by \emph{$[n,k,L]_q$-AS}, if the two properties hold:
  \begin{enumerate}
  \item The family $\cF$ is a partial $k$-spread of $\F_q^n$.
  \item Any $(k+1)$-dimensional subspace in $\F_q^n$ intersects non-trivially at most $L$ subspaces from the family $\cF$.
  \end{enumerate}
\end{definition}
It can be readily seen that an
$[n,k,L]_q$-AS family is also
an $[n,k,L-1]_q$-AAD family.
Hence the following lower bound on the cardinality of AS families holds naturally for AAD families.
Similar to \eqref{eq:polynomial-growth},
we denote the \emph{maximal size} of an $[n,k,L]_q$-AS family by $m^{AS}_q(n,k,L)$ and define the \emph{polynomial growth} of $m^{AS}_q(n,k,L)$ as
\begin{align*}
p^{AS}(n,k,L)&\defeq\limsup_{q\to\infty} \log_q(m^{AS}_q(n,k,L))\ .
\end{align*}
\begin{theorem}[Lower bound for AS families]
  \label{thm:AAD-probabilistic-lower-bound}
  For arbitrary positive integers $L,n,k$ such that $2k<n$, and a prime power $q$, there exists an $[n,k,L]_q$-AS family $\cF$ of size
  $$
  |\cF| \leq m^*_q(n,k,L)\defeq q^{n-2k-\frac{(n-k)(k+1)}{(L+1)}}\parenv*{1-o(1)}\ .
  $$
  For fixed $L$, it follows that $p^{AS}(n,k,L)\ge n-2k - \frac{(n-k)(k+1)}{(L+1)}$.
\end{theorem}
\begin{proof}
	The number of $k$-dimensional subspaces in $\F_q^n$ is
    \begin{equation*}
      \quadbinom{n}{k}_q = \prod\limits_{i=0}^{k-1} \frac{q^n-q^i}{q^k-q^i}=\prod\limits_{i=0}^{k-1} \frac{q^{n-i}-1}{q^{k-i}-1}= \Theta(q^{k(n-k)})\ .
    \end{equation*}
	We form a family of $k$-dimensional subspaces, $\cF=\{\cS_1,\dots, \cS_m\}$, of size
    $$m=q^{n-2k-\frac{(n-k)(k+1)}{(L+1)}}$$
    by choosing each subspace $\cS_i$ independently and uniformly with probability $1/\quadbinom{n}{k}_q$. Note that it is possible that $\cS_i=\cS_j$ for some $i\neq j$.

	Define $\xi\defeq|\{(i,j)\ |\ i,j\in[m],\ i<j, \ |\cS_i\cap \cS_j|\neq 1\}|$ as the number of pairs $(i,j)$ such that $\cS_{i}$ and $\cS_{j}$ intersect non-trivially.
    We now estimate the expectation of $\xi$.
    The number of $k$-dimensional subspaces that only trivially intersect with a fixed $k$-dimensional subspace
    is equal to
	$$
	g_k\defeq \prod_{i=1}^{k-1}\frac{q^n-q^{k+i}}{q^k-q^i}\ .
	$$
	Thus, the probability of two random $k$-dimensional subspaces having only the trivial intersection is $g_k/\quadbinom{n}{k}_q$. The expectation of $\xi$ is then upper bounded as follows:
\begingroup
\allowdisplaybreaks
	\begin{align*}
          \mathbb{E}(\xi) &\leq \sum_{1\leq i<j\leq m}\Pr\{|\cS_i\cap \cS_j|\neq 1\}\hphantom{1\leq} \\
          &=  \sum_{1\leq i<j\leq m}\left(1-\frac{g_k}{\quadbinom{n}{k}_q}\right)\\
	&\leq\binom{m}{2}\left(1-\prod_{i=0}^{k-1}\frac{q^n-q^{k+i}}{q^n-q^i}\right)
	\\
	&<m^2\left.\left(1-\left(q^{nk}-q^{n(k-1)}\sum\limits_{i=k}^{2k-1}q^i\right)\right/q^{nk}\right)
	\\
	&=m^2\parenv*{q^{2k-1-n}+o\parenv*{q^{2k-1-n}}}<m\parenv*{q^{-1}+o\parenv*{q^{-1}}}\ .
	\end{align*}
    \endgroup
    By the Markov inequality, we have
	$$
	\Pr\set*{\xi>q^{0.5}\mathbb{E}\parenv*{\xi}}<\frac{\mathbb{E}(\xi)}{q^{0.5}\mathbb{E}(\xi)}=o(1)\ .
	$$
    Since $q^{0.5}\mathbb{E}(\xi)<m\parenv*{q^{-0.5}+o\parenv*{q^{-0.5}}}=o(m)$, we obtain that
    with probability at least $1-o(1)$, there exists a family $\cF$ of size $m$, which contains at most $o(m)$ pairs of subspaces with non-trivial intersections.
    If we delete one of the intersecting subspaces for each pair, then we obtain a family $\cF'\subset \cF$ of subspaces of size at least $m-o(m)$ satisfying the first property in \cref{def:AS-family} for AS subspace families.

	Now we compute the probability of the second property in \cref{def:AS-family} being violated.
    For a fixed  $(k+1)$-dimensional subspace $\cV$, the number of $k$-dimensional subspaces that only trivially intersect with $\cV$ is equal to
	$$
	u_k\defeq\prod_{i=0}^{k-1}\frac{q^n-q^{k+1+i}}{q^k-q^i}\ .
	$$
	Thus, the probability that $\cS_i$ only trivially intersects $\cV$ is
	$u_k/\quadbinom{n}{k}_q$.
    Let $\cF_\cV\subseteq\cF$ be the set of subspaces in $\cF$ that non-trivially intersect $\cV$, i.e., $\cF_{\cV}:=\{\cS\in\cF\ |\ \cS\cap \cV|\neq 1\}$.
  	Applying the union bound, we can bound from above the probability that $\cV$ non-trivially intersects at least $L+1$ subspaces in $\cF$ by
    \begingroup
    \allowdisplaybreaks
	\begin{align*}
	\Pr[|\cF_\cV|\geq L+1]&\leq \binom{m}{L+1}\left(1-\frac{u_k}{\quadbinom{n}{k}_q}\right)^{L+1}\\
	&=\binom{m}{L+1}\left(1-\prod_{i=0}^{k-1}\frac{q^n-q^{k+1+i}}{q^n-q^i}\right)^{L+1}\\
	&<m^{L+1}\left(1-\left.\left(q^{nk}-q^{n(k-1)}\sum_{i=k+1}^{2k}q^{i}\right)\right/q^{nk}\right)^{L+1}
	\\
	&= m^{L+1}\parenv*{q^{2k-n}+o\parenv*{q^{2k-n}}}^{L+1}\\
	&<q^{-(n-k)(k+1)}(1+o(1))\ .
	\end{align*}
    \endgroup
	Recall that the total number of $(k+1)$-dimensional subspaces is $\quadbinom{n}{k+1}_q=\Theta(q^{(k+1)(n-k-1)})$. By the union bound,
    $$\quadbinom{n}{k+1}_q\cdot \Pr[|\cF_\cV\geq L+1|]<q^{-1-k}+o\parenv*{q^{-1-k}}\ .$$
	Hence, the probability that the second property is violated is $o(1)$. This completes the proof of the existence of an $[n,k,L]$-AS family of size $m-o(m)$.
\end{proof}
\begin{corollary}[Lower bound for AAD families]
  For arbitrary positive integers $L,n,k$ such that $2k<n$, and $q\to \infty$, there exists an $[n,k,L]_q$-AAD family of size
  $$
  |\cF| \leq q^{n-2k-\frac{(n-k)(k+1)}{(L+2)}}(1+o(1))\ .
  $$
  For fixed $L$, $p^{ADD}(n,k,L)\ge n-2k - \frac{(n-k)(k+1)}{(L+2)}$.
\end{corollary}
\begin{proof}
  The statements follows from the fact that an $[n,k,L+1]_q$-AS family is also an $[n,k,L]_q$-AAD family.
\end{proof}
The relation between an $[n,k,L+1]_q$-AS family and an $[n,k,L]_q$-AAD family holds only in one direction in general. Nevertheless, for $k=1$, an $[n,1,L]_q$-AAD family is also an $[n,1,L+1]_q$-AS family.
To see this,
note that any $2$-dimension subspace (a plane) that intersects with a $1$-dimensional subspace (a line) from an $[n,k=1,L]_q$-ADD family, must contain the line. Therefore, any $[n,1,L]_q$-AAD family is also an $[n,1,L+1]_q$-AS family and \cref{construction:AAD-RS} with $k=1$ also gives an $[n,1,L+1]_q$-AS family. Hence, $p^{AS}(n,1,n)=p^{AAD}(n,1,n-1)=n-2$.
\begin{remark}[Related work on the AS families]
  The concept of almost sparse subspace families is closely related to the \emph{weak subspace design} introduced by Guruswami and Xing in~\cite{guruswami2013list} and further studied in~\cite{guruswami2016explicit,guruswami2018subspace}. A collection $\mathcal{F}$ of subspaces in $\F_q^n$ is an $[n,k,L]_q$-\emph{weak subspace design}
  if every $k$-dimensional subspace in $\F_q^n$ intersects non-trivially at most $L$ subspaces from $\mathcal{F}$. Despite not being required by definition, many known constructions of weak subspace design contain subspaces with a \textit{fixed} co-dimension at least $k$.  Weak subspace design and AS families
  have the following (trivial) relations:
\begin{itemize}
\item An $[n,k,L]_q$-AS family is also an $[n,k+1,L]_q$-weak subspace design.
\item For $n\geq 2k+1$, a partial $k$-spread of $\F_q^n$ is an $[n,k+1,L]_q$-weak subspace family if and only if it is also an $[n,k,L]_q$-AS family.
\end{itemize}
By the explicit constructions presented in~\cite{guruswami2016explicit}, we can derive that $p^{AS}(n,k,L)\geq\floor{\frac{n-k}{k+1}}$ for $L\geq\frac{(n-1)(k+1)}{\floor{(n-k)/(k+1)}}$.
\end{remark}

\section{Summary and Outlooks}
The first part of this chapter introduces a new class bivariate evaluation codes, QLRS codes, which have the local property that the codeword symbols, whose coordinates are lying on a quadratic curve, form a codeword of an RS code. Hence, for any coordinate in a codeword, every quadratic curve passing through this coordinate gives a local recovery set of it.
For a QLRS code over $\Fq$, there are $q^2$ local recovery sets for each codeword symbol.
We have presented a necessary and sufficient condition on the monomials which form a basis of the code.
Based on the condition, we give upper and lower bounds on the dimension and show that the asymptotic rate of a QLRS code over $\mathbb{F}_q$ with local redundancy $r$ is $1-\Theta(q/r)^{-0.2284}$.
Moreover, we have provided lower and upper bounds on the minimum distance of this class of codes and compared QLRS codes with lifted Reed-Solomon codes by simulations in terms of the probability that a certain erasure cannot be recovered locally.
The simulation results showed that for short block lengths (e.g., $n=64$) and under the same code dimension, QLRS codes have better performance than lifted Reed-Solomon codes when the erasure probability $\tau\leq 0.7$.

For future research, error-correcting algorithms of QLRS codes can be developed and analyzed.
A promising candidate is the \emph{randomized list decoding algorithm} which has been used for RM codes~\cite{arora1997improved,sudan1999pseudorandom,goldreich2000learning,kim2016decoding}, lifted Reed-Solomon codes~\cite{guo2016list} and lifted affine-invariant codes~\cite{holzbaur2021decoding}. The algorithm uses list decoding in local correction and aggregates the local decoding results with appropriate weight to make a final decision on a symbol.
Moreover, QLRS codes were originally motivated by potential applications in \emph{coded caching}.
The task of coded caching problems is to find \emph{caching} and \emph{delivery} strategies to minimize delay in a communication system, where $K$ users, each has a cache capability for $M$ files, demand some of $N>M$ files stored at a server.
QLRS codes have the property that any two local recovery sets for one coordinate $(a_1,b_1)$ intersect at another coordinate $(a_2,b_2)$.
If we see the symbol at the $(a_1,b_1)$ as a segment of the file that a user demands from the server, then symbol at the intersection $(a_2,b_2)$ can be placed as the cached content and the rest of symbols in one of the recovery sets should be downloaded to reconstruct the file.
New coded caching schemes can be developed and analyzed based on QLRS codes, or the generalization -- the weighted $\eta$-lifted codes \cite{lavauzelle2020weighted}.

The second part studies the
AAD family of $k$-dimensional subspaces, motivated by constructions for batch codes -- a class of local recoverable codes with availability.
The subspaces in an AAD family form a partial spread and any $(k+1)$-dimensional subspace containing a subspace from the family non-trivially intersects with at most $L$ subspaces from the family.
We have presented an explicit construction of the AAD family for $k=1,2$ based on RS codes.
The construction gives $[n,k,L]_q$-AAD families with cardinality $q^{n-2k}$ and $L$ being polynomial in $n$ and $k$.
Another question discussed is the polynomial growth $p^{AAD}(n,k,L)$ (with $q\to \infty$) of the maximal cardinality of these families.
The upper and lower bounds on this quantity derived in this section show that
$$
n-2k-\frac{(k+1)(n-k)}{L+2}\leq p^{AAD}(n,k,L)\leq n-2k,\ \textrm{ for all } 2k<n\ .
$$
With the explicit constructions based on RS codes, we have $p^{ADD}(n,1,L)=n-2$ for $L\geq n-1$ and $p^{AAD}(n,2,L)=n-4$ for $L\geq 1+2(n-2)(2n-6)$, respectively.
For future work, the analysis on the explicit construction (\cref{construction:AAD-RS}) can be done for $k\geq 3$ to show whether the upper bound on $p^{AAD}(n,k,L)$ is also tight for $k\geq 3$.


%% file: figs/QC-LRS/ratePlot_l=5_log.tex
    \begin{tikzpicture}[font=\normalsize,scale = 0.8]
    \pgfplotsset{compat = 1.3}
    \begin{axis}[
    legend style={nodes={scale=0.9, transform shape}},
    cycle list name = {bounds_list},
    width = 0.9\columnwidth,
    height = 0.5\columnwidth,
    xlabel = {{Local Redundancy $r^*$}},
    ylabel = {{Rate of QC-LRS codes}},
    ymode=log,
    log basis y={32},
    xmin = 1,
    xmax = 16,
    ymin = 0,
    ymax = 1,
    legend pos = north east,
    legend cell align=left,
    grid=both]

\addplot table[x=r, y=rate] {./figs/QC-LRS/dim_l=5_CL.dat};

\addlegendentry{{rate $\ell=5$}};



\addplot table[x=r, y=rate_ub] {./figs/QC-LRS/dimBounds_l=5_CL.dat};

\addlegendentry{{rate (ub) $\ell=5$}};

\addplot table[x=r, y=rate_lb] {./figs/QC-LRS/dimBounds_l=5_CL.dat};

\addlegendentry{{rate (lb) $\ell=5$}};
    \end{axis}
    \end{tikzpicture}

%% file: figs/QC-LRS/global_l=3_dim10.tex
    \begin{tikzpicture}
    \pgfplotsset{compat = 1.3}
    \begin{semilogyaxis}[
    legend style={nodes={scale=0.8, transform shape}},
    cycle list name = {sims_list},
    width = 0.9\columnwidth,
    height = 0.5\columnwidth,
    xlabel = {{Erasure Probability $\tau$}},
    ylabel = {{Failure Probability of Local Recovery}},
    xmin = 0.3,
    xmax = 1,
    ymin = 0.00004,
    ymax = 1.1,
    legend pos = south east,
    legend cell align=left,
    grid=both]
\addplot table[x=tau,y=fail_rate] {./figs/QC-LRS/local/fail_rate_l=3_r=4_step=50_LRS.dat};

\addlegendentry{{LRS $(q=8, k=10, r=4)$}};

\addplot table[x=tau,y=fail_rate] {./figs/QC-LRS/local/fail_rate_l=3_r=3_step=50_CL.dat};

\addlegendentry{{QLRS $(q=8, k=10, r=3)$}};

\addplot table[x=tau,y=fail_rate] {./figs/QC-LRS/local/fail_rate_l=3_r=5_step=50_LRS.dat};

\addlegendentry{{LRS $(q=8, k=6, r=5)$}};

\addplot table[x=tau,y=fail_rate] {./figs/QC-LRS/local/fail_rate_l=3_r=4_step=50_CL.dat};

\addlegendentry{{QLRS $(q=8, k=6, r=4)$}};




    \end{semilogyaxis}
    \end{tikzpicture}

%% file: chap_dec_eva.tex
\noindent

An $\intOrd$-interleaved code is the direct sum of $\intOrd$ possible different codes (called \emph{constituent codes}) of the same length $n$, and its codewords can be represented as $\intOrd \times n$ matrices.
A common error model for these codes are \emph{burst errors}~\cite{owsley1988burst,chen1992burst}, where the errors are concentrated in several columns. As a distance metric, the Hamming weight of such an $\intOrd \times n$ matrix is defined as the number of nonzero columns of the matrix.

To decode an interleaved code, a naive approach is to simply decode each constituent codeword, i.e., each row of the $\intOrd \times n$ matrix independently. The error-correction capability of the interleaved code is then $\floor{(\dH-1)/2}$, where $\dH$ is the smallest minimum Hamming distance of the constituent codes.
However, for various algebraic interleaved codes, it is possible to correct a larger fraction of errors by adopting a joint approach.
For this reason, interleaved codes have many applications in which burst errors occur naturally or artificially, for instance, replicated file disagreement location \cite{metzner1990general},
correcting burst errors in data-storage applications \cite{krachkovsky1997decoding,holzbaur2019error,hamburg2024unraveling},
outer codes in concatenated codes \cite{metzner1990general,krachkovsky1998decoding,haslach1999decoding,justesen2004decoding,schmidt2005interleaved,schmidt2009collaborative},
ALOHA-like random-access schemes \cite{haslach1999decoding},
decoding non-interleaved codes beyond half-the-minimum distance by power decoding \cite{schmidt2010syndrome,kampf2014bounds,rosenkilde2018power,puchinger2019improved,couvreur2020power,puchinger2021improved},
distributed computing \cite{subramaniam2019collaborative} and
code-based cryptography \cite{elleuch2018interleaved,holzbaur2019decoding,renner2021liga,porwal2022interleaved,aguilar2022lrpc,aragon2023lowms}.

Generalized Reed--Solomon (GRS) codes are among the most-studied classes of constituent codes for interleaved codes.
There are several decoders for interleaved GRS codes \cite{krachkovsky1997decoding,bleichenbacher2003decoding,brown2004probabilistic,schmidt2009collaborative,nielsen2013generalised,yu2018simultaneous} that decode up to $\tInt\defeq\tfrac{\intOrd}{\intOrd+1}(n-\bar{k})$ errors, where
$\bar{k}$ is the mean dimension of the constituent codes.
All of these decoders fail for \emph{some} error patterns of weight larger than the unique decoding radius of the constituent code with the smallest minimum Hamming distance.
For errors of a given weight $t$, the fraction of errors leading to a unsuccessful decoding is roughly $q^{-m}$ at the maximal decoding radius $\tInt$ (where $q^m$ is the field size of the GRS code), and decreases exponentially in $\tInt-t$, the difference between the maximal decoding radius and the actual error weight.

There are also other decoding algorithms for interleaved GRS codes that decode beyond the radius $\tfrac{\intOrd}{\intOrd+1}(n-\bar{k})$, and even beyond the Johnson radius: \cite{coppersmith2003reconstructing,parvaresh2004multivariate,parvaresh2007algebraic,schmidt2007enhancing,cohn2013approximate,wachterzeh2014decoding,puchinger2017irs,huang2022list}.
For some of these decoders, simulation results suggest that these decoders can successfully decode a large fraction of error matrices of weight up to the claimed maximal radius, and in some very special cases, it is possible to derive bounds on this fraction.
However, in general, only little is known about the fraction of decodable errors by these decoders, which are therefore not considered in this chapter. Other code classes that have been considered as constituent codes of interleaved codes are one-point Hermitian codes \cite{kampf2014bounds,puchinger2019improved,matthews2021fractional} and, more generally, algebraic-geometry codes~\cite{brown2005improved,nakashima2010ag}.

For interleaved decoders of high order, i.e., where $\intOrd$ is at least the number of corrupted columns in received codewords, a simple linear-algebraic decoder was proposed in \cite{metzner1990general,haslach1999decoding} and generalized in \cite{haslach2000efficient,roth2014coding}.
Unlike all decoders mentioned above, this decoder works for interleaved codes with an arbitrary linear constituent code and guarantees to correct any error of weight up to $d-2$ that has full rank, where $d$ is the minimum distance of the constituent code.
This generic decoding approach has been applied to interleaved codes with arbitrary constituent codes, whose parity-check matrix is known, in the Hamming metric~\cite{sendrier2011decoding,porwal2022interleaved}, the rank metric~\cite{chabaud1996cryptographic,ourivski2002new,gaborit2015complexity,aragon2018new,bardet2020algebraic} and the sum-rank metric~\cite{puchinger2022generic,song2023blockwise,aragon2023blockwise}, respectively, to enhance the security of the proposed post-quantum cryptosystems based on codes in these metrics.

This chapter is devoted to the joint decoding approaches applied to interleaved codes with two classes of evaluation codes as constituent codes.
In \cref{sec:decoding-algorithm-interleaved}, we present a syndrome-based joint decoding algorithm for interleaved RS codes
and a necessary and sufficient condition on a successful decoding.
In \cref{sec:int-alternant}, we apply the joint decoding algorithm to interleaved \emph{alternant} codes, where the constituent codes are subfield subcodes of GRS codes
and give a condition of successful decoding tailored for alternant codes.
In \cref{sec:bounds-succ-IA}, we present lower and upper bounds on the probability of successful decoding of interleaved alternant codes by the joint decoding algorithm.
Finally, we give brief summaries on the other results on joint decoding interleaved evaluation codes in \cref{sec:other-joint-decoding}. 

\emph{The results in \cref{sec:decoding-algorithm-interleaved} and \cref{sec:int-alternant} were published in IEEE TIT~\cite{holzbaur2021decodingIA} and partly in the proceeding of 2021 Information Theory Workshop~\cite{holzbaur2021success}.}

\section{Joint Decoding of Interleaved Reed-Solomon Codes}
\label{sec:decoding-algorithm-interleaved}

We first formally introduce the concept of interleaved codes and briefly recap the joint decoding algorithm for interleaved RS codes from \cite{feng1991generalization,schmidt2009collaborative}.

\begin{definition}[Interleaved codes]
  The $\intOrd$-interleaved code $\cI\code^{(\intOrd)}$ with constituent
  code $\code$ is defined as
  \begin{align*}
    \cI\code^{(\intOrd)}\defeq \left\{
 \left.   \begin{pmatrix}
      \bc^{(1)}\\
      \vdots\\
      \bc^{(\intOrd)}
    \end{pmatrix}\ \right| \ \bc^{(i)}\in \code,\ i\in[\intOrd]
\right\}\ .
  \end{align*}
  The parameter $\intOrd$ is referred to as the \emph{interleaving order} of the interleaved code.
\end{definition}
Consider a channel where burst errors of Hamming weight at most $t$ occur. We transmit a codeword $\bC$ of an $\intOrd$-interleaved code $ \cI\code^{(\intOrd)}$. The received word is given by
\begin{align*}
  \bR = \bC + \widetilde{\bE} \in \Fq^{\intOrd \times n} \ ,
\end{align*}
where each row of $\bC \in \Fq^{\intOrd \times n}$ is a codeword of $\code$ and $\widetilde{\bE} \in \Fq^{\intOrd \times n}$ has at most $t$ nonzero columns.
An illustration of a corrupted codeword of $\cI\cC$ is given in \cref{fig:int-code}.
\begin{figure}[htb]
  \centering
\input{figs/Dec-Int/interleavedCodes.tex}
  \caption{An illustration of a corrupted codeword of an $\intOrd$-code by a burst error $\widetilde{\bE}$.} 
  \label{fig:int-code}
\end{figure}

In the following let $\cC$ be a (generalized) RS code $\RS_q[n,k]$ with nonzero code locators $\alpha_1,\dots,\alpha_n$. 
The joint decoding algorithm that we present is also known as 
a \emph{syndrome-based joint decoding algorithm} for interleaved RS codes. Such algorithms, to name a few, can be found in~\cite{feng1991generalization} for BCH codes and \cite{krachkovsky1997decoding,bleichenbacher2003decoding,schmidt2009collaborative} for interleaved RS code.
We briefly recapitulate the decoding method below and summarize a naive version of~\cite[Algorithm 2]{schmidt2009collaborative} in~\cref{algo:SyndromeDecoder}.

Let $\bH$ be a parity-check matrix of the constituent code $\RS_q[n,k]$ and $d=n-k+1$.
From the received matrix $\bR$, we are able to calculate the syndromes of each row of $\bR$ by
\begin{align}\label{eq:Syndromes}
  \begin{pmatrix}
    \bs_{1}\\
    \bs_{2}\\
    \vdots\\
    \bs_{\intOrd}
  \end{pmatrix}
  \ =\ \bR \cdot \bH^\top =  \widetilde{\bE} \cdot \bH^\top \ ,
\end{align}
where $\bs_{i}=(s_{i,1},\dots,s_{i,d-1})\in\Fqm^{d-1},$ for each $i\in[\intOrd]$. 


Assuming that there are exactly $t$ nonzero columns in $\widetilde{\bE}$, we define the \emph{error locator polynomial} as\footnote{Since $\alpha_i \neq 0$, the error locator polynomial is well-defined.}
%
\begin{align}\label{eq:ELP}
  \Lambda(x)\defeq\prod_{i=1}^{t}(1-\alpha^{-1}_{j_i}x)=1+\Lambda_1 x+\dots+\Lambda_t x^t\ ,
\end{align}
where the $t$ roots $\alpha_{j_1},\dots,\alpha_{j_t}$ of $\Lambda(x)$ are the code locators corresponding to the error positions. The coefficients of $\Lambda(x)$ fulfill the following linear equations (cf.~\cite{peterson1960encoding}),
\begingroup
\allowdisplaybreaks
\begin{align}\label{eq:STmatrix}
  &\underbrace{\begin{pmatrix}
      s_{i,1} & s_{i,2} & \dots & s_{i,t}\\
      s_{i,2} & s_{i,3} & \dots & s_{i,t+1}\\
      \vdots & \vdots & & \vdots \\
      s_{i,d-1-t} & s_{i,d-1-t+1} & \dots &s_{i,d-2}
  \end{pmatrix}}_{\bS^{(i)}(t)}%
  \begin{pmatrix}
    \Lambda_t\\
    \Lambda_{t-1}\\
    \vdots\\
    \Lambda_1
  \end{pmatrix}
               =
                 \underbrace{\begin{pmatrix}
                   -s_{i,t+1}\\
                   -s_{i,t+2}\\
                   \vdots\\
                   -s_{i,d-1}
                 \end{pmatrix}}_{\supbrac{\bT}{i}(t)}
\ ,\ \forall i\in[\intOrd]\ .
\end{align}
\endgroup
Thus, determining the error positions in $\widetilde{\bE}$ is equivalent to solving the following linear system of equations $\mathfrak{S}(t)$ for $t$ unknowns, 
\begingroup
\allowdisplaybreaks
\begin{align}\label{eq:KeyEquationLSE}
  \underbrace{\begin{pmatrix}
    \supbrac{\bS}{1}(t)\\
    \supbrac{\bS}{2}(t)\\
    \vdots\\
    \supbrac{\bS}{\intOrd}(t)
  \end{pmatrix}}_{\bS(t)}%
\underbrace{  \begin{pmatrix}
    \Lambda_t\\
    \Lambda_{t-1}\\
    \vdots\\
    \Lambda_1
  \end{pmatrix}
}_{\vLambda}
  &=
                 \underbrace{\begin{pmatrix}
                   \supbrac{\bT}{1}(t)\\
                   \supbrac{\bT}{2}(t)\\
                   \vdots\\
                   \supbrac{\bT}{\intOrd}(t)
                 \end{pmatrix}}_{\bT(t)}\ .
\end{align}
\endgroup

After determining $\vLambda$ from~\eqref{eq:KeyEquationLSE}, we may use a standard method for error evaluation such as \emph{Forney's algorithm}~\cite{forney1965decoding} (cf.~\cite[Section~6.6]{roth2006introduction}) to calculate the error values $\hat{\bE}$. 
Then, by subtracting the calculated error $\hat{\bE}$ from $\bR$, we obtain the estimated codeword $\hat{\bC}=\bR-\hat{\bE}$.

\begin{algorithm}[htb!]
  \caption{Syndrome-based Joint Decoding Algorithm}\label{algo:SyndromeDecoder}
  \SetAlgoLined
  \DontPrintSemicolon
  \KwIn{received word $\bR$}
  \KwOut{$\hat{\bC}$ or \texttt{decoding failure}}
  Calculate the syndromes $\bs_{i,:},\forall i\in[\intOrd]$ \tcp*[r]{See~\eqref{eq:Syndromes}} 
  \lIf{$\bs_{i,:} = \0$ for all $i$}{\Return{$\hat{\bC}=\bR$}\label{step:zeroSyndromes}}
  Find minimal $t^\star$ for which $\bS(t^\star) \cdot \vLambda^\star = \bT(t^\star)$ has a solution
  $\vLambda^\star$ \label{step:solveLSE} \tcp*[r]{See~\eqref{eq:KeyEquationLSE}}
  \textbf{if} the solution $\vLambda^\star$ is not unique \textbf{then} \Return{} \texttt{decoding failure}
  \label{step:uniqueSolution}\;
    \uIf{$\Lambda^\star(x)$ \emph{has} $t^\star$ distinct \emph{roots in} $\Fqm$}{
      Evaluate the errors $\hat{\bE}$ by Forney's algorithm~\cite{forney1965decoding}\cite[Section~6.6]{roth2006introduction}\;
      \Return $\hat{\bC}=\bR-\hat{\bE}$ \label{step:codewordEstimate}}
    \Else{\Return{} \texttt{decoding failure}\label{step:distinctRoots}}
  \SetAlgoLined
\end{algorithm}

For the channel where the burst errors occur, the joint decoding algorithm given in \cref{algo:SyndromeDecoder} may yield three different results:
\begin{itemize}
\item The algorithm returns the correct result, i.e., $\hat{\bC}=\bC$, with \emph{success probability} $\Psuc$.
\item The algorithm returns an erroneous result, i.e., $\hat{\bC}\neq \bC$, with \emph{miscorrection probability} $\Pmisc$.
\item The algorithm returns a \texttt{decoding failure}, with \emph{failure probability}~$\Pfail$.
\end{itemize}

\begin{remark}[Practical Implementations]
  \cref{algo:SyndromeDecoder} is a naive approach.
  It is mainly meant for the proof of the successful probability, instead of for an efficient implementation.


  For practical implementations, one can use some fast algorithm for~\cref{step:solveLSE}, for instance, 1)~\cite[Algorithm 3]{sidorenko2011linear} with the complexity of $O(\intOrd d^2)$ operations in $\Fqm$, 2) the currently fastest algorithm~\cite{rosenkilde2021algorithms} with complexity $O^{\sim}(\intOrd^{\omega-1} d)$ where $O^{\sim}$ omits the $\log$-factors in $d$ and $\omega$ is the matrix multiplication exponent, for which the best algorithm allow $\omega<2.38$~\cite{coppersmith1990matrix,le2014powers}. 
\end{remark}

\cref{algo:SyndromeDecoder} yields a \emph{bounded distance} decoder which can decode beyond half of the minimum distance
with high probability.
Clearly, the solution $\vLambda^\star$ cannot be unique if the number of equations in~\eqref{eq:KeyEquationLSE} is less than the number of unknowns. Thus,
the following maximum decoding radius of~\cref{algo:SyndromeDecoder} can be derived.

\begin{theorem}[{\cite[Theorem~3]{schmidt2009collaborative}}]
  \label{thm:IRS-max-decoding-radius}
  Let $\cI\code^{(\intOrd)}$ be an $\intOrd$-interleaved code with $\cC=\RS_q[n,k]$.
  For a received word $\bR=\bC+\widetilde{\bE}$, where $\bC\in\cI\code^{(\intOrd)}$ and the error $\widetilde{\bE}$ has $t$ nonzero columns, \cref{algo:SyndromeDecoder} may only succeed, i.e., return $\hat \bC = \bC$, if 
\begin{align}
  t\leq t_{\max,\RS} &\defeq \frac{\intOrd}{\intOrd+1}(d-1) \ . \label{eq:tmaxGRS}
\end{align}
\end{theorem}

By the nature of a bounded distance decoder, where the correction balls of each codeword inevitably overlap for some error patterns of weight $t>\floor{\frac{\dH-1}{2}}$, \cref{algo:SyndromeDecoder} is unsuccessful with some probability when $t>\floor{\frac{\dH-1}{2}}$.
The following lemma gives a necessary and sufficient condition such that \cref{algo:SyndromeDecoder} is unsuccessful, i.e., returning an erroneous result or a \texttt{decoding failure}. This will be the foundation to bound the success probability of \cref{algo:SyndromeDecoder} in decoding interleaved alternant codes in \cref{sec:int-alternant}. 
The sufficiency has been shown in the proof of \cite[Lemma~2]{schmidt2009collaborative}.
The proof below completes the necessity.
 \begin{lemma}[Necessary and sufficient condition on unsuccessful decoding]\label{lem:FailRankCond}
   Let $\cI\code^{(\intOrd)}$ be an $\intOrd$-interleaved code with $\cC=\RS_q[n,k]$.
   For a received word $\bR=\bC+\widetilde{\bE}$, where $\bC\in\cI\code^{(\intOrd)}$ and the error $\widetilde{\bE}$ has $t>0$ nonzero columns, \cref{algo:SyndromeDecoder} is not \emph{successful}, i.e., returns $\hat \bC \ne \bC$ or a {\normalfont \texttt{decoding failure}}, if and only if $\rank\parenv*{\bS(t)}<t$.
\end{lemma}
\begin{proof}
  Denote by $\Lambda(x)$ the \emph{true} error locator polynomial corresponding to the $t$ error positions (indices of nonzero columns) in $\widetilde{\bE}$. Then $\Lambda(x)$ has $t$ distinct roots in $\Fqm$ and $\vLambda$ is a solution of the linear system of equations ${\frak S}(t)$ as in~\eqref{eq:KeyEquationLSE}.

  \emph{Sufficiency:} We show that $\rank(\bS(t))<t$ implies unsuccessful decoding.
  Consider two cases, $\rank(\bS(t))=0$ and $0<\rank(\bS(t))<t$.
  For $\rank(\bS(t))=0$, the algorithm outputs $\hat\bC=\bR$ at \cref{step:zeroSyndromes}. However, since $t>0$, this is apparently a miscorrection.
  For the latter case, assume $\rank(\bS(t))=t^{\star}$, which is found at \cref{step:solveLSE}. Then $t^{\star}<t$ by assumption. Suppose the algorithm runs until \cref{step:codewordEstimate}, then $\wtH(\hat\bE)=t^{\star}<t=\wtH(\widetilde{\bE})$ and hence $\hat\bE\neq \widetilde{\bE}$. Then the resulting $\hat\bC$ is not the sent codeword $\bC$. Other termination (\cref{step:uniqueSolution} or \cref{step:distinctRoots}) of the algorithm results in a \texttt{decoding failure}.

  \emph{Necessity:}
  We show that unsuccessful decoding implies $\rank(\bS(t))<t$.
The algorithm returns \texttt{decoding failure} only on \cref{step:uniqueSolution} or \ref{step:distinctRoots}. \cref{step:solveLSE} determines the \emph{minimal} $t^{\star}$ such that ${\frak S}(t^{\star})$ has at least one solution $\vLambda^{\star}$, hence $t^{\star} \leq t$. Note that a solution to ${\frak S}(t^{\star})$ is also a solution to ${\frak S}(t)$. 
 If the algorithm fails on \cref{step:uniqueSolution}, i.e., the system ${\frak S}(t^\star)$ has many distinct solutions, then ${\frak S}(t)$ also has many solutions and therefore $\rank(\bS(t))<t$.
The failure occurs on \cref{step:distinctRoots} if $\Lambda^{\star}(x)$ does not have $t^{\star}$ different roots, which implies $\Lambda^{\star}(x)\ne \Lambda(x)$. This means that the system ${\frak S}(t)$ has at least two solutions $\vLambda$ and $\vLambda^{\star}$. Hence $\rank(\bS(t))<t$.

The algorithm returns a miscorrected codeword only at \cref{step:zeroSyndromes} or \ref{step:codewordEstimate}.
If the decoder outputs $\hat \bC$ on \cref{step:zeroSyndromes}, we have $\hat \bC \ne \bC$ as $t>0$. In this case $\bS(t) = \0$, so $\rank(\bS(t))= 0 <t$.
Note that $\bS(t)=\0$ will not occur if $0<t<\dH(\cC)$, since the error $\widetilde{\bE}$ with $\wtH(\widetilde{\bE})=t$ cannot result in $\bR$ being another codeword in $\cC$. 
Assume the algorithm runs to \cref{step:codewordEstimate} and returns a $\hat \bC\neq\bC$. This implies that either $t^{\star}<t$ or, $t^{\star}=t$ and \cref{step:solveLSE} has two distinct solutions $\Lambda^{\star}$. The former directly implies $\rank(\bS(t))<t$ and the latter will cause a \texttt{decoding failure} at \cref{step:uniqueSolution}, contradicting the assumption.
\end{proof}

\begin{remark}
  \label{rem:applicationToRS}
  In the proof of \cite[Theorem~7]{schmidt2009collaborative}, $\Pr[\rank(\bS(t))<t]$ is used to bound the probability of a \texttt{decoding failure}. Since we have shown in \cref{lem:FailRankCond} that $\rank(\bS(t))<t$ is not only a sufficient condition but also a necessary condition for an unsuccessful decoding (either a miscorrection or a \texttt{decoding failure}), the upper bound \cite[Theorem~7]{schmidt2009collaborative} 
  is in fact an upper bound on
  $1-\Psuc=\Pmisc+\Pfail$.
\end{remark}

\section{Joint Decoding of Interleaved Alternant Codes}
\label{sec:int-alternant}

GRS codes are generalizations of Reed-Solomon codes $\RS_q[n,k]$ with nonzero \emph{column multipliers}. 
\emph{Alternant code} are \emph{subfield subcodes} of GRS codes.
This code family contains some of the best-known and most-often used algebraic codes over small fields, including the BCH~\cite{Hocquenghem_1959,Bose_RayChaudhuri_1960} and the Goppa codes~\cite{goppa1970new,berlekamp1973goppa,sugiyama1976further}.



\subsection{Generalized Reed-Solomon Codes and Alternant Codes}
We begin by formally defining the GRS codes and the alternant codes, then discuss some applications and special cases of alternant codes.
\begin{definition}[Generalized Reed-Solomon codes]\label{def:GRScodes}
  For positive integers $d$ and $n$, let $\balpha=(\alpha_1, \alpha_2, \dots, \alpha_n) \in (\Fqm^*)^n$ be a vector of distinct \emph{code locators} and $\bv \in (\Fqm^*)^n$ be a vector of \emph{column multipliers}.
  A \emph{generalized Reed-Solomon} (GRS) code $\GRSp$ of length $n=|\balpha|$, dimension $k=n-d+1$ and minimum Hamming distance $d$ is defined as\footnote{The results in this section are mostly dependent the minimum Hamming distance, code locators and column multipliers of GRS codes. Hence, we use the notation depend on $d$, $\balpha$ and $\bv$.}
    \begin{align*}
      \GRSp &= \{\bc \in\Fqm^n \ |\  \bH\cdot \diag(\bv)\cdot \bc=\0\}\ ,
    \end{align*}
    with
  \begin{align*}
    \bH =
    \begin{pmatrix}
      1 & 1 & \dots & 1\\
      \alpha_1 & \alpha_2 & \dots & \alpha_n \\
      \vdots & \vdots &  & \vdots\\
      \alpha_1^{d-2} & \alpha_2^{d-2} & \dots & \alpha_n^{d-2}
    \end{pmatrix} \ \in \ \Fqm^{(d-1)\times n}\ .
  \end{align*}
  For a fixed $\balpha$, denote by $\mathbb{G}_{\balpha}^d$ the multiset of GRS codes with different column multipliers, i.e.,
  \begin{equation*}
    \mathbb{G}_{\balpha}^d \defeq \{\{ \GRSp \ | \ \bv \in (\Fqm^*)^n \}\} \ .
  \end{equation*}
\end{definition}
Note that the most general definitions of GRS codes allow for the $\alpha_i = 0$ to be element of $\balpha$, but for consistency with \cite{schmidt2009collaborative} and as this complicates the decoding process, we restrict ourselves to $\alpha_i \neq 0$ here.
GRS codes are well-known to be
MDS codes, i.e., they achieve $\dH=n-k+1$, where $k$ is the dimension of the code.

By design, GRS codes must be defined over finite fields $\Fqm$ with $q^m-1\geq n$ (or $q^m\geq n$ if $\alpha_i=0$ is allowed as a code locator). In many applications it is desirable to work with codes of smaller field size, which can be obtained, e.g., by taking subfield subcodes of codes with good minimum distance.

\begin{definition}[Subfield subcode]\label{def:subfieldSubcode}
  Let $\code$ be an $[n,k]_{q^m}$ code. We define the $\F_q$-subfield subcode of $\code$ as
  \begin{equation*}
    \code \cap \Fq^n = \{\bc\in\Fq^n\ | \ \bc\in \code\} \ .
  \end{equation*}
  Equivalently, let $\bH \in \F_{q^m}^{(n-k) \times n}$ be a parity-check matrix of $\code$. Then $\code \cap \Fq^n$ is given by the $\F_q$-kernel of $\bH$, i.e.,
  \begin{align*}
     \code \cap \Fq^n = \{\bc\in \F_q^n\ | \ \bH \cdot \bc= \0\} \ .
  \end{align*}
\end{definition}


  The subfield subcode of a GRS code is referred to as an \emph{alternant code}~\cite[Ch.~12.2]{macwilliams1977theory}.
  For a fixed $\balpha$ and
  a \emph{designed} distance $d$, we denote by $\mathbb{A}_{\balpha}^d$ the multiset of alternant codes, i.e.,
  \begin{equation}\label{eq:multiset-alternant}
    \mathbb{A}_{\balpha}^d \defeq \{\{ \code \cap \Fq^n \ | \ \code \in \mathbb{G}_{\balpha}^d \}\} \ .
  \end{equation}
We define $\mathbb{A}_{\balpha}^d$ as a multiset, as the multiplicities will be important for the bounds on the probability of successful decoding. An additional advantage is that for a given $\balpha$ of length $n=|\balpha|$, we know the cardinality of $\mathbb{A}_{\balpha}^d$ is
\begin{align}
  |\mathbb{A}_{\balpha}^d| = (q^m-1)^{n} \ . \label{eq:cardinalityAall}
\end{align}

For GRS codes it is known (cf.~\cite{delsarte1975subfield}) that for a fixed vector $\balpha$ of code locators, it holds that $\GRS_{\balpha,\bv}^d = \GRS_{\balpha, \bu}^d$ if and only if $\bv$ is an $\Fqm$-multiple of $\bu$, i.e., any code $\code \in \GRSallp$ occurs with multiplicity exactly $\delta^{\code}_{\GRSallp} = q^m-1$ in $\GRSallp$. This gives a lower bound on the multiplicity of alternant codes as
\begin{align}
  \delta^{\cA}_{\ALTallp} \geq q^m-1, \ \forall  \cA \in \ALTallp  \ . \label{eq:multiplicityAlternant}
\end{align}

 We give some general well-known bounds on the dimension of the $\Fq$-subcode of an $\Fqm$-linear code $\code$ in terms of the parameters of~$\code$.

\begin{lemma}\label{lem:AlternantDimBounds}
  Let $\code$ be an $[n,k]_{q^m}$ code with minimum Hamming distance $d$. Then,
\begin{equation*}
    \max\{n-m(n-k),0\} \leq \dim_q(\code \cap \Fq^n) \leq \min\{k,k_q^{\mathsf{opt.}}(n,d)\} \ ,
\end{equation*}
 where $k_q^{\mathsf{opt.}}(n,d)$ is an upper bound on the dimension of a $q$-ary linear code given the length $n$ and the minimum Hamming distance $d$ (e.g., Singleton, Griesmer, Hamming bounds, etc.)
\end{lemma}
\begin{proof}
  The lower bound $0$ is trivial. The lower bound $n-m(n-k)$ follows from expanding the $n-k$ rows of any parity-check matrix of $\code$ via some basis of $\Fqm$ over $\Fq$. The resulting $m(n-k) \times n$ matrix is a parity-check matrix of the $\F_q$-subcode of $\code \cap \Fq^n$ and the bound follows.
  The upper bound $k$ is trivial because $\dim_q(\code \cap \Fq^n)\leq \dim_q(\code)=k$
  and $k_q^{\mathsf{opt.}}(n,d)$ is an upper bound by definition.
\end{proof}
\subsubsection{Other Applications of Alternant Codes}
In principle, alternant codes can be used as constituent codes in any of the applications of interleaved codes mentioned at the beginning of this chapter.
Several concrete reasons to specifically consider interleaved alternant codes are also worthy mentioning: 
\begin{itemize}
\item Alternant codes (especially BCH codes) are widely used in practice, including data storage and communications. Any system that already uses these codes and is prone to burst errors may be retroactively upgraded to enable a larger error-correction capability.
For instance, in NOR and NAND flash memory, Hamming and BCH codes are considered as the standard error-correction approach~\cite{liu2006low,choi2009vlsi,wang2011error}.
Traditionally, Hamming codes are used in single-level flash memories to correct single errors as they have a simple decoding algorithm and use only a small circuit area. For multi-level flash memories, however, single-error correction is not sufficient and BCH codes with larger distance are employed. In \cite{sun2006use}, the scenario of more than four levels (i.e., storing more than two bits per flash memory cell) was investigated and it was shown that BCH codes of larger correction capability are needed.
To address the fact that errors in flash memories might occur over whole bit or word lines, in \cite{yang2011product} product codes with BCH codes were used. This motivates the use of \emph{interleaved} alternant and in particular interleaved BCH codes.

\item In applications where the cost of \emph{encoding} is dominant (e.g., in storage systems where writing occurs more often than reading an erroneous codeword), encoding in a subfield reduces the complexity. Hence, it might be advantageous to use alternant codes instead of GRS codes in some of the above mentioned applications of interleaved codes. Note that \emph{decoding} is usually done in the field of the corresponding GRS code, so the reduction in complexity is less significant.
\item In some applications, such as code-based cryptography, GRS and algebraic-geometry codes cannot be used due to their vast structure, which can be turned into structural attacks on the cryptosystem.
However, their subfield subcodes are in many cases unbroken, e.g., see in~\cite[Conclusion]{couvreur2017cryptanalysis} and \cite[Section~7.5.3]{couvreur2020algebraic}.
In particular, the codes proposed in McEliece's original paper \cite{mceliece1978public}, binary Goppa codes, have withstood efficient attacks for more than $40$ years.
In a McEliece-type system, the ciphertext is the sum of a codeword of a public code and a randomly chosen ``error'' which hides the codeword from the attacker.
If we encrypt multiple codewords in parallel, we may consider them as an interleaved code and align the errors in bursts of larger weight.
This approach has the potential to increase the designed security parameter, or in turn reduce the key size~\cite{elleuch2018interleaved,holzbaur2019decoding}.
This comes at the cost of a (hopefully very small) probability of unsuccessful decryption/decoding, which corresponds to the probability of unsuccessful decoding of the interleaved decoder.
\end{itemize}

\subsubsection{Dimension vs.~Hamming Distance of Binary BCH and Wild Goppa Codes}
\label{rem:BCHgoppaCodes}
\emph{Wild Goppa codes}~\cite{sugiyama1976further,wirtz1988parameters}, which include \emph{binary square-free Goppa codes}~\cite{goppa1970new,goppa1971rational,berlekamp1973goppa}, are a subclass of Goppa Codes. Along with BCH codes~\cite{Hocquenghem_1959,Bose_RayChaudhuri_1960}, Goppa codes are the best known class of alternant codes, due to their good distance properties in the Hamming metric.
Binary BCH and $q$-ary wild Goppa codes have been shown to be subfield subcodes of GRS codes in $\mathbb{G}^d_{\balpha}$ for some $\balpha$ and $d$.

 Consider a binary BCH code that is a subfield subcode of some GRS code in~$\mathbb{G}^d_{\balpha}$ with length $n=|\alpha|$ and dimension $k=n-d+1$ over $\F_{2^m}$.
It is well-known (cf.~\cite[Ch.~7]{macwilliams1977theory}) that the dimension of the binary BCH code is $k_\mathsf{BCH} \geq n-m\frac{n-k}{2}$, which exceeds the generic lower bound in \cref{lem:AlternantDimBounds}.

Wild Goppa codes are often considered as subclasses of alternant codes of $\mathbb{A}_{\balpha}^{d}$, but with an increased lower bound on the distance $\dGoppa \geq \frac{q}{q-1} d$.
However, the bounds on the probability of decoding success that are studied in \cref{sec:bounds-succ-IA} depend only on the properties of the corresponding GRS and, in particular, its distance $d$, but not on the \emph{actual} dimension or distance of the alternant code itself. Therefore, instead of viewing wild Goppa codes as alternant codes in $\mathbb{A}^{d}_{\balpha}$ with increased distance, it is convenient to view them as alternant codes of $\mathbb{A}_{\balpha}^{\dGoppa}$ with a larger \emph{dimension} than guaranteed by the lower bound in \cref{lem:AlternantDimBounds}. This is possible because the general improvements of wild Goppa codes compared to alternant codes were shown by an equivalence between the Goppa codes obtained from different Goppa polynomials (cf.~\cite{sugiyama1976further}, \cite[Theorem 4.1]{bernstein2011wild}).
In other words, the wild Goppa code is in 
$\ALTallp \cap \mathbb{A}^{\dGoppa}_{\balpha}$. 
It can be shown that the wild Goppa codes with distance $d$ have an increased lower bound on the dimension compared to the generic lower bound in \cref{lem:AlternantDimBounds}, i.e., $\kGoppa\geq n-m\frac{q-1}{q}(d-1)=n-m\frac{q-1}{q}(n-k)$.

\subsection{Condition on Successful Decoding of Interleaved Alternant Codes}
The joint decoder given in \cref{algo:SyndromeDecoder} immediately applies to interleaved alternant codes as well, i.e.,~the subfield subcodes of interleaved Reed--Solomon codes, but the fraction of decodable error matrices differs, since the error is now over the subfield.
Due to this, the bounds on the probability of unsuccessful decoding of interleaved GRS codes does not hold for interleaved alternant codes, which has been shown in the simulation results in \cite{holzbaur2019decoding}.
  
With the help of~\cref{lem:FailRankCond}, we now present the crux in bounding the success probability of decoding interleaved alternant codes by \cref{algo:SyndromeDecoder}, which is the basis of the bounds presented in \cref{sec:bounds-succ-IA}.

In the rest of this chapter, we denote by $\EB{q}{a}{b}$ the set of matrices $\bE$ without zero columns in $ \Fq^{a\times b}$. 
\begin{lemma}\label{lem:FailureCrux}
  Let $\mathcal{IC}^{(\intOrd)}$ be an $\intOrd$-interleaved alternant code with $\code \in \mathbb{A}^d_{\balpha}$, $n=|\balpha|$ and $\cE=\{j_1,j_2,\dots,j_t\}\subset [n]$ be a set of $|\cE|=t$ error positions. For a codeword $\bC\in \mathcal{IC}^{(\intOrd)}$, an error matrix $\widetilde{\bE} \in \Fq^{\intOrd \times n}$ with $\supp(\widetilde{\bE}) = \cE$ and $\bE \defeq \widetilde{\bE}|_{\cE} \in \EB{q}{\intOrd}{t}$, and a received word $\bR = \bC + \widetilde{\bE}$, \cref{algo:SyndromeDecoder} \emph{succeeds}, i.e., returns $\hat \bC = \bC$, if and only if
  \begin{equation}
    \label{eqn:failure_cond}
    \nexists \bv \in \Fqm^t \setminus \{\0\} \text{ such that } \bH \cdot \diag(\bv) \cdot \bE^\top = \0 \ , 
  \end{equation}
  where $\bH\in \F_{q^m}^{d-t-1 \times t}$ is a parity-check matrix of the $[t, 2t-d+1, d-t]_{q^m}$ code $\GRS_{\balpha|_\cE,\1}^{d-t}$.
\end{lemma}
\begin{proof}
  We extend and adapt the proof for interleaved GRS codes from~\cite{schmidt2009collaborative}.

  According to~\cref{lem:FailRankCond}, \cref{algo:SyndromeDecoder} yields a \texttt{decoding failure} or a miscorrection $\hat \bC \neq \bC$ if and only if $\rank(\bS(t))<t$, with $\bS(t)$ as in \eqref{eq:KeyEquationLSE}. In other words, the decoding may only be unsuccessful, if there exists a nonzero vector $\bu\in \Fqm^t$ such that $\bS(t) \cdot \bu=\0$, i.e.,
  \begin{align}\label{eq:ustatement}
    \exists \bu\in \Fqm^t\setminus\{\0\} \text{ such that }\supbrac{\bS}{i}(t)\cdot \bu=\0\ ,\ \forall i\in[\intOrd]\ .
  \end{align}
  It is known (cf.~\cite[Theorem 9.9]{peterson1972error}\cite{schmidt2009collaborative}) that a syndrome matrix $\supbrac{\bS}{i}(t)$ can be decomposed into
  \begin{align*}
    \supbrac{\bS}{i}(t)=\bH \cdot \supbrac{\bF}{i}\cdot \bD\cdot \bV\ ,
  \end{align*}
  where $\bH$ is a parity-check matrix of the $[t, 2t-d+1, d-t]_{q^m}$ code $\GRS_{\balpha|_\cE,\1}^{d-t}$
  as in \cref{def:GRScodes},
  \begin{align*}
    &\bV=
    \begin{pmatrix}
        1 & 1 & \dots & 1\\
      \alpha_{j_1} & \alpha_{j_2} & \dots & \alpha_{j_t}\\
      \alpha_{j_1}^2 & \alpha_{j_2}^2 & \dots & \alpha_{j_t}^2\\
      \vdots & \vdots &  & \vdots\\
      \alpha_{j_1}^{t-1} & \alpha_{j_2}^{t-1} & \dots & \alpha_{j_t}^{t-1}
    \end{pmatrix}^\top\in\Fqm^{t\times t} \ , \\
    &\supbrac{\bF}{i} =\diag(\be_{i}) \in \Fq^{t\times t}\ ,\\
    &\bD = \diag(\bv'|_{\cE})\in \Fqm^{t \times t}\ ,
  \end{align*}
 and $\bv'\in\parenv*{\Fqm^*}^n$ is the vector of column multipliers of the GRS code corresponding to the alternant code $\code$, i.e., $\GRS^d_{\balpha,\bv'}\cap\Fq=\code$.

  We observe that the matrices $\bD$ and $\bV$ are both square and of full rank. Therefore, the product $\bv\defeq\bD\cdot\bV\cdot\bu$ defines a one-to-one mapping $\bu\to\bv$, such that $\0\to\0$. Consequently, the statement~\eqref{eq:ustatement} is equivalent to 
  \begin{align*}
    \exists \bv \in \Fqm^t \setminus \{\0\} \text{ such that } \bH&\cdot \diag(\be_{i})\cdot \bv =\0\ ,\ \forall i\in[\intOrd]\\
    &\Updownarrow\\
    \exists \bv \in \Fqm^t \setminus \{\0\} \text{ such that } \bH&\cdot \diag(\bv)\cdot \be_{i} =\0\ ,\ \forall i\in[\intOrd]\ ,
  \end{align*}
  and the statement follows.
\end{proof}

Note that the upper bound on the probability of unsuccessful decoding interleaved GRS codes from~\cite{schmidt2009collaborative} applies to error matrices $\widetilde{\bE}$ over $\Fqm$ (the field of the GRS code). 
However, for interleaved alternant codes, $\widetilde{\bE}$ is over $\Fq$ (the \emph{subfield} of RS codes) and the bound from~\cite{schmidt2009collaborative} is not valid in this case.


\section{Bounds on Success Probability of Decoding Interleaved Alternant Codes}\label{sec:bounds-succ-IA}
In this section we present lower and upper bounds on the probability of successful decoding of interleaved alternant codes by~\cref{algo:SyndromeDecoder}.
\cref{lem:FailureCrux} gives a necessary and sufficient condition for~\cref{algo:SyndromeDecoder} to succeed for an error $\widetilde{\bE}$ with fixed $\cE=\supp(\widetilde{\bE})$ and $\widetilde{\bE}|_{\cE}\in \EB{q}{\intOrd}{t}$.
With this as the basis, 
we bound the probability of successful decoding for a random error matrix $\widetilde{\bE}$ where $\widetilde{\bE}|_{\cE}$ is i.i.d.~in $\EB{q}{\intOrd}{t}$. 

\subsection{Technical Preliminary Results}\label{sec:TechnicalPreliminaries}
Before deriving the bounds,
we establish some technical preliminary results which are needed to prove the bounds.
\subsubsection{Maximization of Integer Distributions}

To begin, we derive a simple upper bound on the maximization of a sum of integer powers, under a restriction on the base of the power.

\begin{definition}[Majorization relation]\label{def:majorization}
  Let
  $\cM=\{\{m_1, m_2,\dots,m_c\}\}$
  and
  $\cK=\{\{k_1, k_2,\dots,k_c\}\}$
  be two (finite) multisets of real numbers with the same cardinality. We say that the set $\cM$ \emph{majorizes} the set $\cK$ and write
\begin{equation*}
\cM\succ\cK\quad\mathrm{or}\quad \cK\prec\cM
\end{equation*}
if, after a possible renumeration, $\cM$ and $\cK$
satisfy the following conditions:
\begin{enumerate}
    \item[(1)] $m_1\geq m_2 \geq \cdots \geq m_c$ and $k_1\geq k_2 \geq \cdots \geq k_c$;
    \item[(2)] $\sum_{i=1}^jm_i \geq \sum_{i=1}^j k_i, \ \forall \, 1\leq j \leq c$.
\end{enumerate}
\end{definition}
We recap the following well-known result on multisets with this majorization relation.
\begin{lemma}[Karamata's inequality~{\cite[Theorem~1]{kadelburg2005inequalities}}] \label{lem:karamata}
  Let
  $\cM=\{\{m_1, m_2,\dots,m_c\}\}$
  and
  $\cK=\{\{k_1, k_2,\dots,k_c\}\}$
  be two multisets of real numbers from an interval $[a,b]$. If the set $\cM\succ\cK$, and if $f:\bbR\to \bbR$ is a convex and non-decreasing function in the range $[a,b]$, then it holds that
\begin{align}
    \sum\limits_{i=1}^c f(m_i)\geq \sum\limits_{i=1}^c f(k_i) \ .
\end{align}
\end{lemma}

For convenience of notation, we define a fixed notation for the set over which we maximize in the following.
\begin{definition} \label{def:multiset}
  Denote by $\bbM_{c,B}^{[a,b]} = \{\cM,\ldots\}$ the set of all multisets
  $\cM = \{\{m_1,\ldots, m_c\}\}$
  of cardinality $c$ with $b\geq m_1 \geq \ldots \geq m_c\geq a $ and $\sum_{m\in \cM} m = B$.
\end{definition}

With these definitions established, we are now ready to give an upper bound on the sum over the results of a convex non-decreasing function evaluated on the elements of any multiset in $\bbM_{c,B}^{[a,b]}$.

\begin{lemma}\label{lem:maximization}
  Let $a,c\geq 1$, $b\geq a$, $ca\leq B \leq c b$, and $\bbM_{c,B}^{[a,b]}$ be as in \cref{def:multiset}. For any function $f(x)$ that is convex and non-decreasing in the interval $a\leq x \leq b$, it holds that
  \begin{align*}
    \max\limits_{\cM\in\bbM_{c,B}^{[a,b]}} \sum_{m\in \cM} f(M) \!\leq\! \parenv*{\frac{B-ca}{b-a}+1}\parenv*{ f(b) - f(a)} + cf(a) \ .
  \end{align*}
\end{lemma}
\begin{proof}
  Let $\delta^m_{\cM}$ be the multiplicity of an element $m\in\cM$. Denote by $\widetilde{\cM}$ the set of distinct elements in $\cM$.
  By definition, 
  \begin{align*}
    \sum\limits_{M\in\cM}M=\sum\limits_{m \in \widetilde{\cM}}\delta^m_{\cM}\cdot m=B, \ \forall  \cM\in\bbM_{c,B}^{[a,b]}
  \end{align*}
  and it follows that for all $\cM\in\bbM_{c,B}^{[a,b]}$ we have
  \begin{align*}
    \delta^b_{\cM}  &= \frac{1}{b}\bigg({B - \sum\limits_{m \in \widetilde{\cM}\setminus\{b\}}\delta^m_{\cM}\cdot m}\bigg) \leq \frac{B-(c-\delta^b_{\cM})a }{b}\ , \textrm{ and then,}\\
    \delta^b_{\cM} &\leq \frac{B-ca}{b-a}\ .
  \end{align*}

  Let $\cM_{\max} = \{b,\dots,b,a,\dots,a\}$ be a multiset with $\delta^b_{\cM_{\max}}=\ceil{\frac{B-ca}{b-a}}$ and $\delta^a_{\cM_{\max}} = c-\delta^b_{\cM_{\max}}$.
  It can readily be seen that $\cM_{\max}\succ\cM, \ \forall \ \cM\in\bbM_{c,B}^{[a,b]}$ (note that $\cM_{\max}\in \bbM_{c,B}^{[a,b]}$ if $(b-a)|(B-ca)$).
Since $f(x)$ is a convex non-decreasing function for $a\leq x \leq b$, it follows from \cref{lem:karamata} that
\begin{equation}\label{eq:sum_inequality_karamata}
\sum\limits_{m\in \cM_{\max}} f(m) \geq \sum\limits_{m\in \cM} f(m) \ , \ \forall\ \cM \in \bbM_{c,B}^{[a,b]}\ .
\end{equation}
Hence,
\begin{align*}
  \max_{\cM \in \bbM_{c,B}^{[a,b]}} \sum_{m\in \cM} f(m) &\leq \sum\limits_{m\in \cM_{\max}} f(m) \\
  &= \delta^b_{\cM_{\max}} f(b) + (c-\delta^b_{\cM_{\max}})f(a) \\
  &=  \ceil{\frac{B-ca}{b-a}}\parenv*{ f(b) - f(a)} + cf(a)
\end{align*}
and the statement follows.
\end{proof}

\subsubsection{Sum of Cardinalities of Alternant Codes}

Specific subclasses of alternant codes, such as some BCH and Goppa codes, are known to have larger dimension \cite{macwilliams1977theory} than the lower bound given in \cref{lem:AlternantDimBounds}.
However, in general it is a difficult and open problem to predict the dimension of an alternant code for arbitrary column multipliers $\bv$.
Nevertheless, the sum of the cardinalities of subfield subcodes over all nonzero column multipliers can be determined, given the weight enumerator of the code, as we show in the following.
This approach works not only for alternant codes, but also for any linear codes with known weight enumerator.

For a linear $[n,k,d]_{q^m}$ code $\code$, denote by $B_{n,d,w}(\code)$
the sum of the number of codewords of weight $w$ in the $\Fq$-subfield subcodes of $\cC$ over all nonzero column multipliers, i.e.,
\begin{align*}
  B_{n,d,w}(\code) &\defeq \sum_{\bv \in (\Fqm^{*})^n} \Big\lvert \{ \bc \cdot \diag(\bv) \ | \ \bc \in \code, \wtH(\bc) = w\} \cap \Fq^n \Big\rvert \ .
\end{align*}
Since every linear code contains the zero codeword and there is no codeword of Hamming weight $<d$ in the $[n,k,d]_{q^m}$ code, the sum of the cardinalities of the $\Fq$-subfield subcodes over all nonzero column multipliers is given by
\begin{align*}
  B_{n,d}(\code) &\defeq \sum_{\bv \in (\Fqm^{*})^n} \Big\lvert \{ \bc \cdot \diag(\bv) \ | \ \bc \in \code\} \cap \Fq^n\Big\rvert = (q^m-1)^n + \sum_{w=d}^{n} B_{n,d,w}(\code)  \ .
\end{align*}

If $\code$ is a $\GRS_{\balpha,\bv'}^d$ code as defined in \cref{def:GRScodes} for some $\bv' \in (\Fqm^{*})^n$, then $B_{n,d,w}$ is the sum of the number of codewords of weight $w$ in all alternant codes $\ALTallp$, and $B_{n,d}(\code)$ is the sum of the cardinalities of all $\ALTallp$.
Interestingly, while the weight enumerator and cardinality of a specific subfield subcode depend on $\bv$, the sum of these values over all $\bv$ only depends on the weight enumerators of $\code$.

\begin{lemma}\label{lem:sumCardinalitiesAlternant}
Let $\code$ be an $[n,k,d]_{q^m}$ code and denote by $A_w^{\code}$ the $w$-th weight enumerator of $\code$. Then,
\begin{align*}
  B_{n,d,w}(\code) = A^{\code}_{w}\cdot (q^m-1)^{n-w}(q-1)^w \ .
\end{align*}
\end{lemma}
\begin{proof}
Let $\bc$ be a codeword of $\code$. We have $\bc \cdot \diag(\bv) \in \Fq^n$ if and only if $c_iv_i \in \Fq$ for all $i \in [n]$. If $i \in \supp(\bc)$, then there are exactly $q-1$ choices of $v_i$ for which $c_iv_i \in \Fq$. Else, any of the $q^m-1$ possible values of $v_i$ give $c_iv_i=0 \in \Fq$. Hence, we have
\begin{align*}
B_{n,d,w}(\code)\hphantom{\sum} &=  \sum_{\bv \in (\Fqm^{*})^n} \left| \set*{ \bc \cdot \diag(\bv) \ | \ \bc \in \code, \wtH(\bc) = w} \cap \Fq^n \right| \\
 &= \sum_{\substack{\bc \in \code \\ \wtH(\bc) = w}} \left| \set*{ \bv \in (\Fqm^{*})^n \ | \ c_iv_i \in \Fq, \ \forall i \in [n] }  \right| \\
  &= A^{\code}_{w}\cdot (q^m-1)^{n-w}(q-1)^w \ .
\end{align*}
\end{proof}

If $\cC$ is an MDS code, 
then its weight enumerators $A_{w}^\code$, as given in~\cref{thm:MDSWeightEnumerator},
is completely determined by the code parameters (length, dimension/distance) and independent from the specific code constructions.
\begin{theorem}[Weight enumerators of MDS codes {\cite[Ch.~11, Theorem~6]{macwilliams1977theory}}]\label{thm:MDSWeightEnumerator}
  Let $\code$ be an $[n,k,d]_{q^m}$ MDS code. The $w$-th weight enumerator $A_w^{\mathsf{MDS}}$ of $\code$ is $A_0^{\mathsf{MDS}}=1$ and for $w\neq 0$,
  \begin{align*}
    A_w^{\mathsf{MDS}} &\defeq \left|\set*{ \bc\in \code \ | \ \wtH(\bc) = w} \right| 
    = \binom{n}{w} \sum_{j=0}^{w-d} (-1)^j \binom{w}{j} (q^{m(w-d+1-j)}-1)  \ .
  \end{align*}
\end{theorem}
Hence, for an MDS code $\code$, we can write the sum of the cardinalities of the $\Fq$-subfield subcodes of $\cC$ without dependence on $\code$ as
\begin{equation}\label{eq:BforMDS}
  B^{\mathsf{MDS}}_{n,d,w} \defeq B_{n,d,w}(\code)  \ \text{and} \ B^{\mathsf{MDS}}_{n,d} \defeq B_{n,d}(\code)\ .
\end{equation}

\subsubsection{Probability of a Code Containing a Random Matrix}

We begin by proving a technical lemma that bounds the probability that all rows of a randomly chosen matrix with nonzero columns are in a code of a certain dimension. This is a refined version of \cite[Lemma~3]{schmidt2009collaborative}. 
Recall that $\EB{q}{\intOrd}{n}$ is the set of matrix without zero columns in $\Fq^{\intOrd\times n}$.
\begin{lemma}\label{lem:Pwk}
For some integers $\intOrd>0,n\geq k\geq 0$, let $\cA$ be an $[n,k]_q$ code and denote by $A_w^{\cA}$ its $w$-th weight enumerator.
Then, for a matrix $\bE$ taken independently from $\EB{q}{\intOrd}{n}$, we have
\begin{align*}
  \Pr_{\bE}\set*{\be_{i} \in\cA, \ \forall i \in [\intOrd]} \leq \frac{q^{k\intOrd}(q-1)- (q^\intOrd-1)(q^k-1-A_n^{\cA}) -(q-1) }{(q-1)(q^\intOrd-1)^n} \ ,
\end{align*}
where $\be_i$ is the $i$-th row of $\bE$.
\end{lemma}
\begin{proof}
  Let $\cL \subset \F_q^{\intOrd\times n}$ the set of matrices whose rows are codewords of $\cA$ and by $\cL_0\subset \cL$ the subset of all matrices in $\cL$ with at least one all-zero column.
  Denote by $\bar{\cA}\subset \cA$ the set of
  codewords of $\cA$ whose first nonzero entry is $1$.
  Then $\bar{\cA}$ has cardinality $|\bar{\cA}|=\frac{q^k-1}{q-1}$.
  It can be seen that
\begin{equation*}
  \left\{\bE \ \left| \
    \be_{1},\ldots,\be_{\intOrd} \ \text{are $\Fq$-scalar multiples of } 
    \be\in\bar{\cA}\cup\{0\},\wtH(\be)<n
  \right.
\right\} \subseteq \cL_{0}\ .
\end{equation*}
If $\be=\0$ there is only the zero matrix in this set. For all the nonzero $\be$ with $\wtH(\be)<n$, each row $\be_i$ of $\bE$ can be an $\Fq$-multiple of $\be$ and all such matrices $\bE$ are unique, if at least one row is not $\0$. The number of such choices is $q^\intOrd-1$, so 
\begin{align*}
    |\cL_0|&\geq (q^\intOrd-1) (|\bar{\cA}|-\underbrace{|\{\bc\in\bar{\cA} \ | \ \wtH(\bc)=n\}|}_{=\frac{A_n^{\cA}}{(q-1)}}) + 1 
    = \frac{(q^\intOrd-1)}{(q-1)}(q^k-1-A_n^{\cA}) +1 \ .
\end{align*}
Recall that $\EB{q}{\intOrd}{n}$ does not contain any matrices with all-zero columns by definition, so $\cL_0\cap \EB{q}{\intOrd}{n}= \emptyset$. As $\cL_0 \subset \cL$, it follows that
\begin{align*}
  \Pr_{\bE}\set*{\be_{i} \in\cA, \ \forall i=[\intOrd]}&=\frac{|\cL\cap\EB{q}{\intOrd}{n}|}{|\EB{q}{\intOrd}{n}|} = \frac{|\cL\setminus\cL_{0}|}{|\EB{q}{\intOrd}{n}|} = \frac{|\cL|-|\cL_0|}{|\EB{q}{\intOrd}{n}|} \ .
\end{align*}
The statement follows from $|\cL|=|\cA|^\intOrd=q^{k\intOrd}$ and $|\EB{q}{\intOrd}{n}|=(q^\intOrd-1)^n$.

\end{proof}

If $|\cL_0|$ is large, it is worthy to deduct it from $|\cL|$ as in \cref{lem:Pwk}. However, for some parameters, (our best lower bound on) $|\cL_0|$ becomes negligible compared to $|\cL|$. Therefore, we also define a simplified version of this upper bound, where we only exclude the zero matrix from $\cL$.
The difference between $\labelMain$ and $\labelLz$ in the \cref{fig:plots1,fig:plots2} reflects the difference between \cref{lem:Pwk} and \cref{cor:PwkL01}.
\begin{corollary}\label{cor:PwkL01}
For some integers $\intOrd>0,n\geq k\geq 0$, let $\cA$ be an $[n,k]_q$ code.
Then, for $\bE$ that is i.i.d.~in $\EB{q}{\intOrd}{n}$, we have
\begin{equation*}
    \Pr_{\bE}\set*{\be_{i} \in\cA, \ \forall i\in[\intOrd]} \leq  \frac{|\cL\setminus\{\0_{\intOrd\times n}\}|}{|\EB{q}{\intOrd}{n}|} = \frac{q^{k\intOrd} -1 }{(q^\intOrd-1)^n} \ .
\end{equation*}
\end{corollary}

With all the technical tools established, we are now ready to present the bounds on the success probability of decoding interleaved alternant codes using the decoder from \cite{feng1991generalization,schmidt2009collaborative} (see also~\cref{algo:SyndromeDecoder}).

Recall that the success probability is given by
\begin{align*}
  \Psuc = 1- \Pfail - \Pmisc \ ,
\end{align*}
where $\Pfail$ and $\Pmisc$ are the probability of a decoding failure and a miscorrection, respectively. 


\subsection{A Lower Bound on Success Probability}
\label{sec:lb-succ-IA}
In order to derive a lower bound on the success probability, we first establish a connection between the multisets $\mathbb{A}_{\balpha|_{\cV}}^{d-t}$ as defined in \eqref{eq:multiset-alternant} and the probability
of successful decoding.
\begin{theorem}\label{thm:failureProbNewBoundGeneral}
  Let $\mathcal{IC}^{(\intOrd)}$ be an $\intOrd$-interleaved alternant code with $\code \in \mathbb{A}^d_{\balpha}$, $n= |\balpha|$ and $\cE=\set*{j_1,j_2, \dots,j_t}\subset [n]$ be a set of $|\cE|=t$ error positions.
  For a codeword $\bC\in \mathcal{IC}^{(\intOrd)}$, an error matrix $\widetilde{\bE} \in \Fq^{\intOrd \times n}$ with $\supp(\widetilde{\bE}) =\cE$ and $\bE \defeq \widetilde{\bE}|_{\cE}$ i.i.d.~in $\EB{q}{\intOrd}{t}$, and a received word $\bR = \bC + \widetilde{\bE}$, \cref{algo:SyndromeDecoder} \emph{succeeds}, i.e., returns $\hat \bC = \bC$, with probability
\begin{align*}
  \Psuc(\mathcal{IC}^{(\intOrd)},\cE) &\geq 1- \sum_{w=d-t}^t \sum_{\substack{\cV \subseteq [\cE]\\ |\cV|=w}} \sum_{\cA \in \mathbb{A}_{\balpha|_{\cV}}^{d-t}} \!\! \Big(\delta_{\mathbb{A}_{\balpha|_{\cV}}^{d-t}}^{\cA}\Big)^{-1} \underset{\bE}{\Pr}\{ \be_{i}|_{\cV} \in \cA, \ \forall  i\in [\intOrd]  \}\ ,
\end{align*}
where $\be_{i}|_{\cV}$ is the $i$-th row of $\bE$ restricted to the entries indexed in $\cV$ and $\delta_{\mathbb{A}_{\balpha|_{\cV}}^{d-t}}^{\cA}$ is the multiplicity of $\cA$ in the multiset $\mathbb{A}_{\balpha|_{\cV}}^{d-t}$.
\end{theorem}
\begin{proof}
  By \cref{lem:FailureCrux} the decoding of $\widetilde{\bE}$ is unsuccessful if and only if
  \begin{equation*}
    \exists \bv \in \Fqm^t \setminus \{\0\} \text{ such that } \bH \cdot \diag(\bv) \cdot \bE^\top = \0 \ ,
  \end{equation*}
  where $\bH\in\Fqm^{(d-t-1)\times t}$ is a parity-check matrix of the $[t,2t-d+1, d-t]_{q^m}$ code $\GRS_{\balpha|_\cE, \boldsymbol{1}}^{d-t}$.

Therefore, the probability of unsuccessful decoding is upper bounded by
\begingroup
\allowdisplaybreaks
\begin{align}
  1&-\Psuc(\mathcal{IC}^{(\intOrd)},\mathcal{E}) \nonumber \\
  &= \underset{\bE}{\Pr}\{\exists  \bv \in \Fqm^{t} \setminus \{\0\} \textrm{ s.t.}\ \bH \cdot \diag(\bv) \cdot \bE^\top = \0 \}\nonumber \\
                     &= \sum_{w=1}^{t} \underset{\bE}{\Pr}\{\exists  \bv \in \Fqm^{t}\textrm{ with } \wtH(\bv) = w \textrm{ s.t.}\ \bH \cdot \diag(\bv) \cdot \bE^\top = \0 \}\nonumber\\
                     &{=} \! \sum_{w=d-t}^t \! \underset{\bE}{\Pr}\{ \exists \bv \in \Fqm^{t}\textrm{ with } \wtH(\bv) \!=\! w \textrm{ s.t.} \ \bH \cdot \diag(\bv) \cdot \bE^\top \! = \0 \}\label{eq:ieq1}\\
                     &=  \sum_{w=d-t}^t \sum_{\substack{\cV \subseteq [\cE]\\ |\cV|=w}} \underset{\bE}{\Pr}\{ \exists  \cA \in \mathbb{A}_{\balpha|_\cV}^{d-t} \ \textrm{s.t.}\  \be_{i}|_{\cV} \in \cA, \ \forall i\in [\intOrd]  \}\nonumber\\
                     &\leq  \sum_{w=d-t}^t \sum_{\substack{\cV \subseteq [\cE]\\ |\cV|=w}} \sum_{\cA \in \mathbb{A}_{\balpha|_{\cV}}^{d-t}}  \parenv*{\delta_{\mathbb{A}_{\balpha|_{\cV}}^{d-t}}^{\cA}}^{-1} \underset{\bE}{\Pr}\{ \be_{i}|_{\cV}  \in \cA,\  \forall i\in [\intOrd]  \} \ , \nonumber
\end{align}
\endgroup
where \eqref{eq:ieq1} holds because any $d-t-1$ columns of $\bH$ are linearly independent.
\end{proof}

With this connection between the multisets $\mathbb{A}_{\balpha|_{\cV}}^{d-t}$ and the probability of successful decoding $\Psuc(\mathcal{IC}^{(\intOrd)},\cE)$ established, we now apply the technical results of \cref{sec:TechnicalPreliminaries} to obtain a lower bound.

\begin{theorem}[Lower bound on $\Psuc$]\label{thm:failureProbNewBound}
The probability of successful decoding $\Psuc(\cI\code^{(\intOrd)},\cE)$ as in \cref{thm:failureProbNewBoundGeneral} is lower bounded by
\begin{align*}
  \Psuc(\mathcal{IC}^{(\intOrd)},\cE) \geq 1-\sum_{w=d-t}^t \frac{\binom{t}{w}}{(q^m-1)(q^\intOrd-1)^w}
  \cdot &\Bigg( \frac{(q^\intOrd-1)}{(q-1)}\Big(c_w+B_{w,d-t,w}^{\mathsf{MDS}}-B_{w,d-t}^{\mathsf{MDS}}\Big) -c_w \\
  & + \bigg(\frac{B_{w,d-t}^{\mathsf{MDS}}-c_w a_w}{b_w-a_w}+1\bigg)( b_w^\intOrd - a_w^\intOrd) + c_w a_w^\intOrd \Bigg) \ ,
\end{align*}
with
\begin{align*}
  a_w &= \max\{1,q^{w-(d-t-1)m}\} , \quad
        b_w = q^{k_q^{\mathsf{opt.}}(w,d-t)} , \quad \textrm{and} \quad
    c_w = (q^m-1)^{w}\ ,
\end{align*}
where $B_{w,d-t}^{\mathsf{MDS}}$ and $B_{w,d-t,w}^{\mathsf{MDS}}$ are given in \eqref{eq:BforMDS} and $k_q^{\mathsf{opt.}}(w,d-t)$ is an upper bound on the dimension of a $q$-ary code of length $w$ and minimum Hamming distance $d-t$.
\end{theorem}
\begin{proof}
For a $q$-ary code $\cA$ denote $ k_{\cA} \defeq \dim_q(\cA)$.
\newcounter{storeeqcounter}
\newcounter{tempeqcounter}
Starting from \cref{thm:failureProbNewBoundGeneral}, we obtain
  \begin{align*}
  &1-\Psuc(\mathcal{IC}^{(\intOrd)}, \mathcal{E}) \leq  \sum_{w=d-t}^t \sum_{\substack{\cV \subseteq [\cE]\\ |\cV|=w}} \sum_{\cA \in \mathbb{A}_{\balpha|_{\cV}}^{d-t}} (\delta_{\mathbb{A}_{\balpha|_{\cV}}^{d-t}}^{\cA})^{-1} \ \underset{\bE}{\Pr}\{ (\bE|_{\cV})_{i,:} \in \cA \ \forall \ i\in [\intOrd]  \}  \\
                     &\stackrel{\mathsf{(a)}}{\leq}  \sum_{w=d-t}^t \sum_{\substack{\cV \subseteq [t]\\ |\cV|=w}} \sum_{\cA \in \mathbb{A}_{\balpha|_{\cV}}^{d-t}} (q^m-1)^{-1} \frac{(q-1)q^{\intOrd k_\cA}-(q^\intOrd-1)(q^{k_\cA}-1-A^\cA_w)-(q-1) }{(q-1)(q^\intOrd-1)^w}\nonumber\\
    &\stackrel{\mathsf{(b)}}{=} {\footnotesize \sum_{w=d-t}^t \sum_{\substack{\cV \subseteq [t]\\ |\cV|=w}}}\frac{(q^m-1)^{-1}}{(q^\intOrd-1)^w} \Bigg(\frac{(q^\intOrd-1)}{(q-1)}(c_w + B_{w,d-t,w}^{\mathsf{MDS}})-c_w + \sum_{\cA \in \mathbb{A}_{\balpha|_{\cV}}^{d-t}} \bigg(q^{\intOrd k_\cA}-\frac{(q^\intOrd-1)}{(q-1)}q^{k_\cA} \bigg)\Bigg)\nonumber\\
                     &\stackrel{\mathsf{(c)}}{\leq} \sum_{w=d-t}^t \hspace{-1ex}\frac{\binom{t}{w}(q^m-1)^{-1}}{(q^\intOrd-1)^w} \Bigg(\frac{(q^\intOrd-1)}{(q-1)}(c_w + B_{w,d-t,w}^{\mathsf{MDS}})-c_w + \hspace{-3.5ex}{\footnotesize \max_{\cM \in \bbM_{c_w,B^{\mathsf{MDS}}_{w,d-t}}^{[a_w,b_w]} }\! \sum_{M \in \cM}} \bigg(M^\intOrd-\frac{(q^\intOrd-1)}{(q-1)}M\bigg)\Bigg)\nonumber\\
                     &= \sum_{w=d-t}^t \frac{\binom{t}{w}(q^m-1)^{-1}}{(q^\intOrd-1)^w} \left(\frac{(q^\intOrd-1)}{(q-1)}(c_w + B_{w,d-t,w}^{\mathsf{MDS}}- B^{\mathsf{MDS}}_{w,d-t})-c_w +{\footnotesize\max_{ \mathcal{M} \in \bbM_{c_w,B_{w,d-t}}^{[a_w,b_w]}} \sum_{M \in \cM} }M^\intOrd\right) \nonumber
\end{align*}
where $\mathsf{(a)}$ holds by \eqref{eq:multiplicityAlternant} and \cref{lem:Pwk}, $(\mathsf{b})$ holds as $\sum_{\cA \in \mathbb{A}_{\balpha|_{\cV}}^{d-t}}A^\cA_w=B_{w,d-t,w}^{\mathsf{MDS}}$ (see~\eqref{eq:BforMDS}) and $|\mathbb{A}_{\balpha|_{\cV}}^{d-t}|=c_w$ (see~\eqref{eq:cardinalityAall}), and $\mathsf{(c)}$ holds as $a_w$ and $b_w$ are lower and upper bounds on the cardinality of all codes $\cA \in \mathbb{A}_{\balpha|_{\cV}}^{d-t}$ (see~\cref{lem:AlternantDimBounds}) and $\sum_{\cA \in \mathbb{A}_{\balpha|_{\cV}}^{d-t}} q^{k_\cA} = B_{w,d-t}^{\mathsf{MDS}}$ by \cref{lem:sumCardinalitiesAlternant}. The theorem statement follows by \cref{lem:maximization}.
\end{proof}

With the use of \cref{cor:PwkL01} instead of \cref{lem:Pwk} for the inequality at $\mathsf{(a)}$ in the proof we get a slightly simplified (though worse) lower bound.
\begin{corollary}[Simplified Lower Bound on $\Psuc$]\label{cor:failureProbNewBoundL0}
The probability of successful decoding $\Psuc(\cI\code^{(\intOrd)},\cE)$ as in \cref{thm:failureProbNewBoundGeneral} is lower bounded by
\begin{align*}
  \Psuc(&\mathcal{IC}^{(\intOrd)},\cE) \geq 1-\sum_{w=d-t}^t \frac{\binom{t}{w}(q^m-1)^{-1}}{(q^\intOrd-1)^w}
   \cdot \left(\bigg(\frac{B_{w,d-t}^{\mathsf{MDS}}-c_w a_w}{b_w-a_w}+1\bigg)( b_w^\intOrd - a_w^\intOrd) + c_w (a_w^\intOrd-1) \right)\ ,
\end{align*}
with
\begin{align*}
  a_w &= \max\{1,q^{w-(d-t-1)m}\} , \quad
        b_w = q^{k_q^{\mathsf{opt.}}(w,d-t)} , \quad \textrm{and} \quad
    c_w = (q^m-1)^{w}\ ,
\end{align*}
where $B_{w,d-t}^{\mathsf{MDS}}$ is given in \eqref{eq:BforMDS} and $\kopt(w,d-t)$ is an upper bound on the dimension of a $q$-ary code of length $w$ and minimum Hamming distance $d-t$.

\end{corollary}

\subsubsection[A Lower Bound on Success Probability for Large Interleaving Order]{A Lower Bound on Success Probability for Large Interleaving Order  $\intOrd\geq t$}
\label{sec:largeEll}

For large interleaving order $\intOrd \geq t$, the Metzner-Kapturowski generic decoder~\cite{metzner1990general} guarantees to decode any $1\leq t\leq d-2$ errors if $\rank(\bE)=t$ in an $\intOrd$-interleaved code with \emph{any} $[n,k,d]_q$ constituent code. The decoder has been generalized in~\cite{haslach2000efficient} for the case of rank deficiency when $2t-d+2 \leq \rank(\bE)< t$. However, if the structure of the constituent code is unknown, determining the error positions in a rank-deficient error matrix $\bE$ where $\rank(\bE)=\mu<t$ is equivalent to finding a subset $\cU$ of columns of a parity-check matrix $\bH\in\Fqm^{(d-1-\mu)\times n}$ with $\rank(\bH|_{\cU})=t-\mu$. This is known to be a hard problem and no polynomial-time algorithm is known if the rank deficiency $t-\mu$ becomes large~\cite{roth2014coding}. If the code structure is given, efficient syndrome-based algorithms are proposed in~\cite{roth2014coding} and~\cite{YuLoeliger16} to correct linearly dependent error patterns with $\rank(\bE)\geq 2t-d+2$ by interleaved RS codes over $\Fqm$.
These decoders also apply to the class of alternant codes over $\Fq$.
Consider an $\intOrd$-interleaved alternant code $\cI\code^{(\intOrd)}$ where $\code \in \mathbb{A}_{\balpha}^d$, $n= |\balpha|$ and any set $\cE\subset [n]$ of $|\cE|=t$ error positions. A lower bound on the success probability is given in~\cite[Section II.C]{roth2014coding} as
\begin{equation}
  \begin{aligned}
    \Psuc(\cI\code^{(\intOrd)}, \cE)&\geq 1-\Pr\{\rank(\bE)< 2t-d+2\}\\
    & = 1-q^{-(\intOrd+d-1-2t)(d-1-t)}(1+o(1))\\
    & = 1-q^{-2(t-\frac{3(d-1)+\intOrd}{4})^2+\frac{(d-1-\intOrd)^2}{8}}(1+o(1))\ ,
    \label{eq:Pf_large_ell_Roth}
  \end{aligned}
\end{equation}
where $o(1)$ is an expression that goes to $0$ as  $q\to \infty$.

Note that though the decoder in~\cite{roth2014coding} can be applied to interleaved alternant codes, the above lower bound is an asymptotic result. For some applications of alternant codes that we are interested in, e.g., Goppa codes in the McEliece system, the field size $q$ is required to be finite or rather small. Therefore, in order to be self-contained and have a general expression on the failure probability, we prove in~\cref{lem:largeEllNoFailCond} that~\cref{algo:SyndromeDecoder} will always succeed in decoding linearly dependent error patterns if $\rank(\bE)\geq 2t-d+2$ and we then give a lower bound in~\cref{thm:LargeEll} on the success probability for $\intOrd\geq t$.

\begin{lemma}\label{lem:largeEllNoFailCond}
  Assume $\intOrd\geq t$. Let $\mathcal{IC}^{(\intOrd)}$ be an $\intOrd$-interleaved alternant code with $\code \in \mathbb{A}^d_{\balpha}$, $n= |\balpha|$ and $\cE=\{j_1,j_2,\dots,j_t\}\subset [n]$ be a set of $|\cE|=t$ error positions. For a codeword $\bC\in \mathcal{IC}^{(\intOrd)}$, an error matrix $\widetilde{\bE} \in \Fq^{\intOrd \times n}$ with $\supp(\widetilde{\bE}) = \cE$ and $\bE\defeq\widetilde{\bE}|_{\cE} $ i.i.d.~in $\EB{q}{\intOrd}{t}$, and a received word $\bR = \bC + \widetilde{\bE}$, \cref{algo:SyndromeDecoder} \emph{succeeds}, i.e., returns $\hat \bC = \bC$, if
  \begin{align*}
\rank(\bE)\geq 2t-d+2\ .
  \end{align*}
\end{lemma}
\begin{proof}
  Recall from \eqref{eq:ieq1} in the proof of \cref{thm:failureProbNewBoundGeneral} that the decoding does not succeed if and only if
  \begin{equation}
    \exists \bv \in \Fqm^t\setminus\{\0\} \textrm{ with } \wtH(\bv)\geq d-t\ \text{s.t.}\ \bH \cdot \diag(\bv) \cdot \bE^\top = \0 \ , \label{eq:condition-large-ell}
  \end{equation}
  where $\bH$ is a parity-check matrix of the $[t, 2t-d+1, d-t]_{q^m}$ $\GRS_{\balpha|_{\cE},\1}^{d-t}$ code.

  We show that this condition cannot be fulfilled \emph{if} $\rank(\bE)\geq 2t-d+2$.

  Assume $\rank(\bE) \geq 2t-d+2$. Denote
  $\wtH(\bv)=w$ and
  $\cL\defeq\supp(\bv)$.
  Let $\bar{\bH} = \bH|_{\cL}$, $\bar{\bv}=\bv|_{\cL}$, and $\bar{\bE}= \bE|_{\cL}$.
  Observe the equivalence
\begin{equation}\label{eq:failConditionLargeEll}
  \bH|_{\cE} \cdot \diag(\bv) \cdot \bE^\top = \0 \quad \iff \quad \bar{\bH} \cdot \diag(\bar{\bv}) \cdot \bar{\bE}^\top = \0 \ .
\end{equation}
Note that
\begin{align*}
  \rank(\bar{\bE}) &\geq \rank(\bE) - (t-w) 
  \geq 2t-d+2-(t-w) = w-(d-t)+2 
\end{align*}
and $\bar{\bH} \cdot \diag(\bar{\bv})$ is a parity-check matrix of the $[w,w-(d-t)+1,d-t]_{q^m}$ $\GRS_{\balpha|_{\cL},\bar{\bv}}^{d-t}$ code.

Assume for some $\bv$ \eqref{eq:failConditionLargeEll} is fulfilled. Then 
all rows of $\bar{\bE}$ are codewords of the GRS code.
In other words, the code spanned by $\bar{\bE}$ is a subcode of the GRS code,
i.e., $\myspan{\bar{\bE}} \subseteq \GRS_{\balpha|_{\cL},\bar{\bv}}^{d-t}$.
However, since $\dim(\myspan{\bar{\bE}}) = \rank(\bar{\bE}) \geq w-(d-t)+2> w-(d-t)+1 =\dim(\GRS_{\balpha|_{\cL},\bar{\bv}}^{d-t})$,
this is a contradiction.
\end{proof}

\begin{theorem}[Lower bound on $\Psuc$ for $\intOrd\geq t$]
  \label{thm:LargeEll}
  Assume $\intOrd\geq t$. Let $\mathcal{IC}^{(\intOrd)}$ be an $\intOrd$-interleaved alternant code with $\code \in \mathbb{A}^d_{\balpha}$, $n= |\balpha|$ and $\cE=\{j_1,j_2,\dots,j_t\}\subset [n]$ be a set of $|\cE|=t$ error positions. For a codeword $\bC\in \mathcal{IC}^{(\intOrd)}$, an error matrix $\widetilde{\bE} \in \Fq^{\intOrd \times n}$ with $\supp(\widetilde{\bE}) = \cE$ and $\bE \defeq \widetilde{\bE}|_{\cE}$ i.i.d.~in $\EB{q}{\intOrd}{t}$, and a received word $\bR = \bC + \widetilde{\bE}$, \cref{algo:SyndromeDecoder} \emph{succeeds}, i.e., returns $\hat \bC = \bC$, with probability
  \begin{align*}
    \Psuc(\cI\code^{(\intOrd)},\cE)&\geq 
                       \frac{\sum\limits_{s=2t-d+2}^{t}N(\intOrd,t,s)}{(q^\intOrd-1)^t}\ ,
  \end{align*}
  where
  \begin{align*}
    N(\intOrd,t,s)\defeq  |\{\bE \in \EB{q}{\intOrd}{t} \ | \ \rank(\bE)=s\}|
    =  \sum_{j=0}^{t-s}(-1)^j\binom{t}{j}\prod_{i=0}^{s-1}\frac{(q^\intOrd-q^i)(q^{t-j}-q^i)}{q^s-q^i}\ .
  \end{align*}
\end{theorem}
\begin{proof}

  By~\cref{lem:largeEllNoFailCond}, it can be readily seen that the success probability is bounded from below by
   \begin{align*}
     \Psuc(\cI\code^{(\intOrd)},\cE)& \geq \frac{|\{\bE\in\EB{q}{\intOrd}{t} \ | \ \rank(\bE)\geq 2t-d+2 \}|}{|\EB{q}{\intOrd}{t}|}\\
     & = \frac{\sum\limits_{s=2t-d+2}^{t}|\{\bE\in\EB{q}{\intOrd}{t} \ | \ \rank(\bE)=s\}|}{(q^\intOrd-1)^t}\ .
   \end{align*}

   It remains to determine
   \begin{align*}
     N(\intOrd,t,s)=|\{\bE\in\EB{q}{\intOrd}{t}\ | \ \rank(\bE)=s\}| \ ,
   \end{align*}
   the number of matrices of $\F_q^{\intOrd \times t}$ \emph{without} zero column and of a given rank. The number of matrices, including those \emph{with} zero columns, of certain rank is given in~\cite{Landsberg1893}\cite[Theorem~2]{fisher1966matrices}: 
   \begin{align*}
     M(\intOrd,t,s) &\defeq |\{\bE\in\Fq^{\intOrd\times t} \ | \ \rank(\bE)=s\}|
     =\prod_{i=0}^{s-1}\frac{(q^\intOrd-q^i)(q^t-q^i)}{q^s-q^i}\ .
   \end{align*}
   To obtain $N(\intOrd, t,s)$, we need to exclude the matrices with zero columns from $M(\intOrd,t,s)$.
   By the inclusion-exclusion principle, we have
   \begin{align*}
     N(\intOrd,t,s) &= \sum_{j=0}^{t-s}(-1)^j\binom{t}{j}M(\intOrd,t-j,s)\ .
   \end{align*}

\end{proof}
Comparisons between \cref{thm:LargeEll} and \cref{thm:failureProbNewBound} for some parameters can be found in \cref{fig:plots2} with the labels $\labelLarge$ and $\labelMain$ respectively.

\begin{remark}[Upper bound on the miscorrection probability $\Pmisc$]
  An upper bound on $\Pmisc$ of decoding interleaved alternant code by \cref{algo:SyndromeDecoder}
  is given  
  in \cite[Appendix A]{holzbaur2021decodingIA}.
  We expect the bound to be a rather rough upper bound, as it does not depend on the specific alternant code, nor the dimension of the alternant code.
  Nevertheless, we only intend to show that
  the probability of unsuccessful decoding of interleaved alternant codes is dominated by the failure probability,
  and the bound is sufficient for this purpose,
  as evident from the numerical results in \cref{fig:plots1,fig:plots2} under the label $\labelMisc$.
\end{remark}

\subsection[An Upper Bound on Success Probability]{An Upper Bound on Success Probability\footnote{The main contribution of this upper bound is by the co-author L.~Holzbaur of the work \cite{holzbaur2021decodingIA}, we include it here for completeness.}}
\label{sec:ub-succ-IA}
To evaluate the performance of the lower bounds of \cref{sec:lb-succ-IA}, we derive an upper bound on the probability of a decoding success.
The approach is to show that for a certain set of error matrices,
the decoder given in \cref{algo:SyndromeDecoder} is \emph{never} successful, i.e., the condition in \cref{lem:FailureCrux} never holds.

Recall from the proof of \cref{thm:failureProbNewBoundGeneral} that
\begin{align*}
  \Psuc(\mathcal{IC}^{(\intOrd)},\mathcal{E}) 
  = 1- \sum_{w=d-t}^t \underset{\bE}{\Pr}\{ \exists \bv \in \Fqm^{t}\textrm{ with } \wtH(\bv) = w \textrm{ s.t.} \ \bH \cdot \diag(\bv) \cdot \bE^\top  = \0 \}\tag{\textrm{Recall }\eqref{eq:ieq1}}\ ,
\end{align*}
where $\bH\in\Fqm^{(d-t-1)\times t}$ is a parity-check matrix of an $\RS_{q^m}[t, 2t-d+1]$ code and any $d-t-1$ columns of $\bH$ are linearly independent.
It can be readily seen that
\begin{align}
  \Psuc(\mathcal{IC}^{(\intOrd)},\mathcal{E}) 
  \leq 1-\underset{\bE}{\Pr}\{ \exists \bv \in \Fqm^{t}\textrm{ with } \wtH(\bv) = d-t \textrm{ s.t.} \ \bH \cdot \diag(\bv) \cdot \bE^\top  = \0 \}\ . \label{eq:psuc-lb-1}
\end{align}

\begin{lemma}
  \label{lem:psuc-up-by-Ebbad}
  Denote by $\Ebbad{d-t}$ the set of matrices $\bE \in \EB{q}{\intOrd}{t}$ for which there exists a vector $\be\in\Fq^{\intOrd}\setminus\set*{\0}$ that is collinear (i.e., a $\Fq$-scalar multiple) to at least $d-t$ columns of $\bE$.
  Then,
  \begin{align*}
    \bE\in\Ebbad{d-t}\iff \exists \bv \in \Fqm^{t}\textrm{ with } \wtH(\bv) = d-t \textrm{ s.t.} \ \bH \cdot \diag(\bv) \cdot \bE^\top  = \0 \ . 
  \end{align*}
\end{lemma}
\begin{proof}

  We first show the sufficiency $\Longleftarrow$.
Given a matrix $\bE\in\EB{q}{\intOrd}{t}$, let $\cL=\supp(\bv)$ and $|\cL|= d-t$ such that $\bH \cdot \diag(\bv) \cdot \bE^\top  = \0$. 
Let $\bar{\bH} \defeq \bH|_{\cL}\in\Fqm^{(d-t-1)\times (d-t)}$,
$\bar{\bv}\defeq \bv|_{\cL}\in\Fqm^{d-t}$, and $\bar{\bE}\defeq \bE|_{\cL}\in\Fq^{s\times (d-t)}$. We have
the equivalence
\begin{equation}\label{eq:equivalenceSupp}
  \bH \cdot \diag(\bv) \cdot \bE^\top = \0 \quad \iff \quad \bar{\bH} \cdot \diag(\bar{\bv}) \cdot \bar{\bE}^\top = \0 \ .
\end{equation}
As any $d-t-1$ columns of
$\bar{\bH}\cdot \diag(\bar{\bv}) \in \Fqm^{(d-t-1) \times d-t}$ are $\Fqm$-linearly independent (therefore $\Fq$-linearly independent), the right $\Fq$-kernel of $\bar{\bH}\cdot \diag(\bar{\bv})$ is of dimension at most $1$. Since \eqref{eq:equivalenceSupp} holds by assumption and there is at least one nonzero row in $\bar{\bE}$, the dimension of the right $\Fq$-kernel of $\bar{\bH}\cdot \diag(\bar{\bv})$ is at least $1$.
Together with the last argument, the right $\Fq$-kernel of $\bar{\bH}\cdot \diag(\bar{\bv})$ is of dimension exactly $1$ and generated by $\ba\in\parenv*{\Fq^*}^{d-t}$.
It then follows from \eqref{eq:equivalenceSupp} that all the nonzero rows of $\bar{\bE}$ are collinear to the vector $\ba$. Hence, $\rank(\bar{\bE})=1$ and we conclude that there exists a vector $\be\in\Fq^s\setminus\set*{\0}$ that is collinear to all the nonzero columns of $\bar{\bE}$, and $\bE\in\Ebbad{d-t}$ by definition.

Now we show the necessity $\implies$. 
Given a matrix $\bE\in\Ebbad{d-t}$, let $\cL\subset[t]$ with $|\cL|=d-t$ be some set of columns of $\bE$ that are collinear to a vector $\be\in\Fq^s\setminus\set*{\0}$.
Let 
$\bar{\bH} \defeq \bH|_{\cL}\in\Fqm^{(d-t-1)\times (d-t)}$ and
$\bar{\bE}\defeq  \bE|_{\cL}\in\Fq^{s\times (d-t)}$.
By assumption $\rank(\bar{\bE})=1$ and at least one row $\bar{\be}_{i}$ of $\bar{\bE}$ must be a nonzero scalar multiple of some vector $\be\in (\Fq^{*})^{d-t}$, since $\bE$ does not have any zero column.
Note that any $d-t$ columns of $\bH$ form a parity-check matrix of an $\RS_{q^m}[d-t,1]$ code (different sets of columns correspond to different code locators) and so does $\bar{\bH}$. Denote by $\RS_{\bar{\bH}}$ the RS code defined by $\bar{\bH}$.
Since $\wtH(\be)=d-t=\dH(\RS_{\bar{\bH}})$, we can always find a $\bar{\bv}\in\parenv*{\Fqm^*}^{d-t}$ such that the entry-wise multiplication of $\bar{\bv}$ and $\be$ is a codeword of $\RS_{\bar{\bH}}$, i.e., 
\begin{align*}
  \bar{\bH}\cdot (\bar{v}_1e_1, \bar{v}_2e_2,\dots, \bar{v}_{d-t}e_{d-t})^\top=\0\ .
\end{align*}
The $\bv\in\Fqm^{t}$ with $\bv|_{\cL}=\bar{\bv}$ and $0$ entries elsewhere results in $\bH\cdot \diag(\bv)\cdot \bE^\top=\0$ and hence the necessity is proven.

\end{proof}

Similar to $\Ebbad{d-t}$, we defined $\Ebbad{w}$ with $w\leq t$ to be the set of matrices $\bE \in \EB{q}{\intOrd}{t}$ for which there exists a vector $\be\in\Fq^{\intOrd}\setminus\set*{\0}$ that is a scalar multiple to at least $w$ columns of $\bE$.
We give a general quantification on the cardinality of the set $\Ebbad{w}$.
\begin{lemma}\label{lem:rankOneSubmatrices}
  For $w\leq t$, the cardinality of $\Ebbad{w}$ is bounded by
  \begin{align*}
    \max_{w\leq \xi \leq t} \{Z^\xi\} &\leq |\Ebbad{w}| \leq (t-w+1) \max_{w\leq \xi\leq t} \{Z^\xi\} \ , \textrm{ where }\\
    Z^\xi &= \sum_{j=1}^{\floor{\frac{t}{\intOrd}}} (-1)^{j-1} \binom{\frac{q^\intOrd-1}{q-1}}{j} D^\xi_{j}\ , \textrm{ and } \\
    D^\xi_{j} &= \left( \prod_{z=0}^{j-1} \binom{t-z\xi}{\xi}\right) (q-1)^{j\xi} (q^\intOrd - q^{j})^{t-j \xi} \ .
  \end{align*}
\end{lemma}
\begin{proof}
  Consider the equivalence relation $\equiv_q$ on $\Fq^\intOrd \setminus \{\0\}$ defined by $\bv \equiv_{q} \bu$ if  there exists $\lambda \in \Fq^*$ such that $\bv = \lambda \bu$.
  For a fixed vector $\be \in \Fq^\intOrd\setminus \{\0\}$ and a matrix $\bE\in \EB{q}{\intOrd}{w}$, denote by $\delta^{\be}_{\bE} \defeq |\{ i \ | \ \bE_{[:,i]} \equiv_q \be\}|$
  the number of columns of $\bE$ that are equivalent to $\be$ under $\equiv_q$.
  For a set of representatives $\cS \subset \Fq^\intOrd\setminus \{\0\}$ under the given equivalence relation, we have
  \begin{align*}
    D^\xi_{|\cS|} &\defeq |\{ \bE \in \EB{q}{\intOrd}{w} \ | \ \delta^{\be}_{\bE} = \xi, \  \forall  \be \in \cS \}| 
    = \left( \prod_{z=0}^{|\cS|-1} \binom{t-z\xi}{\xi}\right) (q-1)^{|\cS|\xi} (q^{\intOrd} - q^{|\cS|})^{t-|\cS| \xi} \ ,
  \end{align*}
  where the first term counts the ways to position the equivalent (under $\equiv_q$) vectors to $\be\in\cS$ into $\bE$, the second term is the number of choices for the scalar coefficients of these positions, and the third term is the number of choices for the remaining columns, namely any nonzero vector that is not equivalent to any vector in $\cS$.
  By the principle of inclusion-exclusion we get that the size of 
  \begin{align*}
    \cZ^{\xi} &\defeq \set*{ \bE \in \EB{q}{\intOrd}{w} \ | \ \exists \be \in \Fq^{\intOrd} \setminus \set*{\0} \textrm{ s.t.~} \delta^{\be}_{\bE} = \xi  } 
  \end{align*}
  is given by
  \begin{align*}
    Z^{\xi} &\defeq |\cZ^{\xi}| = \sum_{j=1}^{\floor{{t}/{\xi}}} (-1)^{j-1} \binom{\frac{q^{\intOrd}-1}{q-1}}{j} D^{\xi}_{j} \ .
  \end{align*}
  The statement follows from the observation that
  \begin{align*}
     \Ebbad{w} = \bigcup_{j=w}^{t} \cZ^{j} \ .
  \end{align*}
\end{proof}

Using the lower bound on the cardinality of $\Ebbad{w}$, we now derive an upper bound on the probability of successful decoding.

\begin{theorem}[Upper bound on $\Psuc$]\label{thm:lowerBound}
Let $\mathcal{IC}^{(\intOrd)}$ be an $\intOrd$-interleaved alternant code with $\code \in \mathbb{A}^d_{\balpha}$, $n= |\balpha|$ and $\cE=\{j_1,j_2,\dots,j_t\}\subset [n]$ be a set of $|\cE|=t\geq \frac{d}{2}$ error positions. For a codeword $\bC\in \mathcal{IC}^{(\intOrd)}$, an error matrix $\widetilde{\bE} \in \Fq^{\intOrd \times n}$ with $\supp(\widetilde{\bE}) = \cE$ and $\bE \defeq \widetilde{\bE}|_{\cE}$ i.i.d.~in $\EB{q}{\intOrd}{t}$, and a received word $\bR = \bC + \widetilde{\bE}$, \cref{algo:SyndromeDecoder} \emph{succeeds}, i.e., returns $\hat \bC = \bC$, with probability
\begin{align*}
  \Psuc(\mathcal{IC}^{(\intOrd)},\mathcal{E}) 
  &\leq 1-\frac{|\Ebbad{d-t}|}{(q^{\intOrd}-1)^t}
  \leq  1-\frac{\max_{(d-t) \leq \xi \leq t} \{Z^{\xi}\} }{(q^{\intOrd}-1)^t} \ ,
\end{align*}
  where $Z^\xi$ is given in \cref{lem:rankOneSubmatrices}.
\end{theorem}
\begin{proof}
  The statement follows directly from \eqref{eq:psuc-lb-1}, \cref{lem:psuc-up-by-Ebbad}, \cref{lem:rankOneSubmatrices}, and $|\EB{q}{\intOrd}{t}|=\parenv*{q^{\intOrd}-1}^t$.
\end{proof}

\subsection{Discussion and Numerical Results}\label{sec:numericalResults}
\begin{table*}[ht!]
  \centering
  \caption{Overview of the bounds shown in \cref{fig:plots1,fig:plots2}}
  \begin{tabularx}{\linewidth}{llX}
    \textbf{Label} & \textbf{Defined in} & \textbf{Description} \\ \hline \\
    $\labelRS$ & \cref{thm:boundIRS} & Lower bound on the probability of successful decoding for interleaved RS codes\\[.5em]
    $\labelMain$ & \cref{thm:failureProbNewBound} & Lower bound on the probability of successful decoding for interleaved alternant codes where the minimum of the Singleton, Griesmer, Hamming, Plotkin, Elias, and Linear Programming bound is used for $\kopt$. \\
    $\labelSingleton$ & \cref{thm:failureProbNewBound} & Lower bound on the probability of successful decoding for interleaved alternant codes, where the Singleton bound is used for $\kopt$. \\
    $\labelLz$ & \cref{cor:failureProbNewBoundL0} & Simplified version of \cref{thm:failureProbNewBound}. The minimum of the Singleton, Griesmer, Hamming, Plotkin, Elias, and Linear Programming bound is used for $\kopt$. \\[.5em]
    $\labelLarge$ & \cref{thm:LargeEll} & Lower bound on the probability of successful decoding for interleaved alternant codes with $\intOrd \geq t$ \\[.5em]
    $\labelMisc$ &
    \cite[Appendix A]{holzbaur2021decodingIA}
    & Upper bound on the probability of a miscorrection for interleaved alternant codes. We assume that the decoding radius of the interleaved decoder is $\floor{\frac{\intOrd}{\intOrd+1}(d-1)}$, i.e., the largest number of errors for which the \emph{RS interleaved decoder}, given in \cref{algo:SyndromeDecoder}, would succeed (see~\cref{rem:alternantMaxRadius}).\\[.5em]
    $\labelLower$ & \cref{thm:lowerBound} & Upper bound on the probability of successful decoding for interleaved alternant codes.\\[.5em]
    $\labelSim$ & \cref{rem:alternantMaxRadius} & Threshold number of errors such that for all numbers of errors left of the indicated line, the interleaved alternant decoder succeeds with a probability of $\Psuc > 0.9$ obtained by simulation with $100$ decoding iterations per parameter set.\\[.5em]
   \hline 
  \end{tabularx}
  \label{tab:boundsInPlots}
\end{table*}

In \cref{sec:lb-succ-IA,sec:ub-succ-IA} we have established lower and upper bounds on the probability \begin{align*}
  \Psuc = 1-\Pfail-\Pmisc
\end{align*}
of successful decoding of interleaved alternant codes by the decoding algorithm from \cite{feng1991generalization,schmidt2009collaborative} (see also \cref{algo:SyndromeDecoder}),
assuming uniformly distributed burst errors of a given weight.
In the following we present and discuss some numerical results, where we compare these upper and lower bounds\footnote{For better presentation, we plot the respective bounds on the probability $1-\Psuc$ of \emph{unsuccessful} decoding instead of the bounds on $\Psuc$.}.
In order to better emphasize the individual contributions of failures and miscorrections, we further include an upper bound on the probability of miscorrection $\Pmisc$ from \cite[Appendix A]{holzbaur2019decoding}, in the plots of \cref{fig:plots1,fig:plots2}. We label, summarize, and describe the different bounds and versions thereof in \cref{tab:boundsInPlots} and, for convenience and clarity, refer to them by their respective label for the remainder of this section. Further, we fix the code length to be $n = q^m-1$, i.e., given the base field size $q$ and extension degree $m$ we construct the longest possible RS/alternant codes, while excluding $\alpha_i = 0$ as a code locator (see~\cref{def:GRScodes}).

Aside from the comparison of the lower and upper bounds on the success probability, it is also interesting to see how the probability of successful decoding of an interleaved alternant code compares to that of the corresponding interleaved GRS code over $\Fqm$.
Such a bound was derived\footnote{The bound in \cite{schmidt2009collaborative} is presented as a bound on the probability of failure, but it is in fact a bound on the probability of unsuccessful decoding (see~\cref{rem:applicationToRS}).} and shown to be close to the probability of successful decoding obtained from simulation in~\cite{schmidt2009collaborative}.
For the readers' convenience we restate it in~\cref{thm:boundIRS} and assign it the label $\labelRS$. Note that the decoder employed in~\cite{schmidt2009collaborative} is equivalent to the decoder presented in \cref{algo:SyndromeDecoder}. The difference is that the error matrix $\widetilde{\bE}$ is assumed to be over $\Fqm$ (the field of the RS code) in~\cref{thm:boundIRS}.
\begin{theorem}[Probability of successful decoding for interleaved RS codes {\cite[Theorem~7]{schmidt2009collaborative}}]\label{thm:boundIRS}
  Let $\mathcal{IC}^{(\intOrd)}$ be an $\intOrd$-interleaved GRS code with $\GRSp \in \mathbb{G}^d_{\balpha}$ as in~\cref{def:GRScodes}, $n= |\balpha|$ and $\cE=\{j_1,j_2,\dots,j_t\}\subset [n]$ be a set of $|\cE|=t$ error positions.
  For a codeword $\bC\in \mathcal{IC}^{(\intOrd)}$, an error matrix $\widetilde{\bE} \in \Fqm^{\intOrd \times n}$ with $\supp(\widetilde{\bE}) = \cE$ and $\bE = \widetilde{\bE}|_{\cE}$ i.i.d.~$\EB{q^m}{\intOrd}{t}$, and a received word $\bR = \bC + \widetilde{\bE}$, \cref{algo:SyndromeDecoder} \emph{succeeds}, i.e., returns $\hat \bC = \bC$, with probability
\begin{align}\label{eq:boundIRS}
  \Psuc(\mathcal{IC}^{(\intOrd)},\cE)\geq 1-\left(\frac{q^{m\intOrd}-\frac{1}{q^m}}{q^{m\intOrd}-1}  \right)^t\cdot \frac{q^{-m(\intOrd+1)(t_{\max,\RS}-t)}}{q^m-1}\ ,
\end{align}
where $t_{\max,\RS}=\frac{\intOrd}{\intOrd+1}(d-1)$ as given in \cref{thm:IRS-max-decoding-radius}.
\end{theorem}

Before we discuss the numerical evaluations of the bounds, we make an important observation based on the \emph{simulation} results.
\begin{remark}\label{rem:alternantMaxRadius}
  For most parameters, the provided lower bounds on the success probability (i.e., upper bounds on $1-\Psuc$) of decoding interleaved alternant codes do not provide a non-trivial bound for the same decoding radius as the bounds for interleaved RS codes in \cite{schmidt2009collaborative}. To determine the real decoding threshold, i.e., the smallest number of errors for which the decoder succeeds with non-negligible probability\footnote{We arbitrarily choose this probability to be $\Psuc > 0.9$ and run $100$ decoding iterations for each parameter set to determine the decoding threshold.}, we rely on simulation results. This threshold is indicated in the plots and labeled $\labelSim$. Notably, for all tested parameters, \emph{the threshold for interleaved alternant codes is the same as for interleaved RS codes}, i.e., the simulation results imply that the joint decoding of interleaved alternant codes succeeds w.h.p.~in presence of burst errors of weight $t$ with
  \begin{align*}
    t \leq  \frac{\intOrd}{\intOrd+1}(d-1) = t_{\max, \mathsf{RS}} \ .
  \end{align*}
\end{remark}

The numerical evaluations of the bounds are given in \cref{fig:plots1,fig:plots2} for different base field size $q$, extension degree $m$, and distance $d$, each for varying interleaving order $\intOrd$:
\begin{itemize}
\item $\mathbf{q=2, m=10, d=51}$:
  The parameters are chosen such that the rate of the constituent alternant code is $\approx 0.5$, assuming the dimension is $k=n-(d-1)m$ (which tends to be true for most alternant codes). For wild Goppa and BCH codes the rate is
  $\approx 0.75$ (see Page~\pageref{rem:BCHgoppaCodes}).
  The bounds for codes with these parameters are compared for different interleaving order $\intOrd$ in Figs.~\ref{subfig:q=2_m=10_r=50_l=2} ($\intOrd=2$), \ref{subfig:q=2_m=10_r=50_l=5} ($\intOrd=5$), \ref{subfig:q=2_m=10_r=50_l=10} ($\intOrd=10$) and \ref{subfig:q=2_m=10_r=50_l=25} ($\intOrd=25$).
  \cref{subfig:q=2_m=10_r=50_l=50,subfig:q=2_m=10_r=50_l=80}
  are included to show the comparison between $\labelMain$ and $\labelLarge$.
\item $\mathbf{q=2, m=11, d=101}$: We compare to the parameters above by changing the extension degrees $m$ 
  in Figs.~\ref{subfig:q=2_m=11_r=100_l=2}, \ref{subfig:q=2_m=11_r=100_l=5}, \ref{subfig:q=2_m=11_r=100_l=10} and \ref{subfig:q=2_m=11_r=100_l=25}. The designed distance $d$ changes accordingly such that the rate of the constituent alternant code is $\approx 0.5$ (for wilde Goppa and BCH codes the rate is $\approx 0.75$).
\item $\mathbf{q=32, m=2, d=51}$: To illustrate the influence of the base field size $q$, we show some evaluations of the bounds for $q=32$ in Figs.~\ref{subfig:q=32_m=2_r=50_l=2} and \ref{subfig:q=32_m=2_r=50_l=10}. 
\end{itemize}

\begin{figure*}
\begin{subfigure}{.5\textwidth}
\centering
\input{figs/Dec-Int/boundsPlot_q=2_m=10_r=50_l=2.tex}
\caption{$q=2, m=10, d=51, \intOrd=2$}
\label{subfig:q=2_m=10_r=50_l=2}
\end{subfigure}
\begin{subfigure}{.5\textwidth}
\centering
\input{figs/Dec-Int/boundsPlot_q=2_m=11_r=100_l=2.tex}
\caption{$q=2, m=11, d=101, \intOrd=2$}
\label{subfig:q=2_m=11_r=100_l=2}
\end{subfigure}

\begin{subfigure}{.5\textwidth}
\centering
\input{figs/Dec-Int/boundsPlot_q=2_m=10_r=50_l=5.tex}
\caption{$q=2, m=10, d=51, \intOrd=5$}
\label{subfig:q=2_m=10_r=50_l=5}
\end{subfigure}%
\begin{subfigure}{.5\textwidth}
\centering
\input{figs/Dec-Int/boundsPlot_q=2_m=11_r=100_l=5.tex}
\caption{$q=2, m=11, d=101, \intOrd=5$}
\label{subfig:q=2_m=11_r=100_l=5}
\end{subfigure}

\begin{subfigure}{.5\textwidth}
\centering
\input{figs/Dec-Int/boundsPlot_q=2_m=10_r=50_l=10.tex}
\caption{$q=2, m=10, d=51, \intOrd=10$}
\label{subfig:q=2_m=10_r=50_l=10}
\end{subfigure}%
\begin{subfigure}{.5\textwidth}
\centering
\input{figs/Dec-Int/boundsPlot_q=2_m=11_r=100_l=10.tex}
\caption{$q=2, m=11, d=101, \intOrd=10$}
\label{subfig:q=2_m=11_r=100_l=10}
\end{subfigure}
\caption{Comparison of the bounds for different interleaving order $\intOrd$ and extension degree $m$. Rows are with different $\intOrd$ while the two columns are with different $m$. For the bounds $\labelRS, \labelMain, \labelSingleton, \labelLz, \labelLarge,$ and $\labelLower$ on the success probability we show the respective probabilities of unsuccessful decoding $1-\Psuc$. The references of the bounds can be found in \cref{tab:boundsInPlots}.}
\label{fig:plots1}
\end{figure*}
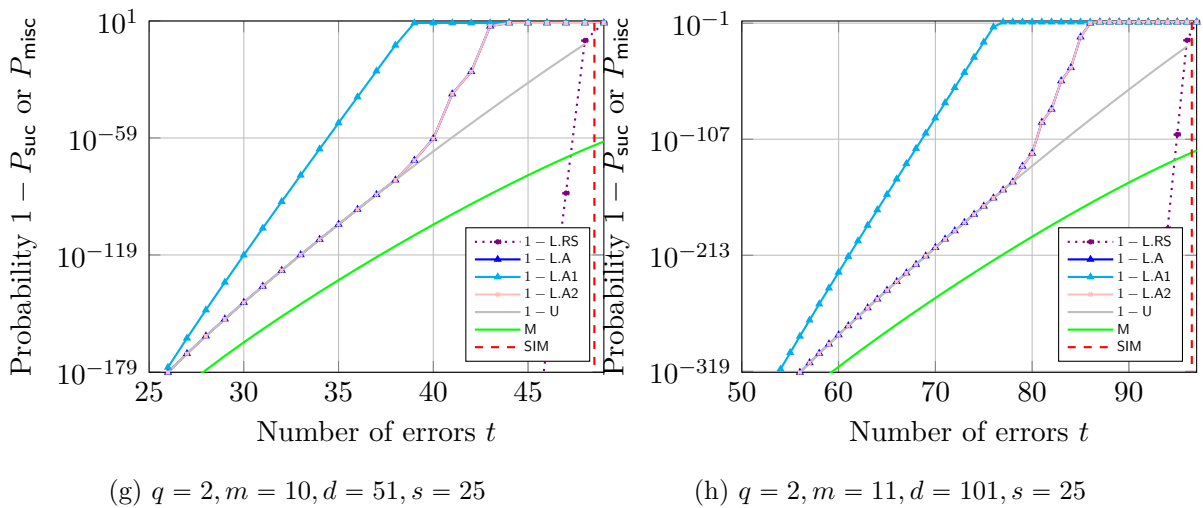
\begin{figure*}
  \ContinuedFloat
\begin{subfigure}{.5\textwidth}
\centering
\input{figs/Dec-Int/boundsPlot_q=2_m=10_r=50_l=25.tex}
\caption{$q=2, m=10, d=51, \intOrd=25$}
\label{subfig:q=2_m=10_r=50_l=25}
\end{subfigure}%
\begin{subfigure}{.5\textwidth}
\centering
\input{figs/Dec-Int/boundsPlot_q=2_m=11_r=100_l=25.tex}
\caption{$q=2, m=11, d=101, \intOrd=25$}
\label{subfig:q=2_m=11_r=100_l=25}
\end{subfigure}
\caption{(Cont'd.) Comparison of the bounds for different $\intOrd$ and $m$.}
\end{figure*}
\begin{figure}[h!]
\begin{subfigure}{.5\textwidth}
\centering
\input{figs/Dec-Int/boundsPlot_q=2_m=10_r=50_l=50.tex}
\caption{$q=2, m=10, d=51, \intOrd=50$}
\label{subfig:q=2_m=10_r=50_l=50}
\end{subfigure}%
\begin{subfigure}{.5\textwidth}
\centering
\input{figs/Dec-Int/boundsPlot_q=2_m=10_r=50_l=80.tex}
\caption{$q=2, m=10, d=51, \intOrd=80$}
\label{subfig:q=2_m=10_r=50_l=80}
\end{subfigure}%

\begin{subfigure}{.5\textwidth}
\centering
\input{figs/Dec-Int/boundsPlot_q=32_m=2_r=50_l=2.tex}
\caption{$q=32, m=2, d=51, \intOrd=2$}
\label{subfig:q=32_m=2_r=50_l=2}
\end{subfigure}%
\begin{subfigure}{.5\textwidth}
\centering
\input{figs/Dec-Int/boundsPlot_q=32_m=2_r=50_l=10.tex}
\caption{$q=32, m=2, d=51, \intOrd=10$}
\label{subfig:q=32_m=2_r=50_l=10}
\end{subfigure}
\caption{Comparison of the bounds for large interleaving order $\intOrd\geq t$ (a and b) and different base field size $q$ (c and d). For the bounds $\labelRS, \labelMain, \labelSingleton, \labelLz, \labelLarge,$ and $\labelLower$ on the success probability we show the respective probabilities of unsuccessful decoding $1-\Psuc$. The references of the bounds can be found in \cref{tab:boundsInPlots}.}
\label{fig:plots2}
\end{figure}

We now briefly discuss the main observations we have taken from the numerical results. As $\labelSingleton$ and $\labelLz$ are simplifications of $\labelMain$ and therefore strictly worse, we leave their comparisons to a later point in the section, and begin by only comparing $\labelRS, \labelMain, \labelLarge, \labelMisc$, and $\labelLower$. All statements on the decoding failure, miscorrection, and success probability refer to the syndrome-based joint decoder of \cite{feng1991generalization,schmidt2009collaborative} given in \cref{algo:SyndromeDecoder}.

\begin{itemize}
  \item For fixed $q,m$ and $\intOrd$, the success probability for interleaved ($q$-ary) alternant codes is significantly smaller than that for interleaved ($q^m$-ary) RS codes, since even the \emph{upper} bound $\labelLower$ on the success probability for interleaved alternant codes is in most cases smaller than the \emph{lower} bound $\labelRS$ on the success probability for interleaved RS codes.

  \item The probability of unsuccessful decoding interleaved alternant codes $1-\Psuc$ is dominated by the probability of failure $\Pfail$, as $\Pmisc \ll 1-\Psuc$, i.e., the bound on the probability of a miscorrection $\Pmisc$, labeled $\labelMisc$, is multiple orders of magnitude smaller than $1-\Psuc=\Pmisc+\Pfail$ for the best bound on $\Psuc$ among $\labelMain$ and $\labelLarge$. This is consistent with the numerical results from \cite{schmidt2009collaborative} for the case of decoding interleaved RS codes.

  \item For most parameters, $\labelMain$ provides the best lower bound on the success probability $\Psuc$. In particular, for higher interleaving order $\intOrd$ ($=10,25$ for example) and relatively small number of errors $t$, it essentially matches the upper bound of \cref{thm:lowerBound} (see~Figs. \ref{subfig:q=2_m=10_r=50_l=10}, \ref{subfig:q=2_m=11_r=100_l=10}, \ref{subfig:q=2_m=10_r=50_l=25}, and \ref{subfig:q=2_m=11_r=100_l=25} for example).

  \item For fixed $q,m$, and $d$, the relative gap between the number of errors for which the lower bounds on the success probability become non-trivial, i.e., give $\Psuc >0$, and the simulated decoding threshold decreases for increasing interleaving order~$\intOrd$ (compare \cref{subfig:q=2_m=10_r=50_l=2}, \ref{subfig:q=2_m=10_r=50_l=5}, \ref{subfig:q=2_m=10_r=50_l=10} and \ref{subfig:q=2_m=10_r=50_l=25} or \cref{subfig:q=2_m=11_r=100_l=2}, \ref{subfig:q=2_m=11_r=100_l=5}, \ref{subfig:q=2_m=11_r=100_l=10} and \ref{subfig:q=2_m=11_r=100_l=25}).

  \item The lower bound $\labelLarge$ on the success probability for $\intOrd > t$ improves upon the bound of $\labelMain$ for large number of errors that is close to the maximum decoding radius (see~\cref{subfig:q=2_m=10_r=50_l=50,subfig:q=2_m=10_r=50_l=80}).
\end{itemize}

Now consider the different versions of the lower bound on $\Psuc$ in \cref{thm:failureProbNewBound} labeled by $\labelMain$ (using the best upper bound on dimension of linear codes for $\kopt$), $\labelSingleton$ (using the Singleton bound for $\kopt$), and $\labelLz$ (\cref{cor:failureProbNewBoundL0}, a simplified version of $\labelMain$).
\begin{itemize}
\item For small $q$, the performance of \cref{thm:failureProbNewBound} is significantly worse when using the field-size-independent Singleton bound for $\kopt$, as evident from comparing $\labelMain$ and $\labelSingleton$ in Figs.~\ref{subfig:q=2_m=10_r=50_l=2} to \ref{subfig:q=2_m=11_r=100_l=25}, \ref{subfig:q=2_m=10_r=50_l=50} and \ref{subfig:q=2_m=10_r=50_l=80}. This can be expected due to the large gap between the field-size-dependent bounds and the Singleton bound for $\kopt$ when $q$ is small.

  \item For larger interleaving order $\intOrd$ ($\geq 10$ for example), the simplified lower bound $\labelLz$ approaches the most accurate version of the bound $\labelMain$ (see Figs.~\ref{subfig:q=2_m=10_r=50_l=10} to~\ref{subfig:q=2_m=11_r=100_l=25}, \ref{subfig:q=2_m=10_r=50_l=50} and \ref{subfig:q=2_m=10_r=50_l=80}).

\end{itemize}


\section{Other Results on Joint Decoding of Interleaved Codes}
\label{sec:other-joint-decoding}
This section briefly summarizes a selection of other works on joint decoding of interleaved codes. For more details the interested reader is referred to the respective publications.
\input{chap_dec_other_results.tex}

\section{Summary and Outlooks}


In this chapter, we investigated joint decoding of interleaved evaluation codes, in particular,
of interleaved alternant codes by the decoder from \cite{feng1991generalization,schmidt2009collaborative} in occurrence of burst errors.
We first recapped the syndrome-based joint decoding algorithm for interleaved Reed-Solomon from \cite{feng1991generalization,schmidt2009collaborative}
and showed that a sufficient condition on decoding success given in \cite{schmidt2009collaborative} for interleaved Reed-Solomon codes is also a necessary condition.
After adapting the
condition to the crux of jointly decoding interleaved alternant codes,
we provided a framework for characterizing the probability of decoding success.
Within this framework,
we derived a lower bound and an upper bound on the success probability that holds for any interleaving order. Inspired by a generic decoding method from \cite{metzner1990general,roth2014coding}, we derived another lower bound that works for the interleaving orders that are larger than the number of error positions.
Moreover, we numerically evaluated the obtained bounds for different code parameters, which show that one of the new lower bounds is tight for some parameters, as it matches the corresponding newly derived upper bound.
Finally, other two works related to joint decoding interleaved evaluation codes are summarized.
It can be seen from the plots of the bounds that there is a gap between the number of errors $t$, where the upper bound $1-\labelMain$ on $1-\Psuc$
provides a non-trivial value, and the $\tmax$, where simulated decoding succeeds with high probability.
For future research, the most apparent open problem is to improve the general lower bound $\labelMain$, in particular for smaller interleaving order, to close this gap.
On the other hand, the current upper bound $\labelLower$ on decoding success is derived by considering one out of $2t-d+1$ cases that the decoding never succeeds, where $d$ is the designed distance of the alternant code. Therefore improvements upon the upper bound should be further investigated.
Another closely related question, out of purely theoretical interest, is to determine the distribution of the dimensions of all alternant codes for a given set of code locators. For specific applications, such as code-based cryptography, improvements of the bounds for specific error distributions, e.g., full-rank errors, could be of practical relevance.

Although we only briefly summarized the recent results on list decoding of interleaved alternant codes, many open problems are left to be investigated. For example, 
the bottle neck of extending the algorithm to larger interleaving order $\intOrd\geq 3$ is the recovery step (root-finding) step. Efficient algorithms in eliminating variables in $\Fq[x][y_1,\dots, y_s]$ are needed. 
Moreover, the upper bound given in \cite{huang2022list} on the list decoding radius is quite far above the simulated number of decodable errors. Tighter upper bound or estimation by simulating on more code parameters should be further studied.

%% file: figs/Dec-Int/interleavedCodes.tex
\def\x{0.9}

\tikzset{
	ncbar angle/.initial=90,
	ncbar/.style={
		to path=(\tikztostart)
		-- ($(\tikztostart)!#1!\pgfkeysvalueof{/tikz/ncbar angle}:(\tikztotarget)$)
		-- ($(\tikztotarget)!($(\tikztostart)!#1!\pgfkeysvalueof{/tikz/ncbar angle}:(\tikztotarget)$)!\pgfkeysvalueof{/tikz/ncbar angle}:(\tikztostart)$)
		-- (\tikztotarget)
	},
	ncbar/.default=0.5cm,
}
\tikzset{round left paren/.style={ncbar=0.5cm,out=100,in=-100}}
\tikzset{round right paren/.style={ncbar=0.5cm,out=80,in=-80}}

\tikzset{square left brace/.style={ncbar=0.1cm}}
\tikzset{square right brace/.style={ncbar=-0.1cm}}
\begin{tikzpicture}

\node (S1) at (0,0) [draw=none] {};
\node (S2)  [right=\x*0.5cm of S1, draw=none] {};
\node (S3)  [right=\x*0.5cm of S2, draw=none] {};
\node (S4)  [right=\x*0.5cm of S3, draw=none] {};


\node (S5)  [right=\x*0.5cm of S4, draw=none] {};
\node (S6)  [right=\x*0.5cm of S5, draw=none] {};
\node (S7)  [right=\x*0.5cm of S6, draw=none] {};
\node (S8)  [right=\x*0.5cm of S7, draw=none] {};


\node (S9)  [right=\x*0.5cm of S8, draw=none] {};
\node (S10)  [right=\x*0.5cm of S9, draw=none] {};
\node (S11)  [right=\x*0.5cm of S10, draw=none] {};
\node (S12)  [right=\x*0.5cm of S11, draw=none] {};


\foreach \i in {1,2,3,4,5,6,7,8,10,11,12}{
  \node (C\i) at ($(S\i)+(0,\x*2)$) [draw,minimum width=\x*0.7cm,minimum height=\x*0.45cm] {};
  \node (C\i) at ($(S\i)+(0,\x*1.4)$) [draw,minimum width=\x*0.7cm,minimum height=\x*0.45cm] {};
  \node (C\i) at ($(S\i)+(0,\x*0.8)$) [draw,minimum width=\x*0.7cm,minimum height=\x*0.45cm] {};
  \node (C\i) at ($(S\i)+(0,\x*0.3)$) [minimum width=\x*0.7cm,minimum height=\x*0.45cm] {$\vdots$};
  \node (C\i) at ($(S\i)+(0,\x*-0.4)$) [draw,minimum width=\x*0.7cm,minimum height=\x*0.45cm] {};
}


  \node (C9) at ($(S9)+(0,\x*2)$) [draw=none,minimum width=\x*0.7cm,minimum height=\x*0.45cm,rounded corners=0pt] {$\hdots$};
  \node (C9) at ($(S9)+(0,\x*1.4)$) [draw=none,minimum width=\x*0.7cm,minimum height=\x*0.45cm] {$\hdots$};
  \node (C9) at ($(S9)+(0,\x*0.8)$) [draw=none,minimum width=\x*0.7cm,minimum height=\x*0.45cm] {$\hdots$};
  \node (C9) at ($(S9)+(0,\x*0.2)$) [draw=none,minimum width=\x*0.7cm,minimum height=\x*0.45cm] {$\hdots$};
  \node (C9) at ($(S9)+(0,\x*-0.4)$) [minimum width=\x*0.7cm,minimum height=\x*0.45cm] {$\hdots$};


\draw[dashed, rounded corners = 1pt, TUMBlue, very thick] ($(S1)+(-\x*.45,\x*2.3)$) rectangle ($(S12)+(\x*.45,\x*1.7)$);


\node[anchor = west] (L1) at ($(S12.east)+(\x*1.0,\x*2.1)$) {\color{TUMBlue}$\bc^{(1)} \in \cC$};

\node[anchor = east] (L1) at ($(S1.west)+(-\x*1.2,\x*0.7)$) {$\bC\quad\;\;=$};

\draw [] ($(S1) + (-\x*0.5,-\x*0.8)$) to [round left paren ] ($(S1) + (-\x*0.5,\x*2.4)$);
\draw [] ($(S12) + (\x*0.5,-\x*0.8)$) to [round right paren ] ($(S12) + (\x*0.5,\x*2.4)$);

\draw [decorate,decoration={brace,amplitude=7pt}]  ($(S12)+(\x*0.8,\x*2.4)$) -- ++(\x*0,-\x*3.2) node [black,midway,xshift=0.6cm] {$\ell$};

\draw [decorate,decoration={brace,amplitude=7pt}]  ($(S1)+(-\x*0.4,\x*2.5)$) -- ($(S12)+(\x*0.4,\x*2.5)$) node [black,midway,yshift=0.6cm] {$n$};


\draw[thick,-{Latex[length=3mm,width=2mm]},TUMRot] ($(S2)+(0,\x*2)$) -- ($(S2)+(-\x*0.25,\x*0.4)$) -- ($(S2)+ (\x*0.25,\x*0.9)$) -- ($(S2)+(0,-\x*0.5)$);
\draw[thick,-{Latex[length=3mm,width=2mm]},TUMRot] ($(S4)+(0,\x*2)$) -- ($(S4)+(-\x*0.25,\x*0.4)$) -- ($(S4)+ (\x*0.25,\x*0.9)$) -- ($(S4)+(0,-\x*0.5)$);
\draw[thick,-{Latex[length=3mm,width=2mm]},TUMRot] ($(S7)+(0,\x*2)$) -- ($(S7)+(-\x*0.25,\x*0.4)$) -- ($(S7)+ (\x*0.25,\x*0.9)$) -- ($(S7)+(0,-\x*0.5)$);
\draw[thick,-{Latex[length=3mm,width=2mm]},TUMRot] ($(S12)+(0,\x*2)$) -- ($(S12)+(-\x*0.25,\x*0.4)$) -- ($(S12)+ (\x*0.25,\x*0.9)$) -- ($(S12)+(0,-\x*0.5)$);


\node (El) at ($(S6)+(\x*0.6,-\x*2.2)$) {\color{TUMRot} Nonzero columns of the burst error $\widetilde{\bE}$};
\path ($(El.north)+(-\x*0.9,0)$) edge[bend left=5, -{Latex[length=2mm,width=1.5mm]}]  ($(S2)+(\x*0.1,-\x*0.8)$) ;
\path ($(El.north)+ (-\x*0.5,0)$) edge[bend left=5, -{Latex[length=2mm,width=1.5mm]}]  ($(S4)+(\x*0.1,-\x*0.8)$) ;
\path ($(El.north) + (\x*0.5,0)$) edge[bend right=5, -{Latex[length=2mm,width=1.5mm]}]  ($(S7)+(-\x*0,-\x*0.8)$) ;
\path ($(El.north)+(\x*0.9,0)$) edge[bend right=5, -{Latex[length=2mm,width=1.5mm]}]  ($(S12)+(-\x*0.1,-\x*0.8)$) ;

\end{tikzpicture}


%% file: figs/Dec-Int/boundsPlot_q=2_m=10_r=50_l=2.tex
\begin{tikzpicture}
\pgfplotsset{compat = 1.3}
\begin{axis}[
	legend style={nodes={scale=0.5, transform shape}},
	width = 0.97\columnwidth,
	height = 0.8\columnwidth,
	xlabel = {{Number of errors $t$}},
	ylabel = {{Probability $1-\Psuc$ or $\Pmisc$}},
	xmin = 25,
	xmax = 34,
	ymin = 1.271205e-10,
	ymax = 10,
	legend pos = south east,
	legend cell align=left,
	ymode=log,
	grid=both]

\addplot [dotted, color=violet, mark=*, mark size=1pt, thick] table[x=t,y=RS] {figs/Dec-Int/data/boundsData_q=2_m=10_r=50_l=2.dat};

\addlegendentry{{$1-\labelRS$}};

\addplot [solid, color=blue, mark=triangle, thick, mark size=1pt] table[x=t,y=Thm1] {figs/Dec-Int/data/boundsData_q=2_m=10_r=50_l=2.dat};

\addlegendentry{{$1-\labelMain$}};

\addplot [solid, color=cyan, mark=triangle, thick, mark size=1pt] table[x=t,y=WoKopt] {figs/Dec-Int/data/boundsData_q=2_m=10_r=50_l=2.dat};

\addlegendentry{{$1-\labelSingleton$}};

\addplot [solid, color=pink, mark=x, thick, mark size=1pt] table[x=t,y=L01] {figs/Dec-Int/data/boundsData_q=2_m=10_r=50_l=2.dat};

\addlegendentry{{$1-\labelLz$}};

\addplot [solid, color=lightgray, mark=none, thick, mark size=1pt] table[x=t,y=LowerIE] {figs/Dec-Int/data/boundsData_q=2_m=10_r=50_l=2.dat};

\addlegendentry{{$1-\labelLower$}};

\addplot [solid, color=green, mark=none, thick, mark size=1pt] table[x=t,y=Miscorrection] {figs/Dec-Int/data/boundsData_q=2_m=10_r=50_l=2.dat};

\addlegendentry{{$\labelMisc$}};

\addplot [solid, color=red, dashed, mark=none, thick, mark size=1pt] table[x=t,y=Sim] {figs/Dec-Int/data/boundsData_q=2_m=10_r=50_l=2.dat};

\addlegendentry{{$\labelSim$}};

\end{axis}
\end{tikzpicture}

%% file: figs/Dec-Int/boundsPlot_q=2_m=11_r=100_l=2.tex
\begin{tikzpicture}
\pgfplotsset{compat = 1.3}
\begin{axis}[
	legend style={nodes={scale=0.5, transform shape}},
	width = 0.97\columnwidth,
	height = 0.8\columnwidth,
	xlabel = {{Number of errors $t$}},
	ylabel = {{Probability $1-\Psuc$ or $\Pmisc$}},
	xmin = 50,
	xmax = 67,
	ymin = 2.996352e-21,
	ymax = 10,
	legend pos = south east,
	legend cell align=left,
	ymode=log,
	grid=both]

\addplot [dotted, color=violet, mark=*, mark size=1pt, thick] table[x=t,y=RS] {figs/Dec-Int/data/boundsData_q=2_m=11_r=100_l=2.dat};

\addlegendentry{{$1-\labelRS$}};

\addplot [solid, color=blue, mark=triangle, thick, mark size=1pt] table[x=t,y=Thm1] {figs/Dec-Int/data/boundsData_q=2_m=11_r=100_l=2.dat};

\addlegendentry{{$1-\labelMain$}};

\addplot [solid, color=cyan, mark=triangle, thick, mark size=1pt] table[x=t,y=WoKopt] {figs/Dec-Int/data/boundsData_q=2_m=11_r=100_l=2.dat};

\addlegendentry{{$1-\labelSingleton$}};

\addplot [solid, color=pink, mark=x, thick, mark size=1pt] table[x=t,y=L01] {figs/Dec-Int/data/boundsData_q=2_m=11_r=100_l=2.dat};

\addlegendentry{{$1-\labelLz$}};

\addplot [solid, color=lightgray, mark=none, thick, mark size=1pt] table[x=t,y=LowerIE] {figs/Dec-Int/data/boundsData_q=2_m=11_r=100_l=2.dat};

\addlegendentry{{$1-\labelLower$}};

\addplot [solid, color=green, mark=none, thick, mark size=1pt] table[x=t,y=Miscorrection] {figs/Dec-Int/data/boundsData_q=2_m=11_r=100_l=2.dat};

\addlegendentry{{$\labelMisc$}};

\addplot [solid, color=red, dashed, mark=none, thick, mark size=1pt] table[x=t,y=Sim] {figs/Dec-Int/data/boundsData_q=2_m=11_r=100_l=2.dat};

\addlegendentry{{$\labelSim$}};

\end{axis}
\end{tikzpicture}

%% file: figs/Dec-Int/boundsPlot_q=2_m=10_r=50_l=5.tex
\begin{tikzpicture}
\pgfplotsset{compat = 1.3}
\begin{axis}[
	legend style={nodes={scale=0.5, transform shape}},
	width = 0.97\columnwidth,
	height = 0.8\columnwidth,
	xlabel = {{Number of errors $t$}},
	ylabel = {{Probability $1-\Psuc$ or $\Pmisc$}},
	xmin = 25,
	xmax = 42,
	ymin = 9.389114e-35,
	ymax = 10,
	legend pos = south east,
	legend cell align=left,
	ymode=log,
	grid=both]

\addplot [dotted, color=violet, mark=*, mark size=1pt, thick] table[x=t,y=RS] {figs/Dec-Int/data/boundsData_q=2_m=10_r=50_l=5.dat};

\addlegendentry{{$1-\labelRS$}};

\addplot [solid, color=blue, mark=triangle, thick, mark size=1pt] table[x=t,y=Thm1] {figs/Dec-Int/data/boundsData_q=2_m=10_r=50_l=5.dat};

\addlegendentry{{$1-\labelMain$}};

\addplot [solid, color=cyan, mark=triangle, thick, mark size=1pt] table[x=t,y=WoKopt] {figs/Dec-Int/data/boundsData_q=2_m=10_r=50_l=5.dat};

\addlegendentry{{$1-\labelSingleton$}};

\addplot [solid, color=pink, mark=x, thick, mark size=1pt] table[x=t,y=L01] {figs/Dec-Int/data/boundsData_q=2_m=10_r=50_l=5.dat};

\addlegendentry{{$1-\labelLz$}};

\addplot [solid, color=lightgray, mark=none, thick, mark size=1pt] table[x=t,y=LowerIE] {figs/Dec-Int/data/boundsData_q=2_m=10_r=50_l=5.dat};

\addlegendentry{{$1-\labelLower$}};

\addplot [solid, color=green, mark=none, thick, mark size=1pt] table[x=t,y=Miscorrection] {figs/Dec-Int/data/boundsData_q=2_m=10_r=50_l=5.dat};

\addlegendentry{{$\labelMisc$}};

\addplot [solid, color=red, dashed, mark=none, thick, mark size=1pt] table[x=t,y=Sim] {figs/Dec-Int/data/boundsData_q=2_m=10_r=50_l=5.dat};

\addlegendentry{{$\labelSim$}};

\end{axis}
\end{tikzpicture}

%% file: figs/Dec-Int/boundsPlot_q=2_m=11_r=100_l=5.tex
\begin{tikzpicture}
\pgfplotsset{compat = 1.3}
\begin{axis}[
	legend style={nodes={scale=0.5, transform shape}},
	width = 0.97\columnwidth,
	height = 0.8\columnwidth,
	xlabel = {{Number of errors $t$}},
	ylabel = {{Probability $1-\Psuc$ or $\Pmisc$}},
	xmin = 50,
	xmax = 84,
	ymin = 9.677669e-72,
	ymax = 10,
	legend pos = south east,
	legend cell align=left,
	ymode=log,
	grid=both]

\addplot [dotted, color=violet, mark=*, mark size=1pt, thick] table[x=t,y=RS] {figs/Dec-Int/data/boundsData_q=2_m=11_r=100_l=5.dat};

\addlegendentry{{$1-\labelRS$}};

\addplot [solid, color=blue, mark=triangle, thick, mark size=1pt] table[x=t,y=Thm1] {figs/Dec-Int/data/boundsData_q=2_m=11_r=100_l=5.dat};

\addlegendentry{{$1-\labelMain$}};

\addplot [solid, color=cyan, mark=triangle, thick, mark size=1pt] table[x=t,y=WoKopt] {figs/Dec-Int/data/boundsData_q=2_m=11_r=100_l=5.dat};

\addlegendentry{{$1-\labelSingleton$}};

\addplot [solid, color=pink, mark=x, thick, mark size=1pt] table[x=t,y=L01] {figs/Dec-Int/data/boundsData_q=2_m=11_r=100_l=5.dat};

\addlegendentry{{$1-\labelLz$}};

\addplot [solid, color=lightgray, mark=none, thick, mark size=1pt] table[x=t,y=LowerIE] {figs/Dec-Int/data/boundsData_q=2_m=11_r=100_l=5.dat};

\addlegendentry{{$1-\labelLower$}};

\addplot [solid, color=green, mark=none, thick, mark size=1pt] table[x=t,y=Miscorrection] {figs/Dec-Int/data/boundsData_q=2_m=11_r=100_l=5.dat};

\addlegendentry{{$\labelMisc$}};

\addplot [solid, color=red, dashed, mark=none, thick, mark size=1pt] table[x=t,y=Sim] {figs/Dec-Int/data/boundsData_q=2_m=11_r=100_l=5.dat};

\addlegendentry{{$\labelSim$}};

\end{axis}
\end{tikzpicture}

%% file: figs/Dec-Int/boundsPlot_q=2_m=10_r=50_l=10.tex
\begin{tikzpicture}
\pgfplotsset{compat = 1.3}
\begin{axis}[
	legend style={nodes={scale=0.5, transform shape}},
	width = 0.97\columnwidth,
	height = 0.8\columnwidth,
	xlabel = {{Number of errors $t$}},
	ylabel = {{Probability $1-\Psuc$ or $\Pmisc$}},
	xmin = 25,
	xmax = 46,
	ymin = 1.564516e-71,
	ymax = 10,
	legend pos = south east,
	legend cell align=left,
	ymode=log,
	grid=both]

\addplot [dotted, color=violet, mark=*, mark size=1pt, thick] table[x=t,y=RS] {figs/Dec-Int/data/boundsData_q=2_m=10_r=50_l=10.dat};

\addlegendentry{{$1-\labelRS$}};

\addplot [solid, color=blue, mark=triangle, thick, mark size=1pt] table[x=t,y=Thm1] {figs/Dec-Int/data/boundsData_q=2_m=10_r=50_l=10.dat};

\addlegendentry{{$1-\labelMain$}};

\addplot [solid, color=cyan, mark=triangle, thick, mark size=1pt] table[x=t,y=WoKopt] {figs/Dec-Int/data/boundsData_q=2_m=10_r=50_l=10.dat};

\addlegendentry{{$1-\labelSingleton$}};

\addplot [solid, color=pink, mark=x, thick, mark size=1pt] table[x=t,y=L01] {figs/Dec-Int/data/boundsData_q=2_m=10_r=50_l=10.dat};

\addlegendentry{{$1-\labelLz$}};

\addplot [solid, color=lightgray, mark=none, thick, mark size=1pt] table[x=t,y=LowerIE] {figs/Dec-Int/data/boundsData_q=2_m=10_r=50_l=10.dat};

\addlegendentry{{$1-\labelLower$}};

\addplot [solid, color=green, mark=none, thick, mark size=1pt] table[x=t,y=Miscorrection] {figs/Dec-Int/data/boundsData_q=2_m=10_r=50_l=10.dat};

\addlegendentry{{$\labelMisc$}};

\addplot [solid, color=red, dashed, mark=none, thick, mark size=1pt] table[x=t,y=Sim] {figs/Dec-Int/data/boundsData_q=2_m=10_r=50_l=10.dat};

\addlegendentry{{$\labelSim$}};

\end{axis}
\end{tikzpicture}

%% file: figs/Dec-Int/boundsPlot_q=2_m=11_r=100_l=10.tex
\begin{tikzpicture}
\pgfplotsset{compat = 1.3}
\begin{axis}[
	legend style={nodes={scale=0.5, transform shape}},
	width = 0.97\columnwidth,
	height = 0.8\columnwidth,
	xlabel = {{Number of errors $t$}},
	ylabel = {{Probability $1-\Psuc$ or $\Pmisc$}},
	xmin = 50,
	xmax = 91,
	ymin = 1.738559e-146,
	ymax = 10,
	legend pos = south east,
	legend cell align=left,
	ymode=log,
	grid=both]

\addplot [dotted, color=violet, mark=*, mark size=1pt, thick] table[x=t,y=RS] {figs/Dec-Int/data/boundsData_q=2_m=11_r=100_l=10.dat};

\addlegendentry{{$1-\labelRS$}};

\addplot [solid, color=blue, mark=triangle, thick, mark size=1pt] table[x=t,y=Thm1] {figs/Dec-Int/data/boundsData_q=2_m=11_r=100_l=10.dat};

\addlegendentry{{$1-\labelMain$}};

\addplot [solid, color=cyan, mark=triangle, thick, mark size=1pt] table[x=t,y=WoKopt] {figs/Dec-Int/data/boundsData_q=2_m=11_r=100_l=10.dat};

\addlegendentry{{$1-\labelSingleton$}};

\addplot [solid, color=pink, mark=x, thick, mark size=1pt] table[x=t,y=L01] {figs/Dec-Int/data/boundsData_q=2_m=11_r=100_l=10.dat};

\addlegendentry{{$1-\labelLz$}};

\addplot [solid, color=lightgray, mark=none, thick, mark size=1pt] table[x=t,y=LowerIE] {figs/Dec-Int/data/boundsData_q=2_m=11_r=100_l=10.dat};

\addlegendentry{{$1-\labelLower$}};

\addplot [solid, color=green, mark=none, thick, mark size=1pt] table[x=t,y=Miscorrection] {figs/Dec-Int/data/boundsData_q=2_m=11_r=100_l=10.dat};

\addlegendentry{{$\labelMisc$}};

\addplot [solid, color=red, dashed, mark=none, thick, mark size=1pt] table[x=t,y=Sim] {figs/Dec-Int/data/boundsData_q=2_m=11_r=100_l=10.dat};

\addlegendentry{{$\labelSim$}};

\end{axis}
\end{tikzpicture}

%% file: figs/Dec-Int/boundsPlot_q=2_m=10_r=50_l=25.tex
\begin{tikzpicture}
\pgfplotsset{compat = 1.3}
\begin{axis}[
	legend style={nodes={scale=0.5, transform shape}},
	width = 0.97\columnwidth,
	height = 0.8\columnwidth,
	xlabel = {{Number of errors $t$}},
	ylabel = {{Probability $1-\Psuc$ or $\Pmisc$}},
	xmin = 25,
	xmax = 49,
	ymin = 6.271928e-180,
	ymax = 10,
	legend pos = south east,
	legend cell align=left,
	ymode=log,
	grid=both]

\addplot [dotted, color=violet, mark=*, mark size=1pt, thick] table[x=t,y=RS] {figs/Dec-Int/data/boundsData_q=2_m=10_r=50_l=25.dat};

\addlegendentry{{$1-\labelRS$}};

\addplot [solid, color=blue, mark=triangle, thick, mark size=1pt] table[x=t,y=Thm1] {figs/Dec-Int/data/boundsData_q=2_m=10_r=50_l=25.dat};

\addlegendentry{{$1-\labelMain$}};

\addplot [solid, color=cyan, mark=triangle, thick, mark size=1pt] table[x=t,y=WoKopt] {figs/Dec-Int/data/boundsData_q=2_m=10_r=50_l=25.dat};

\addlegendentry{{$1-\labelSingleton$}};

\addplot [solid, color=pink, mark=x, thick, mark size=1pt] table[x=t,y=L01] {figs/Dec-Int/data/boundsData_q=2_m=10_r=50_l=25.dat};

\addlegendentry{{$1-\labelLz$}};

\addplot [solid, color=lightgray, mark=none, thick, mark size=1pt] table[x=t,y=LowerIE] {figs/Dec-Int/data/boundsData_q=2_m=10_r=50_l=25.dat};

\addlegendentry{{$1-\labelLower$}};

\addplot [solid, color=green, mark=none, thick, mark size=1pt] table[x=t,y=Miscorrection] {figs/Dec-Int/data/boundsData_q=2_m=10_r=50_l=25.dat};

\addlegendentry{{$\labelMisc$}};

\addplot [solid, color=red, dashed, mark=none, thick, mark size=1pt] table[x=t,y=Sim] {figs/Dec-Int/data/boundsData_q=2_m=10_r=50_l=25.dat};

\addlegendentry{{$\labelSim$}};

\end{axis}
\end{tikzpicture}

%% file: figs/Dec-Int/boundsPlot_q=2_m=11_r=100_l=25.tex
\begin{tikzpicture}
\pgfplotsset{compat = 1.3}
\begin{axis}[
	legend style={nodes={scale=0.5, transform shape}},
	width = 0.97\columnwidth,
	height = 0.8\columnwidth,
	xlabel = {{Number of errors $t$}},
	ylabel = {{Probability $1-\Psuc$ or $\Pmisc$}},
	xmin = 50,
	xmax = 97,
	ymin = 1.096826e-320,
	ymax = 10,
	legend pos = south east,
	legend cell align=left,
	ymode=log,
	grid=both]

\addplot [dotted, color=violet, mark=*, mark size=1pt, thick] table[x=t,y=RS] {figs/Dec-Int/data/boundsData_q=2_m=11_r=100_l=25.dat};

\addlegendentry{{$1-\labelRS$}};

\addplot [solid, color=blue, mark=triangle, thick, mark size=1pt] table[x=t,y=Thm1] {figs/Dec-Int/data/boundsData_q=2_m=11_r=100_l=25.dat};

\addlegendentry{{$1-\labelMain$}};

\addplot [solid, color=cyan, mark=triangle, thick, mark size=1pt] table[x=t,y=WoKopt] {figs/Dec-Int/data/boundsData_q=2_m=11_r=100_l=25.dat};

\addlegendentry{{$1-\labelSingleton$}};

\addplot [solid, color=pink, mark=x, thick, mark size=1pt] table[x=t,y=L01] {figs/Dec-Int/data/boundsData_q=2_m=11_r=100_l=25.dat};

\addlegendentry{{$1-\labelLz$}};

\addplot [solid, color=lightgray, mark=none, thick, mark size=1pt] table[x=t,y=LowerIE] {figs/Dec-Int/data/boundsData_q=2_m=11_r=100_l=25.dat};

\addlegendentry{{$1-\labelLower$}};

\addplot [solid, color=green, mark=none, thick, mark size=1pt] table[x=t,y=Miscorrection] {figs/Dec-Int/data/boundsData_q=2_m=11_r=100_l=25.dat};

\addlegendentry{{$\labelMisc$}};

\addplot [solid, color=red, dashed, mark=none, thick, mark size=1pt] table[x=t,y=Sim] {figs/Dec-Int/data/boundsData_q=2_m=11_r=100_l=25.dat};

\addlegendentry{{$\labelSim$}};

\end{axis}
\end{tikzpicture}

%% file: figs/Dec-Int/boundsPlot_q=2_m=10_r=50_l=50.tex
\begin{tikzpicture}
\pgfplotsset{compat = 1.3}
\begin{axis}[
	legend style={nodes={scale=0.5, transform shape}},
	width = 0.97\columnwidth,
	height = 0.8\columnwidth,
	xlabel = {{Number of errors $t$}},
	ylabel = {{Probability $1-\Psuc$ or $\Pmisc$}},
	xmin = 25,
	xmax = 50,
	ymin = 1.295002e-310,
	ymax = 10,
	legend pos = north west,
	legend cell align=left,
	ymode=log,
	grid=both]

\addplot [dotted, color=violet, mark=*, mark size=1pt, thick] table[x=t,y=RS] {figs/Dec-Int/data/boundsData_q=2_m=10_r=50_l=50.dat};

\addlegendentry{{$1-\labelRS$}};

\addplot [solid, color=blue, mark=triangle, thick, mark size=1pt] table[x=t,y=Thm1] {figs/Dec-Int/data/boundsData_q=2_m=10_r=50_l=50.dat};

\addlegendentry{{$1-\labelMain$}};

\addplot [solid, color=cyan, mark=triangle, thick, mark size=1pt] table[x=t,y=WoKopt] {figs/Dec-Int/data/boundsData_q=2_m=10_r=50_l=50.dat};

\addlegendentry{{$1-\labelSingleton$}};

\addplot [solid, color=pink, mark=x, thick, mark size=1pt] table[x=t,y=L01] {figs/Dec-Int/data/boundsData_q=2_m=10_r=50_l=50.dat};

\addlegendentry{{$1-\labelLz$}};

\addplot [solid, color=orange, mark=triangle, thick, mark size=1pt] table[x=t,y=LargeEll] {figs/Dec-Int/data/boundsData_q=2_m=10_r=50_l=50.dat};

\addlegendentry{{$1-\labelLarge$}};

\addplot [solid, color=lightgray, mark=none, thick, mark size=1pt] table[x=t,y=LowerIE] {figs/Dec-Int/data/boundsData_q=2_m=10_r=50_l=50.dat};

\addlegendentry{{$1-\labelLower$}};

\addplot [solid, color=green, mark=none, thick, mark size=1pt] table[x=t,y=Miscorrection] {figs/Dec-Int/data/boundsData_q=2_m=10_r=50_l=50.dat};

\addlegendentry{{$\labelMisc$}};

\addplot [solid, color=red, dashed, mark=none, thick, mark size=1pt] table[x=t,y=Sim] {figs/Dec-Int/data/boundsData_q=2_m=10_r=50_l=50.dat};

\addlegendentry{{$\labelSim$}};

\end{axis}
\end{tikzpicture}

%% file: figs/Dec-Int/boundsPlot_q=2_m=10_r=50_l=80.tex
\begin{tikzpicture}
\pgfplotsset{compat = 1.3}
\begin{axis}[
	legend style={nodes={scale=0.5, transform shape}},
	width = 0.97\columnwidth,
	height = 0.8\columnwidth,
	xlabel = {{Number of errors $t$}},
	ylabel = {{Probability $1-\Psuc$ or $\Pmisc$}},
	xmin = 25,
	xmax = 50,
	ymin = 5.188752e-304,
	ymax = 10,
	legend pos = north west,
	legend cell align=left,
	ymode=log,
	grid=both]

\addplot [dotted, color=violet, mark=*, mark size=1pt, thick] table[x=t,y=RS] {figs/Dec-Int/data/boundsData_q=2_m=10_r=50_l=80.dat};

\addlegendentry{{$1-\labelRS$}};

\addplot [solid, color=blue, mark=triangle, thick, mark size=1pt] table[x=t,y=Thm1] {figs/Dec-Int/data/boundsData_q=2_m=10_r=50_l=80.dat};

\addlegendentry{{$1-\labelMain$}};

\addplot [solid, color=cyan, mark=triangle, thick, mark size=1pt] table[x=t,y=WoKopt] {figs/Dec-Int/data/boundsData_q=2_m=10_r=50_l=80.dat};

\addlegendentry{{$1-\labelSingleton$}};

\addplot [solid, color=pink, mark=x, thick, mark size=1pt] table[x=t,y=L01] {figs/Dec-Int/data/boundsData_q=2_m=10_r=50_l=80.dat};

\addlegendentry{{$1-\labelLz$}};

\addplot [solid, color=orange, mark=triangle, thick, mark size=1pt] table[x=t,y=LargeEll] {figs/Dec-Int/data/boundsData_q=2_m=10_r=50_l=80.dat};

\addlegendentry{{$1-\labelLarge$}};

\addplot [solid, color=lightgray, mark=none, thick, mark size=1pt] table[x=t,y=LowerIE] {figs/Dec-Int/data/boundsData_q=2_m=10_r=50_l=80.dat};

\addlegendentry{{$1-\labelLower$}};

\addplot [solid, color=green, mark=none, thick, mark size=1pt] table[x=t,y=Miscorrection] {figs/Dec-Int/data/boundsData_q=2_m=10_r=50_l=80.dat};

\addlegendentry{{$\labelMisc$}};

\addplot [solid, color=red, dashed, mark=none, thick, mark size=1pt] table[x=t,y=Sim] {figs/Dec-Int/data/boundsData_q=2_m=10_r=50_l=80.dat};

\addlegendentry{{$\labelSim$}};

\end{axis}
\end{tikzpicture}

%% file: figs/Dec-Int/boundsPlot_q=32_m=2_r=50_l=2.tex
\begin{tikzpicture}
\pgfplotsset{compat = 1.3}
\begin{axis}[
	legend style={nodes={scale=0.5, transform shape}},
	width = 0.97\columnwidth,
	height = 0.8\columnwidth,
	xlabel = {{Number of errors $t$}},
	ylabel = {{Probability $1-\Psuc$ or $\Pmisc$}},
	xmin = 25,
	xmax = 34,
	ymin = 6.283091e-34,
	ymax = 10,
	legend pos = north west,
	legend cell align=left,
	ymode=log,
	grid=both]

\addplot [dotted, color=violet, mark=*, mark size=1pt, thick] table[x=t,y=RS] {figs/Dec-Int/data/boundsData_q=32_m=2_r=50_l=2.dat};

\addlegendentry{{$1-\labelRS$}};

\addplot [solid, color=blue, mark=triangle, thick, mark size=1pt] table[x=t,y=Thm1] {figs/Dec-Int/data/boundsData_q=32_m=2_r=50_l=2.dat};

\addlegendentry{{$1-\labelMain$}};

\addplot [solid, color=cyan, mark=triangle, thick, mark size=1pt] table[x=t,y=WoKopt] {figs/Dec-Int/data/boundsData_q=32_m=2_r=50_l=2.dat};

\addlegendentry{{$1-\labelSingleton$}};

\addplot [solid, color=pink, mark=x, thick, mark size=1pt] table[x=t,y=L01] {figs/Dec-Int/data/boundsData_q=32_m=2_r=50_l=2.dat};

\addlegendentry{{$1-\labelLz$}};

\addplot [solid, color=lightgray, mark=none, thick, mark size=1pt] table[x=t,y=LowerIE] {figs/Dec-Int/data/boundsData_q=32_m=2_r=50_l=2.dat};

\addlegendentry{{$1-\labelLower$}};

\addplot [solid, color=green, mark=none, thick, mark size=1pt] table[x=t,y=Miscorrection] {figs/Dec-Int/data/boundsData_q=32_m=2_r=50_l=2.dat};

\addlegendentry{{$\labelMisc$}};

\addplot [solid, color=red, dashed, mark=none, thick, mark size=1pt] table[x=t,y=Sim] {figs/Dec-Int/data/boundsData_q=32_m=2_r=50_l=2.dat};

\addlegendentry{{$\labelSim$}};

\end{axis}
\end{tikzpicture}

%% file: figs/Dec-Int/boundsPlot_q=32_m=2_r=50_l=10.tex
\begin{tikzpicture}
\pgfplotsset{compat = 1.3}
\begin{axis}[
	legend style={nodes={scale=0.5, transform shape}},
	width = 0.97\columnwidth,
	height = 0.8\columnwidth,
	xlabel = {{Number of errors $t$}},
	ylabel = {{Probability $1-\Psuc$ or $\Pmisc$}},
	xmin = 25,
	xmax = 46,
	ymin = 6.422853e-323,
	ymax = 10,
	legend pos = north west,
	legend cell align=left,
	ymode=log,
	grid=both]

\addplot [dotted, color=violet, mark=*, mark size=1pt, thick] table[x=t,y=RS] {figs/Dec-Int/data/boundsData_q=32_m=2_r=50_l=10.dat};

\addlegendentry{{$1-\labelRS$}};

\addplot [solid, color=blue, mark=triangle, thick, mark size=1pt] table[x=t,y=Thm1] {figs/Dec-Int/data/boundsData_q=32_m=2_r=50_l=10.dat};

\addlegendentry{{$1-\labelMain$}};

\addplot [solid, color=cyan, mark=triangle, thick, mark size=1pt] table[x=t,y=WoKopt] {figs/Dec-Int/data/boundsData_q=32_m=2_r=50_l=10.dat};

\addlegendentry{{$1-\labelSingleton$}};

\addplot [solid, color=pink, mark=x, thick, mark size=1pt] table[x=t,y=L01] {figs/Dec-Int/data/boundsData_q=32_m=2_r=50_l=10.dat};

\addlegendentry{{$1-\labelLz$}};

\addplot [solid, color=lightgray, mark=none, thick, mark size=1pt] table[x=t,y=LowerIE] {figs/Dec-Int/data/boundsData_q=32_m=2_r=50_l=10.dat};

\addlegendentry{{$1-\labelLower$}};

\addplot [solid, color=green, mark=none, thick, mark size=1pt] table[x=t,y=Miscorrection] {figs/Dec-Int/data/boundsData_q=32_m=2_r=50_l=10.dat};

\addlegendentry{{$\labelMisc$}};

\addplot [solid, color=red, dashed, mark=none, thick, mark size=1pt] table[x=t,y=Sim] {figs/Dec-Int/data/boundsData_q=32_m=2_r=50_l=10.dat};

\addlegendentry{{$\labelSim$}};

\end{axis}
\end{tikzpicture}

%% file: chap_dec_other_results.tex
\subsection{Joint Decoding of Generalized Goppa Codes}
\emph{This abstract summarizes the results of the work~\cite{liu2021decoding} published in the proceeding of 2021 IEEE International Symposium on Information Theory (ISIT).}

\emph{Generalized Goppa codes} (GGCs) are an extension of Goppa codes, which are defined by a set of \textit{code locator polynomials} and a \emph{Goppa polynomial}~\cite{shekhunova1981cyclic,bezzateev1997one}.
In~\cite{noskov2020effective}, a code-based cryptosystem 
using binary GGCs with code locator polynomials of degree $1$ and $2$ is proposed.
A special class of binary GGCs which is perfect in the weighted Hamming metric was introduced in~\cite{BezzateevShekhunova2013Perfect} and cyclic GGCs were investigated in \cite{bezzateev2014one,bezzateev2015cyclic}.

  In this work, basic properties, decoding and potential cryptographic applications of binary GGCs are investigated.
  First, we derive a parity-check matrix for GGCs with code locators of any degree (an instance for GGCs with code locator polynomials of degree $2$ was presented in \cite{noskov2020one},\cite{noskov2020effective}).
  We provide a formal proof for the lower bound on the minimum Hamming distance of binary GGCs, which was stated in \cite{noskov2020one},\cite{noskov2020effective},
  and we show that the lower bound for GGCs with even-degree code locator polynomials is improved compared to the general lower bound.
  Then, a quadratic-time decoding algorithm that can decode errors up to half of the minimum Hamming distance is presented.
  We further consider GGCs as the constituent code of interleaved codes. An explicit decoding algorithm based on \cref{algo:SyndromeDecoder} and extended Euclidean algorithm is presented, and new maximum decoding radius for interleaved GGCs are derived.
  Finally, we list some code parameters of GGCs and discuss their applicability to the McEliece cryptosystem. By comparing the public key sizes
  for several code parameters, it can be observed that, under the same security level, the GGCs with degree-$2$ code locator polynomials provides smaller public key size than the binary Goppa codes.


\subsection{List Decoding of $2$-Interleaved Binary Alternant Codes}
\emph{This abstract summarizes the results of the work~\cite{huang2022list} published in the proceeding of 2022 IEEE International Symposium on Information Theory (ISIT).}


Parvaresh \cite{parvaresh2007algebraic} combined list and interleaved decoding by adapting the Guruswami-Sudan algorithm to the decoding of $2$-interleaved GRS codes.
Trivariate polynomials are used to set up the interpolation constraints and \emph{resultants} of polynomials are used to recover the codeword. By combining the approaches of interleaved decoding and the Guruswami-Sudan algorithm, this decoder achieves a larger decoding radius than the Guruswami-Sudan algorithm, however, at the cost of a small probability of failure. 

In this work, a list decoding algorithm for $2$-interleaved binary alternant codes is proposed. The new algorithm combines the approach in \cite{augot2011list} that applies the Koetter-Vardy list decoding algorithm \cite{koetter2003algebraicsoft} to alternant codes, with the Parvaresh's algorithm for interleaved GRS codes \cite{parvaresh2007algebraic}. Similar to Parvareh's algorithm, it is difficult to make a precise statement on the decoding radius of this code. Instead, we present an upper bound on the decoding radius, along with simulation results showing that the decoding radius of the algorithm exceeds the decoding radii of all other algorithms known in literature for the chosen parameters. The drawback of the presented algorithm is that decoding is not guaranteed to succeed (similar to \cite{parvaresh2007algebraic}). However, the simulation results indicate that this probability of failure is small, if the parameters of the algorithm are chosen suitably.



%% file: conclusion.tex
\chapter{Concluding Remarks}

This dissertations concerns new code constructions with properties that are desired in quantum error-correction, distributed storage system and network coding based on non-conventional classes of polynomials, and joint decoding on evaluation codes to decode beyond half the minimum distance.


\textbf{\cref{chap:mod_ring}} is devoted to constructing Euclidean and $\sigma$-Hermitian dual-containing $(\Endom, \Deriv)$-polycyclic codes over finite commutative Frobenius rings
from skew polynomials.
We have developed an algorithm to find all dual-containing codes
by transforming the problem of searching for dual-containing codes into a system of polynomial equations and using Gr\"obner bases to solve it.
By applying this algorithm to several rings of order 4, the results show that there are dual-containing $(\Endom,\Deriv)$-codes that can only be constructed from skew polynomials with non-trivial endomorphisms $\Endom$ (not automorphisms) or nonzero derivations $\Deriv$.
Moreover, we have presented another algorithm with the usage of Gr\"obner basis to test whether the dual code is also a $(\Endom,\Deriv)$-polycyclic code.
Applying this algorithm to the resulting dual-containing codes found by the previous algorithm, we find some of those dual-containing codes whose dual is not a $(\Endom,\Deriv)$-code.

\textbf{\cref{chap:eva_skew}} focuses on the condition of constructing support-constrained codes from evaluation codes based on skew polynomials (i.e., the linearlized Reed-Solomon (LRS) codes) and network coding. 
We have derived a necessary and sufficient condition on the existence of an LRS code fulfilling given support constraints and give an upper bound on the field size to construct such a code.
With the help of the condition, we have proposed a scheme to design \emph{distributed LRS codes}
for multi-source unicast networks via solving an integer linear programming problem.
For a class of multicast networks, the generalized combination networks, we have derived upper and lower bounds on the gap between the minimum required alphabet size of scalar solutions and vector solutions.
The asymptotic behavior of the newly derived upper and lower bounds on the gap show that the number of bits that scalar solutions overpay is increasing sub-linearly with the size of the network. 

\textbf{\cref{chap:eva_multivar}}
contains two new results on locally recoverable codes from evaluation codes based on multivariate polynomials.
The first is on the $[q^2, 1-\Theta\left((q/r)^{-0.2284}\right)]_q$ quadratic lifted Reed-Solomon codes (QLRS), where each codeword symbol has $q^2$ local recovery sets and within each local set $r$ erasures can be corrected locally. We have compared its local recovery performance with the $[q^2, 1-\Theta\left((q/r)^{-0.4150}\right)]_q$ lifted Reed-Solomon codes which has $q$ local recovery sets for each codeword symbol. Simulation results showed that, for a fixed dimension, QLRS codes are more likely to locally recover an erasure at a certain position in the presence of other erasures, if the erasure probability is $\leq 0.7$.
The second result is on the almost affinely disjoint (AAD) subspace family motivated by batch codes. We have given a construction based on Reed-Solomon codes of such family for $k=1,2$. The newly derived upper bound on the asymptotic growth of the cardinality of the family shows the optimality of this construction.

\textbf{\cref{chap:dec_eva}} is dedicated to joint decoding of interleaved evaluation codes.
We have derived a necessary and sufficient condition on decoding success by the Schimidt-Sidorenko-Bossert joint decoding algorithm for interleaved alternant codes and lower and upper bounds on the probability of decoding success based on this condition.
Numerical evaluations show that one of the provided lower bounds is tight for some parameters, as it matches the corresponding newly derived upper bound. 
The short summary on utilizing list-decoding for interleaved codes has shown the potential of this approach to further increase the decoding radius of evaluation codes in the presence of burst errors.

Various future research directions are presented in the outlooks at the end of each chapter.


%% file: appendix.tex
\section{Derivation of $\deg_{\betalt}P_{\bT}$}
\label{apendix:fieldSizeSkewPoly}
This proof is an extension of the analysis on linearized polynomials for Gabidulin codes in \cite[Section II.F]{yildiz2019gabidulin} to skew polynomials.

From \eqref{eq:skewPolyEachRow}, the entry $T_{i,j}$ in $\bT$ is the coefficient of $\SkewVar^{j-1}$ in $\rowpoly_i(\SkewVar)$. Note that we have set $T_{i,k}=1$ in order to have the other entries in $\bT$ uniquely determined given the roots of $f_i$'s.
For $h\in[k-1]$, $T_{i,h}$ is a commutative multivariate polynomial in $R_{\numMulPolyVar}$ (see \eqref{eq:multiVarRing})
and
\begin{align*}
  \deg_{\betalt}T_{i,h}\leq \deg_{\betalt} \rowpoly_i(\SkewVar)\ . 
\end{align*}
For any $l\in[\ell]$ and $t\in[n_l]$, to find 
$\deg_{\betalt}\rowpoly_i(\SkewVar)$, consider the definition of $\rowpoly_i(\SkewVar)$ in \eqref{eq:lclmMinPoly}.
Suppose that $j=\indMap(l,t)\in\zeroSet_i$, otherwise $\deg_{\betalt}\rowpoly_i(\SkewVar)=0$.
Recall that $\alpha_j=a_l\betalt^{q-1}$ and $\rootSet_i$'s, $i\in[k]$ are as defined in \eqref{eq:rootSet}.
Let $\rowpoly_i'\in\FrobSkewPolys$ be the minimal polynomial of $\rootSet_i'\defeq \rootSet_i\setminus{\set*{\alpha_j}}$,
i.e.,
\begin{align*}
  \rowpoly_i'(\SkewVar) = \underset{\alpha\in\rootSet_i'}{\lclm}\{\SkewVar-\alpha\}\ ,
\end{align*}
whose $\SkewVar$-degree is $\deg_{\SkewVar}\rowpoly_i'(\SkewVar)=|\rootSet_i'|=k-2$.
Since $j=\indMap(l,t)\not\in\zeroSet_i'$, the coefficients of $\rowpoly_i'(\SkewVar)$ are independent of $\betalt$, i.e., $\deg_{\betalt}\rowpoly_i'(\SkewVar)=0$.

By the remainder evaluation of skew polynomials in \eqref{eq:Frob-remainder-evaluation},
\begin{align*}
  \rowpoly_i'(a_l\betalt^{q-1}) &= \sum_{h=1}^{k-1} f'_{i,h}N_{h-1}(a_l\betalt^{q-1})\ ,
\end{align*}
and
\begin{align*}
  \deg_{\betalt}\rowpoly_i'(a_l\betalt^{q-1}) &= \deg_{\betalt} N_{k-2}(a_l\betalt^{q-1})\\
                                       &= \deg_{\betalt}(a_l\betalt^{q-1})^{(q^{k-2}-1)/\parenv*{q-1}}\\
                                       &= q^{k-2}-1\ .
\end{align*}
By the Newton interpolation in \eqref{eq:newtonInterpolation}, we can write
\begingroup
\setlength\arraycolsep{3pt}
\allowdisplaybreaks
\begin{align*}
  \rowpoly_i(\SkewVar) &= \left( \SkewVar - \Frobaut( \rowpoly_i'(a_l\betalt^{q-1}))\cdot a_l\betalt^{q-1}\cdot \left(\rowpoly_i'(a_l\betalt^{q-1})\right)^{-1}\right)\cdot \rowpoly_i'(\SkewVar)\\
   &= \left( \SkewVar - \left( \rowpoly_i'(a_l\betalt^{q-1})\right)^{q-1}\cdot a_l\betalt^{q-1}\right)\cdot \rowpoly_i'(\SkewVar)\\
  &= \SkewVar\cdot  \rowpoly_i'(\SkewVar) -  \left( \rowpoly_i'(a_l\betalt^{q-1})\right)^{q-1}\cdot a_l\betalt^{q-1}\cdot \rowpoly_i'(\SkewVar)\ .
\end{align*}
\endgroup
Since $\deg_{\betalt}\rowpoly_i'(\SkewVar)=0$ and so is $\deg_{\betalt}(\SkewVar\cdot \rowpoly_i'(\SkewVar))$, we have
\begin{align*}
  \deg_{\betalt} \rowpoly_i(\SkewVar) &= (q-1)\cdot \deg_{\betalt} \rowpoly_i'(a_l\betalt^{q-1}) + \deg_{\betalt}(a_l\betalt^{q-1})\\
                               &= (q-1)\cdot (q^{k-2}-1)+(q-1)\\
  &=(q-1)\cdot q^{k-2}\ ,
\end{align*}
for all $l,t$ such that $\indMap(l,t)\in\zeroSet_i$.
Hence, $\deg_{\betalt}T_{i,h}\leq \deg_{\betalt}\rowpoly_i(\SkewVar)=(q-1)q^{k-2}, \forall h\in[k-1]$. Then,
\begin{align}
  \deg_{\betalt}P_{\bT} &= \deg_{\betalt} \det \bT \nonumber\\
                        &\leq \max_{\pi\in\xi_k}\sum_{h=1}^{k} \deg_{\betalt}T_{\pi(h),h} \nonumber \\
                        &\leq(k-1)(q-1)\cdot q^{k-2}\ ,\label{eq:degbetapt}
\end{align}
where $\xi_k$ denotes the set of permutations of $[k]$ and the $(k-1)$ in \eqref{eq:degbetapt} is because $T_{ik}=1$ and hence $\deg_{\betalt}T_{ik}=0$.
\section{Proofs of Properties of Skew Polynomials over $R_n$}
\label{appendix:SkewPolyProperty}

\ref{p:noZeroDiv}:
  $\FrobPolysn$ is a ring without zero divisors.
\begin{proof}
The ring properties of $\FrobPolysn$ are trivial, we only need to show that it has no zero divisors.

Note that for any $a,b\in R_\numMulPolyVar$, $\Frobaut(a+b)=\Frobaut(a) + \Frobaut(b)$.
It can be seen from \eqref{eq:prodRn} that if $f,g\neq 0$, then $f\cdot g\neq 0$ since the leading coefficients of $f_{d_f}, g_{d_g}$ are nonzero and therefore $f_{d_f}\Frobaut^{d_f}(g_{d_g})$ is nonzero. Hence, $\FrobPolysn$ does not have zero divisors.
\end{proof}

\ref{p:gcrd}:
  For any sets 
  $\rootSet_1, \rootSet_2\subseteq \mulVarRng$ s.t.~$\rootSet_1\cup\rootSet_2$ is P-independent,
  $\gcrd(f_{\rootSet_1},f_{\rootSet_2})=f_{\rootSet_1\cap\rootSet_2}$.
  In particular, $\rootSet_1\cap\rootSet_2=\varnothing \iff \gcrd(f_{\rootSet_1},f_{\rootSet_2})=1$.
\begin{proof}
  This property has been proven as a part of \cite[Theorem 7]{liu2017matroidal} (cf..~\cref{thm:prop-minimal-skew-poly}) for $\FrobSkewPolys$. For completeness, we also include our proof here.
  We can write the minimal polynomial of the set $\rootSet_1\cap\rootSet_2$ by the least common left multiplier as in \eqref{eq:lclm}, i.e., $f_{\rootSet_1\cap\rootSet_2} = \underset{\alpha\in\rootSet_1\cap\rootSet_2}{\lclm}\{\SkewVar-\alpha\}$. Then we can write
  $f_{\rootSet_1} = g_1\cdot \underset{\alpha\in\rootSet_1\cap\rootSet_2}{\lclm}\{\SkewVar-\alpha\}$
  and
  $f_{\rootSet_2} = g_2 \cdot\underset{\alpha\in\rootSet_1\cap\rootSet_2}{\lclm}\{\SkewVar-\alpha\},$ for some $g_1,g_2\in\FrobPolysn$.
  Therefore, it is clear that $f_{\rootSet_1\cap\rootSet_2}\mid \gcrd(f_{\rootSet_1},f_{\rootSet_2})$.

  Now we only need to show that $\deg f_{\rootSet_1\cap\rootSet_2}=\deg \gcrd(f_{\rootSet_1},f_{\rootSet_2})$.
  Since $\rootSet_1\cup\rootSet_2$ is P-independent, $\rootSet_1,\rootSet_2$ and $\rootSet_1\cap\rootSet_2$ are also P-independent. Then $\deg f_{\rootSet_1\cup\rootSet_2} = |\rootSet_1\cup\rootSet_2|$, $\deg f_{\rootSet_1}=|\rootSet_1|$, $\deg_{\rootSet_2}=|\rootSet_2|$ and $\deg f_{\rootSet_1\cap\rootSet_2} = |\rootSet_1\cap\rootSet_2|$.
  It follows from \cite[Proposition 5.12]{gluesing2021introduction} that the minimal polynomial of $\rootSet_1\cup\rootSet_2$ is
  \begin{align*}
    f_{\rootSet_1\cup\rootSet_2} = \lclm(f_{\rootSet_1},f_{\rootSet_2})
  \end{align*}
  and from \cite[Eq.(24)]{ore1933theory} that
  \begin{align*}
    \deg\gcrd(f_{\rootSet_1},f_{\rootSet_2})=&\deg f_1+\deg f_2-\deg \lclm(f_{\rootSet_1},f_{\rootSet_2}) \\
    =& |\rootSet_1|+|\rootSet_2|-|\rootSet_1\cup\rootSet_2|\\
    =& |\rootSet_1\cap\rootSet_2|= \deg f_{\rootSet_1\cap\rootSet_2}\ .
  \end{align*}
  Together with $f_{\rootSet_1\cap\rootSet_2}\mid \gcrd(f_{\rootSet_1},f_{\rootSet_2})$, the property is proven.
\end{proof}

\ref{p:xlrdivision}:
For $t\in\bbN$ and any $f\in\FrobPolysn$, $\SkewVar^t |_l f \iff \SkewVar^t|_r f$.
In this case, we write $\SkewVar^t| f$.
\begin{proof}
  If $\SkewVar^t|_l f$, then with some $g\in\FrobPolysn$ we can write $f= \SkewVar^t\cdot g = \FrobautPolyt{g}{t}\cdot \SkewVar^t$, where $\FrobautPolyt{g}{t}=\sum_i \Frobaut^t(g_i)\SkewVar^i$. Then it is obvious that $X^t|_rf$. Similarly, if $\SkewVar^t |_r f$, we can write $f = g\cdot \SkewVar^t = \SkewVar^t\cdot \FrobautPolyt{g}{-t}$ and it is obvious that $X^t|_l f$. This property has been also shown in \cite[Theorem 7]{ore1933theory}.
\end{proof}

\ref{p:xdivision}:
  For $t\in\bbN$ and any $f_1, f_2 \in\FrobPolysn$ such that $\SkewVar\not\divides f_2$, then $\SkewVar^t| (f_1\cdot f_2)\iff\SkewVar^t\divides f_1$.
\begin{proof}
  We first show $\SkewVar^t\divides (f_1\cdot f_2)\Longleftarrow \SkewVar^t\divides f_1$. Suppose $\SkewVar^t\divides f_1$, then we can write $f_1 = \SkewVar^t\cdot f_1'$ with some $f_1'\in\FrobPolysn$. Then $f_1\cdot f_2= \SkewVar^t\cdot f_1'\cdot f_2$ and it can be seen that $\SkewVar^t\divides_l (f_1\cdot f_2)$. By \ref{p:xlrdivision}, we have $\SkewVar^t\divides(f_1\cdot f_2)$.

  For the other direction, we first show that $\SkewVar\divides (f_1\cdot f_2)\implies\SkewVar\divides f_1$ by contradiction.
  Assume $\SkewVar\not\divides f_1$, then we can write $f_1= f_1'+a$ with some $f_1'\in\FrobPolysn$ such that $\SkewVar\divides f_1'$ and $a\in R_\numMulPolyVar \setminus\{0\}$.
  Since $\SkewVar\not\divides f_2$, we can write $f_2 =f_2'+b$, with some $f_2'\in\FrobPolysn$ such that $\SkewVar\divides f_2'$ and $b\in R_\numMulPolyVar \setminus\{0\}$. Then,
  \begin{align*}
    f_1\cdot f_2 &= (f_1'+a)(f_2'+b)\\
                 &= f_1'\cdot f_2' +a\cdot f_2' + f_1'\cdot b+a\cdot b
  \end{align*}
  where the first three summands are all divisible by $\SkewVar$ but $a\cdot b\neq 0$ (since $R_\numMulPolyVar $ is a ring without zero divisor) and $\SkewVar\not\divides a\cdot b$. This implies $\SkewVar\not\divides (f_1\cdot f_2)$, which is a contradiction.
  Note that $\SkewVar^2\divides(f_1\cdot f_2) \implies \SkewVar \divides (f_1\cdot f_2) \implies \SkewVar\divides f_1$. Write $f_1=\SkewVar\cdot g$ with some $g\in\FrobPolysn$, then
  \begin{align*}
    \SkewVar^2\divides(f_1\cdot f_2) &\implies \SkewVar \divides (g\cdot f_2) \underset{\SkewVar\not\divides f_2}{\implies} \SkewVar\divides g \\
    &\implies (\SkewVar\cdot \SkewVar)\divides(\SkewVar\cdot g)\implies \SkewVar^2\divides f_1\ .
  \end{align*}
  We can extend steps above $t$ times and the property is proven.
\end{proof}
For any $\zeroSet_i\subseteq[n],i\in[k], l\in[\ell]$, we
denote $\zeroSet_i^{(l)}\coloneq\{t\ |\ \indMap(l,t) \in \zeroSet_i\}$ and $\rootSet_i^{(l)}=\{ a_l \betalt^{q-1} \ |\ t \in \zeroSet_i^{(l)}\}$,
where $\indMap(l,t)$ is defined in \eqref{eq:indMap}.
We need the following results on the set of roots of skew polynomials in order to prove \ref{p:gcrd2polys}.
It follows from \cref{lem:noMoreZero} that $f_{\rootSet_i}$ only vanishes on $\rootSet_i$ while evaluating on $\locSet$. The following lemma gives the structure of the roots of $f_{\rootSet_i}$ while evaluating on $\mulVarRng$.
\begin{lemma}[{\cite[Theorem 4]{liu2015construction}}]
  \label{lem:structureRootsBlock}
  For $l=1,\dots, \ell$, let $f_i^{(l)}$ be the minimal polynomial of $\rootSet_i^{(l)}$ and
  $\overline{\rootSet_i^{(l)}}\coloneq\{\alpha\in\mulVarRng ~|~ f_i^{(l)}(\alpha)=0\}$.
  Then, for all $l=1,\dots, \ell$,
  \begin{align}
      \overline{\rootSet_i^{(l)}} &= \{a_l\beta^{q-1}\ |\ \beta\in \myspan{\betalt}_{t\in\zeroSet_i^{(l)}}\setminus \{0\}\}\subseteq C_\Frobaut(a_l)\label{eq:rootSetBarl}\\
      |\overline{\rootSet_i^{(l)}}|&=q^{|\zeroSet_i^{(l)}|}-1 \label{eq:rootSetBarlSize}
  \end{align}
  where $C_\Frobaut(a_l)$ is the $\Frobaut$-conjugacy class of $a_l$ as defined in \eqref{eq:Frob-conj-classes}.
  \end{lemma}
  \begin{theorem}
  \label{thm:rootsStructure}
  Let $\rowpoly_{i}$ be the minimal polynomial of $\rootSet_i$. Denote the set of roots of $f_i$ while evaluating on $\mulVarRng$ by $\overline{\rootSet_i}\coloneq\{\alpha\in\mulVarRng ~|~ f_i(\alpha)=0\}$. Then
\begin{align}
\label{eq:rootsZi}
  \overline{\rootSet_i} &=  \bigcup_{l=1}^{\ell} \overline{\rootSet_i^{(l)}}, \text{ where }\overline{\rootSet_i^{(l)}} \text{ is as in }\eqref{eq:rootSetBarl}\\
  |\overline{\rootSet_i}| &= \sum_{l=1}^{\ell}|\overline{\rootSet_i^{(l)}}| = \sum_{l=1}^{\ell}q^{|\zeroSet_i^{(l)}|}-\ell\ .
\end{align}
\end{theorem}
\begin{proof}
Note that for all $l\in[\ell]$, $\rootSet_i^{(l)}$ are P-independent and they are from different conjugacy classes. It follows from \cite[Corollary 4.4]{LamLer2004} that for such sets,
$\overline{\bigcup_{l=1}^{\ell} \rootSet_i^{(l)}}=\bigcup_{l=1}^{\ell} \overline{\rootSet_i^{(l)}}$.
\end{proof}
It is clear that the $\alpha$'s in \eqref{eq:fZt} are P-independent. It follows from \cref{def:PindSet} that $\deg \fZt(\zeroSet,\tfzt) = |\zeroSet|+\tfzt$. By \cref{thm:rootsStructure}, the set of roots of $\fZt(\zeroSet,\tfzt)$ is
\begin{align}
\label{eq:fZtRoots}
\{0\}^{\tfzt}\cup\bigcup_{l=1}^{\ell}\set*{a_l\beta^{q-1}\ |\ \beta\in \myspan{\betalt}_{t\in\zeroSet^{(l)}}\setminus\{0\}}
\end{align}
where $\zeroSet^{(l)}=\{t\ |\ \indMap(l,t)\in\zeroSet\}$. The notation $\{0\}^{\tfzt}$ is to imply that $\SkewVar^\tfzt\divides \fZt(\zeroSet, \tfzt)$ and $X^{\tau+1}\nmid \fZt(\zeroSet,\tfzt)$.

\ref{p:gcrd2polys}:
  For any $f_1 = \fZt(\zeroSet_1,\tfzt_1), f_2 = \fZt(\zeroSet_2, \tfzt_2)\in\Skewnk$, we have
  \begin{align*}
    \gcrd(f_1,f_2)=\fZt(\zeroSet_1\cap \zeroSet_2, \min\{\tfzt_1,\tfzt_2\})\in\Skewnk\ .
  \end{align*}
\begin{proof}
  We prove the property by showing that the skew polynomials on both side have the same set of roots.
  Denote by $\widebar{\rootSet_1}, \widebar{\rootSet_2}, \widebar{\rootSet_{1,2}}\subseteq \mulVarRng$ the set of all roots in $\mulVarRng$ of $f_1,f_2, \fZt(\zeroSet_1\cap \zeroSet_2, \min\{\tfzt_1,\tfzt_2\})$, respectively.
  By the structure of roots of $\fZt(\zeroSet,t)$ given in \eqref{eq:fZtRoots},
  \begin{align*}
    \widebar{\rootSet_i} &=\{0\}^{\tfzt_i}\cup\bigcup_{l=1}^{\ell}\set*{a_l\beta^{q-1}\ |\ \beta\in \myspan{\betalt}_{t\in\zeroSet_i^{(l)}}\setminus\{0\}},\quad i=1,2\\
    \widebar{\rootSet_{1,2}} &= \{0\}^{\min\{\tfzt_1,\tfzt_2\}}\cup\bigcup_{l=1}^{\ell}\set*{a_l\beta^{q-1}\ |\ \beta\in \myspan{\betalt}_{t\in\zeroSet_{1,2}^{(l)}}\setminus\{0\}}
  \end{align*}
  where $\zeroSet_i^{(l)}\coloneq\{t\ |\ \indMap(l,t)\in\zeroSet_i\}$ and $\zeroSet_{1,2}^{(l)}\coloneq\set*{t\ |\ \indMap(l,t)\in\zeroSet_{1}\cap\zeroSet_2}$.
  The set of roots of $\gcrd(f_1,f_2)$ is
  \begin{align*}
    \widebar{\rootSet_1}\cap \widebar{\rootSet_2} &= \{0\}^{\min\{\tfzt_1,\tfzt_2\}} \cup\bigcup_{l=1}^{\ell}\set*{a_l\beta^{q-1}\ |\ \beta\in \myspan{\betalt}_{t\in\zeroSet_1^{(l)}\cap\zeroSet_2^{(l)}}\setminus\{0\}}\ .
  \end{align*}
  It can be seen that $\zeroSet_1^{(l)}\cap\zeroSet_2^{(l)} = \zeroSet_{1,2}^{(l)}, \forall l\in[\ell]$. Hence, $\widebar{\rootSet_1}\cap \widebar{\rootSet_2}=\widebar{\rootSet_{1,2}}$.
\end{proof}

\ref{p:reducevar}:
  Let $f=\fZt(\zeroSet,\tfzt)\in\Skewnk$ and let $f'=f|_{\beta_{\ell, n_\ell}=0}\in R_{\numMulPolyVar-1}[\SkewVar;\Frobaut]$ (namely, we substitute $\beta_{\ell, n_\ell}=0$ in each coefficient of $f$). Then $f'\in\cS_{n-1,k}$ and
\begin{align*}
    f'= \begin{cases}
    \fZt(\zeroSet,\tfzt) & n\not\in \zeroSet\\
    \fZt(\zeroSet\setminus\{n\}, \tfzt+1) & n\in \zeroSet
    \end{cases}\ .
\end{align*}
\begin{proof}
  Denote by $\rootSet$ the subset of $\locSet$ corresponding to $\zeroSet$ as in \eqref{eq:rootSet}.
  It is trivial that $f'\in \cS_{n-1,k}$ and $f'=\fZt(\zeroSet,\tfzt)$ when $n\not\in \zeroSet$. Suppose $n\in \zeroSet$, then $a_\ell\beta_{\ell,n_\ell}^{q-1}\in\rootSet$.
  Let $g = \underset{\alpha\in\rootSet\setminus\{ a_\ell\beta_{\ell,n_\ell}\}}{\lclm}\{\SkewVar-\alpha\}$, then
  \begin{align*}
      f'&=\SkewVar^\tfzt \cdot \left.\parenv*{\underset{\alpha\in\rootSet}{\lclm}\{\SkewVar-\alpha\}} \right|_{\beta_{\ell,n_\ell}=0}\\
      &=\SkewVar^\tfzt \cdot\left.\left(\left(\SkewVar-(a_\ell\beta_{\ell,n_\ell}^{q-1})^{g(a_\ell\beta_{\ell,n_\ell}^{q-1})}\right)\cdot g \right)\right|_{\beta_{\ell,n_\ell}=0}\\
      &=\SkewVar^\tfzt \cdot \SkewVar\cdot g \\
      &=\SkewVar^{\tfzt+1} \cdot g \\
      &=\SkewVar^{\tfzt+1} \cdot\parenv*{\underset{\alpha\in\rootSet\setminus\{a_\ell\beta_{\ell,n_\ell}^{q-1} \}}{\lclm}\{\SkewVar-\alpha\}} \\
      &=\fZt(\zeroSet\setminus\{n\}, \tfzt+1)\in\cS_{n-1,k}\ ,
  \end{align*}
  where the second line holds by the Newton interpolation in \eqref{eq:newtonInterpolation}.
\end{proof}
\section{Induction Proof of \cref{thm:sufficientCond}}
\label{appendix:induction-proof-general-result}
In the part we elaborate the induction proof of \cref{thm:sufficientCond} for all the cases on page \pageref{case_s3n2_a}.\\
\ref{case_s3n2} For $\numRows\geq 3$ and $n\geq 2$,

\ref{case_s3n2_a} $\forall i\in[\numRows]$, $\tfzt_i\geq 1$ (i.e., $|Z_i|\leq k-2$).
\begin{proof}
For convenience we denote $k'=k-1$. For all $i\in[\numRows]$, we can write $f_i=\SkewVar\cdot f_i'$, where $f_i'=\fZt(Z_i,\tfzt_i-1)\in \cS_{n,k-1}=\cS_{n,k'}$. Note that since $\min_{i\in[\numRows]} \tfzt_i \geq 1$, we have $\deg f_\Omega\geq 1$ for any $\Omega\subseteq [\numRows]$. For $\Omega = [\numRows]$, \ref{item:equiv2} implies $k-1\geq k-\deg f_{[\numRows]}\geq \sum_{i\in[\numRows]}(k-\deg f_i)\geq s$.

\stepone \ref{item:equiv2} holds for $(f_1',\dots,f_{\numRows}')$ because for any nonempty $\Omega\subseteq [\numRows]$,
\begin{align}
  k'-\deg f_{\Omega}' &= k-\deg f_{\Omega}\nonumber\\
  &{\geq} \sum_{i\in\Omega}(k-\deg f_i)\label{eq:case2astep1}\\
  &=\sum_{i\in\Omega}(k'-\deg f_i')\nonumber
\end{align}
where \eqref{eq:case2astep1} holds because \ref{item:equiv2} holds for $(f_1,\dots, f_\numRows)$ by \ref{H:fromHere}.
By \ref{H:fromPre}, \ref{item:equiv1} then holds for $(f_1',\dots, f_\numRows')\in \cS_{n,k'}$. Note here that we used the induction hypothesis by reducing $k$ to $k'$.

\steptwo We then show that \ref{item:equiv1} also holds for $(f_1,\dots, f_\numRows)$. Suppose that for $g_1,\dots, g_\numRows\in \FrobPolysn$ with $\deg (g_i\cdot f_i)\leq k-1$, we have $\sum_{i=1}^s g_i\cdot f_i=0=\sum_{i=1}^s g_i\cdot (\SkewVar\cdot f_i')\overset{\text{\ref{p:noZeroDiv},\ref{p:xlrdivision}}}{\implies}\sum_{i=1}^{\numRows} g_i f'_i=0$, which implies that $g_1=\cdots=g_\numRows=0$ since \ref{item:equiv1} holds for $(f_1',\dots, f_\numRows')\in\cS^s_{n,k'}$.
\end{proof}

\ref{case_s3n2_b} $\exists$ a unique $i\in[\numRows]$ such that $\tfzt_i=0$.
\begin{proof}
Suppose w.l.o.g.~$\tfzt_\numRows=0$ and write $f_\numRows'=f_\numRows\in\cS_{n,k}$.
For $i\in[\numRows-1], \tfzt_i\geq 1$, then we can write $f_i=\SkewVar\cdot f_i'$, where $f_i'=\fZt(Z_i, \tfzt_i-1)\in\cS_{n,k-1}$. Note that $f_\numRows'=f_\numRows\in\cS_{n,k-1}$ if and only if $\deg f_\numRows\leq k-2$, in which case for $\Omega=[s]$, \ref{H:fromHere} implies
\begin{align*}
  k\geq k-\deg f_{\Omega}\geq \sum_{i\in\Omega} (k-\deg f_i) \geq s+1.
\end{align*}

\stepone We show that \ref{item:equiv2} holds for $(f_1',\dots, f_\numRows')$ when $k$ is replaced by $k'=k-1$. First consider the case of $\Omega\subseteq[s-1]$. Since $\forall i\in [s-1],  \tfzt_i\geq 1$, the claim follows similarly to \ref{case_s3n2_a}. Additionally, by the induction hypothesis for $(k'=k-1,s-1,n)$ we get that \ref{item:equiv1} is true for $(f_1',\dots, f_{\numRows-1}')$.
Then consider the case of $\Omega$ such that $s\in\Omega$. Since $f_\numRows=\lclm_{\alpha\in\{a_l\betalt^{q-1} | \indMap(l,t)\in \zeroSet_\numRows\}}\{(\SkewVar-\alpha)\}$ has no factor $\SkewVar$, we have $\gcrd\{f_\numRows, f_i\}=\gcrd\{f_\numRows', f_i'\}, \forall i\in[s-1]$, hence $f_\Omega=f_\Omega'$ where we define $f'_\Omega=\gcrd_{i\in\Omega}\{f'_i\}$. Then
\begingroup
\setlength\arraycolsep{3pt}
\allowdisplaybreaks
\begin{align}
  k-1-\deg f_{\Omega}'&=-1+k-\deg f_{\Omega}\nonumber\\
&\geq -1+\sum_{i\in\Omega}\deg (k-\deg f_i)\label{eq:case2bstep1}\\
&= k-1-\deg f_\numRows+\sum_{i\in\Omega\setminus\{s\}}(k-\deg f_i)\nonumber\\
&= k-1-\deg f_\numRows'+\sum_{i\in\Omega\setminus\{s\}}(k-1-\deg f_i')\nonumber\\
&=\sum_{i\in\Omega}(k-1-\deg f_i')\ ,\nonumber
\end{align}
\endgroup
where \eqref{eq:case2bstep1} holds from \ref{H:fromHere}.
By \ref{H:fromPre}, \ref{item:equiv2} $\implies$ \ref{item:equiv1} is true for $(f_1',\dots, f_\numRows')$ with parameters $(k'=k-1,s,n)$ if $\deg f_\numRows'\leq k-2$., which implies $k\geq s+1$.

\steptwo Suppose that for some $g_1,\dots, g_s\in\FrobPolysn$ with $\deg(g_i\cdot f_i)\leq k-1$ we have $\sum_{i=1}^s g_i\cdot f_i=0$.
Then $0 = \sum_{i=1}^s g_i\cdot f_i = g_s\cdot f_s + \sum_{i=1}^{s-1}g_i\cdot (\SkewVar \cdot f_i' )$, which implies $\SkewVar\divides (g_s\cdot f_s)$. However, since $\SkewVar\not\divides f_s$, by \ref{p:xdivision}, $\SkewVar\divides g_s$. Then we can write $g_s=g_s'\cdot \SkewVar$ for some $g_s'\in\FrobPolysn$ with $\deg g_s'=\deg g_s-1$.

If $\deg f_s=k-1$, then $\deg g_s'=-1$, implying $g_s=0$. Since \ref{item:equiv1} holds for $(f_1',\dots,f_{s-1}')\in \cS_{n,k-1}^{s-1}$ with the parameter tuple $(k-1,s-1,n)$, $g_1,\dots, g_{s-1}$ are also zero.
Note that here we used the induction hypothesis by reducing $k$ to $k-1$ and $s$ to $s-1$.

If $\deg f_s \leq k-2$, we have
\begingroup
\allowdisplaybreaks
\begin{align*}
  0 &=    \sum_{i=1}^s g_i\cdot f_i\\
  &=  (g_s'\cdot \SkewVar)\cdot f_s+\sum_{i=1}^{s-1}g_i\cdot (\SkewVar\cdot f_i')\\
  &= (g_s'\cdot \SkewVar)\cdot f_s' +\sum_{i=1}^{s-1} (g_i\cdot \SkewVar)\cdot f_i'\ .
\end{align*}
\endgroup
Then $g_1 = \cdots = g_{s-1}=g_s' =0 $ since \ref{item:equiv1} holds for $(f_1',\dots, f_s')\in \cS_{n,k-1}^{s}$ with the parameter tuple $(k-1,s,n)$.
Hence, all $g_1=\cdots = g_s = 0$.
Note that here we used the induction hypothesis by reducing $k$ to $k-1$.
\end{proof}

\ref{case_s3n2_c} $\exists\ \Omega \subset [\numRows]$ with $2\leq |\Omega|\leq \numRows-1$ such that \eqref{eq:equiv2} holds with equality.
\begin{proof}
W.l.o.g., assume that \eqref{eq:equiv2} holds with equality for $\Omega'=\{1,\dots,\subRows\}$, $1< \subRows<\numRows$, i.e.,
\begin{align}
  \label{eq:case1-assumption}
  k-\deg f_0 = \sum_{i\in\Omega'}(k-\deg f_i)\ ,
\end{align}
where $f_0=f_{\Omega'}=\gcrd_{i\in\Omega'}f_i$. Since $f_0\divides_r f_i, \forall i\in\Omega'$, there exists $f_i'\in\FrobPolysn$ such that $f_i=f_i'\cdot f_0$.
Since $\subRows<\numRows$ and $\numRows-\subRows +1 <\numRows$, we split $(f_1,\dots, f_\numRows)\in \Skewnk^{\numRows}$ into two smaller problems $(f_1,\dots,f_\subRows)\in\Skewnk^{\subRows}$ with the parameter tuple $(k,\subRows<\numRows,n)$ and $(f_0, f_{\subRows+1},\dots, f_{\numRows})\in\Skewnk^{\numRows-\subRows+1}$ with the tuple $(k,\numRows-\subRows+1<\numRows, n)$.

\stepone Note that by \ref{H:fromHere}, \ref{item:equiv2} is true for $(f_1,\dots,f_{\subRows})$ and for $(f_0, f_{\subRows+1},\dots, f_{\numRows})$ when $0\not\in\Omega''\subseteq\{0,\subRows+1,\dots,\numRows\}$.
We show in the following that \ref{item:equiv2} is also true for $(f_0, f_{\subRows+1},\dots, f_{\numRows})$ with $0\in\Omega''$:
\begingroup
\allowdisplaybreaks
\begin{align}
  k-\deg f_{\Omega''} \nonumber
  &= k- \deg \gcrd\{f_0, f_{\Omega''\setminus\{0\}}\}\nonumber\\
  &= k- \deg \gcrd\{f_{\Omega'}, f_{\Omega''\setminus\{0\}}\}\nonumber\\
  &= k- \deg \gcrd_{i\in\Omega'\cup \Omega''\setminus\{0\}} f_i\nonumber\\
  &{\geq}  \sum_{i\in\Omega'\cup \Omega''\setminus\{0\}} (k-\deg f_i)\label{eq:case1step1}\\
  &= \sum_{i\in\Omega'}(k-\deg f_i) + \sum_{i\in\Omega''\setminus\{0\}}(k-\deg f_i)\nonumber\\
   & = k- \deg f_0 + \sum_{i\in\Omega''\setminus\{0\}}(k-\deg f_i)\label{eq:case1-equal}\\
  &= \sum_{i\in\Omega''}(k-\deg f_i)\nonumber
\end{align}
\endgroup
Note that $\Omega'\cup\Omega''\setminus\{0\}$ is a subset of $\Omega$. Therefore, the inequality in \eqref{eq:case1step1} follows from \ref{H:fromHere}. The equality \eqref{eq:case1-equal} follows from \eqref{eq:case1-assumption}. 
Now we can conclude that \ref{item:equiv2} is true for $(f_1,\dots,f_{\subRows})$ and for $(f_0, f_{\subRows+1},\dots, f_{\numRows})$.

By \ref{H:fromPre}, \ref{item:equiv1} is true for both smaller problems $(f_1,\dots,f_\subRows)\in\Skewnk^{\subRows}$ and $(f_0, f_{\subRows+1},\dots, f_{\numRows})\in\Skewnk^{\numRows-\subRows+1}$.

\steptwo Then we show \ref{item:equiv1} is also true for $(f_1,\dots, f_{\numRows})$. Suppose that for some $g_1,\dots, g_{\numRows}\in\FrobPolysn$ with $\deg g_i\cdot f_i\leq k-1, \forall i\in [\numRows]$, we have
\begin{align}
  \sum_{i=1}^{\numRows} g_i\cdot f_i=0\ . \label{eq:LHSequiv1}
\end{align}
Since $f_0\divides_r f_i$ for all $i\in \Omega'=[\subRows]$, $f_0$ is a right factor $\sum_{i=1}^{\subRows} g_i\cdot f_i$ and we can then write $\sum_{i=1}^{\subRows}g_i\cdot f_i=g_0\cdot f_0$, for some $g_0\in\FrobPolysn$. Then
\begingroup
\allowdisplaybreaks
\begin{align}
  0 &= \sum_{i=1}^{\numRows} g_i\cdot f_i\nonumber \\
  &= \sum_{i=1}^{\subRows} g_i\cdot f_i + \sum_{i=\subRows+1}^{\numRows} g_i\cdot f_i\nonumber \\
  &= g_0\cdot f_0 + \sum_{i=\subRows+1}^{\numRows} g_i\cdot f_i\label{eq:0equal}
\end{align}
\endgroup
From the conclusion that \ref{item:equiv1} is true for $(f_0,f_{\subRows+1}, \dots, f_{\numRows})$, \eqref{eq:0equal} holds only if $g_0=g_{\subRows+1}=\dots = g_{\numRows}=0$.
Similarly, since \ref{item:equiv1} is true for $(f_1,\dots, f_{\subRows})$, $0=g_0\cdot f_0=\sum_{i=1}^{\subRows}g_i\cdot f_i$ only if $g_1 = \dots = g_{\subRows}=0$.
Therefore, \eqref{eq:LHSequiv1} holds only if $g_1 = \dots = g_{\numRows}=0$ and \ref{item:equiv1} is proven for $(f_1,\dots, f_{\numRows})\in\cS_{n,k}^{\numRows}$ with the parameter tuple $(k,s,n)$.
\end{proof}

\ref{case_s3n2_d} $\forall\ \Omega\subset [\numRows]$ with $2\leq |\Omega|\leq \numRows-1$, \eqref{eq:equiv2} holds strictly and $\exists$ at least two $i\in[\numRows]$ such that $\tfzt_i=0$.
\begin{proof}
Assume w.l.o.g.~that $\tfzt_{\numRows-1}=\tfzt_{\numRows}=0$. Then for $i=\numRows-1, \numRows$, $\deg f_{i}=|\zeroSet_{i}|$.
If $\zeroSet_{\numRows-1}=\zeroSet_{\numRows}$, then for $\Omega = \{\numRows-1,\numRows\}$, \ref{item:equiv2} implies
\begingroup
\allowdisplaybreaks
\begin{align*}
  k-\deg f_\numRows&=k-\deg f_{\numRows-1}\\
  &= k-\deg \gcrd \{f_{\numRows-1}, f_{\numRows}\}\\
  &\geq k-\deg f_{\numRows-1}+k-\deg f_{\numRows}
\end{align*}
\endgroup
which contradicts with $\deg f_{i}\leq k-1$ for any $i\in[\numRows]$. Hence,  $\zeroSet_{\numRows-1}\neq[n]$ or $\zeroSet_{\numRows}\neq [n]$.
W.l.o.g., assume $\zeroSet_{\numRows}\neq  [n]$ and $n\not\in \zeroSet_{\numRows}$.

Note that $n=\indMap(\ell,n_{\ell})$. We will substitute the variable $\beta_{\ell,n_\ell}=0$. For all $i\in[\numRows]$, let $f_i'\coloneqq f_i|_{\beta_{\ell,n_\ell}=0}$. Since $n\not\in\zeroSet_s$, we have $f_s'=f_s\in\cS_{n-1,k}$. For other $i\in[\numRows-1]$, by \ref{p:reducevar}, $f_i'\in\cS_{n-1,k}$ and
\begin{align}\label{eq:fi2c}
  f_i' = \begin{cases}
    \fZt(Z_i,\tfzt_i) & n\not\in Z_i\\
    \fZt(Z_i\setminus\{n\}, \tfzt_i+1) & n\in Z_i
  \end{cases}\ .
\end{align}
In the first case of~\eqref{eq:fi2c} we denote $\zeroSet_i' = \zeroSet_i$ and $\tau'_i=\tau_i$, whereas in the second we denote $\zeroSet_i' = \zeroSet_i\setminus\{n\}$ and $\tau'_i=\tau_i+1$. Additionally, we define $f'_\Omega=\gcrd_{i\in\Omega}f'_i$.

\stepone We will first show that $(f_1',\dots, f_s')$ satisfies \ref{item:equiv2}. That is, we show that $\forall \varnothing\neq \Omega'\subseteq[s]$, $k-\deg f_{\Omega'}'\geq  \sum_{i\in\Omega'} (k-\deg f_i') $.

For $|\Omega'|=1$, it is trivial.

For $2\leq |\Omega'|\leq s-1$,
\begingroup
\allowdisplaybreaks
\begin{align}
  k-\deg f_{\Omega'}' &= k-|\bigcap_{i\in\Omega'}\zeroSet_i'|-\min_{i\in\Omega'} \tfzt_i'\nonumber \\
 & \geq k- |\bigcap_{i\in\Omega'}\zeroSet_i|  - \min_{i\in\Omega'} \tfzt_i - 1 \label{eq:case2cineq1}\\
  &= k-\deg f_{\Omega'} -1\nonumber \\
  &\geq  \sum_{i\in\Omega'} (k-\deg f_i) \label{eq:case2cineq2}\\
  &= \sum_{i\in\Omega'} (k-\deg f_i')\label{eq:case2c-deg-equal}\ .
\end{align}
\endgroup
The inequality \eqref{eq:case2cineq1} is because $|\bigcap_{i\in\Omega}\zeroSet_i'|\leq |\bigcap_{i\in\Omega} \zeroSet_i|$ and $\min_{i\in\Omega} \tfzt_i'\leq \min_{i\in\Omega} \tfzt_i +1$. The inequality \eqref{eq:case2cineq2} is because we assume the inequality \eqref{eq:equiv2} in \ref{item:equiv2} holds strictly for all $2\leq |\Omega|\leq s-1$. The equality \eqref{eq:case2c-deg-equal} holds because $\deg f_i'=\deg f_i,\forall i\in[\numRows]$ by observing \eqref{eq:fi2c}.

For $|\Omega'|=s$, \eqref{eq:equiv2} is not necessarily strict. However, since
\begin{align*}
n\not\in \zeroSet_s &\implies n\not\in \bigcap_{i\in[s]}\zeroSet_i\\
&\implies \left|\bigcap_{i\in[\numRows]}\zeroSet_i'\right|=\left|\bigcap_{i\in[\numRows]} \zeroSet_i\right|\implies f'_{[\numRows]}=f_{[\numRows]},
\end{align*}
we have
\begingroup
\allowdisplaybreaks
\begin{align*}
  k-\deg f'_{[\numRows]} &= k-\deg f_{[\numRows]}\\
  &\geq  \sum_{i\in[\numRows]} (k-\deg f_i)\\
  &=  \sum_{i\in[\numRows]} (k-\deg f_i')\ .
\end{align*}
\endgroup

Hence, \ref{item:equiv2} holds for $(f_1',\dots, f_\numRows')\in\cS_{n-1,k}^{\numRows}$.
By \ref{H:fromPre}, \ref{item:equiv1} holds for $(f_1',\dots, f_\numRows')\in\cS_{n-1,k}^{\numRows}$ with the parameter tuple $(k\geq\numRows\geq 3,n-1)$ where $n\geq 2$. Note that here we used the induction hypothesis by reducing $n$ to $n-1$.

\steptwo Suppose that for some $g_1,\dots, g_s\in\FrobPolysn$, not all zero, with $\deg (g_i\cdot f_i)\leq k-1$, we have $\sum_{i=1}^s g_i\cdot f_i =0$. Let $g_i'=g_i|_{\beta_{\ell,n_\ell}=0}\in R_{n-1}[\SkewVar;\Frobaut]$. Further assume that at least one coefficient of some $g_i$ is not divisible by $\beta_{\ell,n_\ell}$ (otherwise, divide them by $\beta_{\ell,n_\ell}$). Then $g_i'$ are not all zero. We can write
\begin{align*}
  \sum_{i=1}^s g_i'\cdot f_i' =\left. \left(\sum_{i=1}^s g_i\cdot f_i \right)\right|_{\beta_{\ell,n_\ell}=0} = 0|_{\beta_{\ell,n_\ell}=0} = 0\ .
\end{align*}
However, this contradicts \ref{item:equiv1} being true for $(f_1',\dots, f_s')$ with the parameter tuple $(k,s,n-1)$. Therefore, $g_1,\dots, g_s\in\FrobPolysn$ must be all zero to have $\sum_{i=1}^s g_i\cdot f_i =0$.
\end{proof}
\noindent \ref{case_s2n2} For $\numRows=2$ and $n\geq 2$,

\ref{case_s2n2_a} $\forall i\in\{1,2\}$, $\tfzt_i\geq 1$ (i.e., $|Z_i|\leq k-2$).

The proof for this case is the same as for \ref{case_s3n2_a}. We use the induction hypothesis by reducing $k$.

\ref{case_s2n2_b} $\exists$ a unique $i\in\{1,2\}$ such that $\tfzt_i=0$.

The proof for this case is the same as for \ref{case_s3n2_b}. We use the induction hypothesis by reducing $k$. We may need to reduce $s$, too.

\ref{case_s2n2_c} $\forall i\in\{1,2\}$, $\tfzt_i=0$.
\begin{proof}
In this case we have $\Omega=\{1,2\}$ and $\tfzt_1=\tfzt_2=0$.
Similar to \ref{case_s3n2_d}, $\zeroSet_1\neq [n]$ or $\zeroSet_2\neq[n]$. W.l.o.g., assume $\zeroSet_2\neq[n]$ and $n\not \in \zeroSet_2$. Note that $n=\indMap(\ell,n_{\ell})$. We substitute the variable $\beta_{\ell,n_{\ell}}=0$. For $i=1,2$, let $f_i'\coloneqq f_i|_{\beta_{\ell,n_\ell}=0}$ and $f'_{\Omega}\coloneq \gcrd\{f_1', f_2'\}$.
Since $n\not\in\zeroSet_2$, $f_2'=f_2$.
By \ref{p:reducevar}, $f_1'\in\cS_{n-1,k}$ and
\begingroup
\allowdisplaybreaks
\begin{align*}
  f_1' = \begin{cases}
    \fZt(Z_1,0) & n\not\in Z_1\\
    \fZt(Z_1\setminus\{n\}, 1) & n\in Z_1
  \end{cases}\ .
\end{align*}
\endgroup

\stepone We first show that $(f_1',f_2')\in\cS_{n-1,k}^2$ satisfies \ref{item:equiv2}. That is, we show that $\forall \varnothing\neq \Omega'\subseteq\Omega$, $k-\deg f_{\Omega'}'\geq  \sum_{i\in\Omega'} (k-\deg f_i') $.

For $|\Omega'|=1$, it is trivial.

For $\Omega'=\{1,2\}$, since
\begin{align*}
    n\not\in\zeroSet_2\implies n\not\in \zeroSet_1\cap\zeroSet_2 &\implies |\zeroSet_1'\cap\zeroSet_2'|=|\zeroSet_1\cap\zeroSet_2|\\
    &\implies \deg f'_{\Omega'}= \deg f_{\Omega'}\ ,
\end{align*}
we have
\begin{align*}
  k-\deg f'_{\Omega'} &= k-\deg f_{\Omega'}\\
 & \geq  \sum_{i\in\Omega'} (k-\deg f_i)\\
  &=  \sum_{i\in\Omega'} (k-\deg f_i')\ .
\end{align*}
Hence, \ref{item:equiv2} holds for $(f_1',f_2')\in\cS_{n-1,k}^2$.
By \ref{H:fromPre}, \ref{item:equiv1} holds for $(f_1',f_2')\in\cS_{n-1,k}^2$ with parameter tuple $(k\geq \numRows=2, n-1)$ where $n\geq 2$. Here we used the induction hypothesis by reducing $n$ to $n-1$.

\steptwo This step can be shown in the same manner as in \ref{case_s3n2_d}.
\end{proof}
\noindent \ref{case_s2n1} For $\numRows\geq 2$ and $n=1$,

\ref{case_s2n1_a} $\forall i\in[\numRows]$, $\tfzt_i\geq 1$ (i.e., $|Z_i|\leq k-2$).

The proof for this case is the same as for \ref{case_s3n2_a}. We use the induction hypothesis by reducing $k$.

\ref{case_s2n1_b} $\exists$ a unique $i\in\{1,2\}$ such that $\tfzt_i=0$.

The proof for this case is the same as for \ref{case_s3n2_b}. We use the induction hypothesis by reducing $k$. We may need to reduce $s$, too.

\ref{case_s2n1_c} $\exists$ at least two $i\in[\numRows]$, $\tfzt_i=0$.
\begin{proof}
W.l.o.g., assume that $\tfzt_1=\tfzt_2=0$. Since $n=1$, $|\zeroSet_i|\leq n=1, \forall i\in [\numRows]$. By the definition of $f_i\in \cS_{1,k}$ in \eqref{eq:Skewnk}, $\deg f_i \leq n=1$. Assume \ref{item:equiv2} is true for this case, then for $\Omega=\{1,2\}$, we have
\begin{align}
  k-\deg f_{\Omega}\geq k-\deg f_1 + k-\deg f_2\ . \label{eq:s2n1_equiv2}
\end{align}
If $\zeroSet_1=\zeroSet_2=\varnothing$ or $\{1\}$, then $\deg f_{\Omega}=\deg f_1=\deg f_2\leq n=1$ and \eqref{eq:s2n1_equiv2} implies $\deg f_1\geq k$, which contradicts $k\geq s\geq 2$.
Otherwise, w.l.o.g., assume $\zeroSet_1=\varnothing$ and $\zeroSet_2 = \{1\}$, then $\deg f_{\Omega}=0$ and \eqref{eq:s2n1_equiv2} implies $1=\deg f_2\geq k$, which contradicts $k\geq s\geq 2$.
Therefore, if \ref{item:equiv2} is true for $(k\geq s\geq 2,n=1)$, this case cannot happen.
\end{proof}